\documentclass[oneside,12pt]{Classes/CUEDthesisPSnPDF}
\usepackage[latin1]{inputenc}
\usepackage[T1]{fontenc}

\usepackage{epigraph}
\usepackage{ragged2e}

\usepackage[shortlabels]{enumitem}

\usepackage{tabularx}
\usepackage{array}
\usepackage{multirow}

\usepackage{verbatim}

\usepackage{graphicx,float}
\usepackage{graphicx}
\usepackage[left]{lineno}
\usepackage{amsmath,amssymb}

\usepackage{algorithm}
\usepackage{algorithmic}

\usepackage[autostyle]{csquotes}

\usepackage{caption}
\usepackage{subcaption}

\usepackage{pdfpages}

\usepackage{wasysym}

\ifpdf
    \pdfcatalog { /PageMode (/UseOutlines)
                  /OpenAction (fitbh)  }
\fi

\title{La recherche approchée de motifs : théorie et applications}

\ifpdf
  \author{\href{Ibrahim Chegrane ibra.chegrane@gmail.com}{Ibrahim Chegrane} \\ \href{}{Directeur de thèse : Bensaou Nacéra}}
  \collegeordept{\href{http://www.usthb.dz}{Faculté d'électronique et informatique}}
  \university{\href{http://www.usthb.dz}{Université des Sciences et de la Technologie Houari Boumediene}}
  \crest{\includegraphics[width=30mm]{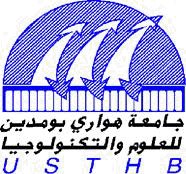}}
\else
    \author{Ibrahim Chegrane \\ Directeur de thèse : Bensaou Nacéra}
    \collegeordept{{Faculté d'électronique et informatique}}
    \university{Université de science et de la technologie Houari Boumediene}
  \crest{\includegraphics[bb = 0 0 292 336, width=30mm]{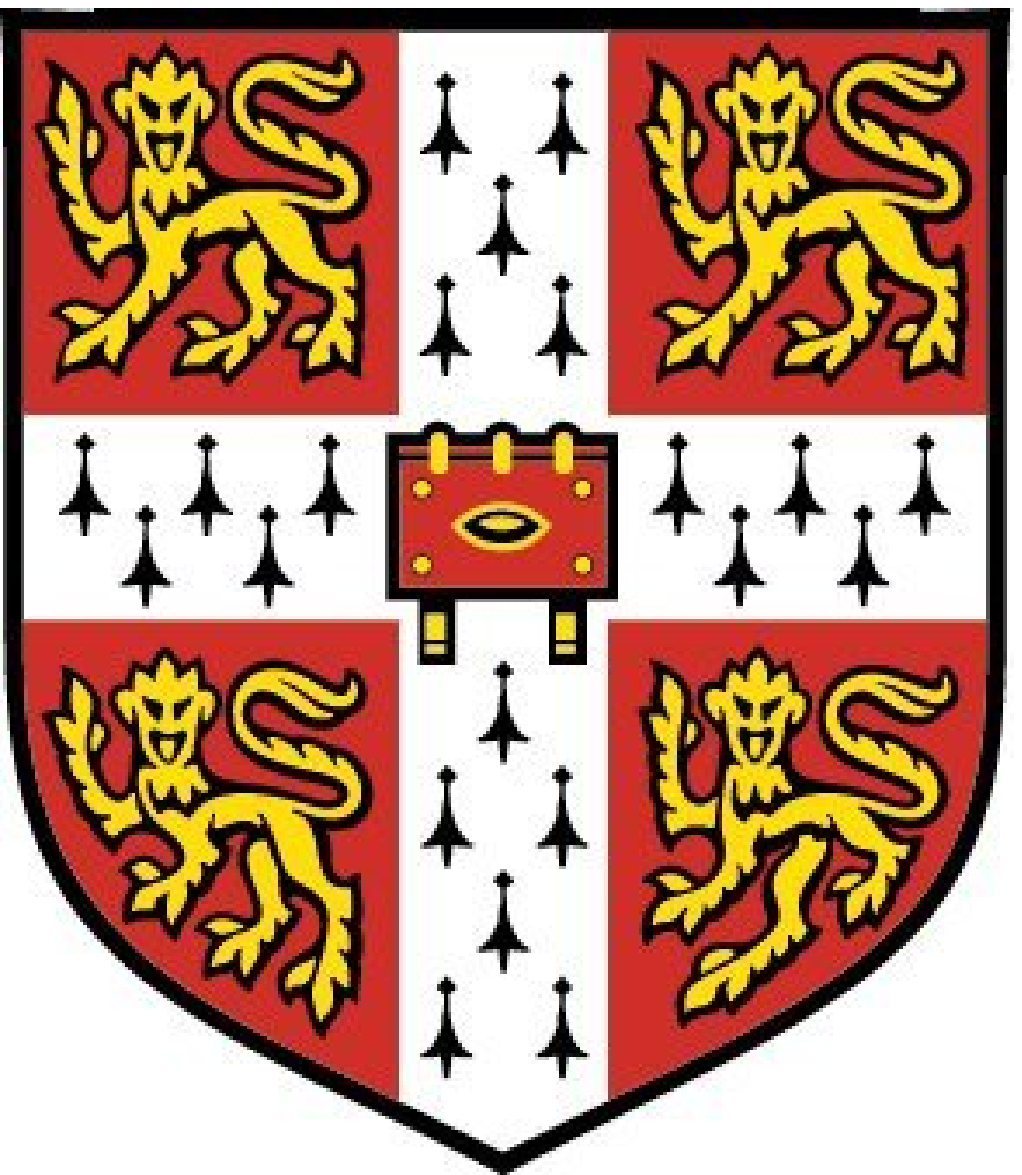}}
\fi
\renewcommand{\submittedtext}{Manuscrit de thèse soumis dans le but de l'obtention du titre de}
\degree{Docteur en informatique}
\degreedate{}

\usepackage{StyleFiles/watermark}

\newtheorem{definition}{Définition}

\newtheorem{theorem}{Théorème}
\newtheorem{proof}{Preuve}

\newenvironment{remarque}{\textbf{Remarque:}}{}

\begin{document}

%\language{english}

% A page with the abstract on including title and author etc may be
% required to be handed in separately. If this is not so, then comment
% the below 3 lines (between '\begin{abstractseparte}' and 
% 'end{abstractseparate}'), normally like a declaration ... needs some more
% work, mind as environment abstracts creates a new page!
% \begin{abstractseparate}
%   \input{Abstract/abstract}
% \end{abstractseparate}

% Using the watermark package which is in StyleFiles/
% and to remove DRAFT COPY ONLY appearing on the top of all pages comment out below line
%\watermark{DRAFT COPY ONLY}



N° d'ordre : 61/2016-C/INF**REPUBLIQUE ALGERIENNE DEMOCRATIQUE ET POPULAIRE**
**Ministère de l'Enseignement Supérieur et de la Recherche Scientifique**
**Université des Sciences et de la Technologie Houari Boumediène**

**Faculté d'électronique et informatique**

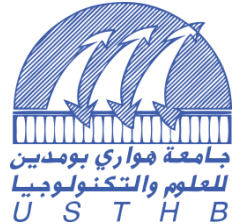

**THESE**

Présentée pour l'obtention du **diplôme** de **DOCTORAT 3$^{eme}$ Cycle**

**En** : INFORMATIQUE

**Spécialité** : Intelligence artificielle

**Par** : **Chegrane Ibrahim**

**Sujet**

# La recherche approchée de motifs : théorie et applications

Soutenue publiquement, le 07 /12/ 2016 , devant le jury composé de :

| | | | | |
|---|---|---|---|---|
| Mme. | Aicha Aissani-Mokhtari | Prof | USTHB/FEI | Président |
| Mme. | Nacéra Bensaou | MCA | USTHB/FEI | Directrice de thèse |
| M. | Thierry Lecroq | Prof | U. Rouen-Normandie | Examinateur |
| M. | Riadh BabaAli | Prof | USTHB/FEI | Examinateur |
| M. | Abdelmajid Boukra | Prof | USTHB/FEI | Examinateur |
| M. | Djamal Belazzougui | M.R | CERIST | Invité |

ibra_chegrane@hotmail.com , ibra.chegrane@gmail.com

<div dir="rtl">

قال **الخوارزمي** في مقدمة كتابه الشهير كتاب الجبر و المقابلة:

**ولم تزل العلماء في الأزمنة الخالية والأمم الماضية يكتبون الكتب، مما يُصنفون من صنوف العلم ووجوه الحكمة**، نظرًا لمن بعدهم واكتسابًا للأجر، بقدر الطاقة ورجاءً أن يلحقهم من أجر ذلك وذخره وذكره ويبقى لهم من لسان الصدق ما صغُر في جنبه كثيرٌ مما كانوا يتكلفونه من المؤونة ويحملونه على أنفسهم من المشقة في كشف أسرار العلم وغامضه.

1. إما رجل سبق إلى ما لم يكن مُستخرجًا قبله، فورَّثه من بعده؛
2. و إما رجل شرح مما أبقى الأولون ما كان مستغلقًا، فأوضح طريقه وسهل مسلكه وقرّب مأخذه ؛
3. وإما رجل وجد في بعض الكتب خللاً فلمّ شعثَه وأقام أوده وأحسن الظن بصاحبه غير زائدٍ عليه ولا مُفتخرٍ بذلك من فعل نفسه.

</div>

**AL–Khawarizmi** a dit dans l'introduction de son fameux livre d'algèbre:

**Les savants du temps passé et des nations d'antan continuaient d'écrire les livres qu'ils composaient dans les diverses sortes de la science et les divers aspects de la sagesse**, en vue de ceux qui leur succéderaient et en prévision d'une récompense, richesse ou renommée, et garder la langue de la vérité devant laquelle s'amenuisent la plupart des ressources engagées et la difficulté qu'ils ont endurée pour étudier les secrets de la science et ce qu'elle renferme d'obscur

1. Ou bien c'est un homme qui est parvenu le premier à découvrir ce qui n'était pas découvert avant lui et l'a légué après lui ;

2. Ou bien un homme qui a expliqué ce que ses prédécesseurs avaient laissé inaccessible, pour en éclairer la méthode, en aplanir la voie et en rapprocher l'accès;

3. Ou bien un homme qui a trouvé une faille dans certains livres, et qui a alors rassemblé ce qui était éparpillé, a redressé sa stature en ayant bonne opinion de son auteur, sans renchérir sur lui et sans s'enorgueillir de ce que lui-même a réalisé.

---

*Al-Khwarizmi, Le commencement de l'algèbre,*
*éd. Roshdi Rashed*



%set the number of sectioning levels that get number and appear in the contents
\setcounter{secnumdepth}{3}
\setcounter{tocdepth}{3}

\frontmatter % book mode only
\pagenumbering{roman}

\ifpdf
    \graphicspath{{Remerciements/}}
\else
    \graphicspath{{Remerciements/}}
\fi

%\chapter*{Remerciements}
%\label{chap:remerciements}
~\\
~\\
{\Huge \textbf{Remerciements}}
~\\
\begin{quote}
Le Prophète  Mohammed (que la paix et le salut soient sur lui) a dit : "\textit{\textbf{Ne remercie pas Dieu, celui qui ne remercie pas les gens}}."~\footnote{Rapporté par Ahmad, Boukhari dans Al-Adab al-Moufrad}
\begin{figure}[H]
\centering
\includegraphics[width=0.5\linewidth]{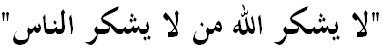}
\label{fig:Merci}
\end{figure}

\end{quote}
% % je pense l'jouter  cela avant le remercment comme citation ou quote

% % -----------------------------------------------------------------
\bigskip
% % Allah

En tout premier lieu, je remercie Dieu (Allah), Alhamdulillah, pour la force qu'il m'a donnée d'arriver au bout de cette thèse, de m'avoir entouré de personnes formidables qui m'ont aidé d'une manière ou d'une autre, tout au long de ce travail. Merci à Dieu (Allah) pour tout ce qu'Il m'a accordé.
\begin{figure}[H]
\centering
\includegraphics[width=0.5\linewidth]{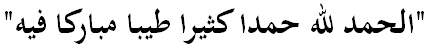}
\label{fig:el_hamd}
\end{figure}

\bigskip
% % >> ma derctrice de thèse

Je remercie très chaleureusement ma directrice de thèse, madame Nacéra Bensaou, d'abord d'avoir accepté de m'encadrer.
Je la remercie pour ces nombreux conseils, pour son soutien et son encouragement durant toute la durée de la thèse. Je la remercie d'avoir partagé son propre bureau (335) avec moi pour que je puisse travailler. Je la remercie de m'avoir mis en contact avec des chercheurs spécialistes du domaine, et de m'avoir encouragé fortement à travailler et à publier avec eux. Je la remercie de m'avoir fait confiance, et pour la liberté qu'elle m'a accordée. Mille Mille Merci Madame Bensaou \smiley .

%------------------------------------------------------------
\bigskip

% Mathieu Raffinot 

Je remercie profondément monsieur Mathieu Raffinot de m'avoir accueilli au laboratoire LIAFA (université de Paris 7) à plusieurs reprises (la première fois pour un mois, la deuxième fois pour trois mois, et la troisième fois pour une simple visite), de m'avoir pris en charge en stage de deux mois, de m'avoir hébergé à Paris dans le domicile de ses parents pendant 8 jours. Je remercie aussi beaucoup les parents de Mathieu Raffinot pour toute leur gentillesse durant cette période. Merci pour tous tes conseils et tes directives, (Merci aussi pour le dîner du sushi \smiley). Merci Mathieu, tu as été mon deuxième encadreur de thèse.

\medskip
% Djamal belazzougui

Djamal Belazzougui, mon co-encadreur, mon exemple, mon ami. Je ne pourrais jamais te remercier assez. Dès les premiers jours de notre rencontre, tu n'a pas cessé de me guider, de m'aider, de me donner des directives et des conseils.
Merci pour tout le temps et séances de travail que tu m'a accordé au LIAFA durant mon premier stage, mais aussi à l'USTHB, à chaque fois que tu rentres à Alger et que tu viens faire une séance de travail avec moi, et jusqu'à maintenant où tu me reçois et m'accordes du temps au CERIST pour m'aider malgré ton importante charge de travail. Je suis très chanceux de ma rencontre avec toi.

\medskip
% % madame ighilaza....
Je remercie également très chaleureusement madame C.Ighilaza mon enseignante d'algorithmique, mon encadreur du mémoire de Licence et de Master, pour l'aide et les conseils qu'elle m'a donnés, pour m'avoir aidé dans mon travail, et pour sa participation à la correction de la thèse.

%------------------------------------------------------------
\bigskip
% % le cométie d'expertise

Je remercie le professeur Slimane Larabi et le Professeur Aicha Aissani-Mokhtari d'avoir accepté d'expertiser mes travaux.

%------------------------------------------------------------
\bigskip
% %  les mombre de jury.

Je remercie le professeur Aicha Aissani-Mokhtari d'avoir accepté de présider le jury de cette thèse. 

Je remercie également le professeur Thierry Lecroq d'avoir accepté de participer à ce jury en tant qu'examinateur et de venir de loin pour participer à ce jury.

Je remercie le professeur Riadh BabaAli d'avoir accepté de participer au jury de cette thèse.

Je remercie le professeur Abdelmajid Boukra d'avoir accepté de participer au jury de cette thèse.

\bigskip
Mes remerciements vont aussi à Athmane Seghier, Meriem Beloucif et Aïcha Botorh pour leurs collaborations.

%----------------------------------------------------
\bigskip
% % Mes amis

Mes remerciements vont aussi à tous mes amis, nombreux pour être cités individuellement, et qui se reconnaitront certainement. Je pense à mes amis de l'université USTHB, mes amis de la cité universitaire RUB 1 et RUB 3, mes amis les doctorant(e)s de la salle Poste Graduation du département informatique, mes amis du bureau 335 de madame Bensaou, mes amis du club informatique Micro-club et mes amis du sport.

%---------------------------------------------------------
\bigskip
% % hamdane
Merci à mon ami et mon prof de sport Nourdine Zirara (Hamdane), pour son soutien moral et matériel, et pour son encouragement.

%---------------------------------------------------------
\bigskip
% % ma fammile

Je tiens à exprimer toute ma gratitude à ma famille pour leur soutien moral et matériel et sans qui je ne serais jamais allé aussi loin dans mes études.

%---------------------------------------------------------
\bigskip
% % mon épouse

Mon épouse Nacera, je te remercie pour ta présence, ton soutien, tes prières, tes pensées, tes encouragements et ta patience qui m'ont été indispensables.

%---------------------------------------------------------
\bigskip
Que tous ceux qui m'ont aidé et qui ont contribué de près ou de loin à ma formation trouvent ici le témoignage de ma sincère gratitude.

%---------------------------------------------------------
\bigskip
% % mon pays

Mes remerciements vont également à mon pays l'Algérie qui, malgré ses insuffisances et ses problèmes, m'a permis d'arriver à ce niveau et de terminer mes études gratuitement.

%\begin{abstractslong}    %uncommenting this line, gives a different abstract heading
% %\begin{Resume}        %this creates the heading for the abstract page

\chapter{Résumé}
\label{chap:resume}

La recherche approchée de motifs est un problème fondamental et récurrent qui se pose dans la plupart des domaines informatiques. Ce problème peut être défini de la manière suivante:

\begin{quote}
\textit{Soit $D=\{x_1,x_2,\ldots x_d\}$ un ensemble de $d$ mots définis sur un alphabet $\Sigma$, soit $q$ une requête définie aussi sur $\Sigma$, et soit $k$ un entier positive.\\
On veut construire une structure de données pour $D$ capable de répondre à la requête suivante :  trouver tous les mots de $D$ distants d'au plus $k$ erreurs de $q$.}
\end{quote}

Dans cette thèse nous étudions les méthodes de la recherche approchée dans les dictionnaires, les textes et les index, pour proposer des méthodes pratiques qui résolvent de façon efficace ce problème. Nous explorons ce problème dans trois directions complémentaires :

1) La recherche approchée dans un dictionnaire. Nous proposons deux solutions à ce problème, la première utilise les tableaux de hachage pour $k \geq 2$, la deuxième, utilise le Trie et le Trie inversé  et se restreint à $k=1$. Les deux solutions sont applicables, sans perte de performances, à la recherche approchée dans un texte.

2) La recherche approchée pour le problème de \textit{l'auto-complétion} consiste à trouver tous les suffixes d'un préfixe contenant éventuellement des erreurs. Nous apportons une nouvelle solution meilleure, en pratique, que toutes les solutions précédentes.

3) Le problème de l'alignement de séquences biologiques peut être interprété comme un problème de recherche approchée. Nous proposons à ce problème une solution pour l'alignement par pairs et pour l'alignement multiple.

\medskip
Tous les résultats obtenus montrent que nos algorithmes  donnent les meilleurs performances sur des ensembles de données pratiques. Toutes nos méthodes sont proposées comme des bibliothèques et sont publiées en ligne.

% %\end{Résume}
%\end{abstractlongs}

% ----------------------------------------------------------------------

%%% Local Variables: 
%%% mode: latex
%%% TeX-master: "../thesis"
%%% End: 

\includepdf[pages={1}]{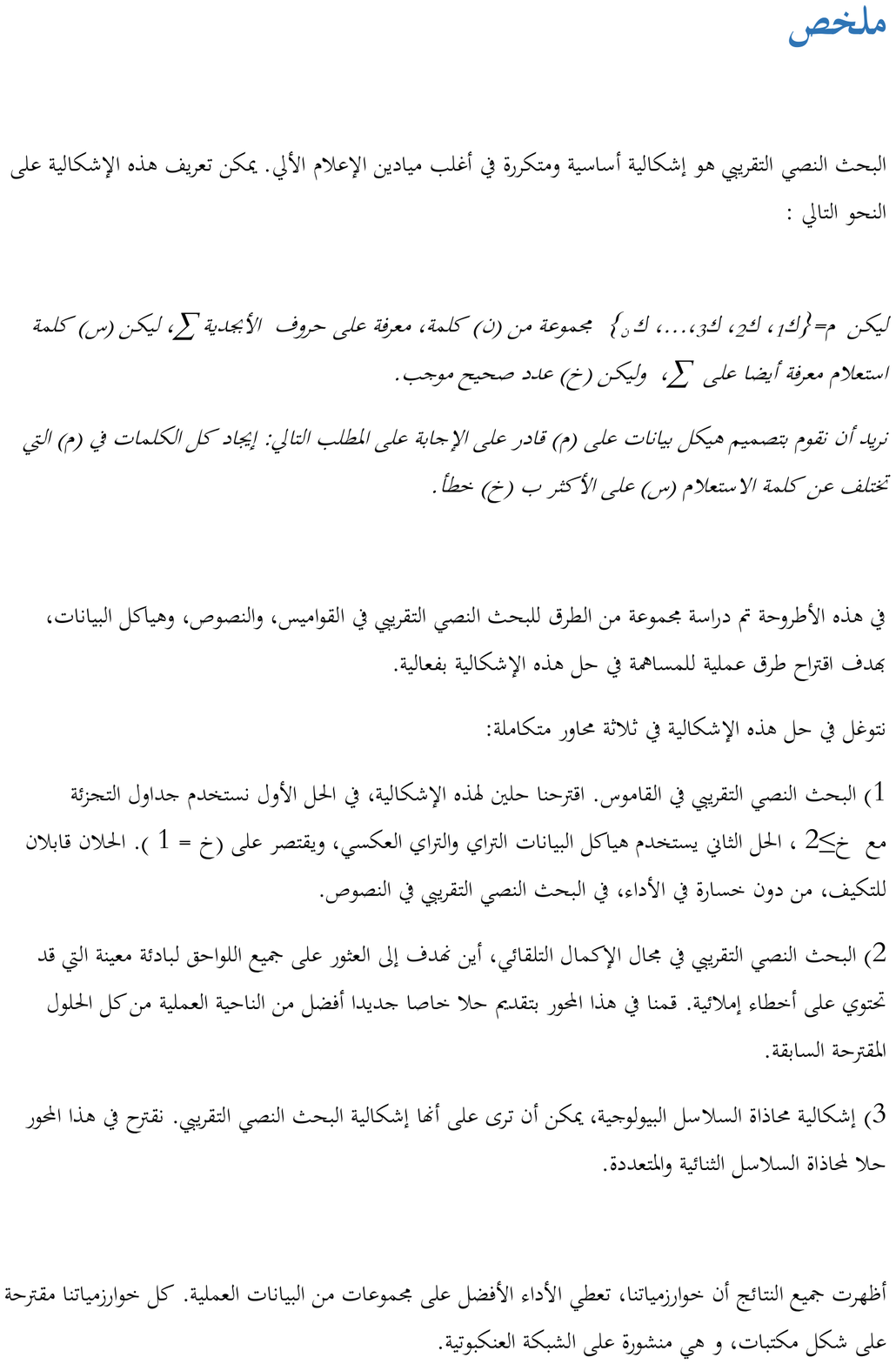}

%\begin{abstractslong}    %uncommenting this line, gives a different abstract heading
% %\begin{Resume}        %this creates the heading for the abstract page

\chapter*{Abstract}
\label{chap:resume_en}

The approximate string matching is a fundamental and recurrent problem that arises in most computer science fields. This problem can be defined as follows:

\begin{quote}
\textit{Let $D=\{x_1,x_2,\ldots x_d\}$ be a set of $d$ words defined on an alphabet $\Sigma$, let $q$ be a query defined also on $\Sigma$, and let $k$ be a positive integer.\\
We want to build a data structure on $D$ capable of answering the following query: find all words in $D$ that are at most different from the query word $q$ with $k$ errors.}
\end{quote}

In this thesis, we study the approximate string matching methods in dictionaries, texts, and indexes, to propose practical methods that solve this problem efficiently. We explore this problem in three complementary directions:

1) The approximate string matching in the dictionary. We propose two solutions to this problem, the first one uses hash tables for $k \geq 2$, the second uses the Trie and reverse Trie, and it is restricted to (k = 1). The two solutions are adaptable, without loss of performance, to the approximate string matching in a text.

2) The approximate string matching for \textit{autocompletion}, which is, find all suffixes of a given prefix that may contain errors. We give a new solution better in practice than all the previous proposed solutions.

3) The  problem of the alignment of biological sequences can be interpreted as an approximate string matching problem. We propose a solution for peers and multiple sequences alignment.

\medskip
All the results obtained showed that our algorithms, give the best performance on sets of practical data (benchmark from the real world). All our methods are proposed as libraries, and they are published online.

% %\end{Résume}
%\end{abstractlongs}

% ----------------------------------------------------------------------

%%% Local Variables: 
%%% mode: latex
%%% TeX-master: "../thesis"
%%% End: 

\tableofcontents
\listoffigures

\printnomenclature  %% Print the nomenclature
\addcontentsline{toc}{chapter}{Nomenclature}

\mainmatter % book mode only

\ifpdf
    \graphicspath{{Introduction/}}
\else
    \graphicspath{{Introduction/}}
\fi

\chapter{Introduction générale}
\label{chap:Introduction_generale}

\begin{figure}[h]
\includegraphics[width=0.7\linewidth]{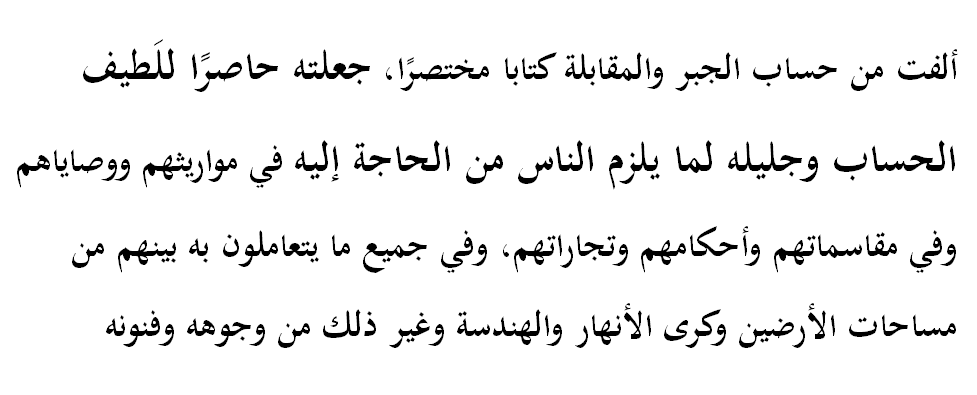}
\end{figure}

\setlength{\epigraphwidth}{0.8\textwidth}
\epigraph{
\justify 
\textit{J'ai composé dans le calcul de l'algèbre et d'al-muqabala un livre concis ; \textbf{j'ai voulu qu'il enferme ce qui est subtil dans le calcul et ce qui en lui est le plus noble, ce dont les gens ont nécessairement besoin} dans leur héritages, leurs legs, leurs partages, leurs arbitrages, leurs commerces, et dans tout ce qu'ils traitent les uns avec les autres lorsqu'il s'agit de l'arpentage des terres, de la percée des canaux, de la mensuration, et d'autres choses relevant du calcul et de ses sortes.}}{\textit{Al-Khwarizmi, Le commencement de l'algèbre,\\ éd. Roshdi Rashed}}

% motivation and context 

\section{Le contexte}
Le Dictionnaire est une des structures de données les plus fondamentales et les plus utilisées dans les problèmes informatiques.

\medskip
Cette structure organise les données comme un ensemble sur lequel il s'agit d'effectuer une \textit{recherche efficace}. Soit $D$ cet ensemble et $n$ sa cardinalité (ou taille): \textit{la recherche} dans $D$ consiste à savoir si un élément $x$ donné appartient ou non à $D$. \textit{L'efficacité} mesure les ressources (temps et espace) nécessaires pour obtenir la réponse. 

\medskip

Le contexte d'une recherche où les données dans $D$ ainsi que $x$ sont considérés comme ne contenant aucune erreur est celui d'une \textit{recherche exacte}. Lorsqu'on admet l'existence d'erreurs dans les données ($D$ et/ou $x$) c'est le contexte d'une \textit{recherche  approchée}.

\medskip

Les performances de la recherche dépendent, en général, du contexte (recherche exacte ou approchée), des tailles de $D$ et $x$ et de l'organisation de $D$.

Lorsque $D$ est de grande taille et la recherche est récurrente, une "\textit{bonne organisation}" de $D$ devient cruciale. L'idée serait d'organiser $D$ de manière à ce que la recherche de $x$ dépende le moins possible de la taille de $D$.

\medskip
Dans sa forme de base, ce problème considère des données textuelles et s'apparente au très ancien domaine de l'algorithmique du texte: $D$ est un ensemble de mots et $x$ un mot. Les premiers travaux ont porté sur la recherche exacte pour laquelle de nombreuses solutions satisfaisantes ont été obtenues.

\medskip

Avec l'explosion de la quantité d'informations textuelles due la croissance des données de l'Internet (bases de données, articles, livres, les titres des images, vidéos, musicales), la recherche exacte est devenue, dans la plupart des cas, moins pertinente: le mot requête et/ou les mots du dictionnaire ou du texte peuvent contenir  des erreurs typographiques. L'erreur peut-être due à plusieurs raisons (selon le domaine): par exemple, une erreur de frappe lors d'une saisie rapide ou due à la méconnaissance par l'utilisateur de l'orthographe correcte (nom de personne, nom d'un produit ... etc.) ou due à la multiplicité de formes d'écriture de ce mot. Les textes qui proviennent d'un outil de reconnaissance de caractère sont susceptibles de contenir des erreurs de mauvaise reconnaissance. Une altération des données lors de leurs transmissions est donc possible.

\medskip

La recherche approchée résout le problème en permettant l'existence d'erreurs entre le motif et ses occurrences. 

\paragraph{}

La recherche approchée dans un dictionnaire ou dans un texte est définie comme suit : 

\begin{quote}
Étant donné un dictionnaire/texte $D=\{{w_1,w_2,...,w_{d}}\}$ de $d$ mots et $n$ caractères, défini sur un alphabet spécifique $\Sigma$,  étant donné un mot requête $q$ de  longueur $m$ et un entier  $k$ représentant un seuil maximal d'erreurs,  étant donnée une fonction  $\mathrm{dist}$ qui mesure la distance entre deux mots, alors la recherche approchée de $q$ dans $D$ consiste à trouver toutes les occurrences dans $D$ qui diffèrent d'au plus $k$  erreurs de $q$. 
\end{quote}

\medskip
La fonction la plus couramment utilisée pour déterminer la distance entre deux mots est appelée la "\textit{distance d'édition}"~\cite{Le66}.

\medskip

Il existe de nombreuses solutions à ce problème dans la littérature \cite{Na01,Bo11}.

\paragraph{}

La recherche approchée de motifs intervient dans plusieurs domaines, principalement dans les moteurs de recherche~\cite{cucerzan2004spelling}, la correction d'orthographe dans les éditeurs de texte~\cite{pollock1984automatic}, dans les systèmes OCR~\cite{reynaert2008non}, dans le craquage de mots de passe~\cite{manber1994algorithm}, dans le nettoyage des données en ligne, dans les bases de données~\cite{chaudhuri2003robust}...etc.

L'une de ces applications  importantes est dans le domaine de la bio-informatique où elle permet de rechercher des séquences biologiques, de les assembler, ou encore de les aligner~\cite{bioinformatics_waterman1995,Gusfield1997,bioinformatics_lesk2013}.

\medskip

D'une manière générale, la recherche approchée de motifs est utilisée dans tous les domaines qui traitent un ensemble de données constitué et défini sur un alphabet $\Sigma$ donné.

\section{La motivation de notre travail}

Dans son article survey sur le problème de la recherche approchée \cite{Na01}, Gonzalo Navarro  constate le nombre important d'\textit{améliorations} apportées à plusieurs méthodes de résolution de ce problème sur le plan de leur complexité théorique sans considérer leur aspect pratique. Ce qui a aboutit à une dichotomie indésirable, à savoir l'existence d'une part, d'un ensemble d'algorithmes performants en théorie et très lents en pratique et d'autre part, des algorithmes rapides en pratique mais dont la complexité théorique n'est pas satisfaisante. 

\medskip

Le fait que les algorithmes dont la complexité théorique performante ne coïncident pas avec les algorithmes efficaces en pratique pose un important problème, en particulier à la communauté des développeurs de logiciels ou des chercheurs en bio-informatique, mais aussi à tous les utilisateurs non spécialistes d'algorithmique de texte, lorsqu'il s'agit de choisir et d'utiliser un bon algorithme de recherche de motifs.  Ce choix a toutes les chances de se porter sur l'algorithme le plus simple à implémenter qui est en général l'algorithme naïf qui a les performances les plus mauvaises, ce qui aura une incidence sur les performances de l'outil qui utilise un tel algorithme   \cite{navarro2002flexible}.

Le travail dans cette thèse est fondamentalement motivé par quelques-unes des questions  posées  à la fin de ce constat, à savoir:

\begin{enumerate}
\item est-il possible de trouver un algorithme en $O(kn)$ pour la complexité du cas pire et efficace en pratique?

\item la complexité minimale du cas moyen est connue comme étant en $O(n+(k+\log \sigma m)/m)$ et il existe un algorithme qui le réalise. Est-il possible de trouver une solution efficace en pratique ayant cette performance théorique?

\item est-il possible, d'améliorer en pratique, la plus performante actuellement des méthodes de filtrage  \cite{Na01}.
\end{enumerate}

Un débat récent \cite{vardi2014boolean}, entre théoriciens de la complexité et  chercheurs impliqués dans le développement de logiciels, pose la question de l'écart qui existe, entre la solution théorique exponentielle du pire cas du problème SAT et ses nombreuses solutions pratiques efficaces. 

\medskip

Le problème est similaire pour les algorithmes non exponentiels. En effet, bien que la complexité des algorithmes de texte soit l'un des problèmes les plus étudiés en informatique théorique, elle reste insuffisante pour juger si un algorithme est efficace en pratique ou non, s'il donne de bons résultats ou non, s'il est optimal ou non en pratique en temps d'exécution et en occupation mémoire.

\medskip

Sans pouvoir encore répondre à la question générale de savoir que manque t-il à la théorie de la complexité pour capturer l'efficacité des solutions pratiques que pose ce débat, cette thèse s'intéresse à la recherche de  solutions efficaces en pratique, sans détériorer les performances des complexités théoriques  du problème de la recherche approchée de motifs dans différents contextes de son application.

\medskip
Les paramètres quantitatifs importants pour juger la qualité des solutions de ce problème sont : l'utilisation de l'espace mémoire de la solution, le temps nécessaire pour construire la structure de données du dictionnaire, et le temps nécessaire pour répondre à une requête.

\medskip

En plus de ces trois paramètres, d'autres considérations importantes doivent être prises en compte, comme le fait que la structure de données de la solution soit dynamique ou non (elle accepte ou non l'addition rapide d'autres chaînes de caractères à l'index du dictionnaire), et que la solution soit \textit{simple} ou non (utilise ou non des structures de données complexes).

\bigskip
Cette thèse étudie le problème de la recherche approchée dans les dictionnaires, les textes et les index et propose des algorithmes qui ne se limitent pas à l'efficacité selon une complexité théorique mais qui  résolvent aussi, en pratique de façon efficace, ce problème.

Dans notre travail, nous explorons la recherche approchée dans trois directions. Nous nous intéressons à la résolution du problème théorique d'une façon générale et à ses applications dans les moteurs de recherche, dans l'alignement des séquences biologiques, et dans l'auto-complétion dans le web du côté serveur et du côté client.

Le but de ce travail est de fournir des algorithmes validés aussi en pratique et très performants et très compétitifs vis-à-vis de toute autre solution qui existe aujourd'hui et de fournir comme résultat des bibliothèques utilisables par les autres développeurs/chercheurs dans leurs travaux.

\section{Notre contribution}
Dans cette thèse, nous apportons de nouvelles solutions au problème de la recherche approchée, toutes fondées sur l'idée d'associer une implémentation pratique au moins aussi efficace que la solution théorique optimale: la première solution contribue à la résolution du problème pour un seuil d'erreurs $k\geq 2$. La seconde, explore de nouveau la possibilité d'utiliser de manière efficace l'utilisation de la structure de Trie associée à une recherche bidirectionnelle et d'observer en pratique pour quel seuil d'erreurs cette approche est satisfaisante. Les deux autres contributions de cette thèse sont deux applications dans des domaines spécifiques et importants de la recherche approchée: la recherche dans le Web et la recherche dans les structures du génome.

\paragraph{La recherche approchée pour $k \geq 2$ en utilisant les tableaux de hachage.}

Nous proposons un algorithme qui utilise une structure de donnée basée sur le hachage avec sondage linéaire et des signatures de hachage associées pour optimiser le temps de la recherche. Pour le dictionnaire exact, nous utilisons une structure de donnée basée sur plusieurs tableaux de hachage organisée par longueur des mots et exploitons l'idée d'utiliser un dictionnaire des listes de substitution \cite{Be09}.

Afin d'optimiser l'espace mémoire, nous compressons ces deux structures de données en utilisant une autre structure de données (le vecteur des bits).
La version non compressée est \textit{progressive}: nous pouvons ajouter de nouveaux mots jusqu'à un certain facteur.

\medskip
La structure de données de cette solution occupe un espace $O(n)$, où $n$ est la taille totale du dictionnaire et en ($O(n\log\sigma)$) bits. Si le compactage de la table de hachage est utilisé, le dictionnaire final occupe au plus un espace mémoire de $n(2\log\sigma+O(1))$ bits. La recherche d'un mot $q$ de $m$ lettres et qui a un nombre d'occurrences $occrs$ s'effectue en un temps proche de $(m+occrs)$ et plus exactement en ($O(m\left\lceil\frac{m\log\sigma}{w}\right\rceil)$ où $w$ représente la taille d'un mot mémoire (qui vaut 32 ou 64 bits).  

\medskip

Les expérimentations montrent que cette solution est en concurrence avec les solutions précédemment proposées et ce travail  est publié dans \cite{ibra.Chegrane.simple}, et son implémentation est disponible sur :\\
\url{https://code.google.com/p/compact-approximate-string-dictionary/}~\footnote{Visité le: 02-07-2016.}

\paragraph{Une recherche approchée en utilisant le Trie et le Trie inversé.}

La structure de Trie est particulièrement intéressante pour représenter un dictionnaire dans le contexte d'une recherche exacte. Il permet de rechercher les mots par préfixe commun. Ainsi le Trie permet de savoir si un mot $q$ de $m$ lettres appartient ou non au dictionnaire en temps proportionnel à la longueur du mot recherché, c'est-à-dire en $O(m)$. 

\medskip

Pour la recherche approchée, l'existence de $k$ erreurs possibles doit être prise en compte tout en préservant une complexité indépendante de la taille du dictionnaire.

\medskip

Dans ce travail, nous utilisons une structure de données bidirectionnelle (Trie, et Trie inversé), pour proposer une méthode qui résout la recherche approchée pour $k=1$.  Si on admet l'existence d'une seule erreur, le mot $q$ peut s'écrire comme étant égal à la concaténation d'un préfixe exact $P_1$, d'une lettre erronée $c$ et d'un suffixe exact $P_2$,  c'est-à-dire que $q= P_1cP_2$. L'idée consiste à rechercher le préfixe $P_1$ dans le Trie et le suffixe $P_2$ dans le Trie inversé en parallèle. À la rencontre de l'erreur $c$ il faut s'assurer que $P_1$ et $P_2$ appartiennent au même mot en réduisant au maximum le nombre de vérifications.

\medskip
Notre méthode réduit, lors de la recherche dans le Trie, le nombre de branches sortantes testées dans chaque n\oe{}ud. Elle ne choisit que les chemins qui peuvent mener à des solutions. La complexité du pire cas de cette recherche s'effectue en $O(m^2 \times \sigma)$. La complexité moyenne est en $O(m^2)$ si $m \geq 2 \frac{log d}{log \sigma}$, et elle est en $O(m + m.occrs)$ si $m \geq 2 ( \frac{log d}{log \sigma} + \frac{log m}{log \sigma} )$ où $occrs$ représente le nombre d'occurrences de la solution.

\medskip
Les résultats expérimentaux montrent que cette méthode surpasse, en temps d'exécution, toutes les autres implémentations testées à ce jour.
La performance de cette solution est en proportion constante par rapport à la recherche exacte, indépendante de la taille du dictionnaire.

\medskip
Le résultat de ce travail est disponible sur : \url{https://github.com/chegrane/TrieRTrie}~\footnote{Visité le: 02-07-2016.}. 
%La publication de l'article est en progression.

\paragraph{L'auto-complétion approchée.}

L'auto-complétion de requêtes sur le Web est un mode de travail qui permet d'assister un utilisateur lors de la saisie de sa requête. Elle consiste à lui proposer une liste de mots, les plus pertinents, pour compléter sa requête. Elle facilite et accélère la saisie en offrant une liste de suggestions pour compléter les caractères introduits/tapés par l'utilisateur dans le champ de texte.

\medskip
L'auto-complétion approchée est un problème de recherche approchée où il s'agit de trouver tous les suffixes d'un préfixe contenant des erreurs. 

\medskip

Dans ce travail, nous présentons une méthode basée sur la structure de données Trie pour effectuer une auto-complétion approchée efficace admettant une seule erreur de type distance d'édition, dans une architecture client/serveur. 

Nous proposons également une méthode qui réduit le nombre de transitions sortantes testées dans chaque n\oe{}ud \-- en particulier dans le premier niveau \-- du Trie.
La complexité temporelle moyenne de cette recherche approchée est $O(m^2)$.

\medskip

Ce travail est publié dans~\cite{chegrane2015jquery}, et son implémentation, une bibliothèque (nommée $\mathtt{appacolib}$), est disponible sur : \url{https://github.com/AppacoLib/api.appacoLib}~\footnote{Visité le: 02-07-2016.}

% %\subsubsection{Alignement multiple des séquences d'ADN:}
\paragraph{Alignement multiple des séquences d'ADN :}

L'une des applications importantes des algorithmes de texte et de la recherche approchée est l'alignement des séquences biologiques.

En bio-informatique, l'alignement des séquences biologiques permet de rechercher dans les séquences (ADN, ARN ou protéines) des régions similaires. 

Trouver le meilleur alignement entre la paire de séquences $x$ et $y$ revient à trouver le plus petit nombre d'erreurs $k$ entre $x$ et $y$, autrement dit la meilleure façon de mettre $x$ face à $y$ avec un nombre minimum de différences.

Dans ce travail, nous présentons un nouvel algorithme d'alignement par paire de séquences biologiques et un algorithme d'alignement multiple basé sur l'algorithme connu DIALIGN.

\medskip
Les résultats des différents tests montrent que notre approche est très efficace dans les deux cas et donne un très bon temps d'exécution comparé à (DIALIGN 2.2).

Ces résultats sont publiés dans deux articles~\cite{DIAWAY,DIAWAY_b}, et les codes sources et les logiciels sont disponibles sur : \url{https://github.com/chegrane/diaway}~\footnote{Visité le: 02-07-2016.} pour la première version, et sur \url{https://github.com/chegrane/DiaWay_2.0}~\footnote{Visité le: 02-07-2016.} pour la deuxième version.

\bigskip

Toutes nos implémentations sont distribuées sous la licence publique générale limitée GNU (GNU LGPL).

Pour évaluer nos méthodes en pratique, nous les avons comparées avec les algorithmes existants qui donnent les meilleurs résultats sur des ensembles de données pratiques.
Les expérimentations sont réalisées sur des dictionnaires des langues (Anglais, WikiTitle, \dots). Ceci pourrait éventuellement expliquer certaines améliorations pratiques obtenues (par rapport à la complexité théorique du pire cas) étant donné que pour les langues les combinaisons du pire cas n'existent pas ou sont très rares.

\section{L'organisation de la thèse}
Le reste de cette thèse est organisé comme suit :

Dans le chapitre qui suit cette introduction, nous donnons quelques préliminaires et définitions nécessaires à la compréhension du domaine et des chapitres qui suivent. 

Ensuite, nous donnons un état de l'art orienté où nous expliquons les concepts généraux des méthodes de la recherche approchée et nous détaillons quelques méthodes qui sont liées à notre travail.

Dans les quatre chapitres suivants nous présentons respectivement nos quatre travaux : la recherche approchée en utilisant les tableaux de hachage, puis la recherche approchée en utilisant le Trie et le Trie inversé, puis l'auto-complétion approchée et enfin l'alignement multiples des séquences biologiques.

Le dernier chapitre conclut cette thèse et donne quelques perspectives de recherche.

\chapter{Préliminaires}
\label{chap:Preliminaires_definitions}

Ce chapitre introduit les définitions de base de l'algorithmique de texte et de la théorie des langages. Il est basé sur les références fondamentales, parmi lesquelles nous citons:  \cite{carton2008langages,cormen2009introduction,Knuth1998_book,Gusfield1997}.
Il rappelle aussi quelques structures de données fondamentales utilisées comme dictionnaire ou structure auxiliaire pour représenter un dictionnaire. Nous donnons les algorithmes de ces structures dans le chapitre suivant.

\section{Alphabet, mot}

\begin{itemize}

\item[\textbf{Alphabet}] Un \textit{alphabet} $\Sigma$ est un ensemble fini non vide d'éléments appelés symboles ou  lettres. La cardinalité de cet ensemble est notée $|\Sigma|=\sigma$.

\noindent Exemples : $\Sigma_1=\{a,b,...,z\}$, 
$\Sigma_2=\{0,1\}$, $\Sigma_3=\{A,C,T,G\}$.

\item[\textbf{Mot}] Un \textit{mot} $v$ défini sur un alphabet $\Sigma$ est une concaténation d'éléments de $\Sigma$, $v=e_1 e_2...e_m$. La \textit{longueur} d'un mot est le nombre de lettres qui le composent, et on la note $|v|=m$. Exemple: $\Sigma_2=\{0,1\}$, $v=01101001$ est défini sur $\Sigma_2$ et $|v|=8$.

Le \textit{mot vide} (sa longueur vaut zéro) est noté par $\epsilon$.

L'ensemble de tous les mots définis sur l'alphabet $\Sigma$ est noté par $\Sigma^*$. $\Sigma^+ = \Sigma^* - \epsilon$.

La concaténation de deux mots $u=aax$ , et $v=bb$ défini sur l'alphabet $\Sigma=\{a,b,x\}$ est le mot noté $uv$ obtenu en mettant bout à bout $u$ et $v$, $uv=aaxbb$.

\end{itemize}

\section{Préfixe, suffixe, facteur, miroir}

\begin{itemize}

\item[\textbf{Préfixe}] Un mot $u$ est un préfixe d'un mot $w$ s'il existe un mot $v$ tel que $w=uv$.  
Soit $w$ un mot dans $\Sigma^*$, $w= a_1\cdots a_n$,  alors tout mot $u \in \{\epsilon, a_1, a_1 a_2,..., a_1\cdots a_n\}$ est préfixe de $w$. 
On note $ Pref(w)$ l'ensemble $\{\epsilon, a_1, a_1 a_2,..., a_1\cdots a_n\}$ de tous les préfixes de $w$. 
L'ensemble $ Pref(w) - \{w\}$ est dit l'ensemble des \textit{préfixes propres} de $w$.

%\item[\textbf{Préfixe complet}] Dans notre travail, lorsque nous utilisons la phrase ''un préfixe complet'', cela signifier que le premier mot $u$ est inclus (un préfixe) dans le 2ème mot $w$.

\item[\textbf{Suffixe}] Un mot $u$ est un suffixe du mot $w$ s'il existe un mot $v$ tel que $w=vu$.  Soit $w= a_1\cdots a_n$ un mot dans $\Sigma^*$, alors tout mot $u \in  \{\epsilon, a_n, a_{n-1} a_{n},..., a_1\cdots a_n\}$ est suffixe de $w$. 
On note $ \mathrm{\textit{Suff(w)}}$ l'ensemble de tous les suffixes de $w$.  L'ensemble $ \mathrm{\textit{Suff(w)}} - \{w\}$ est dit l'ensemble des \textit{suffixes propres} de $w$.

\item[\textbf{Facteur}] Soient $w, u, v, v'$ des mots dans $\Sigma^*$ tels que $w=vuv'$. Alors $u$ est dit  facteur de $w$. Remarquons que si $v=\epsilon$ alors $u$ est préfixe de $w$ et si $v'=\epsilon$ alors $u$ est suffixe de $w$.

L'ensemble des facteurs de $T$ que l'on note $Fact(T)$   est défini par l'expression suivante:

$$Fact(T) = Pref( \mathrm{\textit{Suff(T)}})$$

\item[\textbf{Mot miroir}] Le mot miroir du mot $w= a_1\cdots a_n$ est le mot $\bar{w}= a_n\cdots a_1$.

\end{itemize}

\section{Fonctions de distance}

Soit $E$ un ensemble donné. Une métrique $dist$ sur $E$ est une application:

$$dist:  E \times E\Longrightarrow \mathbb{R^+} \quad\text{ telle que } \forall x, y, z$$
\begin{center}
$\begin{array}{lll}
dist(x,y) &= dist(y,x) \\
dist(x,y) &= 0 \Longleftrightarrow x=y \\
dist(x,y) &\leq dist(x,z)+dist(z,y)\\
\end{array}
$
\end{center}

\subsection{Distance entre mots}

Lorsque $x$ et $y$ sont des chaînes de caractères, alors $dist(x,y)$ mesure le degré de ressemblance entre $x$ et $y$ et est définie comme étant le nombre minimum d'opérations telles que l'insertion d'un caractère, la suppression d'un caractère, la substitution d'un caractère par un autre, etc.,  nécessaires pour transformer le mot $x$ en le mot $y$.

\medskip

Dans ce calcul, le coût de chaque opération vaut 1 et la distance entre mots est donc un nombre entier.

\subsection{Quelques fonctions de distance entre mots}

Il existe plusieurs fonctions de distance qui se distinguent par les opérations que l'on peut appliquer sur les mots, les plus connues en recherche de motifs sont:

\smallskip
\begin{itemize}
\item[\textbf{La distance de Levenshtein \cite{Le66}}] (ou distance d'édition) : elle autorise trois opérations: la substitution ($a\longrightarrow b$), la suppression $a\longrightarrow \epsilon$, et l'insertion ($\epsilon\longrightarrow a$), où $a,b \in \Sigma$, $a\neq b$, et $\epsilon$ : le mot vide. La distance de Levenshtein est la métrique la plus utilisée pour déterminer la différence entre deux mots. 

Exemple : Soit $V=ABCjEF$ et $W=xBCEFy$, pour transformer $V$ en $W$, on substitue le $1^{\text{er}}$ caractère $A$ par $x$, et on supprime le $4^{\text{ème}}$ caractère $j$, et on insère le $y$ à la fin.

\item[\textbf{La distance de Hamming \cite{Hamming50}}] Cette fonction autorise une seule opération: la substitution.

\item[\textbf{La distance de Damerau \cite{damerau1964technique}}] En plus des opérations de la distance de Levenshtein, elle autorise l'opération de transposition ($ab\longrightarrow ba$). Exemple : $V=ABCxyEF$ et $W=ABCyxEF$, on a une seule opération, la transposition de ($xy\longrightarrow yx$).

\end{itemize}

\section{Les structures de données pour dictionnaires et textes}
\label{def:Les structures_donnees}

\noindent
\textbf{Un dictionnaire}  est un ensemble $D=\{w_1,w_2,...,w_d\}$ de $d$ mots, tels que $w_i \neq w_j $ ($i \neq j$). La taille du dictionnaire est $|D|=\sum_{i=1}^{i=d}|w_i| = n$.

\smallskip

\noindent
\textbf{Un texte} $T$ est une suite de $n$ caractères. Il peut être vu aussi comme étant un ensemble de $d$ mots $T=\{w_1,w_2,...,w_d\}$, où chaque mot peut avoir plusieurs occurrences.\\

\smallskip

Pour une recherche efficace dans le texte et/ou dans le dictionnaire, il est important d'organiser et de stocker les données en une structure qui en facilite l'exploitation (l'accès et les modifications).

\smallskip

Généralement une structure de données occupe plus d'espace mémoire que les données sur lesquelles elle est construite. Une \textit{structure de données succincte} occupe un espace proche de la taille des données qu'elle doit représenter. Elle occupe un espace plus petit si elle exploite des techniques de compression.

\medskip

Dans la bibliographie classique de l'algorithmique de texte \cite{Gusfield1997, cormen2009introduction, crochemore2001algorithmique, crochemore2002jewels}, il existe plusieurs structures de données pour index et dictionnaires.
Nous nous limitons dans ce qui suit au rappel de quelques structures de données utilisées dans la représentation d'un dictionnaire/index et exploitées dans cette thèse.

\subsection{Arbre de suffixes}

Soit $T$ un mot défini sur un alphabet $\Sigma$ tel que $|T|=n$.

L'arbre des suffixes~\cite{Gusfield1997,crochemore2009suffix} de $T$ est un arbre enraciné à  $n$ feuilles terminateur tel que :
\begin{enumerate}
\item Chaque feuille $i$ représente un seul et unique suffixe, c'est le mot  $T[i..n]$;
\item Ses branches sont étiquetées par des mots non vides.
\item Ses n{\oe}uds internes sont de degrés (nombre de branches sortantes) \ $>1$ .
\item Les étiquettes des branches sortantes d'un n{\oe}uds ne peuvent pas commencer par le même caractère.
\item La concaténation des étiquettes sur le chemin qui commence à la racine $R$ et qui se
termine à  une feuille $i$ forme le suffixe $T[i..n]$.\\
\end{enumerate}

% % Exemple: (Ajouter ici un exemple d'arbre de suffixes d'un mot) et montrer la définition sur l'exemple

Il existe différents algorithmes permettant la construction de l'arbre de suffixes dans un temps et espace linéaire \cite{Weiner1973,McCreight1976,Ukkonen1995}.

\subsection{Trie}
Soit $D$ un dictionnaire de $d$ mots. Un Trie \cite{crochemore2009trie,de1959file,fredkin1960trie}) pour $D$ est un arbre enraciné de préfixes communs et qui a $d$ feuilles. Chaque n\oe{}ud représente un préfixe commun des mots de $D$ et a un ou plusieurs n\oe{}uds fils. 
  Chaque arête est étiquetée par un caractère des mots de $D$. Deux arêtes qui sortent du même n\oe{}ud ne peuvent pas être étiquetées par le même caractère.
Chaque chemin de la racine à une feuille est un mot du dictionnaire.\\

\noindent La construction d'un Trie est en $O(n)$ où $n$ est la taille du dictionnaire. \\

\noindent La recherche d'un mot $p$ de longueur $|p|=m$ consiste à parcourir le Trie de la racine en suivant les lettres de $p$ sur les arêtes et est en $O(m)$.
 
\medskip

% % Exemple: (Ajouter ici un exemple de TRIE d'un mot) et montrer la définition sur l'exemple et un chemin de la recherche d'un mot

\subsubsection{Trie inversé}
On note par $\bar{w}$ le miroir de $w$ et par $\bar{D}= \{\bar{w_1},\bar{w_2},...,\bar{w_d}\}$ le dictionnaire des mots miroirs de $D$. 
Le Trie inversé de $D$ est le Trie de $\bar{D}$.  Il permet de faire la recherche de droite à gauche.

\subsubsection{Un Trie compact}

Un Trie compact \cite{Gusfield1997,crochemore2009trie}
est une structure de données qui optimise l'espace mémoire par le compactage des chemins du Trie. 

\smallskip

Soient  les n\oe{}uds  ${\{nd_i,nd_{i+1},...,nd_{i+j}\}}$  étiquetés respectivement par les caractères suivant $\{c_i,c_{i+1},...,c_{i+j}\}$, tous sur le même chemin, et chaque $nd_{k+1}$ est un fils unique de $nd_k$ , avec $k \in [i..i+j-1]$. Le compactage de ce chemin consiste à  fusionner ces n\oe{}uds en un seul n\oe{}ud $nd$, et en concaténant tous leurs caractères en une seule sous-chaine $\{c_ic_{i+1}...c_{i+j}\}$ sur la transition sortante du n\oe{}ud parent.\\

\smallskip

Remarquons enfin que l'arbre de suffixes utilisé dans l'indexation d'un texte $T$ est un Trie compact construit sur tous les suffixes de $T$. Chaque chemin de la racine à une feuille représente un suffixe du texte $T$.

\smallskip
% % Exemple: faire un exemple clair avec un $T ---> \{Suff(T)\} ---> Trie(\{Suff(T)\})=AS(T)$

\subsection{Le tableau LCP}

Le tableau des plus long préfixes communs d'un dictionnaire $D$ (\cite{Manber1993,Kasai2001}) est une structure de données qui permet de stocker les longueurs des plus longs préfixes communs entre les paires de mots consécutifs de $D$ dans l'ordre lexicographique.

Soient $D=\{w_1,w_2,...,w_n\}$ le dictionnaire et $LCP(x,y)$ la fonction  qui calcule la longueur du plus grand préfixe commun entre deux mots $x$ et $y$. Le tableau $Tab\_LCP$ contient l'ensemble des entiers $LCP(w_i, w_{i+1})$, pour $i=1, \dots n-1$.
$$Tab\_LCP=\{LCP(w_1,w_2),LCP(w_2,w_3),...,LCP(w_{n-1},w_n)\}$$

% % Exemple: ajouter un exemple de LCP (pour un D donné)  et donner la complexité de construction et de la recherche et l'utilité de cette structure

\subsection{Le Tas}

La structure de données Tas (\cite{williams1964algorithm, cormen2009introduction}) est un arbre binaire presque complet:  l'arbre est complètement rempli à tous les niveaux, sauf (dans certains cas) le niveau le plus bas,  rempli de la gauche jusqu'à un certain point. Il y a deux types de Tas, Tas Max et Tas Min. Dans le Tas Max, la clé dans chaque n\oe{}ud a une valeur supérieure à celle de ses n\oe{}uds fils, et l'inverse dans le Tas Min.

\smallskip

Le Tas peut être implémenté par un tableau d'une taille de $2n+1$ où $n$ est le nombre de n\oe{}uds. Les deux fils d'un n\oe{}ud $i$ sont dans les positions ($2i, 2i+1$).

\smallskip

% % Exemple: ajouter un exemple de tas (max ou min) et donner la complexité de construction et de la recherche.

\subsection{La file de priorité}

La file de priorité~\cite{EmdeBoas1976,Jones1986} est une structure de données qui permet l'accès en premier à l'élément ayant la plus grande priorité. Il existe deux types de file de priorité, celle qui donne la plus grande clé en premier (la plus grande priorité = la plus grande clé), et l'inverse, celle qui donne la plus petite clé en premier (la plus grande priorité = la plus petite clé).

\smallskip

La file de priorité permet trois opérations: 1) ajouter un élément, 2) accéder à l'élément qui a la plus grande priorité, 3) supprimer l'élément qui a la plus grande priorité.

\smallskip

La file de priorité peut être implémentée avec une variété de structure de données comme par exemple un tableau ou une liste chainée, mais généralement, elle est implémentée avec la structure du Tas.

\subsection{Prefix-sum}

La structure de données prefix-sum ~\cite{elias1974,fano1971,fenwick1994new} permet de stocker des éléments numériques (généralement dans un tableau) tels que la valeur stockée dans chaque position est la somme de toutes les valeurs des éléments précédents. $val(i) = \sum_{j=0}^{j=i} val(j)$.

La méthode consiste à transformer un ensemble de valeurs $S=\{v_1,v_2,...,v_n\}$ en un autre ensemble $S'=\{v'_1,v'_2,...,v'_n\}$, tel que $v'_i = \sum_{j=0}^{j=i} v_i = v'_{i-1} + v_i$.

\noindent
Exemple : 

\begin{tabular}{ll}
 S : & \{1,1,1,2,5,5,5\} \\ 
 prefix-sum (S) : &  \{1,2,3,5,10,15,20\}\\  
\end{tabular}

\subsection{Le tableau de bits}

La structure de données tableau de bits (bit-vector) est un tableau utilisé simplement pour stocker et retrouver les bits (1 ou 0) dans une position donnée. Elle permet de savoir si l'élément dans une position existe ou non.

Généralement, chaque case du tableau a une taille d'un mot mémoire $w$. Dans les machines actuelles $w=32$ ou $64$ selon le système utilisé de 32 ou 64 bits respectivement.

Dans une structure de dictionnaire, le tableau de bits est augmenté avec des informations supplémentaires afin de supporter les opérations de rang et de sélection (rank/select)~\cite{jacobson1989space, raman2002succinct, munro1996tables}. L'opération de rang (rank(1/0,i)) calcule le nombre de 1 (ou 0) depuis le début du tableau jusqu'à une position donnée $i$, et de même, l'opération de la sélection (select(1/0,i)) permet de retourner la position de l'occurrence numéro $i$ de l'élément 1 (ou 0).

\section{Les tables de hachage}
\label{def:Le_hachage}

Une table de hachage est une structure de données qui généralise la notion simple d'un tableau ordinaire pour implémenter les dictionnaires.
On utilise une fonction de hachage $h$ pour appliquer l'ensemble des clés $U=\{cl_1,cl_2,...,cl_n\}$ vers les positions $\{1,2,..,n\}$ des cases de la table de hachage $Tab$  de $n$ cases\\
%$h:U \rightarrow \{1,2,..,n\}$.\\

%Ainsi, une clé $cl_i$ est associée à la case h($clé$).

Dans le cas idéal, lorsque la fonction de hachage $h$ est bijective, elle associe à chaque clé une position unique et toute position est associée à une seule clé. Ce hachage est dit \textit{\og parfait \fg }.
Lorsque $h$ n'est pas injective elle retourne la même valeur de position pour  des clés différentes. On dit qu'il y a \textit{\og collision \fg }. 

Il existe plusieurs solutions pour résoudre le problème de collision comme le hachage avec chaînage, l'adressage ouvert, le sondage linéaire, le double hachage, etc.~\cite{cormen2009introduction,Knuth1998_book}.

Les deux sous-sections suivantes présentent respectivement le hachage parfait minimal et le hachage avec sondage linéaire.

\subsection{Le hachage parfait minimal}

Une fonction de hachage parfait est une fonction bijective qui ne donne donc pas de collision.
Étant donné un ensemble $U=\{w_1, w_2, \dots, w_n\}$ de $n$ éléments, la fonction associe chaque élément de $U$ à un numéro distinct dans l'intervalle [0..m-1] :\\

$\begin{array}{llll}
h: & U    &  \longrightarrow & \{1,2,...,m-1\} \text{ avec } m \geq n \\
   & w_i  &  \longmapsto & p \\
\end{array}$

\smallskip

$p=h(w_i)$  est la position de  $w_i$ dans  $D$. \\

\medskip
Une fonction de hachage est dite minimale $mphf$\footnote{Minimal perfect hashing.} (\cite{cichelli1980minimal,jaeschke1981reciprocal,czech1992optimal}) si $m=n$. Une fonction $mphf$ associe $n$ éléments à $n$ valeurs entières successives. Le but est de réduire l'espace mémoire de stockage utilisé par la fonction de hachage parfait.

Habituellement, la construction de la $mphf$ implique des étapes qui consistent à utiliser une fonction de hachage qui applique les mots à des nombres distincts de $O(w)$ bits, puis de ranger ces nombres dans un tableau de $n$ cases en utilisant une autre fonction de hachage.

\subsection{Table de hachage avec sondage linéaire}

Le hachage avec sondage linéaire~\cite{Knuth63noteson} est une méthode (parmi d'autres) pour résoudre le problème de collision. La méthode consiste à placer la donnée $clé$ dans la première case vide dans la table de hachage si la position $h(clé)$ est non vide.

La taille $T$ du tableau est plus grande que le nombre d'éléments $n$. Le taux de remplissage du tableau ou le facteur de chargement (Load Factor) doit être inférieur à 1 ($\alpha<1$). Si un tableau de hachage stocke $n$ éléments, alors sa capacité est $T=\lceil n/\alpha\rceil$.

\medskip
Pour insérer un nouvel élément, si la position $h(clé)$ est vide, alors il n'y a pas de collision et l'élément est inséré à la position $h(clé)$.
Dans le cas contraire, on cherche la première case vide après $h(clé)$ pour insérer l'élément, on vérifie dans les positions $h(clé)+1,h(clé)+2\ldots$, jusqu'à trouver une case vide (si l'une des positions est égale à $T-1$, alors la prochaine position sera la position $0$).

La recherche d'un élément se fait par la comparaison avec le contenu de la table de hachage à la position $h(clé)$. Si l'élément n'est pas à cette position, alors on parcourt toutes les cases consécutives à partir de cette position jusqu'à ce qu'on trouve l'élément recherché, ou on atteint une case vide ce qui implique que l'élément recherché n'existe pas dans la table de hachage.

\medskip
Pour supprimer un élément, on doit faire un décalage (vers la gauche ) aux éléments qui se trouvent directement après la case supprimée, et qui ont la même valeur de hachage de la case supprimée. Cette opération est nécessaire afin de ne pas laisser une case vide et donc perdre l'accès aux éléments qui ont la valeur de hachage que la case supprimée. 

\section{Notion de complexité algorithmique}

La complexité algorithmique~\cite{knuth1997art} est une fonction qui permet d'évaluer la quantité de ressources (temps et espace mémoire) nécessaire utilisée pour le fonctionnement d'un algorithme donné.
Pour calculer la complexité du cas pire, nous utilisons une évaluation asymptotique, où le nombre des éléments $n$ manipulés est assez grand (tend vers l'infini), et nous ignorons les constantes.

\subsection{Les notations de la complexité}
Il existe plusieurs notations pour designer la complexité d'un algorithme donné. Parmi les notations les plus utilisées, nous avons les notations : $O(n)$, $\Omega (n)$, $\Theta (n)$~\cite{knuth1976big_O}.
Soient $f$ et $g$ deux fonctions mathématiques et $n$ le nombre d'éléments manipulés, supposé  suffisamment grand (tend vers l'infini).

\medskip
\begin{itemize}

\item[\textbf{Grand O}] $f(n) = O(g(n))$, la fonction $f$ est bornée asymptotiquement par  la fonction $g$. $|f(n)| \leq k \times |(g(n))|$ avec $k$ une constante strictement positive. Cela est équivalent à dire que $ \frac{f(n)}{g(n)} \leq k$. La notation $O$ décrit une borne supérieure asymptotique.

\item[\textbf{Grand Omega}] $f(n) = \Omega (g(n))$, la fonction $f$ est minorée asymptotiquement par la fonction $g$ (à un facteur près), $|f(n)| \geq k \times |(g(n))|$, pour un $k>0$. La notation $\Omega$ décrit une borne inférieure asymptotique.

\item[\textbf{Grand Theta}] $f(n) = \Theta (g(n))$, la fonction $f$ est dominée et soumise asymptotiquement par la fonction $g$, $k_1 \times |(g(n)) \leq |f(n)| \leq k_2 \times |(g(n))|$, pour un $k_1>0$ et un $k_2>0$. 

La notation $\Theta$ décrit une borne supérieure et inférieure en même temps, $\Theta (g(n)) = \Omega (g(n)) \text{ et } O (g(n))$.

\end{itemize}

\subsection{Les modèles de calcul de la complexité}

L'évaluation des ressources  nécessaires à l'exécution d'un algorithme se fait dans un  modèle abstrait de machines.

Il existe plusieurs modèles de calcul dans la littérature, comme la machine de Turing, la machine RAM (Random Access Machine), la machine avec mémoire externe EMM~\cite{aggarwal1988input}, le modèle Cache-Oblivious (machine à plusieurs niveaux de mémoire)~\cite{frigo1999cache}, la machine PRAM (Parallel Random Access Machine)~\cite{karp1988survey}, etc.

Dans ce qui suit, nous décrivons brièvement les deux modèles les plus utilisés qui sont la machine de Turing et la machine RAM.
 
\medskip
\begin{itemize}

\item[\textbf{La machine de Turing}] 
Une machine de Turing~\cite{turing2004essential} est une machine abstraite composée d'un ruban infini divisé en cases consécutives sur lesquelles la machine écrit/efface des symboles (un alphabet fini) en déplaçant une tête de lecture/écriture vers l'avant ou vers l'arrière.

La machine de Turing est caractérisée par la position de la tête dans le ruban, le contenu de la case sur laquelle pointe la tête de lecture/écriture, et une table des actions à faire (les règles de transition entre les cases) qui sert comme un programme à la machine. Le ruban est initialement vide, toutes les cases contiennent le caractère $\sqcup \notin \Sigma$.

\medskip
Une machine de Turing est formellement définie par : $(Q,\Sigma,\Gamma,\sqcup,q_0,q_f,\delta)$

\begin{itemize}
\item $Q$ un ensemble fini des états.
\item $\Sigma$ un ensemble fini d'alphabet (les entrées-sorties); le choix typique est $\Sigma=\{0,1\}$.

\item $\Gamma$ un ensemble fini d'alphabet du ruban, $\Sigma\subseteq\Gamma$.

\item $\sqcup\in \Gamma$ un symbole (caractère) blanc, avec $\sqcup \notin \Sigma$.

\item $q_0 \in Q$ un état initial.
\item $q_f \in Q$ un état final.

\item $\delta: Q\setminus \{q_f\}\times\Gamma \rightarrow Q\times\Gamma\times\{L, R, N\}$, une fonction de transition partiellement définie. $\{L, R, N\}$ signifie que la tête de lecture peut se déplacer vers la gauche ou vers la droite par une seule case, ou ne se déplace pas.

Si $\delta$ n'est pas défini sur l'état actuel et le symbole actuel du ruban alors la machine s'arrête.

\end{itemize}

\medskip
\item[\textbf{La machine RAM}] 

Le modèle RAM (Random Access Machine)~\cite{luginbuhl1990computational,cook1972time} est un modèle d'une machine abstraite assez proche de l'architecture d'un ordinateur. La machine abstrait RAM est composée principalement d'une unité du calcul, des registres, d'une mémoire divisée en deux parties, l'une, pour stocker les données, et l'autre pour stocker les instructions (le programme).

La machine RAM utilise l'adressage indirect pour accéder aux différentes cases de la mémoire. 
La taille d'une cellule mémoire est notée par $w$, toutes les cellules ont le même coût d'accès; le coût pour accéder à une case mémoire est unitaire.

Dans ce modèle, chaque opération arithmétique ou logique à un coût unitaire, sauf les deux opérations la multiplication et la division. Certains chercheurs les considèrent comme deux opérations à coût unitaire, et d'autres non.

\end{itemize}

%CPU stalls :
%modern CPUs can execute hundreds of instructions in the time taken to fetch a single cache line from main memory.
%from wikiédia.

\subsection{La notion de la complexité moyenne}

La complexité en moyenne~\cite{levin1986average,bogdanov2006average} calcule la quantité des ressources (typiquement, le temps) nécessaire utilisée par un algorithme qui agit sur des données en entrée qui sont équiprobables, c'est-à-dire que les données en entrée ont une distribution qui ne provoque pas le pire cas.

La complexité en moyenne est nécessaire pour calculer la performance qui est proche de la réalité, lorsque les données en entrée ne provoquent pas le pire cas (ou le pire cas se produit rarement).
La complexité en moyenne permet de distinguer l'algorithme le plus efficace en pratique entre l'ensemble de tous les autres algorithmes qui peuvent avoir la même complexité dans le pire des cas.

L'analyse de la complexité moyenne dans certains domaines (comme la cryptographie) permet de bien comprendre le comportement de l'algorithme et le réadapter par exemple pour générer des instances difficiles dans le domaine de la cryptographie.

\chapter{Algorithmes et structures d'index pour le problème de la recherche approchée.}
\label{chap:etat_art}

%\section{Introduction}

L'algorithmique du texte, en tant que domaine qui porte sur l'étude des structures de données textuelles et d'algorithmes les traitants, a donné lieu à de nombreuses contributions, en particuliers  les articles survey de Navarro~\cite{Na01,Na01b} et Boytsov~\cite{Bo11} ou bien les livres : de M.Crochemore \cite{crochemore2001algorithmique,crochemore2002jewels}, de D. Gusfield \cite{Gusfield1997}, de C.Charras et T.Lecroq \cite{charras2004handbook} et de G.Navarro et M.Raffinot \cite{navarro2002flexible}.

%\cite{Na01, Bo11, pollock1984automatic, reynaert2008non, crochemore2001algorithmique, crochemore2002jewels, charras2004handbook}. 

Ce chapitre présente quelques-uns de ces travaux, en particuliers ceux fondés sur quelques idées que nous exploitons pour proposer de nouvelles solutions présentées dans les prochaines chapitres.

\paragraph{}

Les algorithmes utilisés pour résoudre le problème de la recherche approchée peuvent être regroupés selon le type de données traitées, leurs tailles, le scénario choisi etc. Ils peuvent être divisés en deux grandes catégories: 1) les algorithmes en-ligne (online), 2) les algorithmes hors-ligne (offline). On peut aussi les classer en deux autres catégories: 1) algorithmes traitants le texte, 2) algorithmes traitants le dictionnaire.

\paragraph{Les algorithmes en-ligne et hors-ligne :} % online vs offline

Les méthodes en-ligne sont utilisées lorsque le texte ne peut être pré-traité et indexé. C'est le cas lorsque ce texte n'est par exemple pas connu à l'avance ou s'il est susceptible d'être modifié dans le temps de très nombreuses fois (par exemple l'édition de texte et l'OCR). La complexité temporelle de la recherche par ces méthodes est proportionnelle à la taille du texte.

\medskip

Les méthodes hors-ligne sont utilisées lorsque la quantité d'informations textuelles est importante et le texte est connu à l'avance. Les méthodes hors-ligne construisent des structures de données pour stocker et pré-traiter le texte afin de répondre rapidement à une requête dans un temps proportionnel à la longueur du mot requête (linéaire dans le meilleur des cas).

\paragraph{Les algorithmes du texte et de dictionnaire :} % % texte vs dictionnaire

La recherche dans un dictionnaire est différente de la recherche dans un texte car dans un dictionnaire, les mots sont ordonnés et ils n'apparaissent qu'une seule fois contrairement à un texte où les mêmes mots apparaissent plusieurs fois dans différentes positions, en plus de l'utilisation de différents types de mots vides (stopwords) et les signes de ponctuation...etc. 

Il arrive que le texte soit une seule chaîne de caractères comme par exemple les chaînes génomiques comme l'ADN.

\medskip

La recherche approchée dans un dictionnaire doit retourner tous les mots du dictionnaire qui diffèrent d'au plus $k$ erreurs du mot requête $q$. La recherche dans un texte consiste à trouver toutes les solutions approchées et leurs occurrences (leurs positions).

En général, toutes les méthodes peuvent être appliquées dans les deux contextes (texte/dictionnaire) avec de légères modifications ou bien par l'ajout de simples structures de données (par exemple, un index inversé). Comme exemple typique de structures de données pour la recherche dans un dictionnaire, on peut citer, le Trie. Pour la recherche dans un texte, on peut citer l'arbre des suffixes qui est un Trie (compact) contenant tous les suffixes du texte.\\

%De nos jours, il y a beaucoup de livres et de surveys qui traitent le sujet des algorithmes sur les textes.
%Le lecteur peut se référer aux serveys de Navarro~\cite{Na01,Na01b} et BoyTsov~\cite{Bo11} ou bien aux livres : de M.Crochemore \cite{crochemore2001algorithmique,crochemore2002jewels}, de D. Gusfield \cite{gusfield97}, de C.Charras et T.Lecroq \cite{charras2004handbook} et de G.Navarro et M.Raffinot \cite{navarro2002flexible}.

Dans les sections de ce chapitre, nous donnons un aperçu sur quelques concepts de la recherche approchée, ensuite nous présentons les travaux qui sont fortement liés à notre travail.

%\section{Classer les méthodes de la recherche approchée hors ligne (offline)}
\section{Classification des méthodes de la recherche approchée}

Les différents algorithmes utilisés pour résoudre le problème de la recherche approchée dans ses deux versions (en-ligne / hors-ligne) peuvent être regroupés dans des catégories selon leurs méthodes et techniques communes. Dans ce travail, toutes les méthodes sont regroupées comme suit :

\begin{enumerate}
\item La programmation dynamique.
%\item Les automates.
\item Indexation (seulement dans le cas hors-ligne).
\item Les méthodes de filtrage.
\item La génération de voisinage.
\item Les méthodes qui utilisent le hachage.
\item Les méthodes de bit-parallélisme.
\item Les méthodes qui utilisent le parallélisme.
\item Les méthodes hybrides.
\end{enumerate}

Dans les sous-sections qui suivent, nous expliquons chaque catégorie et nous donnons comme exemples quelques algorithmes qui utilisent les méthodes et les techniques de cette catégorie.

\subsection{La programmation dynamique}

La programmation dynamique est un paradigme de programmation permettant de résoudre des problèmes complexes en les décomposant en une collection de sous-problèmes plus simples. La méthode examine les solutions des sous-problèmes déjà calculées et combine leurs résultats pour trouver la meilleure solution du problème donné  \cite{bellman1952theory,bellman1956dynamic}. La programmation dynamique est utilisée dans plusieurs domaines pour résoudre des problèmes d'optimisation.

Dans l'algorithmique du texte, la programmation dynamique est utilisée pour calculer la distance entre deux mots $dist(x,y)$, lors d'une recherche approchée des motifs ou d'un alignement de deux séquences.

Le principe de base est de construire une matrice $M$ de $|x+1| \times |y+1|$ cases où $M[i,j]$ représente le nombre minimum d'opérations nécessaires pour que la sous chaîne $x[1..i]$ ($x[1..i]=\{x_1, x_2,...,x_i\}$) corresponde à la sous-chaîne $y[1..j]$. La matrice est remplie selon les règles suivantes :

\bigskip
\noindent
$M[i,0] = i$\\
$M[0,j] = j$\\
$M[i,j] = M[i-1,j-1] \:\:\: si \:(x_i = y_j) $\\
$M[i,j] = 1+ min(M[i-1,j] , M[i,j-1] , M[i-1,j-1]) \:\:\: sinon $\\

Le résultat final (la distance d'édition) se trouve à la dernière case $M[|x|,|y|]$, la  distance des sous-chaînes ($x[1..i]$ et $y[1..j]$) se trouve dans la case $M[i,j]$. Nous donnons ci-dessous l'explication de ces règles.

\begin{itemize}

\item $M[i,0] = i$ ($M[0,j] = j$) : la distance d'édition entre une chaîne $x$ et la chaîne vide $\epsilon$ est $|x|$. L'opération pour transformer $x$ en $\epsilon$ est la suppression. La case $M[i,0]$ représente la distance d'édition entre une sous-chaîne $x[1..i]$ de longueur $i$ et la chaîne vide, donc $dist(x[1..i],\epsilon) = i$, (de même pour la case $M[0,j]$, $dist(\epsilon,y[1..j]) = j$).

\item $M[i,j] = M[i-1,j-1] $ : dans le cas où les deux caractères $x_i$ et $y_j$ sont égaux, alors nous n'avons aucune opération. Donc la valeur de la distance d'édition revient à la valeur de 2 sous-chaînes précédentes c'est-à-dire $x[1..i-1]$ et $y[1..j-1]$. Donc si $(x_i = y_j) $ alors $dist(x[1..i],y[1..j]) = dist(x[1..i-1],y_[1..j-1])$.

\item $M[i,j] = 1+ min(M[i-1,j] , M[i,j-1] , M[i-1,j-1])$ : dans le cas où ($x_i \neq y_j $), cela signifie que nous avons l'une des trois opérations (substitution, suppression, insertion), donc la distance d'édition de ($x_i, y_j$) revient à ajouter 1 à la valeur minimale des trois cases ($M[i-1,j] , M[i,j-1] , M[i-1,j-1]$) qui contient la valeur de la distance d'édition des trois cas suivants : ($x[1..i-1], y[1..j]$), ($x[1..i], y[1..j-1]$), ($x[1..i-1], y[1..j-1]$).

Si nous prenons  la valeur($M[i-1,j-1]$), alors l'opération de ($x_i,y_j$) est la substitution. Si nous prenons  $M[i-1,j]$ alors l'opération est la suppression, et pour $M[i,j-1]$ c'est l'insertion. 

Les deux opérations (la suppression et l'insertion) sont interchangeables, selon le sens de transformation des deux mots. Transformer $x$ en $y$ où l'inverse.
\end{itemize}

\medskip

Pour plus de détails le lecteur pourra consulter le survey de Navarro \cite{Na01}.

L'algorithme du calcul de la distance d'édition peut être adapté pour la recherche de motifs dans un texte~\cite{sellers1980theory,masek1980faster,Ukkonen1993,q-gram_sutinen1995using}.

Pour trouver l'alignement entre une séquence $x$ et une séquence $y$, il suffit juste de retrouver le chemin de la distance d'édition dans le sens inverse depuis la case $M[|x|,|y|]$ vers la case $M[0,0]$. Parmi les algorithmes les plus célèbres qui calculent l'alignement entre les séquences, nous citons: \cite{Needleman_Wunsch,Smith_waterman}.

%\subsection{Les automates:}
%However, there is a time and space exponential dependence on m and k that limits its practicality \cite{Na01}
%\subsection{Filtrage avec partitionnement intermidière :}

\subsection{L'indexation}

L'idée de base est de simplement indexer le texte/dictionnaire et d'appliquer une méthode récursive pour vérifier, en un temps relativement rapide, toutes les occurrences qui sont des solutions approchées.

Il existe de nombreuses structures de données utilisées à la fois pour la recherche exacte et approchée. Parmi ces structures, le Trie est considéré comme une des structures les plus fondamentales pour représenter un dictionnaire, l'arbre de suffixes pour représenter le texte \cite{Weiner1973,McCreight1976,Ukkonen1995}, les structures de données compressées comme le Trie compact  \cite{Gusfield1997,Morrison1968}, le tableau des suffixes \cite{Manber1993} compressé~\cite{suffix_arrays_grossi2005}, les graphes orientés acycliques de mots  $DAWG$ (directed acyclic word graphs) \cite{DAWG_blumer1983,DAWG_blumer1985}, et les graphes orientés acycliques compacts de mots  $CDAWG$ (compact directed acyclic word graphs) \cite{CDAWG_blumer1984} ou 
le FM-index qui compresse les données et l'index en même temps~\cite{FM_index_ferragina2000,FM_index_ferragina2005}.

%La structure d'indexation la plus fondamentale pour le dictionnaire est le Trie et pour le texte c'est l'arbre des suffixes \cite{Weiner1973,McCreight1976,Ukkonen1995}.
%Le Trie compact  \cite{Gusfield1997,Morrison1968}, le tableau des suffixes \cite{Manber1993}, la compression du tableau des suffixes et l'arbre des suffixes avec des applications à l'indexation et la recherche de texte (\cite{suffix_arrays_grossi2005}), graphes de mots orientés acycliques $DAWG$ (directed acyclic word graphs) \cite{DAWG_blumer1983,DAWG_blumer1985}, et les graphes de mots orientés acycliques compacts $CDAWG$ (compact directed acyclic word graphs) \cite{CDAWG_blumer1984}.
%Le FM-index qui compresse les données et l'index en même temps (\cite{FM_index_ferragina2000,FM_index_ferragina2005}).

% % ajouter les wrvlet tree, et trie suffixe tree bidirectionlle.... les structur bidrictionnele kayan wahda bzf haylaa nessitha ta3 ihawasss fiha les erreur

Un algorithme de recherche parcourt d'une façon récursive la hiérarchie de l'arbre à partir de la racine vers les feuilles. \`A chaque étape, on détermine si le chemin mène à une solution approchée ou non ; s'il n'y a pas de solution, l'algorithme retourne au n\oe{}ud précédent pour choisir une autre transition sortante pour la vérifier. Cet algorithme récursif, vérifie toutes les transitions sortantes dans chaque 
n\oe{}ud d'erreur possible.

Des algorithmes basés sur les différentes structures de données d'indexation proposent des solutions au problème de la recherche approchée de motifs, comme par exemple les algorithmes  connus et présentés dans les références suivantes: \cite{Ukkonen1993, Tries_shang1996, russo2009_compressed_indexes_approximate, bieganski1994generalized_suffix_trees, zobel1995finding_approximate_large_lexicons}.

\subsection{Les méthodes basées sur le filtrage}

Chaque mot ayant un certain nombre d'erreurs contient des sous-chaînes exactes (des parties qui ne contiennent pas d'erreurs). Afin d'accélérer la recherche approchée, on parcourt rapidement le dictionnaire/texte pour trouver des pièces qui peuvent être des solutions approchées possibles.
L'approche consiste à faire une recherche exacte sur certains morceaux du mot requête dans le dictionnaire/texte pour obtenir des motifs candidats, et ensuite, faire les vérifications pour s'assurer s'ils représentent des solutions approchées ou non.
La première étape utilise généralement des index (Trie, arbre des suffixes) pour localiser rapidement des parties exactes, puis des algorithmes classiques sont utilisés pour la vérification.

Les algorithmes de la recherche exacte sont beaucoup plus rapides que ceux de la recherche approchée. Par conséquent, les algorithmes de filtrage peuvent améliorer la recherche approchée d'une manière considérable \cite{Na01}.

Étant donnés un motif $P$ et un seuil $k$, le motif $P$ peut être divisé en $k+1$ morceaux  $P=\{P_1,P_2,...,P_{k+1} \}$. La première étape consiste à faire une recherche exacte sur ces $k+1$ morceaux. Au moins l'un des $P_i$ ($i\,\in\,[1..k+1]$) est une sous-chaîne exacte dans $P$. La deuxième étape, consiste à vérifier s'il y a des solutions approchées \cite{Na01b}.

Nous citons quelques algorithmes qui utilisent ce concept~\cite{Ri76,wu1992fast,T_R_T_brodal1996approximate,Amir2000}.

\subsection{La génération du voisinage}

Une approche de génération de voisinage génère, pour un mot requête, une liste de mots ayant une certaine distance d'édition avec le mot original, puis cette liste de mots générés est recherchée en utilisant des méthodes exactes. Dans la génération de voisinage, il y a trois méthodes différentes :

\paragraph{Le voisinage complet :}

La génération de voisinage complet implique le calcul de tous les mots possibles ayant un nombre d'erreurs de distance d'édition $k$ par rapport au mot requête. Chaque élément des mots générés est recherché dans le dictionnaire par la méthode exacte. Puisque la taille de la génération de voisinage complet est en $O(m^k\,\sigma^k)$ \cite{Ukkonen1993}, cet algorithme est seulement pratique lorsque les paramètres suivants sont de petites tailles : la taille de l'alphabet $\sigma$, le nombre maximum autorisé d'erreurs $k$, et la longueur du mot requête $m$.

\paragraph{Le voisinage avec joker :}

Diminuer la taille du voisinage complet en remplaçant certains caractères par un caractère joker (wild-card).  Un caractère joker peut remplacer  n'importe quel autre caractère dans la chaîne.

Soit le mot requête $ABCD$, et $\phi$ le caractère joker, on obtient seulement 4 mots par une seule substitution : ${\phi}BCD$, $A{\phi}CD$, $AB{\phi}D$, $ABC{\phi}$.

Donc au lieu de substituer tous les caractères de l'alphabet dans chaque position, il suffit juste de placer un caractère joker qui peut matcher tous les caractères.

\paragraph{Le voisinage réduit : }

Dans cette approche, les mots sur l'alphabet d'origine sont associés à des mots sur un alphabet réduit, plus petit.
L'objectif de cette approche consiste à compresser la génération de voisinage en diminuant la taille de l'alphabet. La technique consiste à associer l'alphabet d'origine $\Sigma$ à un alphabet réduit $\Sigma^{'}$ par l'intermédiaire d'une fonction de hachage.

\bigskip
Il existe plusieurs algorithmes basés sur cette approche \cite{mor1982hash,russo2005efficient,myers1994sublinear,Bocek2008,karch2010improved,Be09} (voir \cite{Bo12} pour plus de détails sur l'approche de génération de voisinage).

\subsection{Les algorithmes qui utilisent le hachage}
\label{sub:etat:methodes_hachage}

Les méthodes de hachage (voir la section \ref{def:Le_hachage}) sont utilisées généralement dans la version hors-ligne du problème de la recherche approchée bien qu'on puisse utiliser le hachage sans indexation et donc l'utiliser dans la version en-ligne du problème.

Parmi les algorithmes célèbres pour la recherche exacte et qui utilise le hachage: l'algorithme de ''Rabin-Karp''~\cite{KR87}.
Les performances de l'algorithme Rabin-Karp proviennent de l'utilisation d'une fonction de hachage efficace crée par Rabin-Karp. Cette fonction est utilisée par \cite{Be09}. Nous avons utilisé la même fonction dans notre travail. Les détails de l'utilisation de cette fonction sont présentés dans la sous-section \ref{sub:etat:dictionnaire_exacte_mphf}.

Nous donnons quelque autres algorithmes qui utilisent le hachage : \cite{Be09,q-gram_ukkonen1992approximate,wu1992agrep,mor1982hash,Bo12}.

\subsection{Les méthodes de bit-parallélisme}
\label{sub:etat:methodes_bit-parallelisme}

Ce type d'algorithmes est basé sur l'exploitation du fait qu'un ordinateur manipule des données par des blocs de $w$ bits, où $w$ est la taille d'un seul mot mémoire. Dans le modèle standard du calcul de complexité RAM, $w=\Omega(\log n)$ avec $n$ la taille totale du texte/dictionnaire. Généralement $w = 32 $ ou $ w = 64 $ bits dans les ordinateurs actuels.

Cette technique est appliquée pour tous les algorithmes qui utilisent les approches d'indexation directe, de génération de voisinage, de filtrage, de programmation dynamique, parce que l'idée est d'encoder plusieurs éléments de données d'un algorithme dans un seul mot mémoire pour traiter de nombreux éléments en même temps lors d'une seule opération du processeur. 

Le vocable \textit{bit-parallélisme} traduit l'idée de  traiter plusieurs éléments encodés en bits en même temps, c'est comme si on traitait plusieurs éléments en parallèle avec plusieurs processeurs.

Les algorithmes de bits parallélisme ne sont pas nouveaux dans le domaine de la recherche approchée de motifs, ils ont été proposés dans les années 1990.

Parmi les algorithmes qui utilisent cette technique, nous citons les travaux suivants:  \cite{baeza1992new,wu1992fast,wu1996subquadratic,myers1999fast,navarro2000fast,navarro1998bit,hyyro2001explaining,hyyro2002faster,Hyyro2003,hyyro2005increased}.

\subsection{Les méthodes hybrides}
\label{sub:etat:methodes_hybrides}

Il y a de nombreux algorithmes qui utilisent des approches hybrides (la programmation dynamique, l'indexation directe, la génération de voisinage, le filtrage) afin d'obtenir de meilleurs résultats.

Dans toutes les approches qui traitent le problème de la recherche approchée dans sa version hors-ligne, les méthodes (la génération de voisinage, le filtrage...) sont utilisées avec des index pour accélérer la recherche. Par exemple, l'indexation est utilisée avec la génération de voisinage, le filtrage, le hachage : \cite{Bocek2008, karch2010improved,Be09,CGL04,T_R_T_brodal1996approximate}.

Le bit-parallélisme peut aussi être combiné avec les autres approches afin d'accélérer le traitement comme nous l'avons expliqué dans la sous-section \ref{sub:etat:methodes_bit-parallelisme}. Par exemple, les algorithmes expliqués dans \cite{myers1999fast,hyyro2001explaining} combinent le bit-parallélisme avec la programmation dynamique.

Le calcul de la distance d'édition entre deux mots $x$ et $y$ se fait avec la programmation dynamique, cela implique généralement que les différentes méthodes (le filtrage, la génération de voisinage ...) vérifient la distance d'édition des derniers mots avec le mot requête. Nous citons quelques exemples pour les algorithmes qui utilisent la programmation dynamique avec les autres approches \cite{Ukkonen1993, q-gram_sutinen1995using,cole2002approximate}.

\subsection{Le parallélisme}
\label{sub:etat:parallelisme}

Le parallélisme permet de traiter des informations de manière simultanée sur plusieurs unités de calculs. Au lieu d'exécuter les tâches sur une seule unité de calcul, on les divise sur plusieurs. Le but est de réaliser le plus grand nombre d'opérations en un temps record (voir~\cite{pacheco2011introduction,grama2003introduction}).

Il existe un modèle de machine parallèle qui fût très en vogue dans les années 80 (un modèle théorique) appelée PRAM (Parallel Random Access Machine)~\cite{karp1988survey}. Nous citons quelque algorithmes datant de cette époque : \cite{landau1989fast,landau1986introducing} et quelques travaux récents : \cite{GPU_kouzinopoulos2009string,GPU_lin2010accelerating,GPU_zha2011multipattern}.

%\section{Les Résultats théoriques et pratiques} % lorsque j'ajoute les résultat pratique j'active ce titre.
\section{Quelques résultats connus}

%Le domaine de la recherche approchée n'est pas nouveau, et il y a beaucoup de chercheurs qui ont proposer plusieurs solutions avec des différent complexité.

%Parmi les résultat les plus imporatnt dans ce dommaine on cite le travaille de Amir et al, (petit exlication), Buchman et al (petit explication + voir la sous-section pour plus de détail), cole et al ....,  djamal.

%% je ne sais pas si je doit ou pas ajouter plus de méthode.... je pense pas....

%Les Résultat pratique, utilisée dans la pratique : ..... grep, agrep....fastss, karch et al, boytsov, la thèse de master almang....

\subsection{K-errata trie de Cole et al}

Dans~\cite{CGL04} , Cole et al ont proposé une nouvelle structure de données appelée K-errata Trie. Cette structure permet de résoudre le problème de la recherche approchée pour n'importe quel nombre d'erreurs $k$.

La solution est basée sur le concept de la génération de voisinage et l'utilisation du Trie avec erreur pour créer l'index K-errata Trie.

La méthode consiste à créer un arbre de suffixes et ensuite de trouver les ''centroid path''\footnote{centroid path : Le barycentre, ou centre de masse, parfois appelé centre de gravité.} pour décomposer des sous-arbres en ''centroid path decompositions''. La méthode insère des erreurs dans l'arbre de suffixes, puis  fusionne récursivement les sous-arbres de ''centroid path decompositions''.\\

Rappelons les définitions du \textit{Centroid paths} et du \textit{Centroid path decompositions} avant d'expliquer brièvement la construction de K-errata Trie.

\textbf{\textit{Centroid paths (chemins barycentre)}} : le chemin commence à la racine de l'arbre $T$. Chaque n\oe{}ud $nd$ sur le chemin, se branche au n\oe{}ud fils qui a le plus grand nombre de feuilles dans son sous-arbre  (donc ce n\oe{}ud devient la racine de ce sous-arbre).

%Dans une arbre on peut avoir plusieurs centroid path, par exemple si l'arbre est équilibrer (dans chaque n?ud on a le même nombre de fils qui sort) donc on aura autant de centroid path que de chemin dans l'arbre.

\textbf{\textit{Centroid path decompositions (décomposition de chemin barycentre)}} :
après l'identification des \textit{centroid path}, l'arbre est décomposé de façon récursive en sous-arbres. Chaque transition qui sort du \textit{centroid path} va représenter un sous-arbre.\\

Soit $C$  un des \textit{centoid path} résultants et soit $v$ un n\oe{}ud sur  le \textit{centroid path} $C$. On forme un nouveau sous-arbre qui commence avec un caractère joker qui sort de $v$, ce sous-arbre comprenant la fusion de tous les sous-arbres qui sortent de $v$ et qui ne sont pas sur le \textit{centroid path} en remplaçant le premier caractère de chaque sous-arbre par un nouveau symbole joker $\phi$ ($\phi \notin \Sigma$ ).

Ainsi,  1-errata Trie est créé. Pour créer le 2-errata Trie, on prend chaque sous-arbre qui commence par $\phi$, et on applique la même procédure, c-à-d, trouver les \textit{centroid path}, et créer des sous-arbres avec le caractère joker, ensuite, fusionner tous les sous-arbres qui sortent d'un n\oe{}ud sur le \textit{centroid path}.
On applique cette méthode d'une façon récursive afin de créer le k-errata Trie.\\

La complexité temporelle et spatiale atteinte est comme suit :

La recherche approchée dans un texte : l'espace mémoire de la structure de données : $O(n \frac{(c_1 log\, n)^{k} }{k!} )$. Le temps de la construction : $O(n \frac{(c_1 log\, n)^{k} }{k!} )$. Le temps de la requête (3 types d'erreurs): $O( \frac{(c_2 log\, n)^{k} \, log \, log \, n }{k!} + m + 3^{k}\times occ )$. $c_1 , \, c_2$ sont des constants $>1$, et le nombre d'erreurs $k$ est $k \leq log \, n $.
	
La recherche approchée dans un dictionnaire: l'espace mémoire de la structure de données : $O(n \, log^{k} \, n)$. Le temps de la construction : $O(n \, log^{k} \, n + n \, log\,\Sigma)$. Le temps de la recherche : $2^{k} log \, log \, n + m + occ $.

\subsection{La méthode de Buchsbaum et al}

Buchsbaum et al ont travaillé dans~\cite{buchsbaum2000range} sur le problème de la recherche dans une plage \cite{agarwal1999geometric} (rechercher un ensemble d'éléments $S$ dans un ensemble $S'$ plus grand). Ils ont introduit le produit vectoriel des arbres ''tree cross-product'', leur but est de faire un pré-traitement sur des arbres et des hypergraphes \cite{berge1973graphs} afin de rendre le traitement des requêtes plus efficace. Leur résultat est appliqué dans la visualisation des graphes, l'analyse des logiciels et la recherche de motifs.

Buchsbaum et al. améliorent la méthode de Amir et al~\cite{Amir2000} en augmentant les structures de données Trie et Trie inversé. La méthode consiste à ajouter des arcs entre les feuilles des deux arbres qui ont les mêmes étiquettes. Ensuite, ils utilisent la méthode sur les graphes pour faire l'intersection des n\oe{}uds trouvés dans le Trie et le Trie inversé dans l'étape 3 de Amir et al.

La complexité temporelle et spatiale atteinte par Buchsbaum et al. est comme suit :
Pour un texte $T$ de taille $n$ et un mot requête d'une longueur $m$, l'espace mémoire utilisé est de l'ordre de $O(n \, log \, n)$. Le temps de pré-traitement est de  $O(n \, log \, n)$. Le temps de la recherche avec $k=1$ erreur est de $O(m \, log \, log \, n + occrs)$ ($occrs$ représente le nombre d'occurrences de la solution).

\section{L'algorithme de D.Belazzougui}
\label{sec:etat:algorithme_Djamal_Belazzougui}

D.Belazzougui dans \cite{Be09} a proposé une nouvelle structure de données qui permet de résoudre le problème de la recherche approchée pour $k=1$ erreur dans un temps proportionnel à la longueur $m$ du mot requête $O(m+occrs)$ ($occrs$ le nombre d'occurrences de la solution), et dans un espace mémoire optimal où l'index occupe $O(n(lg(n) log log(n))^{2})$ avec $n$ la taille de texte. 

Pour arriver à ce résultat, il utilise plusieurs outils combinés entre eux : les fonctions de hachage parfait minimal (mphf), les tableaux de hachage dynamique, le Trie et le Trie inversé, les tableaux de bits, et la structure de données préfixe-sum.

Étant donnée une collection $S$ de $d$ motifs (clés) de longueur $n$ caractères.
On note le mot requête par $q$ et sa longueur par $m=|q|$.
On met $\sigma$ comme la taille de l'alphabet $\Sigma$, donc $\sigma = |\Sigma|$).

La solution simple (naïve) pour répondre aux requêtes approchées sur le mot $q$ est de chercher exhaustivement pour tous les mots possibles qui peuvent être obtenues par l'opération de distance d'édition avec une seule erreur sur le mot requête $q$.
Si un dictionnaire exact est utilisé donc cela prend un temps de $O(m)$ pour répondre à chaque requête candidats pour les différentes combinaisons des mots générés avec les 3 opérations de distance d'édition, donc cela donne au total un temps de $O(m^{2}\sigma)$ pour répondre à tous les mots candidats.

Afin de réduire le temps des requêtes, D.Bellazougui a réduit le nombre de mots candidats pour l'insertion et la substitution de $O(m\sigma)$ à juste $O(m)$, par l'utilisation d'une structure de données ''dictionnaire des listes de substitution'' qui donne un caractère candidat pour chaque position de l'insertion ou la substitution dans le mot $q$.
Ainsi, au lieu de tester pour chaque position $\sigma$ caractères, cette structure nous donne une liste de caractères qui lorsqu'on les teste donnent tous des solutions.

Il effectue un pré-traitement pour le mot $q$, avec un temps $O(m)$, pour pouvoir vérifier (avec la comparaison exacte) chaque mot candidat dans un temps de $O(1)$ au lieu $O(m)$.

L'idée utilisée pour améliorer le temps de comparaison d'une manière déterministe se base sur le fait que l'on peut comparer deux blocs mémoires dans un temps seulement $O(1)$, tel que chaque bloc est de taille $w$ (un mot mémoire), cela permet de comparer deux morceaux de chaîne de caractères tels que chaque morceau contient $u$ caractères où $u b = O(w)$ et $b$ est le nombre de bits pour un seul caractère. Ceci implique que les mots de moins de $u$ caractères peuvent être comparés dans un temps O(1).

Pour améliorer le temps de comparaison pour les mots longs, l'idée est de calculer les signatures de tous les préfixes et suffixes des mots dans le dictionnaire dont les longueurs sont multiples de $u$. Ces signatures occupent moins que $w$ bits, et l'espace total utilisé par les signatures est ainsi du même ordre que l'espace occupé par le mot. La comparaison d'un mot candidat comprend la comparaison des signatures des préfixes et des suffixes de $q$, avec les signatures des préfixes et des suffixes de mot $s$ de dictionnaire. Le nombre de comparaisons est constant, impliquant un nombre de blocs de moins de $|q|/u$.

Afin d'obtenir des signatures déterministes (non-collision) pour les préfixes et les suffixes, on utilise un Trie et un Trie inversé construit sur l'ensemble des mots du dictionnaire. Comme cela, on peut trouver des signatures déterministes pour n'importe quels préfixe et suffixe de chaque mot du dictionnaire.\\

\subsection{La construction de la structure de données}
\label{sub:la_structure_de_donnee}

La structure de données utilise les composants suivants :

\subsubsection{Un Trie et un Trie inversé}

Étant donnée une collection de mots $S$ de taille total $n$ caractères ($nb$ bits), on veut construire un Trie sur $S$ qui occupe un espace de $O(n)$, et qui peut être traversé dans un temps de $O(|q|)$ pour un mot requête $q$.

Pour la requête de mot $q$, on a besoin que le Trie renvoie un numéro identifiant unique pour chaque n\oe{}ud traversé pendant la recherche de motif $q$. Cet identifiant doit occuper un espace de $O(w)$ bits. Dans chaque n\oe{}ud on ajoute un identifiant unique (un nombre), la racine à comme indice 0, ensuite, on donne des indices au fils gauche, ensuite, les fils droite.

Ce problème peut être résolu par l'utilisation d'une structure de données décrite dans \cite{benoit2005}. La structure de données peut renvoyer un nombre entier unique dans l'intervalle $[1, Nb\_nd]$ pour chaque n\oe{}ud traversé, où $Nb\_nd$ est le nombre total des n\oe{}uds.\\

On a besoin de deux Tries :

\begin{itemize}

\item Un Trie $Tr$ construit sur l'ensemble des mots $S$. Le rôle principal de ce Trie est qu'il permet de comparer les préfixes des clefs dans $S$ avec les préfixes de n'importe quel mot requête $q$, et cela juste en comparant les identifiants retournés par le Trie pour les deux mots. C'est la partie essentielle de l'algorithme qui permet d'obtenir un temps de requête borné en $O(1)$ pour chaque mot candidat.

\item Un Trie inversé $\overline{Tr}$ construit sur l'ensemble $\bar{S}$, afin de comparer les suffixes des clefs de $S$ avec les suffixes de n'importe quel mot $q$, et cela juste on comparant les identifiants retournés par le Trie inversé pour les deux mots.
\end{itemize}

\medskip
On note que le parcours du Trie, ou du Trie renversé pour un motif $q$ retourne au plus $|q|$ identifiants dans l'intervalle $[1,Nb\_nd]$ correspondant aux n\oe{}uds traversés.

\subsubsection{Un dictionnaire exact de mot avec une fonction du hachage parfait minimal}
\label{sub:etat:dictionnaire_exacte_mphf}

Étant donné un ensemble $S$ de $d$ mots de taille totale $n$. En utilisant la fonction du hachage parfait minimal, on peut construire un dictionnaire (index du dictionnaire) qui occupe un espace optimal afin de répondre à des requêtes exactes dans un temps proportionnelle à la longueur de mot requête. Le dictionnaire occupe $O(nb)$ bits et la recherche prend $O(m)$, où $m$ est la longueur de mot requête.\\

La $1\iere{}$ étape de construction du dictionnaire consiste à construire la fonction du hachage parfait minimal $mphf$ de l'ensemble $S$.

La $mphf$ mappe chaque mot de $S$ vers un numéro distinct dans l'intervalle $[0,d-1]$.

On stocke les mots l'un après l'autre dans un tableau $T$ consécutif dans l'ordre donné par $mphf$. Le premier caractère du mot $s_i$ mappé par $mphf$ dans le tableau $T$ est donné par la formule suivante : \\

$Pos(s_i) = \sum_{j=0}^{j < i} |s_j|$ \\

Dans le but de trouver la position d'un mot dans $T$, on utilise la structure de données prefix-sum qui donne pour n'importe quelle position $i$ la somme des longueurs de tous les mots mappés dans $T$ par $mphf$ avant la position $i$.\\

\paragraph{La construction de la $mphf$ :} soit un nombre premier $P$ tel que $P > n\times d^{2}$, et $P > 2^{b}$ (n: le nombre total des caractères de l'ensemble $S$, d: le nombre de motifs de $S$, b: le nombre de bits d'un seul caractère).
Avant de construire la fonction de hachage parfait minimal $mphf$, on doit mapper l'ensemble de l'alphabet de l'ensemble $S$ vers des valeurs entières, chaque caractère $c$ va être mappé à une valeur unique $v$ avec $v \in [0, P-1]$.\\

La première étape dans la construction est de mapper l'ensemble $S$ de $d$ mots aux valeurs du hachage distinctes dans l'intervalle $[0, P-1]$.
Pour cela, on utilise une fonction de hachage $h$ paramétrée avec un nombre entier $t$ aléatoirement choisi tel que  $t \in [0, P-1]$.

On calcule la valeur de hachage pour un mot $s$ en utilisant la formule suivante :\\

$h(s) = (s[1] \otimes t) \oplus (s[2] \otimes t^{2}) \oplus ... \oplus (s[m] \otimes t^{m})$\\

où l'addition et la multiplication sont fait modulo $P$, (les caractères d'une chaine $s$ sont considérés comme des nombres entiers dans l'intervalle $[0, P-1]$).

Une fois qu'on a calculé les valeurs de hachage liées à tous les mots de $S$, on vérifie si on a des collisions entre les valeurs du hachage générées pour les clefs (les mots) de l'ensemble $S$.
Si c'est le cas, on répète le calcul de l'ensemble des valeurs de hachage en utilisant une nouvelle valeur $t$ aléatoirement choisie. Ce traitement est répété jusqu'à obtenir un ensemble sans collision.
Le temps total prévu pour générer les valeurs de hachage \textbf{sans collisions} est $O(n)$.

Une fois que l'on a mappé tous les mots de $S$ aux nombres distincts, on utilise ces nombres comme clefs pour construire la $mphf$. On utilise une fonction de hachage qui mappe ces $d$ valeurs vers un tableau de taille de $d$ éléments.

\paragraph{Comment stocker les mots dans le dictionnaire :}
Le détail final dans la construction du dictionnaire exact basé sur $mphf$, est comment stocker les mots dans le dictionnaire.

Les mots courts et longs sont traités différemment. Les mots courts (de longueur $m \leq w$) sont stockés tels quels sans modification dans le dictionnaire.

Pour un mot long $s$ de taille $m>w$, on stocke un mot modifié $s'$ de longueur $3m$ caractères. On divise le mot $s'$ en trois parties consécutives $s^{'}_{1}$, $s^{'}_{l}$ et $s^{'}_{r}$, chacun de longueur de $mb$ bits.

Le mot $s'_1$ va contenir une copie du mot $s$. Il reste les deux parties $s'_l$ et $s'_r$, on va expliquer comment les composer.

À cet effet, premièrement on met une valeur $u= \lceil log(n)/b \rceil$ qui représente un bloc (on peut le considérer comme un mot mémoire).

On considère le mot $s'_l$ et $s'_r$ comme des tableaux de $m' = \lfloor m/u \rfloor$ éléments avec $ub$ bits pour chacun (noter que $log(n) \leq ub < log(n)+b$ ), on ignore les bits de remplissage (le dernier $mb-m'ub$ bits).
Donc chaque tableau va être décomposé à un nombre de cases $m'$, tel que $m'$ représente le nombre de blocs ou les parties dans le mot $s$, (si on considère que $u=4$ caractères et on a un mot de 12 caractères, alors on aura 3 partie, $m'=3$ ). Les éléments du tableau de $s'_l$ et $s'_r$ sont numérotés en commençant par $1$.

On met $s'_l[i]$ à la valeur $L[ui]$ (l'identifiant du n\oe{}ud atteint à l'étape $u \times i$ en traversant le Trie $Tr$ pour le mot $s$), et on met $s'_r[i]$ la valeur $R[ui]$ (l'identifiant du n\oe{}ud atteint à l'étape $u \times i$ en traversant le Trie inversé $\overline{Tr}$ pour le mot $\bar{s}$) ($i \in [1..m']$). Cela signifie que dans chaque case on met l'identifiant correspond au longueur de préfixe du mot $i \times u$ (exemple: soit un mot de longueur 12 et $u=4$, donc lorsque $i=1$, on met l'identifiant des 4 premiers caractères de mot, ensuite, $i=2$, donc la partie avec 8 caractères, $i=3$ donc la partie avec 12 caractères).

Bien évidemment, on stocke les mots (les mots courts non modifiés, et les mots longs qui sont modifiés) dans un tableau contigu dans l'ordre donné par la $mphf$, et on utilise la structure de données prefix-sum pour stocker l'emplacement de chaque mot.

\paragraph{Utiliser la structure de données prefix-sum :} 
Tous les mots du dictionnaire sont rangés dans un tableau contigu l'un après l'autre selon l'ordre donné par $mphf$. Chaque mot court $i$ occupe juste sa taille qui est $|m_i|$, et chaque mot long $j$ occupe trois fois sa taille $3 \times |m_j|$. 

Noter que les mots courts et les mots longs peuvent avoir des longueurs différentes. Cela veut dire que le tableau qui stocke ces mots n'est pas divisé en blocs de même taille. Le problème est comment ranger ces mots dans le tableau, et le plus important comment les retrouver par la suite?
Pour cela, on utilise une structure de données supplémentaire appelée \textbf{prefix-sum} dans le but de garder la position donnée au mot par $mphf$ et retrouver sa position dans la mémoire.

On utilise un tableau intermédiaire $Tab\_1$ avant de créer le tableau prefix-sum. Pour chaque mot on récupère sa position depuis le tableau de hachage parfait minimal calculé précédemment, et on met sa taille dans la case correspondante à l'indice donnée par $mphf$ dans $Tab\_1$.

Comme cela, la position du mot $s_i$ dans la mémoire est après tous les mots ($s_j$, avec $j \in [1..i-1]$) précédents, donc la somme de toutes les tailles des mots précédents dans le tableau. Donc $Pos(s_i) = \sum_{j=0}^{j < i} |s_j|$.

On construit le tableau prefixe-sum  qui va contenir pour chaque case de $Tab\_1$ la somme de toutes les tailles précédentes.\\

Lorsqu'on cherche un mot $q$, on calcule son hachage parfait minimal pour trouver la position dans le tableau prefix-sum, ensuite, on récupère sa position dans la mémoire. La taille de mot dans la mémoire est la différence entre les deux cases successives de tableau prefix-sum, si l'indice de mot est $i$ donc sa taille est le contenu de la case $i+1$ moins le contenu de la case $i$.

%% on peut ici mettre un exemple. et aussi avant cela dans la partie fonction de hachage.

\subsubsection{Le dictionnaire des listes de substitution}
\label{sub:djamal_dictionnaire_des_listes_de_substitution}

Le but de cette structure de données est de réduire le nombre de possibilité des mots candidat pour l'insertion et la substitution de $O(m\sigma)$ à juste $O(m)$ avec $m$ est la taille de mot requête $q$, donc on élimine toutes les combinaisons des caractères possibles de la solution naïve qui est $\sigma$ de chaque position.

Soit un ensemble $S^{'} = \{x_1, x_2,..., x_{d'}\}$ de $d'$ clés et de $n'$ caractères au total. Soit la collection non vide $L = \{l_1, l_2,..., l_{d'}\}$ où chaque $l_i$ est une liste d'éléments (caractères) qu'on l'associe avec une clef $x_i$ de $S'$. Le nombre total d'éléments stocker dans toutes les listes est $n'$ où chaque élément est de taille $b$ bits.

Le but est de construire une structure de données de taille $O(n'b)$ bits qui supporte les opérations suivantes dans un temps constant :

\begin{itemize}

\item $list\_element(x, i)$, cette opération retourne l'élément numéro $i$ (les éléments de la liste sont numérotés à partir de 1) de la liste associer avec la clef $x$ si $x \in S'$. L'élément retourné est non défini si $x \notin S'$.

\item $size\_of(x)$, cette opération retourne le nombre d'éléments dans la liste associée avec $x$ si $x \in S'$. Le résultat de l'opération est non défini si $x \notin S'$.

\end{itemize}

La solution est basée sur la structure de données prefix-som et le hachage parfait minimal, et elle supporte les deux opérations précédentes dans un temps constant. La solution a une certaine similarité avec la solution du dictionnaire exact avec hachage parfait minimal décrite précédemment.

L'idée est de stocker tous les caractères possibles qui peuvent donner des solutions approchées pour l'insertion et/ou la substitution dans une position d'erreur $i$. 
Soit un mot $x$ du dictionnaire qui a un préfixe $u$ et un suffixe $v$ et un caractère au milieu $c$, donc $x=u c v$. Si on suppose qu'on a un mot requête $q = u \phi v$, et $\phi$ la position de l'erreur (on prend le cas de la substitution). Le mot $x = u c v$ est une solution approchée, on substitue le caractère $c$ par $\phi$. Ainsi, on a juste un seul caractère possible pour cette position. Si dans le dictionnaire on a un autre mot de la même forme $y = u c_2 v$, alors ce mot aussi représente une $2\ieme{}$ solution (on substitue le caractère $c_2$ par $\phi$), dans ce cas on a deux caractères qui donnent deux solutions. Il y a des cas où on a juste une solution, et des cas où on trouve plusieurs solutions. Ces caractères vont être stockés dans un tableau contigu, l'ordre est donnée par une fonction de hachage parfait minimal, les caractères qui donnent des solutions à la même requête dans la même position (comme l'exemple précédent $c$ et $c_2$) sont rangés dans la même place et donc il vont représenter une liste de solutions. Pour trouver ces caractères lors de la recherche, on utilise la structure de données prefix-sum afin de trouver la position dans la mémoire donc le début et la fin de chaque liste de caractères qui donne des solutions pour la même position.\\

La première étape consiste à construire la fonction de hachage parfait minimal sur l'ensemble $S^{'}$. On note par $x_i$ l'élément mappé avec $mphf$ de $S^{'}$ vers la position $i$. On stocke les éléments de chaque liste d'une façon contiguë dans un tableau $T'$. La position d'élément numéro $j$ de la liste qui est associée à la clef $x_i$ est donnée par :\\

$Pos(i, j) = (\sum_{k=0}^{k < i} |x_k|) + j$ \\

Le parcours de la liste liée à une clef $q$, commence par le calcul de $mphf$ qui donne un numéro $i$. Ensuite, on interroge la structure de données prefix-sum qui donne la position dans la mémoire $pos_i$. L'élément numéro $j$ de la liste associée avec la clef $q$ est dans la position $T'[pos_i +j]$.

Pour avoir le nombre d'éléments qui est dans la liste associée avec la clef $q$, on interroge simplement la structure de données prefix-sum avec les numéros $i$ et $i+1$ pour nous donner $pos_i$ et $pos_{i+1}$. Le nombre d'éléments dans la liste associée avec $q$ est simplement $T'[pos_{i+1}] - T'[pos_i]$.

\bigskip
On détaille maintenant la construction de cette structure de données. On traite successivement chaque mot $s$ de l'ensemble $S$ (soit $m$ la longueur de mot $s$). On stocke deux tableaux temporaires de nombres entiers $L[0..m]$ et $R[0..m]$ pour chaque mot, les nombres sont dans l'intervalle $[0,n]$ avec $n$ le nombre total de tous les caractères du dictionnaire.

On met $L[0] = R[0] = 0$. En premier on parcourt le Trie $Tr$ pour le mot $s$, et on stocke dans la case $L[i]$ l'identifiant du n\oe{}ud atteint à l'étape $i$ (les étapes sont numérotées à partir de $1$). 

On fait la même chose pour générer les éléments $R[1..m]$. On parcourt le Trie renversé $\overline{Tr}$ pour le mot $\bar{s}$ (l'inverse du mot $s$) et on stocke l'identifiant du n\oe{}ud atteint à l'étape $i$ dans la case $R[i]$.

Pour chaque mot $s$ de longueur $m$ pour lequel on a calculé les tableaux $L$ et $R$, on ajoute le caractère $s[i]$ à la liste correspondante à la paire $(L[i-1], R[m-i])$ pour chaque $i$, tel que $1 \leq i \leq m$.

Donc on aura pour chaque mot une liste $\{(L[0],R[m-1],s[1]), (L[1],R[m-2],s[2]), ... , (L[m-1],R[0],s[m])\}$.

Ensuite, pour chaque paire $(L[i], R[j])$ on calcule le hachage parfait minimal. Chaque caractère $s[i]$ est stocké dans la position donnée par $mphf$ de $(L[i-1], R[m-i])$.

En fait, on insère les éléments dans le dictionnaire qui est implémenté en utilisant une table de hachage dynamique temporaire et des listes liées.

Si deux ou plusieurs mots ont la même valeur paire $(L[i], R[j])$ cela vaut dire qu'ils ont le même préfixe et le même suffixe, juste le caractère à la position $i$ qui est différent (comme on a expliqué dans un exemple précédent pour $c$ et $c_2$), donc leurs valeurs de hachages sont les mêmes, donc tous les caractères liés à la paire $(L[i], R[j])$ vont être stockés dans la même position (la même liste).

Par la suite, on peut construire le dictionnaire des listes de substitution depuis la table de hachage temporaire et les listes liées en utilisant la méthode décrite précédemment dans la construction de dictionnaire exact (avec la structure prefix-sum).

On va créer un seul tableau contigu, et on stocke les listes des éléments (des caractères) liste après liste, il y a des listes qui contiennent juste un seul élément, et d'autres qui en contiennent plusieurs.
L'ordre de ces listes est donné par $mphf$ (l'ordre des listes dans le tableau de hachage dynamique temporaire).

À la fin, on construit le tableau prefix-sum, pour cela on utilise un tableau temporaire qui va contenir dans chaque case $k$, la taille de la liste numéro $k$, ensuite, on construit le tableau prefix-sum, chaque case $k$ va indiquer le début de la liste $k$ dans la mémoire (la somme de toutes les tailles des listes précédentes sachant que chaque élément occupe $b$ bits).

%%meme chose , on peut mettre l'exemple ici

%%\subsubsection{Données satellites "Satellite data":}
%%Dans beaucoup d'applications, il se peut que on aurai besoin d'associer des données satellites à chaque mot du dictionnaire. Ceci peut facilement être fait dans notre cas, simplement en ajoutant une table de $n$ cases indexer par le mphf utilisé pour le dictionnaire.

\subsection{La recherche du mot requête}
\label{sub:la_requete_djamal}

Étant donné un mot requête $q$ de longueur $m$. 

En premier, on calcule $\bar{q}$ qui est le mot inverse de mot $q$.

Ensuite, on calcule les tableaux $L[0..m]$ et $R[0..m]$. Pour cela, on met $L[0] = R[0] = 0$.
Ensuite, on parcourt le Trie $Tr$ pour le mot $q$ et on met l'identifiant de n\oe{}ud atteint dans l'étape $i$ dans la case $L[i]$. Si la recherche s'arrête à l'étape $i$ alors on place une valeur spéciale $\bot$ dans les cases qui reste $L[j]$ pour $j \in [i + 1, m]$. De même, on remplit le tableau $R$ en parcourant le Trie inversé $\bar{Tr}$ avec le mot $\bar{q}$ et on met l'identifiant de n\oe{}ud atteint à l'étape $i$ dans la case $R[i]$. Si la recherche s'arrête à l'étape $i$ on place une valeur spéciale $\bot$ dans les cases qui reste $R[j]$ pour $j \in [i + 1, m]$.

On prépare aussi trois tableaux en plus, notés par $A_t[0, m + 1]$, $F[0..m]$ et $G[1..m + 1]$ :
\begin{enumerate}

\item Le tableau $A_t$ stocke toutes les puissances de $t$ jusqu'à $t^{m+1}$. En premier on met $A_t[0] = 1$, ensuite, on met $A_t[i] = A_t[i - 1] \otimes t$ pour chaque $i$ dans l'intervalle $[1, m+1]$.

\item Pour générer le tableau $F$, on met en premier $F[0] = 0$, ensuite, on met $F[i] = F[i-1] \oplus (q[i] \otimes A_t[i])$ pour chaque $i$ dans l'intervalle $[1..m]$. (F : tableau qui calcule les valeurs du hachage pour les préfixes de $q$).

\item Pour générer le tableau $G$, en premier on met $G[m+1] = 0$, ensuite, on met $G[i] = (G[i+1] \oplus q[i]) \otimes t$ pour chaque $i$ dans l'intervalle $[1..m]$. (G : tableau qui calcule les valeurs du hachage pour les suffixe de mot $q$).
\end{enumerate}

On a quatre types de mots dans le dictionnaire qui pourrait correspondre au mot requête :
un mot qui est à une distance $0$ (exact), et les mots du dictionnaire qui peuvent être obtenus par l'application de l'un des trois types d'erreurs de distance d'édition (suppression, insertion, et substitution).

\subsubsection{La recherche exacte}

La recherche exacte pour les mots avec une distance d'édition $0$. On fait simplement un accès au dictionnaire exact pour le mot requête $q$. Pour cela, on calcule la valeur du hachage de mot $q$.\\
On sait que $h(q) = (q[1] \otimes A_t[1]) \oplus (q[2] \otimes A_t[2]) \dots \oplus (q[k] \otimes A_t[k]) = F[k] = G[1]$. La valeur du hachage est déjà pré-calculée, lorsque on a préparé les deux tableaux $F$ et $G$.
Ensuite, avec cette valeur de hachage on calcule la valeur de hachage parfait minimal $mphf$ qui donne la position $i$ dans le tableau prefix-sum. On récupère la position d'accès au dictionnaire exact (le tableau contigu qui stocke les mots du dictionnaire). La position se trouve dans la case de tableau prefix-sum numéro $i$. La taille de mot dans la mémoire est calculée par le contenu du tableau prefix-sum $i+1$ moins le contenu de la case à la position $i$.

Si le mot trouvé est un mot court (sa longueur est $<w$), soit ce mot court $s$. Donc on le compare directement avec le mot requête $q$ , la comparaison prend un temps de $O(1)$, car la vérification d'un seul mot mémoire $w$ ce fait d'un seul coup. Si les deux mots sont égaux ($s=q$) alors on retourne $s$ comme un résultat exact.

Dans le cas contraire, on a un mot long, sa longueur est au moins $3w$. Il faut se rappeler que dans la phase de construction, si le mot du dictionnaire $s$ est un mot court on le stocke tel quel, et si sa longueur est $> w$, on le traite et on le stocke dans un nouveau mot $s'$ qui occupe un espace de 3 fois sa taille. Dans ce cas on a un mot $s'$ qui contient trois parties $s'_1$, $s'_l$ et $s'_r$, le mot est stocké dans la première partie $s'_1$.

On peut comparer directement $q$ avec $s'_1$  dans un temps $O(m)$ et si on trouve qu'ils sont égaux on retourne $s'_1$ comme résultat exact. Mais le but est de faire cette comparaison dans un temps de $O(1)$.
Pour cela, on calcule le nombre de blocs dans le mot requête $q$ donc $l_q'= \lfloor (m-1)/u \rfloor$. Ensuite, on vérifie les conditions suivantes :

\begin{itemize}

\item La longueur de $s'$ est $3m$. Car $s'$ contient trois parties égaux. Cette vérification prend un temps de $O(1)$.

\item $s'_l[l_q'] = L[u \times l_q']$, vérifier si l'identifiant de mot $s'_1$ (pour les caractères qui constituent $l_q'$ blocs) est le même que l'identifiant de mot $q$ pour le même nombre de blocs. Les deux identifiants sont récupérés du Trie, et donc on a une égalité si et seulement si les caractères sont égaux. Cette opération prend un temps de $O(1)$.

\item $q[u \times l_q' + 1..m] = s'_1[u \times l_q' + 1..m]$, tester les caractères qui restent qui ne sont pas dans les blocs $u \times l_q'$, le nombre de ces caractères est moins qu'un bloc $u$, donc la comparaison est en $O(1)$.

\end{itemize}

Le pré-traitement des tableaux $A_t ,\, F ,\, G ,\, L ,\, R$ est en $O(m)$ chacun. Si on considère qu'on a que la recherche exacte seule, le temps de recherche est en $O(m)$.

Le temps de comparaison est en $O(1)$ car on ne compte pas le temps de pré-traitement qui est fait spécialement pour qu'on puisse faire la comparaison des différents mots candidats dans la recherche approchée en $O(1)$.

\subsubsection{la recherche approchée}

\paragraph{L'erreur de type insertion :} on commence par l'erreur de type insertion, les mots obtenus par l'insertion d'un seul caractère dans $q$.

On va faire $m + 1$ étapes pour $i \in [0, m]$. À chaque étape $i$ on fait un accès au dictionnaire des listes de substitution pour trouver la liste associée avec la paire $(L[i],R[m-i])$ si et seulement si $L[i] \neq \bot$ et $R[m-i] \neq \bot$ (on suppose qu'on a inséré un caractère après la position $i$).

Pour trouver cette liste il suffit juste de calculer le hachage parfait minimal de la paire $(L[i],R[m-i])$ pour avoir une position dans le tableau prefix-sum qui nous donne la position de la liste dans la mémoire (et on déduit aussi sa taille).

Pour vérifier la validité de la liste retournée, on doit seulement vérifier son premier élément. Si le premier élément est valide, alors on conclut que la liste existe vraiment et que tous les éléments restants sont également valides.

Soit $c$ le premier élément (caractère) de la liste associée avec la paire $(L[i],R[m-i])$.
On doit faire un accès au dictionnaire pour vérifier l'existence de mot $q' = q[1..i] \,c\,q[i+1..m+1]$. Les étapes sont les mêmes que celles de la recherche exacte.

On a $h(q') = F[i] \oplus  (c \oplus G[i + 1]) \otimes A_t[i+1]$. On utilise cette valeur de hachage $h(q')$ afin de calculer le $mphf$, ensuite, on utilise la position retournée par $mphf$ pour accéder au dictionnaire à travers le tableau prefix-sum.

Si le mot retourné est un mot court (donc $s$), on peut directement le comparer à $q'$ dans un temps de $O(1)$, et on dit qu'on a trouvé une solution dans le cas où ils sont égaux.

Si c'est un mot long (donc $s'$ de longueur au moins $3w$), on divise $s'$ en trois parties égaux $(s'_1), (s'_l)$ et $(s'_r)$.
Ce dont on a besoin maintenant est de pouvoir comparer les mots $s'_1$ et $q'$, on peut le faire d'une façon naïve dans un temps $O(m)$. Cependant, on peut comparer $s'_1$ et $q'$ dans un temps de $O(1)$.

Pour pouvoir comparer $s'_1$ et $q'$ dans un temps constant, on utilise les parties $s'_l$ et $s'_r$ qui sont considérées comme des tableaux d'éléments de longueur $u \times b \geq log(n)$ bits chacun.  $u= \lceil log(n)/b \rceil$, $u$ représente un bloc qu'on peut vérifier d'un seul coup (un bloc de $u$ caractères représente un mot mémoire $w$).

On initialise la variable $l_q'$ par le nombre de blocs qu'on a dans le préfixe $q'[0..i]$ ($l_q' = \lfloor i/u \rfloor $), et de même, on initialise $r_q'$ par le nombre de blocs qu'on dans le suffixe $q'[m-i..m]$ ($r_q'= \lfloor (m-i)/u \rfloor $).

On retourne un match si et seulement si toutes les conditions suivantes sont réunies :

\begin{itemize}

\item La longueur de $s'$ est $3(m + 1)$, donc les tailles du mot requête et $s'$ sont les mêmes. La taille du mot requête est $m$, on met +1 car on a une insertion.
 
\item $l_q' = 0$, dans le cas où l'insertion est au début donc on a que la partie suffixe après l'insertion, \\
ou\\
$s'_l[l_q'] = L[u \times l_q']$\footnote{D.Belazzougui dans son article  \cite{Be09}, à fait une petite erreur d'indice dans la comparaison, il compare  $s'_l[l_q'] = L[l_q']$ au lieu de $s'_l[l_q'] = L[u \times l_q']$ (donc il a mis $L[l_q']$ au lieu de $L[u.l_q']$). Car dans le tableau $s'_l$ on a les identifiants multiples de $u$ (donc $i\, \times \,u$), alors que dans le tableau $L$ on a les identifiants de chaque position $i$.}.
Ce test est utilisé afin de vérifier si le préfixe avant la position de l'insertion ($u \times l_q'$ caractères de $s'_1$) et le préfixe de $q'$ sont les mêmes.  Ce test prend un temps de $O(1)$.
Il faut se rappeler que $s'_l$ est un tableau qui contient les identifiants des parties multiples de $u$ de $s'_1$ dans le Trie, et $L$ aussi contient tous les identifiants de mot requête pour chaque position, mais qui ne sont pas multiples de $u$.

On peut facilement voir que $s'_l[l_q'] = L[u \times l_q']$ est vrai si et seulement si l'identifiant retourné du Trie $Tr$ est identique pour le premier préfixe ($u \times l_q'$ caractères) de $s'_1$ et le premier préfixe $u \times l_q'$ de $q$, cela est vrai seulement si leurs caractères de ces deux parties sont les mêmes.

\item $q[u \times l_q' + 1..i] = s'_1[u \times l_q' + 1..i]$. Ce test se fait dans un temps de $O(1)$ car on compare deux mots qui ont une longueur au max $(u-1)b = O(w)$ bits.
Ce test est fait pour couvrir la partie des caractères qui reste jusqu'à la position $i$ et qui ne rentre pas dans les blocs $u \times l_q'$.\\
Exemple : si on a un préfixe de taille $10$ et $u=4$ , dans ce cas $l_q'=10/4=2$, donc on a 2 blocs qui couvrent 8 caractères $(u \times l_q'= 4 \times 2)$, ces blocs on les teste par $s'_l[l_q'] = L[u \times l_q']$, et il reste 2 caractères, donc on les teste séparément.

\item $s'_1[i + 1] = c$. Ce test prend clairement un temps de O(1). (le test de caractère  de l'insertion).

\item $q[i + 1..m - u \times r_q'] = s'_1[i + 2..m + 1 - u \times r_q']$. Ce test prend un temps de O(1) car on compare deux mots de taille $O(w)$ bits. Le test des caractères qui ne rentre pas dans les blocs de suffixe de $[m-i..i]$.

\item $r_q' = 0$ si l'insertion est à la fin, donc on a que le préfixe,\\
ou\\
$s'_r[r_q'] = R[u.r_q']$. Ce test est fait pour vérifier si les derniers $u \times r_q'$ caractères de $s'_1$ et $q$ sont les mêmes. Ce teste prend un temps de $O(1)$.
$s'_r[r_q'] = R[u \times r_q']$ est vrai si et seulement si l'identifiant retourné du Trie inversé $\overline{Tr}$ est identique pour les derniers $u \times r_q'$ caractères (suffixe) de $s'_1$ et les derniers $u \times r_q'$ caractères de $q$, et cela est vrai seulement si leurs caractères sont les mêmes.

\end{itemize}

\medskip
Noter qu'à chaque étape $i$, on a besoin de vérifier seulement le mot obtenu en insérant à la position $i$ le premier caractère de la liste retournée par le dictionnaire des listes de substitution.

Si on a une égalité, on peut continuer à récupérer les mots obtenus en insérant les caractères restants de la liste à la position $i$, et on est sûr d'avoir une égalité pour les mots obtenus.
 
La récupération de chaque élément supplémentaire de la liste prend un temps de O(1).

La vérification pour les deux types d'erreurs (la substitution et la suppression) prend aussi un temps de $O(m)$. La procédure pour vérifier ces erreurs est semblable à la procédure pour vérifier l'insertion. L'explication est décrite dans le paragraphe suivant.
La récupération des mots valides pour chaque erreur prend un temps supplémentaire de O(1) par élément. Ainsi le temps total de la requête est de $O(m + occrs)$.

\paragraph{L'erreur de type substitution :}

Le cas de la substitution est très semblable à celui de l'insertion.
En premier, on récupère une liste de caractères depuis le dictionnaire des listes de substitution pour chaque paire $(L[i-1],R[m-i])$ pour une position donnée $i \in [1, m]$, avec $L[i-1] \neq \bot$ et $R[k-i] \neq \bot$.

Soit $c$ le premier élément de la liste associé avec $(L[i-1],R[m-i])$, et soit $q' = q[1..i-1]\,c\,q[i + 1..m]$. On a $h(q') = F[i-1] \oplus (c \oplus G[i + 1]) \otimes A_t[i]$. 

On fait un accès au dictionnaire en utilisant $h(q')$. Si le mot retourné est un mot court (donc $s$), on le compare directement avec $q'$ en $O(1)$. Autrement (on a un mot long $s'$), on décompose $s'$ en trois parties égaux $s'_1$,$s'_l$ et $s'_r$. Et on compare $s'_1$ avec $q'$. Pour cela, on initialise les deux variables $l_q' = \lfloor (i-1)/u \rfloor $, et $r_q'= \lfloor (m-i)/u \rfloor $. On retourne une égalité si et seulement si les conditions suivantes sont réunies :

\begin{itemize}

\item La longueur de $s'$ est $3m$.

\item $l_q' = 0 \vee s'_l[l_q'] = L[u \times l_q']$.

\item $q[u \times l_q' + 1..i - 1] = s'_1[u \times l_q' + 1..i - 1]$.

\item $s'_1[i] = c$.

\item $q[i + 1..m - u \times r_q'] = s'_1[i + 1..m - u \times r_q']$.

\item $r_q' = 0 \vee s'_r[r_q'] = R[u \times r_q']$.
\end{itemize}

Il est clair que la vérification de chaque condition prend un temps de $O(1)$. Pour plus de détail sur ces conditions, voir la partie précédente (type d'erreur insertion).

De même dans le cas de l'insertion, on doit faire la vérification seulement pour le premier élément de la liste. Si on n'a pas une égalité pour le premier élément, on conclut que la liste n'existe pas et on arrête immédiatement. Autrement, on retourne le premier élément et on continue à récupérer les éléments restants de la liste sans les vérifier dans un temps $O(1)$ par élément, et on retourne chaque élément comme une égalité (on fait la substitution dans la position $i$ pour chaque élément).

\paragraph{L'erreur de type suppression :}

Les mots obtenus en supprimant un caractère de $q$. Dans ce cas, on essaye d'accéder au  dictionnaire exact avec chaque mot candidat possible qui peut être obtenu en supprimant un des caractères de $q$ à la position $i$. Cela revient à faire une recherche exacte pour chaque mot candidat. On a exactement $m$ mots candidats, qui signifie qu'on doit passer seulement un temps de $O(1)$ pour chaque accès pour vérifier s'il y a une égalité ou non.

Si on souhaite accéder au dictionnaire pour le mot $q'$ obtenu en supprimant le caractère numéro $i$ de $q$ ($q'= q[1..i-1]q[i + 1..m]$), on fait les étapes suivantes :

En premier, on vérifie que $L[i-1] \neq \bot$ et $R[m-i] \neq \bot$. Si ce n'est pas le cas, on conclut immédiatement qu'on a pas une égalité pour le mot $q'$.

Dans le cas contraire, on calcule la valeur de hachage $h(q') = F[i-1] \oplus G[i+1] \otimes At[i-1]$ dans un temps constant. Ensuite, on calcule la valeur $mphf$. En utilisant la position retournée par $mphf$, on fait la requête au dictionnaire exact, qui va retourner le mot $s$ si c'est un mot court, ou $s'$ dans le cas contraire.

Si le mot est un mot court, donc on peut faire la comparaison directement de $s$ et $q'$ dans un temps de $O(1)$ (avant de comparer les mots, on compare leur longueur) et on retourne une solution si les deux mots sont égaux.

Si le mot est long, on divise $s'$ en trois parties égaux $s'_1, s'_l$ et $s'_r$. On retourne une solution si $s'_1 = q'$. Pour faire cette vérification on met $l_q' = \lfloor  (i-1)/u \rfloor$ et $r_q' = \lfloor (m-i)/u \rfloor$, et on vérifie toutes les conditions suivantes :

\begin{itemize}

\item La longueur de $s'$ est $3(m-1)$.

\item $l_q' = 0 \vee s'_l[l_q'] = L[u \times l_q']$.

\item $q[u \times l_q' + 1..i - 1] = s'_1[u \times l_q' + 1..i - 1]$.

\item $q[i + 1..m - u \times r_q'] = s'_1[i..m - 1 - u \times r_q']$.

\item $r_q' = 0 \vee s'_r[r_q'] = R[u \times r_q']$.

\end{itemize}

\subsubsection{Le comptage du nombre des solutions d'une requête}

Le comptage du nombre d'occurrences de la solution peut être fait dans un temps de $O(m)$. Ceci est fait par la somme des mots solutions de chaque type d'erreur.
Pour la suppression, on peut avoir au maximum $m$ mots solutions, et la vérification de chaque mot candidat prend un temps de $O(1)$. Pour l'insertion et la substitution on a besoin de $m + 1$ et $m$ étapes respectivement.
Dans chaque étape on a besoin seulement de vérifier le premier élément de la liste retournée par le dictionnaire des listes de substitution, et si cet élément est valide, on ajoute la taille de la liste (obtenue par l'utilisation de l'opération $size\_of$) au nombre total des mots solutions $occrs$.

\subsection{La recherche approchée dans un texte indexé avec une distance d'édition $1$}

La structure de données de la recherche approchée dans le dictionnaire avec une distance d'édition "1" peut-être utilisée pour avoir une solution pour la recherche approchée dans le texte. Pour cela on combine cette solution avec celle décrite dans \cite{CGL04}. Donc on construit l'index de \cite{CGL04}.

On sélectionne toutes les sous-chaines de texte de longueur jusqu'à $log(n) log log(n)$ caractères, et on les stocke dans la structure de données (dictionnaire exact + dictionnaire des listes de substitution) qui va stocker au plus $n(log(n) log log(n))$ chaines de caractères.
Les chaines qui apparaissent plus qu'une seule fois dans le texte sont stockées seulement une fois dans la structure de données.
On associe avec chaque chaine (dans le dictionnaire exact) en tant que données satellites\footnote{Données satellites : des données supplémentaires qui ne sont pas une partie de la structure de données.} dans le dictionnaire un pointeur qui pointe vers un vecteur qui stocke tous les positions où la chaine apparaît dans le texte.
  
Chaque pointeur de vecteur occupe $log(n)$ bits, et chaque chaine occupe au maximum  $log(n) log log(n)$ caractères, cela rend l'espace total utilisé par le dictionnaire $O(n(log(n) log log(n))^{2}b)$ bits.

On peut stocker tous les vecteurs qui contiennent les positions des chaines dans un seul tableau contigu de $O(n log(n) log log(n))$ pointeurs, qui occupe un espace total de
 $O(n log(n)^{2} log log(n))$ bits. On fait la somme de l'espace utilisé par tous les composants de la structure de données , on obtient un espace total de $O(n(log(n) log log(n))^{2}b)$ bits.\\

Pour faire une requête avec une chaine de caractères $q$, on vérifie simplement si $|q| > log(n) log log(n)$ dans ce cas on utilise la structure de données de \cite{CGL04} pour répondre à la requête. Dans le cas contraire, on utilise cette structure de données pour répondre à la requête.

\section{La méthode de Amir et al.}

Amir \textit{et al.} \cite{Amir2000} proposent une solution pour résoudre le problème de la recherche approchée avec $k=1$ erreur dans un dictionnaire et dans un texte avec l'utilisation du Trie et du Trie inversé. Avec une seule erreur on a deux parties exactes dans le mot requête, ces deux parties sont vérifiées avec les deux Tries, ensuite une étape de vérification est faite afin de trouver les solutions approchées. Ils obtiennent un résultat dans un temps de pré-traitement de $O(n\,log^{2}n)$ où $n$ est la taille du texte dans le problème de l'indexation, et la taille de dictionnaire dans le problème de la recherche dans un dictionnaire. Pour le temps de la requête, ils obtiennent un temps de recherche de $O(m\,log^{2}n\,log\,log\,n + occrs)$ dans l'indexation de texte où $m$ est la longueur de mot requête et $occrs$ est le nombre des occurrences de la solution. Pour la recherche dans un dictionnaire ils obtiennent un temps de $O(m\,log^{3}d\,log\,log\,d + occrs)$ avec $n$ est la taille de texte et $d$ est la taille du dictionnaire.

Nous commençons par expliquer la construction de la structure de données bidirectionnelle (Trie et Trie inversé), ensuite, nous donnons l'algorithme de la recherche pour un mot requête de longueur $m$.

\subsection{La construction de la structure de données}

La construction de la structure de données se fait comme suit :

\begin{enumerate}
\item Construire le Trie du dictionnaire et enregistrer les numéros des mots dans les feuilles.
\item Construire le Trie inversé du dictionnaire et enregistrer les numéros des mots dans les feuilles.
\item Dans chaque structure de données, relier toutes les feuilles de la gauche vers la droite, pour obtenir une liste de numéros de mots.

\item Enregistrer dans chaque n\oe{}ud le pointeur de la feuille la plus à gauche et la plus à droite, donc le début et la fin de la liste des sous-listes qui contiennent tous les enfants sortant de ce n\oe{}ud.

\end{enumerate}

La construction de chaque Trie prend $O(n)$, où $n$ est le nombre total de caractères dans le dictionnaire.

\subsection{La recherche de mot requête }

Soit un mot requête $q=c_1\, c_2 \, ... \,c_{m}$ de longueur $m$.

La étapes pour rechercher toutes les solutions approchées sont comme suit :\\

\noindent
\textbf{Pour $i=1,...,m$  faire :}\\
 $i$ représente la position de l'erreur($\phi$) $q'=c_1\, c_2 \, ... \,c_{i-1} \, \phi \, c_{i+1}\,c_{i+2}\,...\,c_m$.

\begin{enumerate}
\item Trouver le n\oe{}ud $v$, qui représente l'emplacement de  $\{c_1,c_2,...,c_{i-1}\}$ dans le Trie, si ce n\oe{}ud existe.

\item Trouver le n\oe{}ud $w$, qui représente l'emplacement de  $\{c_{i+1},c_{i+2},...,c_m\}$ dans le Trie inversé, si ce n\oe{}ud existe.

\item Si $v$ et $w$ existent, alors trouver l'intersection entre les deux sous-listes enracinés par les deux n\oe{}uds, (l'intersection des numéros qui sont stockés dans les feuilles des deux n\oe{}uds $v$ et $w$).
\end{enumerate}

\paragraph{Comment faire l'intersection :}
Pour une solution efficace au problème de l'intersection entre les deux ensembles des nombres qui sont dans les feuilles de deux n\oe{}uds $v_{err}$ et $w$,
Amir \textit{et al.} utilisent la technique expliqué dans \cite{overmars1988efficient}.

\chapter{La recherche approchée dans un dictionnaire en utilisant des tables de hachage}

\label{chap:recherche_approchee_avec_hachage}

\ifpdf
    \graphicspath{{Work_with_djamal/}}
\else
    \graphicspath{{Work_with_djamal/}}
\fi

%\section*{Keywords}
%recherche approchée, dictionnaire, distance d'édition, hachage.

\section{Introduction}
\label{sec:Introduction}

Le problème de la recherche approchée dans les dictionnaires est largement étudié. Le problème est défini comme suit:
Étant donné (comme entrée donnée à l'avance) un dictionnaire $D=\{x_1,x_2,\ldots x_d\}$ avec $d$ mots de longueur totale $n$ sur un alphabet $\Sigma$ de taille $\sigma$~\footnote{Pour des raisons pratiques, on suppose que l'alphabet est le domaine entier $[1..\sigma]$.} et un seuil $k$. On veut construire une structure de données sur $D$ d'une manière à être en mesure de répondre aux requêtes suivantes: étant donné un mot requête $q$ de longueur $m$, retourner tous les mots du dictionnaire à une distance d'édition au plus $k$ du mot requête $q$.

Dans ce chapitre, nous nous intéressons aux algorithmes pratiques pour la recherche approchée des motifs dans un dictionnaire avec un nombre d'erreurs d'édition $k \geq 2$. 

Nous avons proposé une méthode qui utilise une structure de données basée sur le hachage avec sondage linéaire et des signatures de hachage associées afin d'optimiser le temps de recherche. Pour le dictionnaire exact, nous utilisons une structure de données basée sur plusieurs tables de hachage selon la longueur des mots. 

Dans notre travail nous utilisons l'idée de dictionnaire des listes de substitutions qui a été proposée dans~\cite{Be09}. Pour plus de détail voir dans le chapitre état de l'art, la sous-section \ref{sub:djamal_dictionnaire_des_listes_de_substitution}.

Notre solution pratique utilise un espace mémoire de $O(n\log\sigma)$ bits, et un temps de requête proche de $O(m+occrs)$ et plus exactement en $O(m\left\lceil\frac{m\log\sigma}{w}\right\rceil)$ où $w$ représente la taille d'un mot mémoire, et $occrs$ est le nombre d'occurrences trouvés. Le temps de construction est de $O(n)$. Nous montrons que la performance de notre solution est en $O(m\left\lceil\frac{m\log\sigma}{w}\right\rceil)$ sous des hypothèses qui sont susceptibles d'être vérifiées en pratique.
Nous vérifions expérimentalement que notre solution est compétitive avec les solutions précédemment proposées et que nos hypothèses sont réalistes (elles sont généralement vraies en pratique).

\medskip
\paragraph{Le reste de ce chapitre est organisé comme suit : }
Dans la section \ref{sec:structure_de_donne}, nous détaillons les étapes de la création de notre structure de données. Dans la section \ref{sec:Verification_des_Occurrences}, nous expliquons comment faire la recherche exacte et approchée avec k=1 erreur. Dans la section \ref{sec:recapitulatif_des_performances_de_notre_structure_de_donnees}, nous donnons un récapitulatif des performances de notre structure de données.
Dans la section \ref{sec:Extension_deux_erreurs_ou_plus}, nous expliquons comment étendre notre algorithme pour fonctionner avec deux erreurs ou plus.
La section \ref{sec:hash:experimentation} est dédiée aux tests et aux expérimentations. 
Dans la section \ref{sec:hash:application_notre_methode_dans_texte}, nous expliquons comment adapter notre méthode dans l'indexation de texte.
La dernière section \ref{sec:auto:conclusion} conclut ce chapitre.

\section{La structure de données}
\label{sec:structure_de_donne}

Dans notre travail, nous utilisons seulement deux ($2$) structures de données afin de minimiser le nombre d'accès mémoire. 
\begin{enumerate}
\item Un dictionnaire exact basé sur le hachage. 
\item Un dictionnaire des listes de substitutions implémenté avec une structure de données basée sur le hachage avec sondage linéaire et des signatures de hachage.
\end{enumerate}

Dans le but d'optimiser l'espace mémoire, nous compressons ces deux structures de données, en utilisant pour cela une autre structure de données (le vecteur des bits). La version non compressée est progressive, donc on peut rajouter des nouveaux mots au dictionnaire (sous réserve de ne pas dépasser la capacité maximale).\\

Dans notre travail, nous utilisons le hachage avec sondage linéaire~\cite{Knuth63noteson} pour implémenter les deux dictionnaires (index) exact et approché.
La raison pour laquelle nous utilisons ce type de hachage est qu'il a été démontré qu'il est parmi les méthodes de hachage les plus pratiques.
La bonne performance est essentiellement due a la bonne localité des accès mémoire~\cite{HL05}.

Tous nos tableaux de hachage avec sondage linéaire sont paramétrés avec un facteur de chargement (LF: LoadFactor) $\alpha<1$. Pour plus de simplicité, nous utilisons le même paramètre pour tous les tableaux de hachage.

Dans notre travail, nous utilisons une fonction de hachage polynomiale, la fonction de hachage de Rabin-Karp~\cite{KR87}) modulo un nombre premier $P$.
Les valeurs de hachage sont calculées modulo $P = 2^{32}-5$, le plus grand nombre premier qui est plus petit que $2^{32}$.
La fonction de hachage est très simple, on utilise un nombre $t$ choisi aléatoirement dans l'intervalle $[1..P-1]$.
La valeur de hachage pour une chaîne de caractère $x$ est calculée par la formule suivante: $h(x)=\sum_{i=1}^{m}x_i\cdot t^i$ \footnote{Il y a une petite erreur dans notre article \cite{ibra.Chegrane.simple}, nous avons mis $m$ à la place de $i$ et donc la formule suivante $h(x)=\sum_{i=1}^{m}x_m\cdot r^m$ est fausse.}.

La fonction de hachage polynomiale a comme propriété intéressante d'être incrémentale. On peut dans un temps $O(m)$ pré-traiter la chaîne $x$, en calculant un certain nombre de vecteurs d'entiers en temps $O(m)$, de telle sorte que le calcul de la valeur de hachage de n'importe quelle chaîne à une distance $1$ de $x$, prend un temps constant (en utilisant un nombre constant d'additions et de multiplications modulo le nombre premier $P$).
Voir les détails dans l'état de l'art dans la sous-section~\ref{sub:la_requete_djamal}, et dans ce chapitre section \ref{sec:Verification_des_Occurrences}.\\

Notre algorithme permet de traiter deux erreurs d'édition ou plus ($k \geq 2$). Dans les prochaines sous-sections, on donne plus de détails sur notre structure de données.

\subsection{Dictionnaire exact}
\label{sub:Dictionnaire_exact}

Le dictionnaire exact (basé sur le hachage avec sondage linéaire) est utilisé pour localiser les occurrences exactes et permet de vérifier si le motif candidat est dans le dictionnaire ou pas.

Dans notre travail, nous utilisons une seule structure de données pour la vérification exacte des mots, cette structure de données étant constituée de plusieurs tables de hachage selon les longueurs des mots du dictionnaire.

L'insertion d'un mot de longueur $m$ dans le dictionnaire exact prend un temps moyen de $O(m)$, et la recherche prend aussi un temps moyen de $O(m)$ si on applique la méthode naïve. Avec l'utilisation du \textit{bit-parallelism}, la vérification peut être améliorée en un temps de $O(\lceil\frac{m\log\sigma}{w}\rceil)$.

Nous notons que le nombre $t$ de longueurs de mots distinctes (les mots avec une longueur $l_1, l_2,\dots l_t$) dans le dictionnaire ne peut excéder $2\sqrt{n}$. Ceci peut être aisément prouvé comme suit. En premier lieu considérons 
les longueurs de mots excédant $\sqrt{n}+1$. On note que le nombre total des mots de longueur au moins égales à $\sqrt{n}$ ne peut dépasser $\sqrt{n}$, sinon, leur longueur totales serait au moins $\sqrt{n}\cdot (\sqrt{n}+1)>n$. Donc le nombre de 
longueurs distinctes supérieur à $\sqrt{n}$ ne peut dépasser $\sqrt{n}$, puisque qu'il ne peut y a voir plus de mots (de longueurs au moins $\sqrt{n}$) que de longueurs distinctes. Comme le nombre de longueurs distinctes au plus $\sqrt{n}$ est au maximum égal à $\sqrt{n}$, nous déduisons que le nombre maximal de longueurs distinctes est au plus égal à $2\sqrt{n}$.

Afin d'implémenter le dictionnaire exact, on pourrait partitionner les mots du dictionnaire en (au plus) $2\sqrt{n}$ groupes, contenant chacun des mots de même longueur, et on utilise ensuite une table de hachage séparé pour chaque groupe.

\paragraph{Le détail de notre schéma :} Pour mettre en place le dictionnaire exact, au lieu de stocker $2\sqrt{n}$ \footnote{ Dans notre article \cite{ibra.Chegrane.simple} il y a une petite erreur dans la section 3.7. Nous avons écrit $\sqrt{d}$ ($d$ le nombre total des mots) au lieu de $2\sqrt{n}$ ($n$ la somme de toutes les longueurs des $d$ mots de dictionnaire). De même pour la section 3.1, on doit mettre $2\sqrt{n}$ au lieu de $\sqrt{n}$.} différentes tables de hachage pour les $2\sqrt{n}$ groupes des longueurs de mots distincts, nous utilisons un seuil $\beta$ et on stocke $\beta-1$ tables de hachages.

Tous les mots de longueur $i<\beta$ sont stockés dans la table de hachage numéro $i-2$\footnote{Il y a une petite erreur d'indice dans notre article \cite{ibra.Chegrane.simple}, dans la section 3.7, nous avons écrit la table numéro $i$ au lieu du numéro $i-2$.} qui contient $\lceil n_i/\alpha\rceil$ cases, où chaque case est de longueur $i$ caractères.

Le mot de longueur $i$ est stocké dans le table numéro $i-2$ car on considère que le mot le plus petit est de longueur $2$. La table numéro $0$ va stocker les mots de longueur $2$ ($i=2$,donc $i-2=2-2=0$), et les mots de longueur $3$ vont être stockés dans la table numéro $1$, jusqu'à ce qu'on arrive à l'avant-dernière table dont le numéro est $(\beta-3)$.
Tous les mots de longueur $i\geq \beta$ sont stockés dans une table de hachage (la dernière table dont le numéro est $(\beta-2)$ ) qui stocke des pointeurs vers les mots originaux au lieu des mots eux-mêmes.\\

Dans notre implémentation, nous avons utilisé $\beta=16$ caractères.
Lorsque nous avons fait des calculs sur l'optimisation de l'espace mémoire, nous avons trouvé qu'il est préférable de stocker tous les mots de longueur $i \geq 16$ dans un seul tableau de pointeur. Pour plus de détails et d'explications, voir la sous-section \ref{subsub:explication_beta=16}.
 
Tous les $\beta -1$ tableaux sont gérés par le hachage avec sondage linéaire.
Chaque tableau $i-2$ à une taille multiple de la longueur $i$ des mots qu'il va stocker. Chaque élément de la table est un bloc de longueur $i$ permettant de stocker les mots de longueur $i$. Le numéro de bloc pour stocker un mot donné est calculé par une fonction de hachage. 
La taille de tableau $T$ en terme d'octets est donc $(\mathtt{NbMots\_long\_i}/\mathtt{LF})\times i$. On arrondit ensuite le résultat de la division à l'entier immédiatement supérieur. Donc la formule finale est: $T = \left\lceil \frac{\mathtt{NbMots\_long\_i}}{\mathtt{LF}} \right\rceil \times i$. La taille de la table en terme de blocs est $\mathtt{Tnb} =  \left\lceil \frac{\mathtt{NbMots\_long\_i}}{\mathtt{LF}} \right\rceil$.

La dernière table (celle contenant les pointeurs vers des mots de longueur $i \geq 16$), est constituées de cases de taille $w$, et sa taille est $T=\left\lceil \frac{\mathtt{Tab\_NbMots}[\beta-2]}{\mathtt{LF}}\right\rceil$.

La structure de données est illustrée dans la figure \ref{fig:tableau_hachage}.

\begin{figure}[h]
\centering
\includegraphics[width=0.9\linewidth]{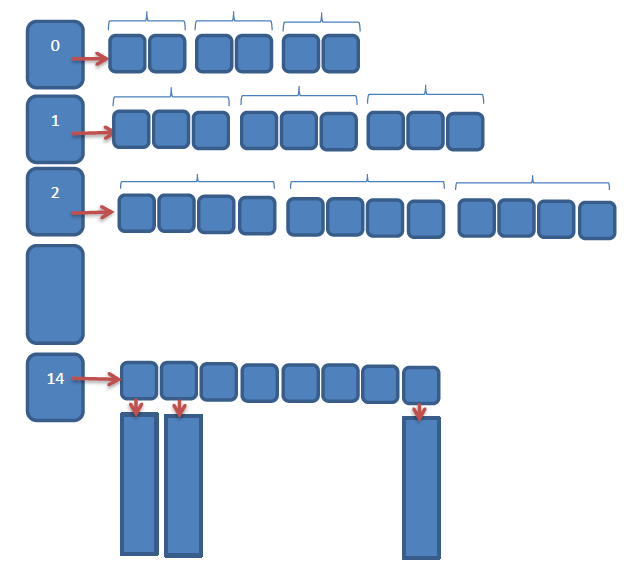}
\caption{La structure de données du dictionnaire exact. Cette structure comporte plusieurs tableaux pour stocker les mots selon leurs longueurs, le mot d'une longueur $m$ va être stocké sans son marqueur de fin de chaine dans le tableau $m-2$. Le dernier tableau stocke les mots de longueur $\geq 16$ avec leurs marque de fin. Tous tableaux sont gérés par le hachage avec sondage linéaire.}
\label{fig:tableau_hachage}
\end{figure}

La construction de la structure de données est donnée en résumé dans l'algorithme suivant:\\

% -----------------------------------------------------------------
\noindent
\rule{8cm}{0.1pt}\\
\textbf {Algorithme de construction de la structure de données exacte:}\\
\textbf{Entrée :} fichier qui contenant les mots de dictionnaire.\\
\textbf{Sortie :} structures de données composant le dictionnaire exact.\\

\begin{enumerate}

\item Mettre le nombre des mots de longueur $i$ dans la case $i-2$ du tableau $\mathtt{Tab\_NbMots}$, et le nombre des mots de longueur $i\geq \beta$ dans la case $\beta-2$.

\item Créer un tableau de pointeurs $\mathtt{Tab\_exact}$ de $\beta-2$ cases.

\item Pour chaque case $i$ ($i \in [2..\beta-3]$) de $\mathtt{Tab\_exact}$, créer un tableau de caractères de taille $T = \lceil \frac{\mathtt{Tab\_NbMots}[i-2]}{\mathtt{LF}} \rceil \times i$.

\item Créer un tableau de pointeurs de taille $T=\lceil \frac{\mathtt{Tab\_NbMots}[\beta-2]}{\mathtt{LF}} \rceil$, et le relier avec la dernière case $\mathtt{Tab\_exact}[\beta-2]$.

\item Pour chaque mot $x$ de longueur $i$ dans le dictionnaire:

	\begin{itemize}
	\item Calculer la valeur de hachage de $x$ donc $h(x)$. 
	
	\item Si $i<\beta$ insérer $x$ dans le tableau $\mathtt{Tab\_exact}[i-2]$, à la position $\mathtt{pos}=h(x)\bmod (\lceil \frac{\mathtt{Tab\_NbMots}[i-2]}{\mathtt{LF}}  \rceil)$ sans son marqueur de fin ($\backslash 0$).
	
	\item Si $i \geq \beta$, stocker le mot $x$ dans la mémoire avec son marqueur de fin, et garder son adresse dans le tableau pointé par $\mathtt{Tab\_exact}[\beta-2]$ à la position $\mathtt{pos}=h(x)\bmod (\lceil \frac{\mathtt{Tab\_NbMots}[\beta-2]}{\mathtt{LF}}  \rceil)$.

	\end{itemize}

\end{enumerate}
\rule{8cm}{0.1pt}\\
% -----------------------------------------------------------------

\subsubsection{ Explication du paramètre $\beta = 16$}
\label{subsub:explication_beta=16}

Le choix du seuil $\beta$ utilisé pour guider le choix du stockage des mots du dictionnaire (soit dans des tables de hachage stockant des caractères ou alors dans une table stockant des pointeurs vers les mots dans la mémoire) dépend de l'espace mémoire occupé par chaque structure. 
Le but étant de faire une comparaison afin de choisir le paramètres $\beta$ qui minimise l'utilisation de l'espace mémoire.

Dans une table de caractères, les mots sont stockés à l'intérieur de la table sans le marqueur de fin ($\backslash 0$). Par contre dans la table des pointeurs, les mots sont stockés dans la mémoire (avec leur marqueur de fin) et on garde leurs adresses dans la table des pointeurs.

Afin d'éviter toute confusion, on appelle \textit{longueur de mot}, le nombre des caractères du mot sans compter le marqueur de fin du mot.

Dans une table de hachage stockant des caractères, chaque mot $x$ de longueur $m$ occupe un espace mémoire de $\lceil m/\mathtt{LF} \rceil$. Lorsqu'on utilise une table de pointeurs, chaque mot $x$ occupe un espace de taille $m+1$ (la longueur du mot plus le marqueur de fin) auquel on ajoute la taille de pointeur $w$ (un mot mémoire) divisé par le taux de chargement :$ (m+1) + \lceil w/\mathtt{LF} \rceil$. Dans notre cas, on prend $\mathtt{LF} = 0,7$.

En utilisant ces deux informations, on peut trouver quelles structures utiliser, selon l'espace mémoire occupée pour le stockage des mots du dictionnaire.

On nomme la première structure (Table de caractères) \textbf{Tab\_char}, et la deuxième structure (table de pointeurs vers les mots) \textbf{Tab\_pointeur}.

Les calculs pour un seul mot montre que la structure \textbf{Tab\_char} occupe moins d'espace mémoire pour les mots de longueurs $2$ à $15$. Lorsque la longueur des mots est supérieure à $16$, la structure \textbf{Tab\_pointeur} donne de meilleurs résultats.\\

On va faire la même vérification pour les deux structures en utilisant des groupes de mots de $10$, $100$, et $1000$ mots pour toutes les longueurs, voir la figure~\ref{fig:ensemble_mots}.

\begin{figure}[h]
\centering
\includegraphics[width=0.32\linewidth]{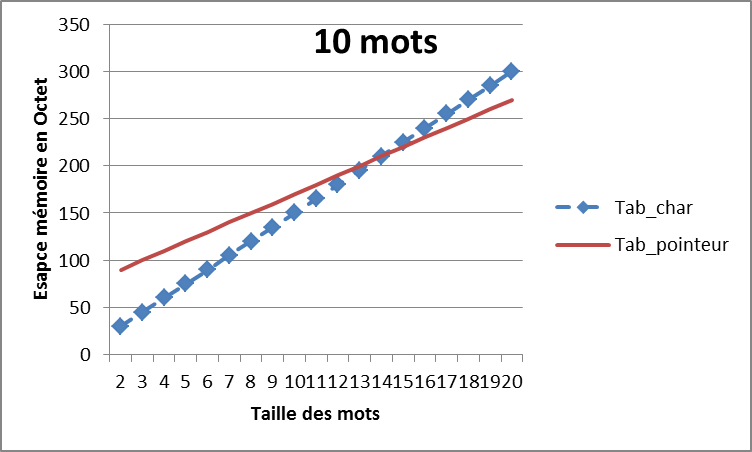}
\includegraphics[width=0.32\linewidth]{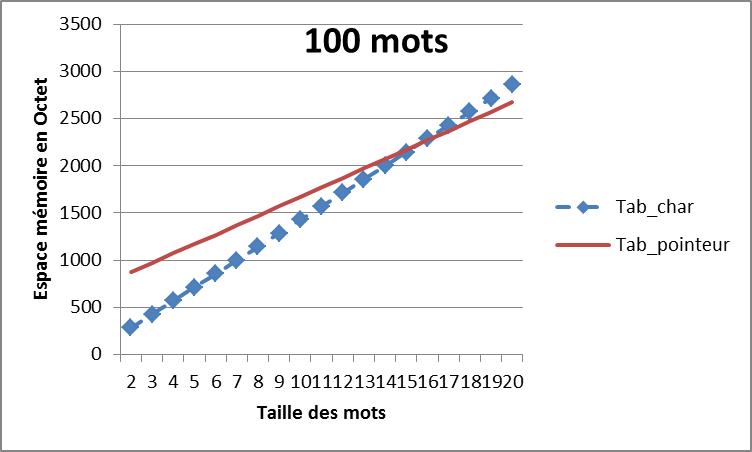}
\includegraphics[width=0.32\linewidth]{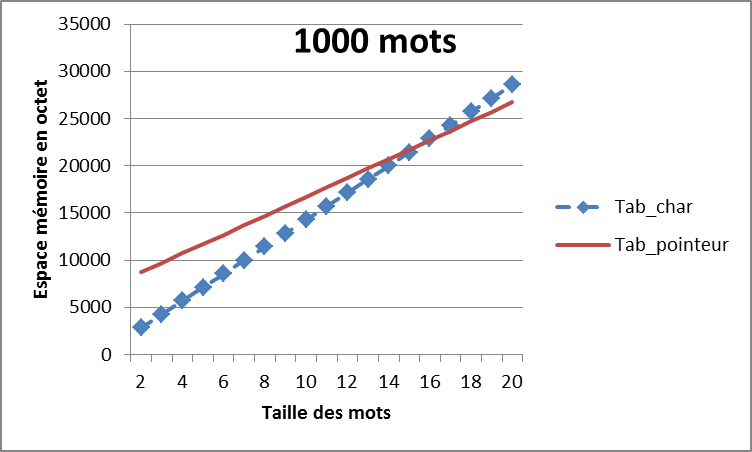}
\caption{Comparaison de l'occupation de l'espace mémoire par les deux structures de données avec différents nombres de mots.}
\label{fig:ensemble_mots}
\end{figure}

On remarque dans la figure \ref{fig:ensemble_mots}, qu'avec juste $10$ mots, la structure  $\mathtt{Tab\_pointeur}$ commence à donner des meilleurs résultats lorsque la longueur des mots devient $15$, mais il y a pas une grand différence avec la longueur $16$. Lorsque le nombre des mots est $100$, la structure $\mathtt{Tab\_pointeur}$ donne de meilleurs résultats lorsque la longueur des mots devient $16$, et c'est la même chose avec le nombre des mots égale à $1000$. Il est clair que la structure $\mathtt{Tab\_char}$ occupe moins d'espace mémoire pour les mots de longueur $2$ à $15$.

On fixe donc le seuil $\beta$ à $16$ car la structure avec pointeurs $\mathtt{Tab\_pointeur}$ donne de meilleurs résultats à partir de ce seuil.

En fin de compte, pour les mots de longueur $2$ à $15$ on utilise des tables de caractères, et pour les mots de longueur $\geq 16$, on stocke des pointeurs vers les mots stockés dans la mémoire.

\subsection{Dictionnaire des listes de substitutions}
\label{sub:dictionnaire_des_liste_de_substitution}

Pour pouvoir faire une recherche approchée, nous avons besoin d'ajouter une autre structure de données nommée \textit{dictionnaire des listes de substitutions}.

L'idée est de faire un pré-traitement, pour stocker pour chaque position donnée, la liste des caractères dont la substitution à cette position conduit à un mot existant dans le dictionnaire.

Donc, au lieu de tester pour chaque position tous les $\sigma$ caractères possibles, la structure nous donne une liste de caractères à insérer ou à substituer dans une position donnée du mot requête $q$. Cette liste contient généralement un seul caractère si on a juste une seule solution (ce qui est le cas pour la majorité des requêtes).\\

Pour implémenter le dictionnaire des listes de substitutions nous utilisons une table de hachage avec sondage linéaire, dans laquelle chaque élément est un caractère. le calcul des valeurs de hachage se base sur les mots du dictionnaire dans lesquels on substitue un caractère spécial dans la position de l'erreur.
Pour chaque mot $x_i$ de longueur $m$, et pour chaque position $j\in[1..m]$, on insère le caractère $x_i[j]$ dans la table de hachage à la position $i=h(x_{i,j})\bmod T$, où $x_{i,j}=x_i[1..j-1]\phi x_i[j+1..m]$ et $T$ est la taille de la table de hachage.
 
Tous les caractères qui appartiennent à la même liste de substitutions sont mappés vers la même position dans la table de hachage.
En d'autres termes, deux caractères $c_1$ et $c_2$ seront mappés vers une même position $i=h(p\phi v)$ dans la table de hachage si et seulement si il existe une paire de chaînes de caractères $(p,v)$, tels que les deux mots $(p c_1 v)$ et $(p c_2 v)$ existent dans le dictionnaire.

% % les 4 bites :
\paragraph{Diminuer les candidats avec des signatures de $r$ bits:}
Nous proposons une deuxième méthode d'implémentation du dictionnaire de listes de substitutions dans le but de diminuer le nombre des collisions.

On ajoute une signature de hachage de taille $r$ bits avec chaque caractère, pour permettre ainsi de filtrer une grande partie des caractères qui sont en collision.
La signature de hachage est calculée sur la base de la valeur de hachage de la chaîne de caractère  $x_i[1..j-1]\phi x_i[j+1..m]$, et sa longueur est seulement $r$ bits (pour un petit nombre $r$).

Quand un caractère $c$ est inséré dans la liste de substitutions, on stocke avec lui sa signature correspondante. Lors de la recherche, on considère qu'un caractère appartient à la liste de substitutions si et seulement si sa signature est la même que celle stockée dans la liste de substitutions.

\paragraph{Détail de la mise en place de la structure de données:}
Comme il a été expliqué précédemment, nous avons deux variantes pour l'implémentation de la table de hachage des caractères de substitution: celle dans laquelle chaque entrée est juste un caractère, et l'autre dans laquelle chaque entrée est un caractère et une signature de $r$ bits.

On code un caractère en utilisant un octet ($8 \: bits$), et on code la signature de hachage $r$ en utilisant $4$ bits ($r=4$). Le codage de la signature avec $4$ bits simplifie l'implémentation.

Dans la deuxième variante, la table de hachage est divisé en blocs de 3 octets. Le premier octet stocke le premier caractère, le troisième octet stocke le deuxième caractère, et le deuxième octet central stocke les 4 bits de chaque caractère. Les 4 bits de poids fort (celle à gauche) de l'octet central pour le premier caractère, et les 4 bits de poids faible (celle à droite) pour le deuxième caractère.\\

La table de hachage est constituée avec des blocs de 3 octets (soit 2 cases, deux caractères avec leurs signatures). Voir la figure \ref{fig:tab_char_4bits}.

\begin{figure}[h]
\centering
\includegraphics[width=0.9\linewidth]{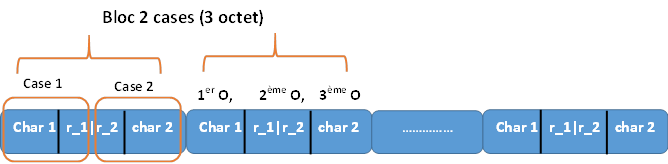}
\caption{La deuxième variante du dictionnaire des listes de substitutions.}
\label{fig:tab_char_4bits}
\end{figure}

Si le nombre des caractères est impair, donc il reste un caractère seul, on lui donne aussi un bloc de 3 octets. Au final, le tableau à une taille $T$ qui est multiple de $3$.
La taille de la table est donc $T =\left\lceil \frac{n+ \frac{n}{2}}{\mathtt{LF}} \right\rceil $, si $T \bmod 3 = 2$ alors $T$ devient $T+1$, et si $T \bmod 3 = 1$ alors $T$ devient $T+2$.\\

La signature de hachage de $r=4$ bits est constituée à partir de la valeur de hachage $h(x)$. Dans notre cas, on prend simplement les $4$ bits de poids faible de la valeur de hachage $h(x)$. Nous avons fait ce choix, car, quelle que soit la taille de la valeur de hachage (petite ou grande), généralement les bits de poids faible ne sont pas nuls, contrairement aux bits de poids fort, si la valeur de hachage est petite cela signifie qu'ils sont nuls.\\

Le calcul de la valeur de hachage pour les mots candidats générés depuis le mot $x_i$ se fait en $O(1)$. Pour cela, on fait le même pré-traitement utilisé par~\cite{Be09}. On prépare trois tableaux: 
\begin{enumerate}
\item $A_t[0, m + 1]$ pour stocker toutes les puissances de $t$ jusqu'à $t^{m+1}$ où $m$ est la longueur de mot $x_i$, et $t$ est une valeur aléatoire choisi entre $1$ et $P$, comme il est expliqué dans le début de cette section.
\item $F[0..m]$ un tableau qui stocke les valeurs de hachage pour les préfixes de mot $x_i$. 
\item $G[1..m + 1]$ un tableau qui stocke les valeurs de hachage pour les suffixes de mot $x_i$.
\end{enumerate}

Pour plus de détail sur ces trois tableaux, voir dans le chapitre état de l'art la sous-section \ref{sub:la_requete_djamal}.

Pour le mot $x_i$ de longueur $m$, pour chaque $j\in[1..m]$, on substitue le caractère  $x_i[j]$ avec un caractère spécial $\phi$ où $x_{i,j}=x_i[1..j-1]\phi x_i[j+1..m]$. La valeur de hachage de mot candidat $x_{i,j}$ est donnée en $O(1)$ par  $h(x_{i,j}) = F[j-1] \oplus (\phi \oplus G[j + 1]) \otimes A_t[j]$. Cette valeur va permettre de stocker le caractère $x_i[j]$ dans la structure de données.\\

Nous donnons l'algorithme qui résume la construction de notre structure de donnée \textbf{le dictionnaire des listes de substitutions}.\\

\noindent
\rule{8cm}{0.1pt}\\
\textbf {Algorithme de construction du dictionnaire des listes de substitutions (DLS) :}\\
\textbf{Entrée :} un fichier qui contient les mots du dictionnaire.\\
\textbf{Sortie :} la structure de données (DLS) pour la recherche approchée.\\

\begin{enumerate}

\item Calculer le nombre total des caractères des mots du dictionnaire $n$.

\item Créer une table de hachage avec sondage linéaire:

	\begin{itemize}
	\item La première variante: la taille de la table est simplement $ T = \lceil \frac{n}{\mathtt{LF}} \rceil$.
	
	\item La deuxième variante: la taille de la table est $T = \lceil \frac{n+ \frac{n}{2}}{\mathtt{LF}} \rceil $, si $T \bmod 3 = 2$ alors $T$ devient $T+1$, et si $T \bmod 3 = 1$ alors $T$ devient $T+2$.
	\end{itemize}

\item Préparer le tableau $A_t[0, m' + 1]$, pour stocker toutes les puissances de $t$ jusqu'à $t^{m'+1}$ ($m'$ le plus long mot du dictionnaire).
 
\item Pour chaque mot $x_i$ de longueur $m$ ($i\in [1..d]$ et $d$ le nombre des mots du dictionnaire): préparer les deux tableaux $F[0..m], G[1..m + 1]$ qui stockent les valeurs de hachage pour les préfixes et les suffixes de mot $x_i$, ensuite pour chaque position $j\in[1..m]$ :

	\begin{enumerate}
   
	\item Substituer le caractère $x_i[j]$ par un caractère spécial: $x_{i,j}=x_i[1..j-1]\phi x_i[j+1..m]$.

	\item Calculer la valeur de hachage de $x_{i,j}$ en $O(1)$ par la formule $h(x_{i,j}) = F[j-1] \oplus (\phi \oplus G[j + 1]) \otimes A_t[j]$.
  
	\item Trouver la position dans le tableau $\mathtt{pos}=h(x_{i,j})\bmod T$.
   
	\item Stocker le caractère $x_i[j]$

	\smallskip
	La première variante :
		\begin{itemize}
		\item Si la case numéro $\mathtt{pos}$ est vide, alors insérer $x_i[j]$.
		\item Sinon, avancer dans la table case par case jusqu'à ce qu'on trouve une case vide, $\mathtt{pos}=\mathtt{pos}+1 \bmod T$.
		\end{itemize}

	\medskip
	% %\item Stocker le caractère $x_i[j]$, 
	La deuxième variante :

		\begin{itemize}
        \item Si $\mathtt{pos} \bmod 3 = 0$, alors stocker le caractère dans cette position, et stocker les 4 bits de la signature dans $\mathtt{pos}+1$ dans la partie poids fort (partie gauche).
        
        \item Si $\mathtt{pos} \bmod 3 = 2$, alors stocker le caractère dans cette position, et stocker les 4 bits dans l'octet précédent dans les $4$ bits de poids faible qui est à droite dans $\mathtt{pos}-1$.
        
        \item Si $\mathtt{pos} \bmod 3 = 1$, on passe à l'octet suivant ($\mathtt{pos}=\mathtt{pos}+1$), et on applique le cas de $\mathtt{pos} \bmod 3 = 2$.

        \item Si la case est non vide, avancer vers la case suivante, si on est dans  $\mathtt{pos} \bmod 3 = 0$, alors $\mathtt{pos}=\mathtt{pos}+2$, et si $\mathtt{pos} \bmod 3 = 2$, alors $\mathtt{pos}=\mathtt{pos}+1 \bmod T$. 
        \end{itemize} 
        
	\end{enumerate}
	
\end{enumerate}        
\rule{8cm}{0.1pt}

\subsection{La compression de notre structure de données}
\label{sec:bit_compact}

Dans notre travail, toutes les tables ont le même facteur de chargement $\alpha = 0,7$, ce qui signifie que pour chaque table ($16$ en tout) on a $30 \%$ de cases vides. Le but est de trouver une méthode permettant de réduire l'espace mémoire utilisé par notre structure de données tout en gardant le même temps d'accès au données. L'inconvénient du compactage de la structure de données est qu'elle devient non incrémentale.\\

Nous utilisons la compression basée sur un tableau de bits. Le résultat du compactage donne deux vecteurs: un table de hachage compact de $n$ cases, et un vecteur de bits de $n'= \left \lceil \frac{n}{\alpha} \right \rceil$ bits.
 
La compression est faite de la manière suivante: on parcourt toutes les cases de la table de hachage, et on génère le vecteur de bits. On met dans le vecteur de bits un $0$ pour chaque case vide de la table de hachage originale, et un $1$ pour chaque case non-vide.
Durant le parcours, on rajoute toutes les cases non-vides dans la table de hachage compact (qui était initialement vide).

Afin de prendre en charge les requêtes dans la table de hachage compactée, on indexe le vecteur de bits pour supporter les opérations de rangs ~\cite{jacobson1989space,munro1996tables} (Le rang d'une position $i$ dans un vecteur de bits $b$ est défini comme étant le nombre de $1$ dans $b[1..i]$).

L'idée est très simple, une position dans le tableau de bits est la même que celle dans la table de hachage originale. On marque les cases non-vides avec $1$, donc la position de ce même élément dans la table compactée est le nombre des $1$ dans le tableau de bits depuis le premier bit jusqu'à cette position $i$.\\

Compter le nombre de $1$ depuis le début jusqu'à une position $i$ donnée à chaque fois augmente le temps de la requête. Afin d'y remédier, nous utilisons une méthode qui permet de donner le nombre des $1$ pour n'importe quelle position dans un temps constant.
On implémente un simple algorithme de rang qui augmente l'espace original du vecteur de bits d'un facteur de $1+1/\delta$, tel que $\delta$ est une petite valeur constante. On décompose le vecteur de bits en blocs de $\delta$ mots mémoires ($w \times \delta$ bits), et on sauvegarde le rang calculé jusqu'au début de chaque bloc. L'espace final obtenu est de $N(1+1/\delta)$ bits pour un vecteur de bits originalement constitué de $N$ bits.

Dans notre travail, nous assumons un mot mémoire $w$ de $32$ bits, et nous prenons $\delta = 4$. Chaque bloc a une taille de $4$ mots mémoire ($32 \times 4$ bits), auquel, on ajoute une case mémoire de taille $w$ (si on pose $N=4$, et on applique la formule $N(1+1/\delta)$, alors $4(1+1/4) =4+4/4 = 4+1$). Donc pour chaque bloc de $4$ mots mémoire, on ajoute un mot mémoire pour stocker le nombre des $1$ (voir la figure \ref{fig:bits_vecteurs}).
La requête de rang sur le vecteur de bits retourne la position dans la table de hachage compacte. On compte le nombre des $1$ consécutifs dans le vecteur de bits pour avoir le nombre des cases non-vides, et donc la position dans la table de hachage compacte.
Le nombre de comptage partiel $i$ va stocker tous les nombres des $1$ dans les $(i-1)\delta$ premiers mots. Au lieu de stocker les comptes partiels dans un vecteur distinct, on les stocke entrelacées avec le vecteur de bits.

Afin de supporter le comptage efficace des $1$ jusqu'à la position $i$, on utilise le vecteur de bits pour obtenir le comptage des $1$ jusqu'au mot numéro $\left \lfloor \frac{i}{\delta} \right \rfloor$ en utilisant les comptes partiels, puis compter le nombre des $1$ jusqu'à $32\delta-1$ bits, qui couvrent au plus $\delta$ mots en utilisant l'opération $\mathtt{PopCount}$~\cite{bithacks}~\footnote{Même si le temps d'exécution n'est pas constant dans la théorie, dans la pratique, elle est très rapide et peut être considérée comme constante. Notons également que les processeurs récents possèdent souvent une implémentation matérielles de cette instruction.}, voir la figure \ref{fig:bits_vecteurs}. Ainsi, l'espace total sera de $n'+n'/\delta = \lceil n/\alpha\rceil(1+\delta)$ bits, et le temps sera de $O(\delta)=O(1)$.

\begin{figure}[h]
\centering
\includegraphics[width=0.9\linewidth]{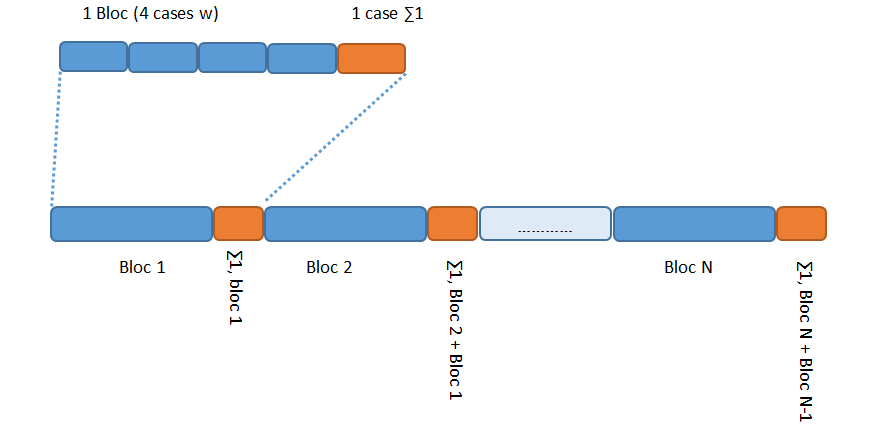}
\caption{Le vecteur de bits et les cases qui contiennent les compteurs des nombres des $1$ pour tous les blocs précédent.}
\label{fig:bits_vecteurs}
\end{figure}

On note que toutes nos méthodes non-compactes sont entièrement incrémentales, dans le sens où l'ajout d'un élément au dictionnaire est possible et prend à peu près le même temps que le temps de construction total divisé par le nombre total des mots. En revanche, les structures compactes ne sont pas incrémentales, en raison des étapes de constructions et du résultat final du compactage.

Notre implémentation suppose des entiers et des pointeurs de $32$ bits, mais il est trivial de changer le code pour fonctionner avec $64$ bits (et donc gérer de plus grands dictionnaires) sans encourir beaucoup de dépassement d'espace mémoire ou de temps.
Le seul élément qui utilise plus d'espace mémoire est le dictionnaire exact pour les grands mots qui sont stockés dans une table de pointeurs (les mots de longueur $\geq 16$). Tous les autres composants n'utilisent pas les pointeurs, et donc leur utilisation de l'espace ne devraient pas être affectés.\\

\begin{remarque}
Pour plus de performance avec $64$ bits, et pour optimiser l'espace mémoire, le seuil $\beta$ va être égale à $32$. On utilise les mêmes étapes expliquées dans la sous-section \ref{subsub:explication_beta=16} pour le trouver. Cela implique le changement de nombre de tableaux utilisé dans le dictionnaire exact, on va utiliser $30$ tableaux de caractères afin d'optimiser l'espace mémoire pour les mots qui ont une longueur $\geq 16$ (la nouvelle condition est $i\geq 32$).
\end{remarque}

\section{La vérification des occurrences}
\label{sec:Verification_des_Occurrences}

Soit un mot requête $q$ de longueur $m$. Pour vérifier les occurrences potentielles, nous utilisons une stratégie similaire à Boytsov~\cite{Bo12}. Au lieu de calculer la distance d'édition entre le mot requête et le mot candidat, nous construisons directement une chaîne modifiée, en appliquant l'opération d'édition candidate, ensuite, nous faisons une comparaison directe entre la chaîne obtenue et le mot candidat.

Avant de commencer les vérifications des mots candidats, on prépare trois tableaux $A_t, G, F$ comme décrit dans la partie \textbf{dictionnaire des listes de substitutions} dans la sous-section \ref{sub:dictionnaire_des_liste_de_substitution}: 
\begin{itemize}
\item $A_t[0..m' + 1]$, pour stocker toutes les puissances de $t$ jusqu'à $t^{m'+1}$, $m'$ est la longueur de mot le plus long supposé être utilisé comme un mot requête. On utilise un nombre suffisamment grand dans le but de calculer ce tableau une seule fois, et on le réutilise pour tous les mots requête, car ces valeurs sont indépendantes des mots requête.
\item $F[0..m]$, un tableau qui stocke les valeurs de hachage pour les préfixes du mot requête $q$.
\item $G[1..m + 1]$ pour stocker les valeurs de hachage pour les suffixes de $q$. Ces deux tableaux sont calculés pour chaque mot requête en $O(m)$.
\end{itemize}

\medskip
On commence par expliquer comment vérifier l'existence d'un mot requête dans le dictionnaire exact (une recherche exacte). Par la suite, on expliquera les étapes de la recherche approchée avec une seule erreur pour les trois types d'erreurs d'édition.

\subsection{La recherche exacte}
\label{sub:Recherche_exact}

La recherche exacte est très simple, il suffit juste de trouver la position dans la structure de données à l'aide de la fonction de hachage et ensuite de faire une simple comparaison. Voici les étapes à faire pour une vérification dans le dictionnaire exact :\\

\noindent
\rule{8cm}{0.1pt}\\
\textbf {Algorithme de la recherche exacte :}\\
\textbf{Entrée :} un mot requête $q$ de longueur $m$, le dictionnaire exact (la structure de données exacte).\\
\textbf{Sortie :} booléen (existe ou non ).\\

\begin{enumerate}

\item Calculer la valeur de hachage du mot requête $q$ en $O(1)$, la valeur de hachage est déjà pré-calculée $h(q) = F[m] = G[1]$.

\item Calculer la position avec la valeur de hachage modulo la taille de la table qui stocke les mots de longueur $m$, $\mathtt{pos}=h(q) \bmod T$.

	\begin{itemize}
	
	\item Si $m<\beta$, alors on utilise la taille de la table en nombre de blocs (le tableau $Tab\_exact[m-2]$), $T = \left \lceil (\mathtt{Tab\_NbMots}[m-2]/\mathtt{LF}) \right \rceil$.
	
	\item Si $m \geq \beta$, on utilise le tableau de pointeur $Tab\_exact[\beta-2]$. La taille est en nombre de cases mémoire de taille $w$, donc $T = \left \lceil (\mathtt{Tab\_NbMots}[\beta-2]/\mathtt{LF}) \right \rceil$.
	
	\end{itemize}

\item Faire une comparaison directe entre le mot requête $q$ et le mot stocké. Si la case à la position $\mathtt{pos}$ est vide, alors, on dit qu'on n'a pas de résultat. Dans le cas contraire, on teste tous les mots jusqu'à la première case vide.

\end{enumerate}
\rule{8cm}{0.1pt}\\

\subsection{La recherche approchée}
\label{sub:Recherche_approchee}

Dans la recherche approchée, on construit des mots candidats en appliquant l'une des trois opérations de distance d'édition dans une position donnée dans un temps constant, ensuite, on fait une vérification dans le dictionnaire exact. On récupère les caractères depuis le \textbf{dictionnaire des listes de substitutions}, qui donnent des mots candidats qui peuvent être une solution, afin de les utiliser dans la substitution et l'insertion. Pour la suppression, il suffit juste de vérifier l'existence de mot candidat dans le dictionnaire exact.

\subsubsection{La substitution}
\label{subsub:operation_substitution}

On commence par la création d'un mot copie $\mathtt{sub}_q$ de notre mot requête $q$ de la même longueur $m$. Ceci permettra de générer des mots candidats en $O(1)$, et de substituer les caractères récupérés depuis le dictionnaire des listes de substitutions dans une position donnée $i$.\\

\noindent
\rule{8cm}{0.1pt}\\
\textbf {Algorithme de la recherche approchée (substitution)}\\
\textbf{Entrée :} un mot requête $q$ de longueur $m$, le dictionnaire exact et le dictionnaire des listes de substitutions.\\
\textbf{Sortie :} liste des mots solutions approchées.\\

\begin{enumerate}

\item Pour chaque position $i$ dans $\mathtt{sub}_q$ (le mot copie du mot requête) avec $i \in [1..m]$:

\item Calculer la valeur de hachage pour la position $i$ en $O(1)$, avec la formule suivante: $h(\mathtt{sub}_q) = F[i-1] \oplus (\phi \oplus G[i + 1]) \otimes A_t[i]$.

\item Récupérer une liste de caractère $\mathtt{list\_chars}$ depuis le dictionnaire des listes de substitutions:

	\begin{enumerate}
	
	\item Trouver la position dans la table de hachage qui stocke les caractères de substitutions $\mathtt{pos}=h(\mathtt{sub}_q{i})\bmod T$.

	\item La première variante de la structure de données des listes de substitutions: 

		\begin{itemize}
		
		\item La taille de la table $ T = \lceil \frac{n}{\mathtt{LF}} \rceil$. 
		
		\item Si la case $\mathtt{pos}$ est vide, alors on retourne une liste vide. 
		
		\item Sinon, on récupère tous les caractères de cette position jusqu'à ce qu'on arrive à une case vide. L'avancement se fait case par case $\mathtt{pos} = \mathtt{pos}+1 \bmod T$.
	
		\end{itemize}

	\item La deuxième variante (avec la signature de $4$ bits) :
	
		\begin{itemize}
		
		\item La taille de tableau $T = \lceil \frac{n+ \frac{n}{2} }{\mathtt{LF}} \rceil $, si $T \bmod 3 = 2$ alors $T$ devient $T+1$, et si $T \bmod 3 = 1$ alors $T$ devient $T+2$.
		
		\item Si la case est vide, alors retourner une liste $\mathtt{list\_chars}$ vide.

		\item Sinon, pour tous les caractères de cette position jusqu'à une case vide, vérifier si la signature de ($h(\mathtt{sub}_q)$) est la même que celle stockée dans la table de hachage. Si c'est le cas, alors ajouter le caractère à la liste $\mathtt{list\_chars}$. Si $\mathtt{pos} \bmod 3 = 0$, alors les $4$ bits de la signature sont dans $\mathtt{pos}+1$ dans la partie poids fort (partie gauche). Si $\mathtt{pos} \bmod 3 = 2$, alors les $4$ bits sont dans l'octet précédent $\mathtt{pos}-1$ dans les $4$ bits de poids faible (à droite). Si $\mathtt{pos} \bmod 3 = 1$ alors on passe à l'octet suivant ($\mathtt{pos}+1$).
		
		\item Pour avancer vers les cases suivantes, si $\mathtt{pos} \bmod 3 = 0$, alors on fait $\mathtt{pos}=\mathtt{pos}+2$, et si on est dans le cas $\mathtt{pos} \bmod 3 = 2$, alors on fait ($\mathtt{pos}=\mathtt{pos}+1 \bmod T$).
		\end{itemize}
  
  \end{enumerate}
  
\item Pour chaque caractère $j$ dans la liste $\mathtt{list\_chars}$ de taille $L$ ($j \in [1..L]$):

	\begin{itemize}
	\item Substituer le caractère $\mathtt{list\_chars}[j]$ dans la position $i$ du mot requête $\mathtt{sub}_q[i]=list\_chars[j]$, pour construire le mot candidat $\mathtt{sub}_q = \mathtt{sub}_q[1..i-1] \mathtt{list\_chars}[j] \mathtt{sub}_q[i+1..m]$.
	
	\item Vérifier l'existence du mot candidat dans le dictionnaire exact. La valeur de hachage se calcule en $O(1)$ avec la formule $h(\mathtt{sub}_q) = F[i-1] \oplus (\mathtt{list\_chars}[j] \oplus G[i + 1]) \otimes A_t[i]$.
	
	\end{itemize}

\item Passer à la position suivante dans le mot requête ($i+1$). Avant cela, on doit remettre le caractère originel de la position $i$ en faisant $\mathtt{sub}_q[i]=q[i]$.

\end{enumerate}
\rule{8cm}{0.1pt}\\

On remarque que la génération des mots candidats après la récupération de la liste des caractères prend un temps constant par candidat: il suffit juste de substituer les caractères de la liste dans la position $i$ et à la fin remettre le caractère original avant de passer à la position suivante.

Dans les deux variantes de la structure de données, l'obtention de la liste de substitutions pour une position donnée, peut nécessiter la traversée de nombreuses cases dans la table de hachage avant d'atteindre une case vide. Pour borner le temps de la requête dans le pire cas, nous avons implémenté la stratégie suivante:
à chaque fois qu'une requête pour une liste de substitutions traverse plus que $\sigma$ cases dans la table de hachage, on arrête le parcours des cases de la table de hachage, et on crée $\sigma$ chaînes candidates parmi tous les $\Sigma$ caractères possibles~\footnote{Il faut noter que lorsque qu'on traverse $\sigma$ cases, cela ne veut pas dire qu'on va collecter $\sigma$ caractères candidats distincts. Le nombre peut même être beaucoup plus petit lorsque les signatures de hachage sont utilisées}.

\subsubsection{L'insertion}
\label{operation_insertion}

Dans cette partie, le traitement est le même qu'auparavant avec juste quelques petits changements dans les indices. Nous en décrivons les détails 
dans ce qui suit.

On commence par créer un mot copie $\mathtt{ins}_q$ avec une longueur $m+1$ afin de nous permettre de générer des mots candidats en $O(1)$ avec insertion dans une position donnée $i$. On copie le mot $q$ dans $\mathtt{ins}_q$ à partir de la deuxième position, on laisse la première position vide, donc $\mathtt{ins}_q[2..m+1]=q[1..m]$.

\rule{8cm}{0.1pt}\\
\textbf {Algorithme de la recherche approchée (insertion)}\\
\textbf{Entrée :} un mot requête $q$ de longueur $m$, le dictionnaire exact, le dictionnaire des listes de substitutions.\\
\textbf{Sortie :} liste des mots solutions approchées.\\

\begin{enumerate}

\item Pour chaque position dans le mot copie du mot requête $\mathtt{sub}_q$ $i \in [1..m+1]$:

\item Calculer la valeur de hachage du mot $\mathtt{sub}_q$ après une insertion du caractère spécial $\phi$ à la position $i$, $h(\mathtt{ins}_q) = F[i] \oplus (\phi \oplus G[i + 1]) \otimes A_t[i+1]$.

\item Calculer la position $\mathtt{pos}=h(\mathtt{ins}_q{i})\bmod T$, et ensuite, récupérer une liste de caractères ($\mathtt{list\_chars}$) si elle existe.
  
\item Construire les mots candidats par la substitution des caractères de la liste $\mathtt{list\_chars}$ à la position $i$ (où on a inséré le caractère $\phi$): $\mathtt{ins}_q = \mathtt{ins}_q[1..i] \mathtt{list\_chars}[j] \mathtt{ins}_q[i+1..m+1]$. Ensuite, faire la vérification dans le dictionnaire exact, la valeur de hachage étant calculée par la formule $h(\mathtt{ins}_q) = F[i] \oplus (\mathtt{list\_chars}[j] \oplus G[i + 1]) \otimes A_t[i+1]$).

\item Passer à la position d'erreur suivante ($i+1$). Décaler le caractère qui est à la position $i+1$ vers la position $i$ du mot $\mathtt{ins}_q$ ($\mathtt{ins}_q[i]=\mathtt{ins}_q[i+1]$).

\end{enumerate}
\rule{8cm}{0.1pt}\\

\subsubsection{La suppression}
\label{susbub:operation_suppression}

Dans la suppression, nous n'avons pas une liste de caractères à insérer ou à substituer. Il suffit juste de faire une suppression dans la position $i$ et de calculer la valeur de hachage pour ensuite faire une vérification dans le dictionnaire exact.\\

\noindent
\rule{8cm}{0.1pt}\\
\textbf {Algorithme de recherche approchée (suppression)}\\
\textbf{Entrée :} un mot requête $q$ de longueur $m$, le dictionnaire exact, le dictionnaire des listes de substitutions.\\
\textbf{Sortie :} liste des mots solutions approchées.\\

\begin{enumerate}

\item On crée un mot $\mathtt{del}_q$ de longueur $m-1$.

\item Construire les $m$ mots candidats :

	\begin{enumerate}[i.]
	
	\item \label{item:sup_1} Construire le premier mot candidat de $i=1$. On copie le mot $q$ à partir de sa deuxième position, donc $\mathtt{del}_q[1..m-1]=q[2..m]$.
	
	\item \label{item:sup_boucle} Construire les mots candidats pour les autres positions $i \in [2..m]$. Il suffit juste de faire $\mathtt{del}_q[i-1]=q[i]$.
	
	\item \label{item:sup_last} Construire le dernier mot candidat, on ajoute $\mathtt{del}_q[i-1]=q[i-1]$.
	\end{enumerate}
	
\item Pour chaque mot candidat $\mathtt{del}_q = q[1..i-1] q[i+1..m]$ ($i\in [1..m]$ ), calculer la valeur de hachage avec la formule suivante: $h(\mathtt{del}_q) = F[i-1] \oplus G[i+1] \otimes A_t[i-1]$, ensuite, faire la vérification dans le dictionnaire exact. 

\end{enumerate}
\rule{8cm}{0.1pt}\\

\subsection{La recherche dans la structure compacte}
\label{sub:Recherche_dans_la_structure_compact}

Avec la structure compacte, on utilise exactement les mêmes étapes précédentes, sauf qu'au lieu d'accéder directement au tableau de hachage compact, on passe par le tableau de vecteur de bits pour trouver la position de l'élément dans le tableau compact.

On commence par trouver la position dans le tableau de bits $\mathtt{pos}_{\mathtt{bv}}=h(q') \bmod T$. Cette position $\mathtt{pos}_{\mathtt{bv}}$ va nous donner la position réel $\mathtt{pos}_{\mathtt{htc}}$ dans la table de hachage compacte, pour cela on récupère le nombre de $1$ depuis le début du tableau de bits jusqu'à cette position $\left(\mathtt{pos}_{\mathtt{htc}}= \sum_{i=0}^{i=\mathtt{pos}_{\mathtt{bv}}}1 \right)$.  Comme on a expliqué dans la section \ref{sec:bit_compact}, cette opération ce fait dans un temps constant, à cause des cases mémoire ajoutées dans chaque $\delta$ cases qui contient la somme des $1$.\\

On récupère directement le nombre d'éléments de la liste de substitutions depuis le vecteur de bits qui est le nombre de $1$ successifs depuis la position trouvée dans le vecteur de bits. Pour cela, on avance case par case depuis la position $\mathtt{pos}_{\mathtt{bv}}$ pour vérifier si on est arrivé à une case vide, et on calcule le nombre d'éléments de $1$.

\section{Récapitulatif des performances de notre structure de données}
%% le calcule de la complexité temporale et spacial.
\label{sec:recapitulatif_des_performances_de_notre_structure_de_donnees}

Nous pouvons maintenant définir deux théorèmes pour récapituler les performances de notre structure de données en se basant sur deux hypothèses formalisées avec les définitions suivantes:

\begin{definition}
Une fonction de hachage pour une chaîne de caractères est dite incrémentale, si après quelques pré-traitements faits sur une chaîne de caractères $x$ en entrée, le résultat du hachage de n'importe quelle chaîne de caractères à une distance $1$ de $x$ prend un temps constant. 
\end{definition}

\begin{definition}
Une fonction de hachage mappant les éléments de l'ensemble $U$ vers l'ensemble $V$ est dite complètement aléatoire, si elle se comporte comme si elle a été choisie aléatoirement parmi l'ensemble de toutes les fonctions possibles de $U$ dans $V$.
\end{definition}

L'espace mémoire utilisé lors de la construction et l'espace mémoire final (espace utilisé par notre structure de données) sont résumés par le théorème suivant :

\begin{theorem}
\label{theo:space_one_error}
Le dictionnaire présenté dans cette section occupe $O(n\log\sigma)$ bits (où $n$ est la taille totale des chaînes dans le dictionnaire). L'espace utilisé durant la construction est également de $O(n\log\sigma)$ bits. Si le compactage de la table de hachage est utilisé, le dictionnaire final occupe au plus un espace mémoire de $n(2\log\sigma+O(1))$ bits.
\end{theorem}

\begin{proof}
Les mots du dictionnaire donnés comme entrée occupent un espace de $n\log\sigma$ bits. Le dictionnaire des listes de substitutions a $n/\alpha=O(n)$ cases de taille $\log\sigma+r$ bits chacune dont $n$ seulement sont occupées.
Le dictionnaire exact a $(d)$ cases (rappelons que $d$ est le nombre total des mots dont la taille totale est $n$) qui occupent un espace total de $(n/\alpha)\log\sigma=O(n\log\sigma)$ bits. L'espace mémoire total utilisé par l'algorithme de construction est clairement en $O(n\log\sigma)$ bits, car il lit uniquement les entrées, et écrit les sorties qui occupent un espace total de $O(n\log\sigma)$ bits.
Quand le compactage est utilisé, l'espace du dictionnaire exact est réduit à $n\log\sigma+d(1+1/\alpha)+o(d)=n\log\sigma+O(d)$ bits, et l'espace utilisé par le dictionnaire des listes de substitutions est réduit à $n(1+r+\log\sigma)+o(n)=n\log\sigma+O(n)$ bits, pour un espace total de $n(2\log\sigma+O(1))$ bits.
% %\qed
\end{proof}

La performance de la construction de la structure de données et le temps de requête sont résumés dans le théorème suivant :

\begin{theorem}
\label{theo:time_one_error}
En supposant un hachage complètement aléatoire et incrémental, et que chaque liste de substitutions contient un seul élément, et étant donné le mot requête de longueur $m$, alors le temps de requête moyen est égal à $O(m\left\lceil\frac{m\log\sigma}{w}\right\rceil)$. Et étant donné un dictionnaire avec une longueur des mots totale égale à $n$, alors le temps de la construction de la structure de données est $O(n)$ en moyenne.
\end{theorem}

\begin{proof}
\label{prof:chap:hachage:liste_sub}
En supposant un hachage incrémental, et étant donné un mot de longueur $m$. Le calcul de la valeur de hachage pour chercher une liste de caractères (dans le dictionnaire des listes de substitutions) pour insertion à une position $i$ donnée dans le mot requête prend un temps amorti de $O(1)$. Cela donne comme résultat, un temps total de $O(m)$ pour les $O(m)$ positions du mot requête. %$m$ listes de substitutions accessibles par une recherche ou une insertion. 
L'hypothèse qui dit que chaque liste de substitutions est de taille $1$, signifie que tous les $n$ mots ($x_{i,j}=x_i[1..j-1]\phi x_i[j+1..m_i]$  avec $i\in [1..n]$) utilisés comme clés pour l'insertion des caractères dans le dictionnaire des listes de substitutions vont être tous distincts. 
De plus, il est bien connu que le coût d'insertion d'un élément dans une table de hachage avec sondage linéaire est de $O(1)$ en moyenne si on suppose l'utilisation d'un hachage aléatoire ~\cite{Knuth63noteson}.
Donc il est évident que l'insertion de chaque élément de la liste de substitutions prend un temps constant en moyenne. Cela donne un temps total moyen de $O(n)$ pour insérer tous les $n$ éléments (caractères) dans le dictionnaire des listes de substitutions pour tous les $d$ mots du dictionnaire.
Par conséquent, le temps total moyen pour la construction de la version non-compactée du dictionnaire des listes de substitutions est de $O(n)$
~\footnote{ Pour voir pourquoi l'hypothèse que chaque liste de substitutions de taille $1$ est utile pour assurer un temps d'insertion moyen constant pour les éléments de la liste de substitutions, on peut considérer la liste de substitutions de taille maximale de $\sigma$. Vu que tous les éléments de cette liste vont être mappés au même endroit dans la table de hachage, l'insertion de chaque élément doit parcourir tous les éléments qui sont déjà insérés dans la même liste de substitutions avant d'atteindre une case vide, résultant en un temps d'insertion $\Omega(\sigma)$.}.
Avec des arguments similaires, le temps de construction du dictionnaire exact est estimé à $O(n)$, car l'insertion d'un mot de longueur $m$ requiert un nombre constant d'accès au tableau de hachage, et chaque accès à un coût de $O(m)$. En faisant la somme de tous les temps d'insertion pour tous les $d$ mots du dictionnaire, le temps total moyen d'insertion devient $O(n)$.
La génération de la version compacte du dictionnaire des listes de substitutions depuis le dictionnaire non-compact prend trivialement $O(n)$. En fin de compte, le temps de construction de la version compacte et non-compacte du dictionnaire présentées dans ce travail est de $O(n)$ en moyenne.\\

Maintenant, on va prouver que le temps de recherche moyen pour un mot requête de longueur $m$ est $O(m\left\lceil\frac{m\log\sigma}{w}\right\rceil)$.
Rappelons que l'extraction des éléments d'une liste de substitutions requiert l'accès à des cases consécutives dans la table de hachage, en démarrant par la position $i$ jusqu'à arriver à une case vide $i+k$ et en collectant ensuite tous les $k$ caractères des cases parcourues.
On distingue deux cas: 
\begin{enumerate}
\item Le cas où la liste de substitutions n'existe pas, la position initiale $i$ va suivre une distribution uniforme sur toutes les cases de la table de hachage. Alors le nombre total des cases parcourues va suivre la borne d'une recherche infructueuse dans une table de hachage avec sondage linéaire, une borne connue comme étant constante en moyenne selon~\cite{Knuth63noteson}. 
\item Dans le cas d'une requête pour une liste de substitutions existante, le nombre de cases parcourues est également constant en moyenne. 
Le nombre de cases parcourues dans ce cas donnés par le temps moyen d'une recherche fructueuse dans une table de hachage, qui est connu également 
comme étant constant~\cite{Knuth63noteson}.
\end{enumerate}

Considérons maintenant que la table de hachage est décomposée en un ensemble de cases non-vides continues, et un ensemble de cases vides. De fait que les cases vides sont une fraction constante de toutes les cases de tableau, on conclut que choisir une position de départ d'une manière uniforme et aléatoire à partir des cases non-vides va générer un nombre constant de cases parcourues avant d'atteindre une case vide. Comme chaque position non-vide représente une liste de substitutions, on conclut que le nombre moyen de caractères candidats collectés lors de la recherche d'une liste de substitutions existante est aussi constant.
Cela implique que le nombre total des mots candidats générés depuis la recherche de toutes les listes de substitutions est égal à $O(m)$ en moyenne.
Le temps de vérification dans le dictionnaire exact pour un mot candidat est égal à $O(\left\lceil\frac{m\log\sigma}{w}\right\rceil)$, qui est le temps nécessaire pour comparer le mot candidat avec le mot du dictionnaire.

L'explication est que le dictionnaire exact est représenté en utilisant la table de hachage avec sondage linéaire, dans laquelle une requête parcourt un nombre de cases constant en moyenne. À chaque fois, nous avons besoin d'un temps de $O(\left\lceil\frac{m\log\sigma}{w}\right\rceil)$ pour comparer le mot du dictionnaire avec le mot candidat. Ceci prouve que le temps d'exécution final d'une requête est de $O(m\left\lceil\frac{m\log\sigma}{w}\right\rceil)$ en moyenne.
% %\qed
\end{proof}

Dans la pratique, nous nous attendons à ce que la plupart des listes de substitutions contiennent un seul élément (nous l'avons vérifié précisément sur les jeux de données testés). 
%% tabulation hashing : https://en.wikipedia.org/wiki/Tabulation_hashing
Comme il est fait dans la majorité des travaux pratiques, nous avons préféré d'utiliser les fonctions de hachage les plus rapides qui sont efficaces en pratique au lieu d'utiliser celles qui sont moins rapides avec de meilleures garanties théoriques~\footnote{Le lecteur peut se référer à~\cite{MV08} pour avoir plus de détails sur les raisons pour lesquelles les fonctions de hachage simples se comportent dans la pratique suivant le principe de l'aléatoire pure (elles se comportent en pratique comme si elles étaient entièrement aléatoires).}.

\section{Extension à deux erreurs ou plus }
\label{sec:Extension_deux_erreurs_ou_plus}

Notre algorithme peut être étendu pour fonctionner avec deux erreurs ou plus. Pour deux erreurs, on stocke deux \textbf{dictionnaire de listes de substitutions}. Le premier dictionnaire est utilisé pour une erreur (le niveau $1$), pour chaque mot $x$ dans le dictionnaire de longueur $m$, et pour toutes les positions $i \in [1..m]$, on place un caractère spécial (joker) $\phi$ dans la position $i$, donc ($x' = u \phi v$). On calcule la valeur de hachage du mot $x'$ donc ($h(x')$), et ensuite, on stocke le caractère $x[i]$ dans le dictionnaire des listes de substitutions. Au total, on a $m$ caractères à insérer. Ce dictionnaire est donc identique à celui utilisé plus haut pour le cas d'une seule erreur.\\ 

Le deuxième dictionnaire (le niveau $2$) va stocker pour chaque mot et chaque paire de positions distinctes, le caractère de la première position. Pour chaque mot $x$ de longueur $m$, et toute paire de positons $(i,j)$ tels que $i \in [1..m-1]$ et $j \in [i+1..m-1]$, on place deux caractères spéciaux $\phi$ dans les positions $i$ et $j$, donc $x''=u \phi v \phi w$ ($u,v,w$ sont des sous-chaines), et on calcule la valeur de hachage de $h(x'')$ afin de stocker le caractère original de la position $i$ dans ce deuxième dictionnaire.\\

Au moment de la requête, on interroge d'abord le dictionnaire des listes de substitutions de niveau $2$ avec le mot requête $q''=u \phi v \phi w$ pour récupérer la liste des caractères de substitutions $L_2=\{c_1,c_2,...,c_l\}$, grâce à la valeur de hachage $h(q'')$, qui nous donne la position dans la table de hachage de niveau $2$. Ensuite, pour chaque caractère $c_i$ récupéré de cette liste, on le substitue dans le premier joker du mot $q''$ (donc $q'=u c_i v \phi w$). Ensuite, on passe au dictionnaire des listes de substitutions de niveau $1$ pour trouver tous les caractères à substituer au deuxième joker à l'aide de la valeur de hachage $h(q')$. On obtient alors une liste $L_1=\{e_1,e_2,...,e_f\}$, et on substitue alors les caractères de la liste (tous les éléments $e_j$) dans la deuxième position, pour constituer les mots candidats de la forme $q=u c_i v e_j w$.  Enfin, on interroge le dictionnaire exact pour les mots résultants (avec la substitution des $2$ caractères trouvés dans le premier et deuxième dictionnaires de substitutions).\\

\noindent
Exemple :\\
Soit le mot $\texttt{ALABAMA}$. Pour illustrer cet exemple, on considère juste une seule combinaison de $2$ positions (les autres combinaisons peuvent être traitées de la même façon). Dans la phase de construction de la structure de données, on fait ce qui suit:
\begin{enumerate}

\item On stocke le mot $\texttt{ALABAMA}$ dans le dictionnaire exact. 

\item On stocke une liste de substitutions associée à $\texttt{ALABA}\phi\texttt{A}$ dans le dictionnaire des listes de substitutions de niveau $1$.

\item On stocke une liste de substitutions associée à $\texttt{A}\phi\texttt{ABA}\phi\texttt{A}$ dans le dictionnaire de niveau $2$.
\end{enumerate}

Au moment de la requête, on fait les étapes suivantes :

\begin{enumerate}
\item On interroge le deuxième dictionnaire des listes de substitutions avec $\texttt{A}\phi\texttt{ABA}\phi\texttt{A}$, pour obtenir un ensemble de caractères ($L_2=\{c_1,c_2,...,c_l\}$) qui pourraient être substitués au premier joker. Dans notre exemple, nous avons un seul élément dans la liste $L_2 : \{ c=L \}$.

\item Pour chaque caractère $c_i$ dans cet ensemble, on  interroge le niveau $1$ avec $\texttt{A} \, c_i \, \texttt{ABA}\phi\texttt{A}$, pour obtenir la liste de substitutions de la deuxième position. Nous avons un seul élément dans la première liste $L$, et donc on va interroger le dictionnaire des listes de substitutions avec $\texttt{ALABA}\phi\texttt{A}$. On trouve un seul élément aussi, le caractère $M$.

\item On fait la substitution pour le deuxième joker, et donc on trouve le mot $\texttt{ALABAMA}$, on vérifie l'existence dans le dictionnaire exact. Le mot est trouvé dans le dictionnaire et est donc ajouté à la liste des solutions approchée avec $2$ erreurs pour tous les mots de la forme $\texttt{A}\phi\texttt{ABA}\phi\texttt{A}$.
\end{enumerate}

Concernant l'utilisation de l'espace mémoire de ce schéma étendu à $2$ erreurs, on note que le mot $\texttt{ALABAMA}$ étant de longueur $\ell=6$, va générer $\frac{\ell(\ell-1)}{2}=\frac{6\times 5}{2}=15$ caractères pour les listes de substitutions.\\

Dans le cas général, on peut donner les théorèmes suivants qui résument la performance de notre structure de données pour $2$ erreurs. On commence d'abord par l'utilisation de l'espace mémoire.

\begin{theorem}
Le dictionnaire présenté dans cette section occupe $O(N\log\sigma)$ bits où :
$N=\sum_{i=1}^{d}{n_i(n_i-1)/2}$ et où $n_i$ est la longueur du mot numéro $i$ du dictionnaire. L'utilisation de l'espace mémoire lors de la construction est de $O(N\log\sigma)$ bits. Si le compactage est utilisé, alors le dictionnaire (l'index) final occupe un espace mémoire au plus de $(N+2n)(\log\sigma+O(1))$ bits.
\end{theorem}

\begin{proof}
On peut utiliser les mêmes arguments que pour le théorème~\ref{theo:space_one_error}. La seule différence étant que l'on a maintenant le dictionnaire des listes de substitutions de niveau $2$ dans lequel on insère $N$ éléments. La borne spatiale se calcule facilement.
% % le dictionnaire excate + dictionare des listes de substition niveau 1  + dictionnaire de niveau 2 où on utilise 'N'.
% %\qed
\end{proof}

Ensuite, analysant le temps de la requête on obtient le théorème suivant : 
\begin{theorem}
En supposant un hachage totalement aléatoire et incrémental, et que chaque liste de substitutions contient un seul élément, avec un mot requête donné de longueur $m$, alors le temps de requête moyen est de $O(m^2\times \left\lceil\frac{m\log\sigma}{w}\right\rceil)$.

Étant donné un dictionnaire de $d$ mots, tels que $N=\sum_{i=1}^{d}{n_i(n_i-1)/2}$
où $n_i$ est la longueur du mot numéro $i$ du dictionnaire, le temps de construction de la structure de données est de $O(N)$ en moyenne.
\end{theorem}

\begin{proof}
De même que pour la borne spatiale, la borne temporelle peut être également prouvée en utilisant les mêmes arguments utilisés pour prouver le théorème~\ref{theo:time_one_error}.
Le coût des requêtes et des insertions dans le dictionnaire des listes de substitutions de niveau $2$ peut être délimité de la même manière que celui du dictionnaire de niveau $1$, sauf qu'on remplace le terme $m$ par $m^2$ pour les requêtes, car chaque requête fait $O(m^2)$ accès en plus au dictionnaire de niveau $2$. On  remplace également le terme $n$ par $N$ dans le temps de la construction, puisqu'on insère $N$ éléments dans le dictionnaire des listes de substitutions de niveau $2$.
% %\qed
\end{proof}

Nous avons seulement implémenté l'algorithme pour deux erreurs, mais notre dictionnaire (la structure de données) peut être étendu pour gérer $K>2$ erreurs. Pour cela, on utilise des dictionnaires de listes de substitutions jusqu'à $K$ niveaux. Le nombre de caractères spéciaux substitués dans chaque mot est égal à $k$, et le caractère stocké dans le niveau $K$ est le premier caractère (parmi les $K$ caractères à substituer par $\phi$). Pour chaque niveau, on doit faire toutes les combinaisons possibles pour remplacer les $k$ caractères par $\phi$. Un mot de longueur $n_i$, ajoute ${{n_i}\choose{k}}\times \log\sigma$ bits au niveau $k$ pour le dictionnaire ($1\leq k\leq K$).

\section{Expérimentation}
\label{sec:hash:experimentation}

Notre implémentation est modulaire.
Nous avons fait notre expérimentation avec les deux jeux de données suivants: l'ensemble des titres de Wikipédia \textbf{WikiTitle}, qui a environ 1,8 millions de mots, et le dictionnaire Anglais qui comporte environ 213 milles mots~\footnote{Nous remercions Dennis Luxen pour nous avoir fournit les jeux de données.}. Ces deux ensembles de données sont également utilisés dans les expérimentations de Karch et al~\cite{karch2010improved}. Ils ont également utilisé deux autres jeux de données plus petits: Mobydick et Town.

Nous avons expérimenté avec \textbf{Mobydick}, mais étant donné la petite taille de ce jeu de données (37 mille mots), le résultat n'a pas donné un bon aperçu des performances (la structure de données totale tient dans le cache du processeur). Nous n'avons pas fait d'expérimentations sur \textbf{Town} car il n'était pas disponible et de toute façon, le résultat n'aurait pas été très différent de \textbf{Mobydick} étant donné que le fichier ne contient que 47 mille mots.

%% je peut donnée ici les résultat de mobydick, mais c pas oubligatoire.... c juste un plus...

Nous avons implémenté la structure de données, les algorithmes de construction et de recherche en \textbf{langage C} utilisant le compilateur GNU GCC version 4.4.1. Les tests ont été effectués avec un processeur Intel E8400 3.0 Gigahertz Core 2 duo, exécutant Windows 7. Nous avons utilisé un seul c\oe{}ur de processeur.

Nous avons comparé nos résultats avec ceux de Karch et al. Les raisons pour lesquelles nous avons choisi Karch et al. est que ces derniers ont trouvé que leurs résultats étaient meilleurs que ceux des approches concurrentes.
Comme l'implémentation de Karch et al. n'est pas accessible au public, nous avons juste pris les résultats tels que publiés dans l'article de Karch et al. Nous notons que le matériel utilisé par Karch et al. est comparable au nôtre, bien qu'ils ne soient pas identiques~\footnote{Leurs expérimentations sont effectuées sur un Intel Xeon X5550 CPU avec 2.67 GHZ, qui est une machine avec une performance très proche de la n\^{o}tre.}.
Ainsi, la comparaison des temps de requête n'est pas tout à fait exacte. Toutefois, étant donné que la capacité des deux machines est très proche, nous pouvons considérer la comparaison précise jusqu'à une marge d'erreur entre $10\%$ et $20\%$. D'autre part, l'utilisation de l'espace mémoire est directement comparable, car elle ne dépend pas du matériel utilisé.

Nous n'avons pas comparé avec les résultats de Boytsov. Bien que le code source soit disponible, nous n'avons pas réussi à le faire fonctionner.
Nous aurions pu comparer aux résultats de Boytsov~\cite{Bo12} tels qu'indiqués dans son article, mais il n'y a pas de tableaux indiquant la performance des algorithmes dans son papier. La seule alternative aurait été de deviner les résultats des performances depuis les figures, mais cela aurait donné des résultats trop approximatifs. Néanmoins, d'après les similitudes avec notre méthode, nous nous attendons à ce que la méthode de Boytsov ait les mêmes bornes en termes d'espace mémoire et de temps de requête. Cependant, sa méthode est inférieure à la nôtre dans les aspects suivants :

\begin{enumerate}
\item La structure de données est très lente à construire ~\cite{Bo12}, beaucoup plus lente que les autres compétiteurs étudiés dans son article. Notre structure de donnés est plus rapide que celle de Karch et al. qui est la plus compétitive.
Le temps de construction de sa structure peut aller jusqu'à $1$ heure pour les grands jeux de données. En revanche, notre temps de construction est de moins d'une minute pour tous les ensembles de données étudiés.
 
\item La construction de notre structure de données nécessite $O(n\log\sigma)$ bits contre $O(n\log n)$ bits pour l'implémentation de Boytsov. Ceci est très important pour les périphériques mobiles et embarqués où l'espace mémoire est un problème.

\item L'utilisation du hachage parfait minimal conduit à utiliser un espace mémoire élevé au moment de la construction, même si l'espace final est petit.
La méthode de hachage parfait minimal exige habituellement au moins $12$ octets par entrée, même si l'espace final est inférieur à $3$ bits. En revanche, l'espace de construction de notre variante non-compacte est exactement le même que l'espace final.
\item La méthode n'est pas incrémentale. L'insertion de nouveaux mots nécessite la reconstruction de toute la structure de données.
\end{enumerate}

\medskip
Pour réaliser nos tests, nous avons choisi au hasard $1000$ mots de notre dictionnaire et nous avons appliqué une opération d'édition choisie aléatoirement sur chacun d'eux, et ensuite, interrogé le dictionnaire. Pour obtenir le temps de requête final, nous divisons le temps par $1000$. Nous avons répété la procédure $20$ fois, et nous avons pris la moyenne des résultats obtenus.\\

Nous avons expérimenté avec différents facteurs de chargement (le paramètre $LF$ ou $\alpha$) allant de $0,3$ jusqu'à $0,7$. Les résultats sont résumés dans les tableaux~\ref{table:compar_english1}~,~\ref{table:compar_wiki1}~,~\ref{table:compar_english2}~,~\ref{table:compar_wiki2}. Les résultats pour les versions compactées sur le dictionnaire anglais sont présentés dans la figure~\ref{fig:plot_compact_english}. Pour Karch et al. juste un seul point est montré (dans le graphe d'une erreur) parce que le schéma pour une erreur n'admet qu'une seule variante (il ne permet pas de faire varier le temps de réponse en fonction du changement de la taille du dictionnaire).

\begin{figure}
\centering
\includegraphics[width=0.5\linewidth]{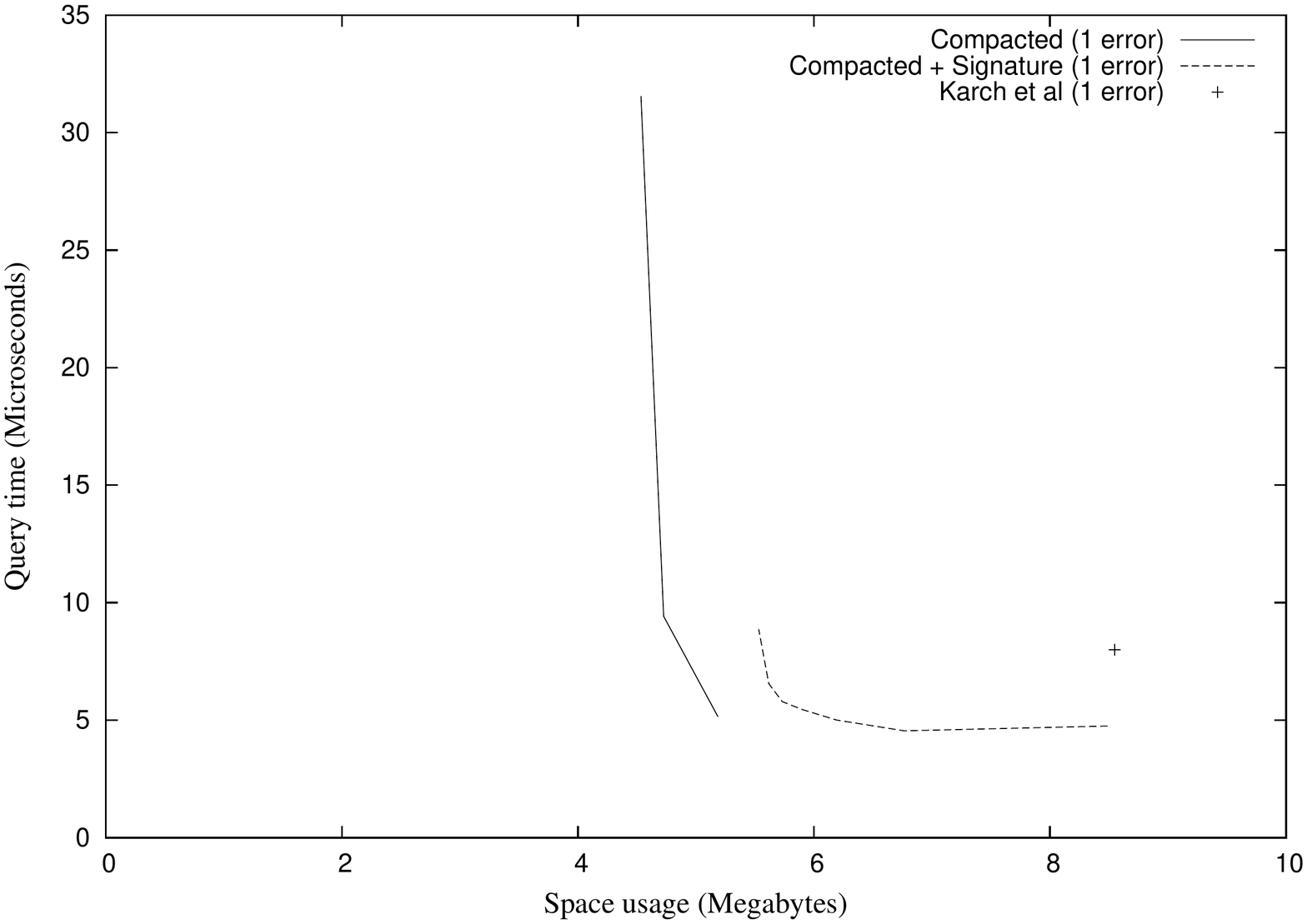}\hfill
\includegraphics[width=0.5\linewidth]{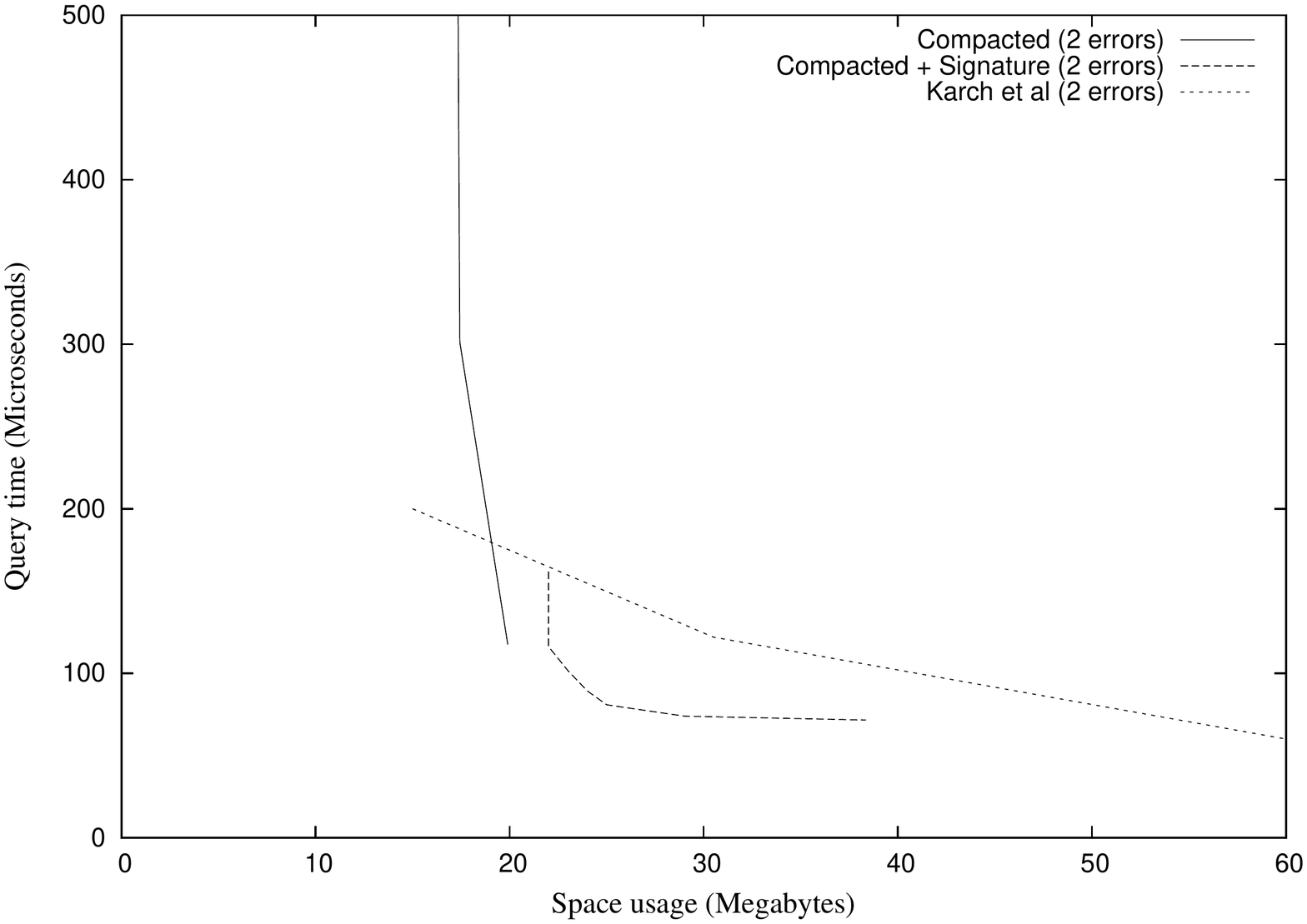}
\caption{L'utilisation de l'espace mémoire en fonction du temps de requête sur le dictionnaire Anglais (En) avec une (1) et deux (2) erreurs (versions compactes).}
\label{fig:plot_compact_english}
\end{figure}

%% --------------------------------
\begin{table}
\small
\begin{center}
\begin{tabular}{|l|c|c|c|}
\hline
Méthode & Temps Constr. (secondes) & Espace (Mo) & Temps de requête ($\mu s$)\\
\hline
Karch et al     & 3.45
		& 8.55 
		& 8 \\

Non comp.	(LF 0.7) & 0.4
		& 5.72 
		& 27.6 \\
Non comp.  + Sign. (LF 0.7) 
		& 0.44 
	     	& 7.15
		& 6.24 \\
Compacte (LF 0.7)	& 0.43
		& 4.53
		& 31.6 \\

Compacte (LF 0.3)	& 0.47
		& 5.19
		& 5.154 \\

Comp. + Sign. (LF 0.7)
		& 0.48
		& 5.53
		& 8.9 \\
Comp. + Sign. (LF 0.3)
		& 0.625
		& 6.77
		& 4.55 \\
\hline

\end{tabular}

\end{center}
\caption{Comparaison des méthodes existantes sur le dictionnaire Anglais (En) pour une erreur.}
\label{table:compar_english1}
\end{table}

\begin{table}
\small
\begin{center}
\begin{tabular}{|l|c|c|c|}
\hline
Méthode & Temps Constr. (secondes) & Espace (Mo) & Temps de requête ($\mu s$)\\
\hline
Karch et al     & 32.29
		& 55.84
		& 34 \\
Non comp. + Sign. (LF 0.7) 
		& 4.61 
	     	& 54.44
		& 8.44 \\
Compacte (LF 0.5)	& 4.72
		& 36.27
		& 16.6 \\
Comp. + Sign. (LF 0.7)
		& 4.91
		& 42.3
		& 13.5 \\
Comp. + Sign. (LF 0.3)
		& 5.17
		& 47.5
		& 8.28 \\

\hline

\end{tabular}

\end{center}
\caption{Comparaison des méthodes existantes sur le dictionnaire WikiTitle (Wi) pour une erreur.}
\label{table:compar_wiki1}
\end{table}

\begin{table}
\small
\begin{center}
\begin{tabular}{|l|c|c|c|}
\hline
Méthode & Temps Constr. (seconds) & Espace (Mo) & Temps de requête ($\mu s$)\\
\hline
Karch et al     & 12.596
		& 30.49 
		& 122 \\
Non comp. + Sign.	(LF 0.7) 
		& 2.612 
	     	& 27.87
		& 117 \\
Compacte (LF 0.7) & 2.54
		& 16.6
		& 2309\\

Comp. + Sign. (LF 0.7)
		& 2.802
		& 22.8
		& 162 \\
Comp. + Sign. (LF 0.4)
		& 2.95
		& 24.44
		& 89.27 \\
Comp. + Sign. (LF 0.2)
		& 3.1
		& 29.08
		& 73.89 \\
\hline
\end{tabular}

\end{center}
\caption{Comparaison des méthodes existantes sur le dictionnaire Anglais (En) pour deux erreurs.}
\label{table:compar_english2}
\end{table}

\begin{table}
\small
\begin{center}
\begin{tabular}{|l|c|c|c|}
\hline
Méthode & Temps Constr. (seconds) & Espace (Mo) & Temps de requête ($\mu s$)\\
\hline
Karch et al     & 107.289
		& 170.79
		& 502 \\
Non comp. + Sign.	(LF 0.7) 
		& 21.29
	     	& 203.82
		& 503 \\

Compacte (LF 0.5)
		& 21.36
		& 127.87
		& 1993	\\
Comp. + Sign. (0.7)
		& 22.86
		& 162.57
		& 1271 \\
\hline

\end{tabular}
\end{center}
\caption{Comparaison des méthodes existantes sur le dictionnaire WikiTitle (Wi) pour  deux erreur.}
\label{table:compar_wiki2}
\end{table}

D'après les résultats, nous pouvons faire les observations suivantes:
\begin{enumerate}
\item Tous nos algorithmes sont beaucoup plus rapides dans la construction que l'algorithme de référence. Ils sont généralement entre $5$ et $10$ fois plus rapide.

\item L'algorithme pour une erreur est très robuste. Il est $4$ fois plus rapide que le résultat de Karch et al. sur les très grands jeux de données (les titres de Wikipédia "WikiTitle").
 
\item Le compactage donne un compromis très intéressant entre l'espace mémoire et le temps.

\item Le temps de requête de notre méthode pour deux erreurs n'est pas si bon, surtout sur les très grands jeux de données.
Toutefois, il fournit toujours un compromis pertinent espace/temps. Même avec le plus grand jeu de données, nous pouvons obtenir un avantage très significatif en terme  d'espace mémoire au détriment d'un temps de requête plus important.
% % au détriment de = 3ala Hissabe ...
\end{enumerate}

Les variantes non-compactes de notre index supportent naturellement l'insertion de nouveaux mots. Nous avons fait quelques expérimentations pour appuyer notre affirmation selon laquelle notre index prend en charge efficacement l'insertion de nouveaux éléments.
Pour cela, nous avons construit d'abord notre index sur une fraction de $0,7$ des mots dans le jeu de donnée WikiTitle avec un facteur de chargement de $0,7$. Nous avons ensuite inséré une fraction de $0,25$ de l'ensemble de données original, résultant en un facteur de chargement final de $0,95$.

Nous avons mesuré le temps d'insertion~\footnote{Nous avons fait des expérimentations sur une machine différente avec des performances légèrement plus faible (Intel Xeon avec 2.27 GHz, 2 Gigaoctets de RAM,  sous Windows XP).} pour des fractions de $0,05$ successives du fichier original. Les résultats sont présentés dans le tableau~\ref{table:insert_test_wiki}. Comme nous pouvons le voir dans le tableau, l'insertion ne prend que quelques microsecondes pour le dictionnaire pour une erreur, et de quelques dizaines de microsecondes dans le dictionnaire pour deux d'erreurs.

%% table here
\begin{table}
\small
{\centering
\begin{tabularx}{\textwidth}{|X|X|X|}

\hline

Facteur de chargement (LF)& \multicolumn{2}{>{\centering\setlength{\hsize}{2\hsize}\addtolength{\hsize}{2\tabcolsep}}X|}
{Temps d'insertion moyen en $\mu s$} \\

\cline{2-3}&1 erreur      &2 erreurs\\
\hline
0.7-0.75&2.42&13.34\\
0.75-0.80&2.44&13.83\\
0.80-0.85&2.52&16.66\\
0.85-0.90&3.52&30.11\\
0.90-0.95&5.56&43.74\\

\hline

\end{tabularx}
}
\caption{Le temps d'insertion dans l'index avec le jeu de données WikiTitle pour $1$ et $2$ erreurs.}
\label{table:insert_test_wiki}

\end{table}

Nous avons vérifié d'une façon expérimentale l'hypothèse que la plupart des listes de substitutions (pour une erreur) contiennent juste un seul élément. À partir du tableau~\ref{table:subst_list_stats}, nous pouvons voir que plus de $96\%$ des caractères sont dans des listes de substitutions de taille $1$.

%% table here
\begin{table}[h]
\small
\begin{center}
\begin{tabular}{|c|c|c|}
\hline
Taille de liste & Pourcentage Anglais (En) & Pourcentage WikiTitle (Wi)\\
\hline
1	        & 98.89\%
		& 96.16\%\\

2	        & 0.71\% 
		& 1.63\%\\

3		& 0.19\% 
		& 0.52\%\\

4		& 0.09\%
		& 0.29\%\\

5		& 0.05\%
		& 0.19\%\\

$\geq 6$	& 0.08\%
		& 1.20\%\\

\hline

\end{tabular}

\end{center}
\caption{Pourcentage des éléments dans les listes de substitutions pour une taille donnée pour les deux jeux de données.}
\label{table:subst_list_stats}
\end{table}

\section{L'application de notre méthode dans l'indexation du texte}
\label{sec:hash:application_notre_methode_dans_texte}

Notre solution est adaptable à la recherche des occurrences des mots dans un texte sans perte de performances.
Rappelons qu'un dictionnaire $D=\{w_1,w_2,...,w_d\}$ avec $d$ mots, et $w_i \neq w_j $ ($i \neq j$).
Un texte 'T' peut être vu comme étant un ensemble de mots, où chaque mot peut avoir plusieurs occurrences, $T=\{w_1,w_2,...,w_n\}$, et il se peut que $w_i = w_j$ ($i \neq j$).
L'index de Texte 'T' est un ensemble de mots où un mot peut avoir plusieurs occurrences (des positions).
On peut écrire la formule 'T' comme suit : $T'=\{l_1,l_2,...,l_d\}$ où chaque mot peut avoir plusieurs occurrences $l_i=\{pos_1,pos_2,...,pos_{nb}\}$.\\

Afin d'adapter notre solution pour la recherche des mots dans un texte, il suffit d'ajouter dans la structure de données de dictionnaire exact, pour chaque mot un champ qui pointe sur une liste finie de ses positions.

\paragraph{La recherche des facteurs dans le texte : }~\\
Pour la recherche de n'importe quel facteur dans le texte, on peut remplacer notre dictionnaire exact avec un arbre des suffixes pour la vérification des mots candidats, le dictionnaires des listes de substitutions des deux niveaux restent les mêmes. On insère le caractère de substitution seulement s'il n'existe pas dans le dictionnaire des listes de substitutions.

\section{Conclusion}

Dans ce chapitre, nous avons présenté un algorithme robuste et efficace pour la recherche approchée pour $k \geq 2$ dans un dictionnaire en utilisant les tables de hachages.
Les algorithmes que nous avons proposés sont faciles à implémenter et très efficaces dans la pratique.

L'algorithme pour une erreur utilise grossièrement deux fois la taille du dictionnaire exact. Le temps de la requête est attendu comme étant linéaire, en se basant sur certaines hypothèses qui sont presque toujours rencontrées en pratique.
L'algorithme pour deux erreurs à un espace d'utilisation linéaire, et en pratique, il est plus lent que le premier. Toutefois, Il atteint un compromis pertinent entre temps/espace.

Les versions non-compactes de nos structures sont incrémentales (avec un temps d'insertion, au pire quelques microsecondes).\\

Dans le chapitre suivant, on va explorer une autre méthode pour résoudre le problème de la recherche approchée. L'idée se base sur le concept de filtrage. Le but est de trouver les parties exactes rapidement de la requête en utilisant le Trie, et de localiser la position d'erreur afin de trouver les solutions possibles rapidement.

\chapter{Un algorithme rapide de recherche approchée de motifs dans un dictionnaire basé sur le Trie et le Trie inversé.}
\label{chap:Trie_R_Tire}

\ifpdf
    \graphicspath{{Trie_R_Trie/}}
\else
    \graphicspath{{Trie_R_Trie/}}
\fi

\section{Introduction}

Dans ce chapitre, nous proposons une solution basée sur le concept de filtrage \cite{Na01, Na01b}. L'algorithme parcourt tous les mots du dictionnaire afin de localiser rapidement les candidats possibles qui peuvent être des solutions.

Chaque mot approché contient quelques sous-chaînes exactes. Nous cherchons ces pièces exactes en utilisant une structure de données bidirectionnelle (un Trie et un Trie inversé). Nous utilisons une méthode de vérification pour obtenir les mots candidats et retourner toutes les occurrences qui sont au plus différentes du mot requête par une seule erreur.

\medskip
\paragraph{Ce chapitre est organisé comme suit : }
Dans la section \ref{sec:our_method}, nous détaillons les étapes de notre nouvelle approche (de la recherche approchée) qui réduit le nombre de transitions sortantes testées dans chaque n\oe{}ud dans le Trie.
Dans la section \ref{sec:the_second_method}, nous proposons une méthode qui utilise le Trie et le Trie inversé, et qui améliore la méthode de Amir et al \cite{Amir2000}, grâce à l'utilisation d'heuristique.
Dans la section \ref{sec:the_third_method}, nous expliquons notre troisième méthode, qui est une combinaison des deux méthodes précédentes.
La section \ref{sec:testes} porte sur les tests et les expérimentations.
Dans la section \ref{sec:TRT:application_notre_methode_dans_texte}, nous expliquons comment adapter notre méthode dans l'indexation de texte. Dans la dernière section \ref{sec:auto:conclusion} nous donnons une conclusion de ce chapitre.

\section{Une nouvelle approche pour réduire le nombre de transitions sortantes testées dans chaque n\oe{}ud}
\label{sec:our_method}

\subsection{L'idée de base}

Soit $ D = \{{x_1, x_2, ..., x_d} \} $ un dictionnaire défini sur l'alphabet  $\Sigma$, et soit $P$ un mot requête de longueur $|P|=m$, $P=c_1 c_2 \cdots c_m$.
Supposons l'existence dans $P$ d'une seule erreur; alors  $P$ peut s'écrire $P = P_1 \phi P_2$ où $P_1=c_1 c_2 \dots c_{i-1}$, $\phi$: représente l'erreur à la position $i$, et  $P_2= c_{i+1} c_{i+2} \dots c_m$.

Si on suppose qu'on a une solution approchée $P'\in D$ et $P' = P_1\:a\:P_2 , \; a\in \Sigma$  où $P_1$ et $P_2$ sont deux fragments exacts, alors $P_1 a$ et  $a P_2$ doivent exister respectivement dans 
le Trie et dans le Trie inversé.

\medskip
Nous pouvons conclure de cette observation que, seul le chemin sortant marqué par la lettre $a$ doit être vérifié pour trouver l'autre pièce exacte $P_2$ dans le Trie.

Nous devons obtenir toutes les lettres $a \in \Sigma $ qui satisfont cette propriété $P' = P_1 a P_2$.

\medskip
L'idée revient à faire une recherche exacte de $P_1$ dans le Trie (recherche de gauche vers la droite), et une recherche exacte pour $P_2$ dans le Trie inversé (recherche de droite à gauche) pour le mot requête, puis faire une intersection entre toutes les transitions sortantes des deux n\oe{}uds qui représente la position de l'erreur, pour avoir les caractères qui pourraient conduire aux solutions.

\medskip

La dernière étape consiste à vérifier si $P_1$ et $P_2$ appartiennent au même mot, pour s'assurer que la réponse n'est ni de la forme $P_1 a Z $, ni de la forme $Q a P_2$ où $Q, Z$ sont deux sous-chaînes différentes de $P_1, P_2$ respectivement.

\bigskip
Dans la prochaine sous-section, nous expliquons et nous détaillons notre méthode qui s'appuient sur le Trie et le Trie inversé afin de réduire le nombre des transitions sortantes dans chaque n\oe{}ud, et vérifier que quelques chemins seulement, ceux qui peuvent conduire à des solutions 
(voir la figure \ref{fig:reductiondechemin}).

\subsection{Principe de l'algorithme}

Soit $P = P_1 \phi P_2$ le mot requête  où $P_1$  et $P_2$ sont deux fragments exacts, et $\phi$ un caractère joker qui représente la position de l'erreur.

Notons $v_{err}$ le n\oe{}ud du Trie de la position de l'erreur (où la recherche s'arrête). De façon similaire, notons $w$ le n\oe{}ud où la recherche de $P_2$ s'arrête dans le Trie inversé.

\medskip

Soit $A=\{a_1,a_2,...,a_{L1} \}$ , et $B=\{b_1,b_2,...,b_{L2}\}$ les ensembles des caractères qui étiquettent les transitions sortantes de $v_{err}$ et $w$ respectivement.

\medskip

Afin de trouver les caractères communs, on effectue une intersection entre les deux ensembles :\\
$C= A  \cap B = \{c_1,c_2,...,c_{L3}\}$.\\
Si une solution approchée existe, cela signifie qu'il existe un caractère $c_k$ avec $k \in [1..L3]$, et le mot $P'= P_1 c_k P_2 , \: P'\in D$ est une solution.\\

\medskip

Après l'opération de l'intersection, l'ensemble résultat $C$ contiendra 2 types d'éléments :
\begin{enumerate}
\item Les caractères qui mènent à des solutions : dans ce cas, on a  $P_1 c_k$, et $c_k P_2$ et les deux sous-chaînes ($P_1, P_2$) sont dans le même mot.
\item  Les caractères qui ne mènent pas aux solutions : lorsqu'on a $P_1 c_k$ et $c_k P_2$  ne sont pas dans le même mot, autrement dit  $P_1 c_k Z$ et $Q c_k P_2$ sont deux mots différents.
\end{enumerate}

\medskip
En raison de ces deux cas, nous devons vérifier dans le Trie que $P_1 c_k P_2$ est dans le même mot. \\

Remarque :\\
Sachant  ($P_1 c_k$) est déjà trouvé dans le Trie, il reste juste à tester l'existence de ($P_2$) dans le Trie sur le même chemin étiqueté par ($c_k$). De même, nous pouvons faire cette vérification dans le Trie inversé, c'est-à-dire tester l'existence de $P_1$ sur le chemin étiqueté par $c_k$ qui sort de $P_2$.

\medskip
Pour réduire le nombre de comparaisons, nous vérifions la plus petite chaîne d'entre $P_1$ et $P_2$. En effet, si ($|P_1|>|P_2|$) on vérifie $P_2$ dans le Trie, autrement, on vérifie $P_1$ dans le Trie inversé.

\bigskip
Pour trouver toutes les solutions, nous devons déterminer toutes les positions possibles de l'erreur.
 
\subsection{Les positions possibles de l'erreur}
\label{Possible_error_positions}

Pour un motif donné $P = P_1 \phi P_2$, $|P| =|P_1|+1+|P_2|=m$. La recherche exacte de ($P$) dans le Trie donne ($P_1$); la position d'erreur est automatiquement trouvée lorsqu'on ne peut pas avancer dans le Trie.

\medskip
Puisque $k = 1$, seuls les cas suivants sont possibles (voire Fig. \ref{fig:error_position_T_RT}) :

\begin{enumerate}[a.]
\item Il n'est pas possible de trouver une autre erreur après $|P_1|+1$, sinon nous comptabilisons deux erreurs. Donc, Lorsque nous cherchons $P_2$ ($|P_2| = m-|P_1|-1$) dans le Trie inversé, si nous trouvons une erreur avant la longueur $|P_2|$, par conséquent, il n'existe aucune solution. Soit $q$ le fragment (la sous-chaîne) trouvé dans le Trie inversé avant la position d'erreur $w$ (qui satisfait $|P_2|$), avec $|q|<|P_2|$. Cela signifie que nous avons $|P_2|-|q|$ positions à ne pas vérifier.

\item Soit $V=\{v_1, v_2, ..., v_k, v_{err} \}$ l'ensemble des n\oe{}uds sur le chemin qui mène de la racine à $v_{err}$. Tous les ($v_i \in V$) peuvent être une position d'erreur.

\medskip
Lorsque nous appliquons une distance d'édition, un autre chemin (différent de celui de la recherche exacte) peut conduire à une solution, voire Fig. \ref{fig:error_position_T_RT}.

\begin{figure}[h!]
\centering
\includegraphics[width=0.7\linewidth]{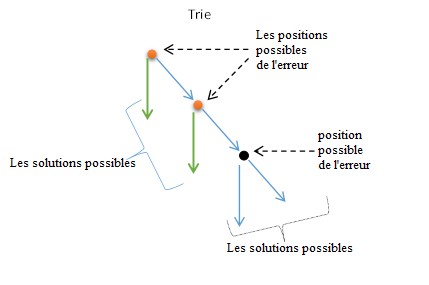}
\caption{Toutes les positions possibles de l'erreur  dans le Trie et le Trie inversé. Toutes les transitions sortantes de ces n\oe{}uds mènent à des solutions possibles.}
\label{fig:error_position_T_RT}
\end{figure}

Comme expliqué précédemment, dans chaque position possible de l'erreur, seulement quelques chemins seront examinés (voire Fig. \ref{fig:reductiondechemin}).

\begin{figure}[h!]
\centering
\includegraphics[width=0.7\linewidth]{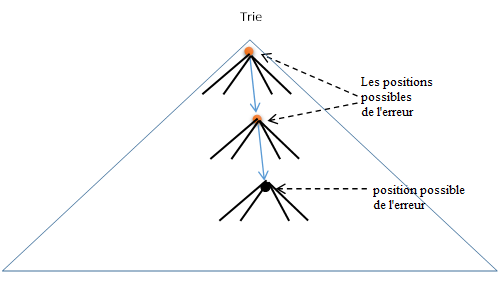}
\caption{Quelques transitions sortantes à partir du n\oe{}ud erreur seront vérifiées.}
\label{fig:reductiondechemin}
\end{figure}

\end{enumerate}

% % cett algo est juste pour le cas de la substitution, donc on doit ajouter les deux autres cas, l'insertion et la suppresion :)
% % déjà cette remarque été parmet les point sulver par les reviwer de CPM.

\subsection{L'algorithme de recherche}
\label{algo:TRT_CI}

% -----------------------------------------------------------------
\noindent
\rule{8cm}{0.1pt}\\
\textbf {Algorithme de la recherche approchée TRT\_CI}\\
\textbf{Entrée :} un mot requête $q$ de longueur $m$, les structures de données Trie et Trie inversé.\\
\textbf{Sortie :} liste des mots solutions approchées.\\

\begin{enumerate}[a.]
\item \label{itm:e_1} Faire une recherche exacte dans l'arbre Trie, jusqu'à ce que l'on arrive à la position de l'erreur. Soit  $v_{err}$ le n\oe{}ud qui représente cette erreur. Soit $P_1$ le préfixe avant cette position.

\item \label{itm:e_2} Faire une recherche pour le suffixe $P_2$ dans le Trie inversé, jusqu'à ce que la longueur de ce suffixe soit consommé, ($|P_2|=m-|P_1|-1$). Soit $w$ le n\oe{}ud dans lequel la recherche s'arrête dans le Trie inversé.

\item \label{itm:e_3} Si on ne peut pas trouver $P_2$,  l'algorithme s'arrête (il n'y a pas de solution).

\item \label{itm:e_4} Sinon, on a deux n\oe{}uds $v_{err}$ et $w$;  on calcule une intersection entre les deux ensembles des caractères des transitions sortantes de ces deux n\oe{}uds. Soit $C$ l'ensemble des caractères résultant.

\item \label{itm:e_5} Pour chaque caractère $c_k \in C$, continuer une recherche exacte en suivant la transition sortante étiquetée par $c_k$ comme suit:

	\begin{itemize}
	\item Si ($|P_1|>|P_2|$) alors continuer la recherche exacte de $P_2$ dans le Trie.
	\item Sinon continuer la recherche exacte de $P_1$ dans le Trie inversé.
	\end{itemize}

\item \label{itm:e_6} Pour chaque $v_i \in V$ où $V=\{v_1, v_2, ..., v_k \}$ sont les n\oe{}uds sur le chemin qui mène de la racine à $v_{err}$, répétez les étapes : \ref{itm:e_2}, \ref{itm:e_3}, \ref{itm:e_4} et \ref{itm:e_5}

\end{enumerate}
\rule{8cm}{0.1pt}\\

\subsection{Des cas particuliers}
\label{Special_cases}

Deux cas particuliers doivent être traités séparément: le cas où $P_1$ ou bien $P_2$ sont vides, et le cas où l'on trouve une solution exacte pour le motif recherché $P$.

\begin{itemize}

\item Le cas où l'une des deux sous-chaînes est vide (préfixe ou suffixe) :\\
Soit $P = P_1 \phi P_2$ et $P_1=\epsilon$ ou $P_2=\epsilon$, donc, on a $P= P_1 \phi$  ou bien  $P = \phi P_2$. Dans ce cas, l'intersection est faite entre un n\oe{}ud interne et le n\oe{}ud racine. Exécuter une intersection dans ce cas serait inefficace. Pour éviter ce calcul, nous appliquons une distance d'édition (que la suppression et la substitution) à la position finale de la sous-chaîne non-vide.
\medskip
\item Le cas où l'on trouve une solution exacte :
\medskip
Lorsqu'on effectue une recherche exacte sur le mot requête dans le Trie et dans le Trie inversé et aucune erreur n'est rencontrée.

\medskip
Soient $n_1$ et $n_2$ les n\oe{}uds dans lesquels la recherche exacte se termine dans le Trie et le Trie inversé, respectivement.
Nous traitons ce cas comme suit :
\medskip
	\begin{itemize}
	
	\item La première étape : dans les deux n\oe{}uds $n_1$ et $n_2$, seule l'erreur d'insertion est examinée.
	\medskip	
	\item La deuxième étape : après la première étape, on aura $P = P_1 \phi P_2$, où $|P_1 |\geq 1$ et  $|P_2|\geq 1$, tous les n\oe{}uds sur le chemin qui mène de la racine à $n_1$ ( ou $n_2$) seront considérés comme des n\oe{}uds de l'erreur.
	\end{itemize}

\end{itemize}

\subsection{Le cas de l'insertion et la suppression}

\paragraph{La suppression : }
Soit un mot requête $P = P_1 \phi P_2$ avec $|P|= m$. Soit $v_{err}$ et $w$ deux n\oe{}uds qui représentent la profondeur de $|P_1|$ et $|P_2|$ respectivement.
Pour faire une recherche avec une suppression, il suffit juste de sauter un seul caractère (la position de l'erreur) dans le mot requête. Cela revient à faire une recherche exacte de mot $P' = P_1 \: P_2$.

Lorsque on applique l'algorithme \ref{algo:TRT_CI} pour faire la recherche des solutions avec substitution, en même temps, on cherche les solutions avec une suppression comme suit :
Si $|P_1| > |P_2|$, alors on continue à vérifier $P_2$ dans le Trie à partir du n\oe{}ud $v_{err}$. Sinon, on continue à vérifier $P_1$ dans le Trie inversé à partir du n\oe{}ud $w$.

\paragraph{Le cas de l'insertion :}
Soit le mot requête $P = P_1 \: P_2$ d'une longueur $m$. Soit le pattern $P' = P_1 \phi P_2$ qui représente une solution approchée avec une insertion, avec $|P'|=m+1$, et $\phi$ est la position de l'erreur de type insertion.

Pour trouver toutes les solutions approchées avec une insertion, on fait une recherche avec substitution avec le même algorithme  \ref{algo:TRT_CI} sur le mot requête $P' = P_1 \phi P_2$.
On cherche la première partie $|P_1|$ dans le Trie pour déterminer le n\oe{}ud $v_{err}$, ensuite, on cherche $P_2$ dans le Trie inversé ($|P_2| = m - |P_1|$) pour déterminer le n\oe{}ud $w$, ensuite, on fait une intersection entre les caractères des fils des deux n\oe{}uds afin de continuer la vérification de l'un des parties ($P_1$ ou $P_2$) en suivant les caractères en commun.

\subsection{L'intersection}

Chaque n\oe{}ud à au maximum $\sigma$ fils (le nombre total des caractères de l'alphabet). 
Le but est de faire une intersection entre deux ensembles des caractères ($S_1$ et $S_2$) des transitions sortantes de deux n\oe{}uds.

Chaque ensemble de l'alphabet est ordonné et représenté par un tableau de bits (si le caractère existe donc 1, et 0 sinon).

Si $|S_1| < |S_2|$, alors on vérifie l'existence des éléments de $S_1$ dans $S_2$. Comme les caractères sont ordonnés, alors on a un accès direct à la bonne case (le bon bit) dans le tableau de bits. Cela donne comme complexité $O(|S_1|)$  avec  $|S_1| < |S_2|$.

\subsection{Évaluation de la complexité}

\begin{theorem}
Étant donné un dictionnaire $D$ équiprobable de $d$ mots, et un mot requête $q$ d'une longueur $m$.
La complexité en pire cas est en $O(\sigma m^2)$. La complexité temporelle moyenne est $O(m^2)$ si $m \geq 2 \frac{\log d}{\log \sigma}$  (elle est en $O(m)$ pour chaque position).  La complexité est en $O(m + m.occrs)$ si $m \geq 2 ( \frac{log d}{log \sigma} + \frac{log m}{log \sigma} )$, $occrs$ représente le nombre des occurrences de la solution.
\end{theorem}

\begin{proof}
Étant donné un dictionnaire équiprobable $D$ de $d$ mots, et un mot requête $q=P_1 c P_2$ d'une longueur $m$.
La complexité temporelle en pire cas est de $O(\sigma m^2)$ :\\
Dans chaque position $i \in [1..m]$ on teste tous les $\sigma$ caractères possibles (pour insertion ou substitution): 
nous avons $m$ positions, ce qui donne $\sigma m$ branches 
du Trie à explorer. Pour chaque branche nous avons un mot candidat que l'on doit vérifier dans le Trie en un temps $O(m)$ en descendant 
depuis cette branche en faisant une recherche exacte pour tous les caractères restant du mot. Donc, le temps total est de $O(\sigma m^2)$
dans le pire cas. 

\begin{enumerate}
\item Le cas en $O(m^2)$. 
L'algorithme cherche $P_1$ dans le Trie, et $P_2$ dans le Trie inversé, ensuite après l'étape de l'intersection, l'algorithme doit vérifier la partie la plus courte entre $P_1$ et $P_2$. Si on suppose que le résultat de l'intersection est non nul. Pour $m \geq 2 \frac{\log d}{\log \sigma}$, la taille de l'une des deux parties $P_1$ ou $P_2$ est au moins égale à $m' \geq \frac{1}{2} m$. Soit $P'=P_1$ si $P_1>P_2$ et $P'=P_2$ sinon. Nous avons donc $|P'|\geq m'$. En supposant que $m' \geq \frac{\log d}{\log \sigma}$ et donc $|P'|\geq \frac{\log d}{\log \sigma}$. Dans ce qui suit on suppose que $P'=P_1$, le cas $P'=P_2$
étant symétrique. La probabilité pour que les $m'$ premiers caractères d'un mot du dictionnaire soit égaux 
aux $m'$ premiers caractères de $P_1$ est donc inférieure à $1/\sigma^\frac{\log d}{\log \sigma}=1/d$, et donc 
le nombre de mots candidats en moyenne est de $O(d\cdot (1/d))=O(1)$. Comme la recherche exacte sur les caractères restant prend un temps $O(m)$, 
le temps moyen est de $O(m)$. Donc, pour les $m$ positions la complexité devient $O(m^2)$.

\item Le cas en $O(m + m \times occrs)$. On suppose maintenant que $m \geq 2 (\frac{\log d}{\log \sigma}+\frac{\log m}{\log \sigma})$. 
On raisonne de la même façon que précédemment avec pour seul changement la valeur de $m'$ qui est maintenant égale à $m/2=\frac{\log d}{\log \sigma}+\frac{\log m}{\log \sigma}$. La probabilité pour qu'un mot du dictionnaire soit candidat devient $1/\sigma^{m'}=1/(m \times d)$ et donc le nombre de mots candidats pour chaque position devient en moyenne égal à $d/(m \times d)=1/m$. Donc le temps moyen pour vérifier chaque candidat devient $O(m \times (1/m))=O(1)$, qui donne un temps de $O(m)$ pour l'ensemble des positions.  On ajoute un temps de $m \times occrs$ pour les autres occurrences de la solution si elles existent (car chaque solution est en $O(m)$).

La complexité moyenne finale est donc : $O(m + m \times occrs)$.

\end{enumerate}

\end{proof}

\bigskip
Nous nous référons à notre méthode par \textbf{TRT\_CI}, qui est un acronyme pour la phrase en Anglais : (Trie and reverse Trie Characters Intersection), qui ce traduit par Trie et Trie inversé Caractères Intersection.

\section{Une amélioration de la méthode de Amir et al.}
\label{sec:the_second_method}

Dans la méthode précédente \emph{TRT\_CI} nous avons expliqué comment localiser la position de l'erreur, et ensuite, comment réduire le nombre de transitions afin de tester uniquement celles qui mènent à des solutions possibles. 

Soit $P= P_1 \phi P_2$, d'une façon très simple, l'erreur divise le mot en deux parties $P_1$ et $P_2$. Le but est de s'assurer que ces deux parties appartiennent au même mot dans le dictionnaire.

L'idée simple suivante peut être utilisée: supposant que l'on puisse déterminer le numéro du mot auquel les parties $P_1$ et $P_2$ appartiennent,  cela veut dire que le motif qui a la forme $P= P_1 \phi P_2$ représente une solution approchée avec $k=1$ erreur si $P_1$ et $P_2$ appartiennent au même mot.

Pour faire cela, supposons que dans le Trie, au moment de sa construction, on stocke les numéros des mots dans les feuilles. Ensuite, on utilise le Trie pour chercher $P_1$, et on récupère tous les numéros qui sont stockés dans les fils du n\oe{}ud qui est à la fin de $P_1$. On fait le même traitement pour $P_2$ dans le Trie inversé. On obtient deux ensembles qui contiennent les numéros des mots. On fait une intersection pour vérifier si $P_1$ et $P_2$ appartiennent au même mot, et obtenir la solution.

Cette idée été utilisée par Amir et al. \cite{Amir2000}. Dans leur travail, ils font un pré-traitement afin de remédier au problème de l'intersection. Ils obtiennent un résultat qui un temps de pré-traitement en $O(n\,log^{2}n)$, et un temps de recherche en $O(m\,log^{3}d\,log\,log\,d + occr)$ pour le dictionnaire, et $O(m\,log^{2}n\,log\,log\,n + occr)$ pour l'indexation de texte, où $n$ est la taille du texte et $d$ est la taille du dictionnaire, et $occr$ est le nombre des occurrences de la solution.

Amir et al. dans \cite{Amir2000} font toutes les combinaisons de la position de l'erreur dans le mot requête. Chaque fois, ils modifient la position de l'erreur et essaye de trouver des n\oe{}uds dans le Trie et dans le Trie inversé, pour faire une intersection entre les numéros des mots stockés dans les feuilles des deux structures de données. Partant de cette observation, nous notons les points suivants :

\begin{enumerate}

\item  Comme nous avons expliqué dans la sous-section \ref{Possible_error_positions}, nous avons juste certaines positions possibles à considérer comme position de l'erreur. Donc si à la position $i$ il y a une erreur, nous allons perdre du temps pour vérifier les autres positions dans l'intervalle $[i+1..m]$ ($m$ la longueur de mot requête), parce que nous aurons plus d'une seule erreur.

\item  Si l'erreur est à la fin de mot, dans ce cas, l'intersection se fait avec tous les mots du dictionnaire, et cela implique un temps de traitement élevé.

\item Si nous utilisons un Trie compact, nous ne pouvons pas trouver un n\oe{}ud à la position $i$, si cette position est au milieu d'une arête qui a de nombreux caractères.

\end{enumerate}

Pour améliorer la méthode de Amir et al \cite{Amir2000}, nous utilisons les mêmes heuristiques utilisées dans la première méthode \emph{TRT\_CI} pour trouver les positions possibles de l'erreur  et pour éviter de faire une intersection avec le n\oe{}ud racine comme nous l'avons expliqué dans la sous-section \ref{Special_cases}.

\medskip

Dans ce que suit nous présentons la structure de données et notre algorithme de recherche pour améliorer la méthode de Aimr et al.

\subsection{La construction de la structure de données}
\label{sub:Methode_2_construction_structure_donnees}

Notre algorithme de recherche est basé sur la même structure de données que celle de Amir et al. \cite{Amir2000}. La construction de la structure de données est comme suit :

\noindent
\rule{8cm}{0.1pt}\\
\textbf {Algorithme de construction de la structure de données TRT\_WNI}\\
\textbf{Entrée :} fichier contenant les mots du dictionnaire.\\
\textbf{Sortie :} structure de données TRT\_WNI.\\

\begin{enumerate}

\item Construire le Trie du dictionnaire et stocker les numéros des mots dans les feuilles.

\item Construire le Trie inversé et stocker les numéros des mots dans les feuilles.

\item Dans chaque structure de données, relier toutes les feuilles de la gauche vers la droite, pour obtenir une liste des numéros de mots.

\item Stocker dans chaque n\oe{}ud interne le pointeur de la feuille de son fils le plus à gauche et son fils le plus à droite, donc le début et la fin de la sous-liste qui contienne tous les fils sortant de ce n\oe{}ud.
\end{enumerate}
\rule{8cm}{0.1pt}\\

% % après , ajouter ou expliquer les operation de distance d'edition, bon dans cette méthode, les 3 operation son inclue dans le meme algorithme. cette algo permet de les faire toutes les 3 en meme temps. donc je doit juste donnée des exmple et expliquer cela :).

\subsection{L'algorithme de recherche}

Les étapes de l'algorithme ressemble à celles de notre méthode \emph{TRT\_CI} expliqué dans la sous-section \ref{algo:TRT_CI}. La différence avec \emph{TRT\_CI} est dans l'étape de la vérification de l'appartenance de $P_1$ et $P_2$ au même mot.

\noindent
\rule{8cm}{0.1pt}\\
\textbf {Algorithme de la recherche approchée TRT\_WNI}\\
\textbf{Entrée :} un mot requête $q$ de longueur $m$, la structure de données TRT\_WNI.\\
\textbf{Sortie :} liste des mots solutions approchées.\\

\begin{enumerate}[a.]

\item Faire une recherche exacte dans le Trie, jusqu'à ce qu'on arrive à la position de l'erreur, on trouve le n\oe{}ud $v_{err}$ qui représente cette erreur. Soit $P_1$ le préfixe avant la position de l'erreur.

\item \label{item:step_2} Faire une recherche pour le suffixe $P_2$ dans le Trie inversé, jusqu'à ce que la longueur de ce suffixe soit consommée, ($|P_2|=m-|P_1|-1$). Soit $w$ le n\oe{}ud dans lequel la recherche s'arrête dans le Trie inversé.

\item \label{item:step_3} Si on ne peut pas trouver $P_2$,  l'algorithme s'arrête (il n'y a pas de solution).

\item \label{item:step_4} Sinon, on a deux n\oe{}uds $v_{err}$ et $w$, on fait une intersection entre les deux ensembles des numéros des mots stockés dans les feuilles sortantes de ces deux n\oe{}uds.

\item Pour chaque $v_i \in V$ où $V=\{v_1, v_2, ..., v_k \}$ sont les n\oe{}uds sur le chemin qui mènent de la racine à $v_{err}$, répétez les étapes : \ref{item:step_2} \ref{item:step_3} et \ref{item:step_4}

\end{enumerate}
\rule{8cm}{0.1pt}\\

\subsection{Un vecteur de bits simple pour faire l'intersection}

Dans notre travail on utilise une méthode simple pour faire l'intersection sans pré-traitement. Nous utilisons un vecteur de bits de longueur $d$, initialisé à $0$, où $d$ est le nombre total des mots du dictionnaire.

On prend le premier ensemble des numéros des mots du Trie et on marque leurs positions dans le vecteur de bits par $1$. Ensuite, on prend le deuxième ensemble des numéros des mots du Trie inversé, et pour chaque numéro de mot, on vérifie si dans sa position dans le vecteur de bits  est marqué $1$. Si oui, donc cette position est le numéro de mot qui représente une solution.

Dans le pire des cas, cette opération peut se faire dans une complexité temporelle de $O(n)$, mais si le dictionnaire est trié, nous pouvons la faire dans $O(n/w)$ où $w$ est un mot mémoire dans le modèle RAM et $w = log (n)$ avec $n$ le nombre total de caractères, car il suffit juste de faire l'opération (AND), où chaque mot mémoire prend un temps de $O(1)$.

\bigskip
Nous nous référons à notre deuxième méthode par \emph{TRT\_WNI} qui est un acronyme pour la phrase en Anglais : (Trie and Reverse Trie Words Number Intersection) qui signifie Trie et Trie inversé, l'intersection des numéros des mots.

% cette methode est applicable derecetement sur le compact Trie. 

\section{Une méthode hybride: une combinaison des deux méthodes précédentes}
\label{sec:the_third_method}

Afin d'améliorer la seconde méthode \emph{TRT\_WNI} (voir la section \ref{sec:the_second_method}), nous combinons les deux techniques, le filtrage des transitions sortantes de la méthode \emph{TRT\_CI}, et l'intersection des numéros de mots de la méthode \emph{TRT\_WNI}.

Nous utilisons la première approche pour trouver les n\oe{}uds de l'erreur, et pour choisir seulement les n\oe{}uds qui sont à la fin des premières transitions sortantes depuis le n\oe{}ud erreur ($v_{err}/w$), ces transitions peuvent probablement conduire à des solutions.
Nous appliquons la deuxième approche pour effectuer une intersection entre les numéros des mots stockés dans les feuilles sortantes des n\oe{}uds choisis dans la première étape.\\

La structure de données est exactement la même, celle  définie dans la méthode précédente, voir la sous-section \ref{sub:Methode_2_construction_structure_donnees}.
Les étapes de l'algorithme de recherche de cette troisième méthode sont comme suit :

\noindent
\rule{8cm}{0.1pt}\\
\textbf {Algorithme de la recherche approchée TRT\_CWNI}\\
\textbf{Entrée :} un mot requête $q$ de longueur $m$, la structure de données TRT\_WNI.\\
\textbf{Sortie :} liste des mots solutions approchées.\\

\begin{enumerate}
\item Obtenir les deux n\oe{}uds de l'erreur $v_{err}$ et $w$ :

	\begin{enumerate}
	\item Faire une recherche exacte dans le Trie, jusqu'à ce qu'on arrive à la position de l'erreur, on trouve le n\oe{}ud $v_{err}$ qui représente cette erreur. Soit $P_1$ le préfixe avant la position de l'erreur.
	
	\item \label{item:M3_1} Faire une recherche pour le suffixe $P_2$ dans le Trie inversé, jusqu'à ce que la longueur de ce suffixe est consommée, ($|P_2|=m-|P_1|-1$). Soit $w$ le n\oe{}ud dans lequel la recherche s'arrête dans le Trie inversé.
	
	\item \label{item:M3_2} Si on ne peut pas trouver $P_2$, on arrête l'algorithme de recherche (il n'y a pas de solution).
	\end{enumerate}

\item \label{item:M3_3} À partir des deux n\oe{}uds ($v_{err}$ et $w$), obtenir tous les n\oe{}uds obtenus par les transitions sortantes qui peuvent conduire à des solutions.
Soit $v_{err}^{'}$ l'ensemble résultants du Trie, $v_{err}^{'}=\{v_{err1},v_{err2},...,v_{errL1}\}$ , chaque $v_{err\,i}$  conduit à une solution possible. Et de même pour le Trie inversé, on a $w^{'}=\{w_{1},w_{2},...,w_{L2}\}$ , donc on a deux ensembles $v'_{err}$  et  $w'$.

\item \label{item:M3_4} Faire une intersection entre les numéros de mots stockés dans les feuilles obtenus de chaque n\oe{}ud dans les deux ensembles. $v'_{err}$  et  $w'$.

\item Pour chaque $v_i \in V$ où $V=\{v_1, v_2, ..., v_k \}$ sont les n\oe{}uds sur le chemin qui mènent de la racine à $v_{err}$, répéter les étapes : \ref{item:M3_1}, \ref{item:M3_2}, \ref{item:M3_3}, et \ref{item:M3_4}.
\end{enumerate}
\rule{8cm}{0.1pt}\\

\bigskip
La troisième méthode est appelée \emph{TRT\_CWNI} pour la phrase en Anglais : (Trie and reverse Trie, Characters and Words Number Intersection) qui signifie Trie et Trie inversé, l'intersection des caractères et l'intersection des numéros des mots.

\section{Les expérimentations et les analyses}
\label{sec:testes}

Les expérimentations ont été faites avec les deux ensembles de données suivants : une version extraite du dictionnaire des titres de wikipédia (Wikititle) qui a environ 1,2 millions de mots, et le dictionnaire Anglais qui contient environ deux cent treize mille mots.

L'implémentation est faite dans le langage C, en utilisant le compilateur GNU C (GCC), la version 4.7.1. Les tests ont été effectués sur Windows 8.1 Pro 64 bits, avec un processeur Intel Core i7-2670QM 2,20 GHz et 6 Gigabyte de RAM. Nous avons utilisé un seul coeur.\\

\noindent
Le code source de la méthode \emph{TRT\_CI} est disponible sur :\\
\url{https://github.com/chegrane/TrieRTrie/tree/master/TrieRTrie\_char\_inter}.

\noindent
Les méthodes de Amir et al., \emph{TRT\_WNI}, et \emph{TRT\_CWNI} sont disponibles sur :\\ \url{https://github.com/chegrane/TrieRTrie}.

\noindent
Les deux ensembles de données sont disponibles sur :\\
\url{https://github.com/chegrane/TrieRTrie/tree/master/dataset}.\\

Toutes nos implémentations sont distribuées sous la licence publique générale limitée GNU (GNU LGPL) (Anglais: GNU Lesser General Public License).

Nous avons comparé nos résultats aux travaux de Karch et al \cite{karch2010improved}, et Chegrane et Belazzougui \cite{ibra.Chegrane.simple}, car leurs résultats sont très compétitifs par rapport à d'autres approches qui ont une implémentation pratique.

Nous comparons, aussi, nos résultats avec un travail pratique et récent de Aleksander Cis{\l}ak et Szymon Grabowski~\cite{cislak2015practical}, qui traite la distance de Hamming seulement et qui donne des résultats intéressants.

\medskip
Pour réaliser nos tests, nous avons choisis $100$ mots d'une façon aléatoire de notre fichier du dictionnaire. Nous avons appliqué une opération d'édition choisie aléatoirement. À la fin, nous avons effectué une recherche approchée sur ces mots requête. Nous avons répété la procédure $10$ fois, et nous avons calculé la moyenne des résultats. Les résultats sont résumés dans le tableau \ref{tab:results}.

\begin{table}[h]
\centering
\begin{tabular}{|l|c|c|}
\hline 
Méthode & Anglais & Wikititle \tabularnewline
\hline 
\hline 
Karch et al & 8 & 34\tabularnewline
\hline 
Chegrane et Belazzougui & 4,55  & 8,28\tabularnewline
\hline 
TRT\_CI & 0,77 & 0,78\tabularnewline
\hline 
\end{tabular}

\caption{Comparaison entre des méthodes existantes et notre méthode \emph{TRT\_CI} sur les ensembles de données WikiTitle, et Anglais. (Le temps d'exécution est en  ($\mu s$)).}
\label{tab:results}
\end{table}

Dans le tableau \ref{tab:results}, on voit clairement que la méthode \emph{TRT\_CI} est plus efficace en terme de temps d'exécution comparée aux deux autres méthodes testées dans ce travail.
En fait, notre méthode \emph{TRT\_CI}, est d'environ $6$ fois plus rapide que celle de \cite{ibra.Chegrane.simple}, dans le dictionnaire Anglais, et environ $10$ fois plus rapide dans le dictionnaire WikiTitle. Elle est aussi d'environ $10$ fois plus rapide que celle de \cite{karch2010improved} dans l'ensemble de données Anglais, et environ $43$ fois plus rapide dans le dictionnaire WikiTitle.\\

Les expérimentations sont effectuées sur des différentes tailles du dictionnaire. Nous avons fait varier le nombre de mots dans le dictionnaire, pour observer l'influence de ce paramètre sur le temps d'exécution, voir la figure \ref{fig:englisg_wiki_test}, (Pour Karch et al. juste un seul point est montré, parce que nous n'avons qu'un seul résultat, les résultats de Karch et al. ne permettent pas de faire varier le temps de réponse en fonction du changement de la taille du dictionnaire).

\begin{figure}[h]
\centering
\includegraphics[width=1.0\linewidth]{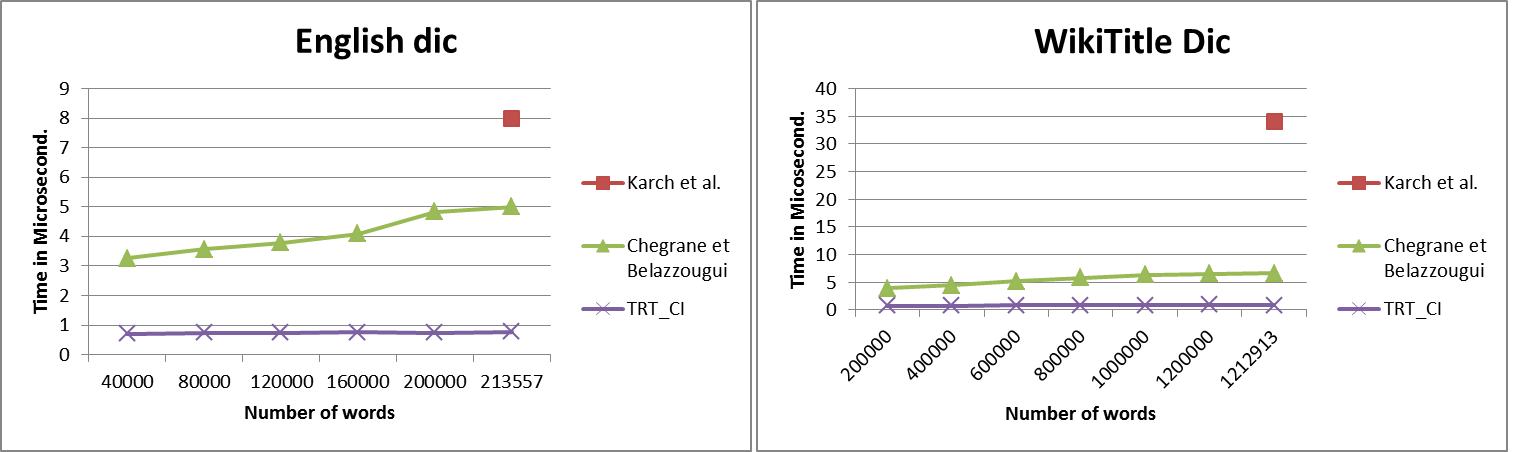}
\caption{Temps d'exécution en fonction du nombre de mots dans les dictionnaires Anglais et WikiTitle.}
\label{fig:englisg_wiki_test}
\end{figure}

Dans la figure \ref{fig:englisg_wiki_test}, on voit que contrairement aux méthodes \cite{karch2010improved} et \cite{ibra.Chegrane.simple}, l'augmentation de la taille du dictionnaire n'affecte pas le temps d'exécution de notre méthode \emph{TRT\_CI}.
Dans la méthode \emph{TRT\_CI}, le temps est presque constant (environ 0,77 $\mu s$), car elle ne dépend que du résultat de l'intersection sur les caractères communs des transitions sortantes des deux n\oe{}uds de l'erreur dans le Trie et dans le Trie inversé, et dans le pire cas, on a $|\Sigma| = \sigma$ chemins à vérifier.

En pratique, l'intersection donne un ensemble très petit, beaucoup moins que le nombre de caractères de l'alphabet $\sigma$. De plus la vérification s'arrête immédiatement si une erreur est rencontrée, car on ne peut pas avoir deux erreurs.
La recherche exacte des deux chaînes $P_1$ et $P_2$ se fait une seule fois dans le Trie et le Trie inversé respectivement. Tous les n\oe{}uds traversés dans la recherche exacte de $P_1$ et $P_2$ sont stockés dans deux vecteurs, pour nous permettre d'effectuer une intersection entre les caractères qui sont sur les transitions sortantes des deux n\oe{}uds (qui sont dans ces vecteurs directement).
Après le filtrage des transitions sortantes, la vérification de chaque chemin s'arrête directement si une erreur est rencontrée, autrement cette voie conduira à une solution.\\

% % je peut ajouter l'explicatio de ** dans la partie explication de notre méthode, je peut ajouter une section (ou sous section) détail d'implimentation.
% ** : (Tous les n\oe{}uds traversés dans la recherche exacte de $P_1$ et $P_2$ sont stockés dans deux vecteurs, pour nous permet d'effectuer une intersection entre les caractères qui sont sur les transition sortantes des deux n\oe{}uds de ces vecteurs directement.)

Aleksander Cis{\l}ak et Szymon Grabowski dans leur travail ~\cite{cislak2015practical}, partitionnent chaque mot de dictionnaire $d_i$ en $k+1$ morceaux $d = \{p_1, p_2, ... , p_{k+1}\}$. Chaque morceau joue le rôle d'une clé pour stocker le mot d'origine dans une table de hachage (avec chaînage externe). Au total, il y a $k+1$ listes de mots à stocker dans la table de hachage pour chaque mot $d_i$. Au lieu de stocker la totalité du mot, ils stockent seulement les parties qui sont différentes du morceau qui joue le rôle d'une clé, et ils stockent comme indice le numéro du morceau qui manque dans le mot stocké.
Les auteurs de~\cite{cislak2015practical}, se sont basés sur le cas de $k=1$, où ils ne stockent que le préfixe et le suffixe du mot $d_i$ dans la table de hachage. Ils ne stockent que le préfixe et le suffixe de mot dans la table de hachage, ils traitent les listes de mots sans leurs préfixes en premier, ensuite, les listes de mots sans leurs suffixes, et ajoutent, dans chaque liste, la position où le dernier morceau commence.
Le travail de~\cite{cislak2015practical} concerne seulement la distance de Hamming.

\medskip
La recherche de mot requête se fait comme suit :
décomposer le mot requête $q$ en $k+1$ morceaux. Pour chaque morceau $p_i$, chercher la liste $l_i$ correspondante depuis la table de hachage (pour $k=1$, on a seulement le préfixe et le suffixe). Pour chaque mot $m_j$ (le mot d'origine moins le morceau $p_i$) stocker dans la liste, si $|m_j|= |q|-|p_i|$, alors vérifier si $dist\_Hamming(m_j,q-p_i) \leq k$, ($q-p_i$ signifie qu'on enlève le morceau $p_i$ du mot requête $q$ ). Si la distance de Hamming est satisfaite, alors les morceaux ($m_j$ et $p_i$) sont combinés dans un seul mot pour le présenter comme une solution approchée.

\medskip
Les résultats expérimentaux de~\cite{cislak2015practical} ont été obtenus sur une machine équipée du processeur Intel i5-3230M cadencé à 2,6 GHz et 8 Go de mémoire DDR3, et le code C++ a été compilé avec la version clang 3.4-1 et exécuter sur Ubuntu 14.04 OS.

Cette configuration est proche de celle utilisée dans ce chapitre. Nous n'avons pas testé leur code sur notre machine par manque de temps (et aussi, nous n'avons pas leur code source). Nous allons utiliser le graphe (voir la figure \ref{fig:ASM_Cislak}) de leurs résultats afin de faire la comparaison.

\begin{figure}[h]
\centering
\includegraphics[width=0.8\linewidth]{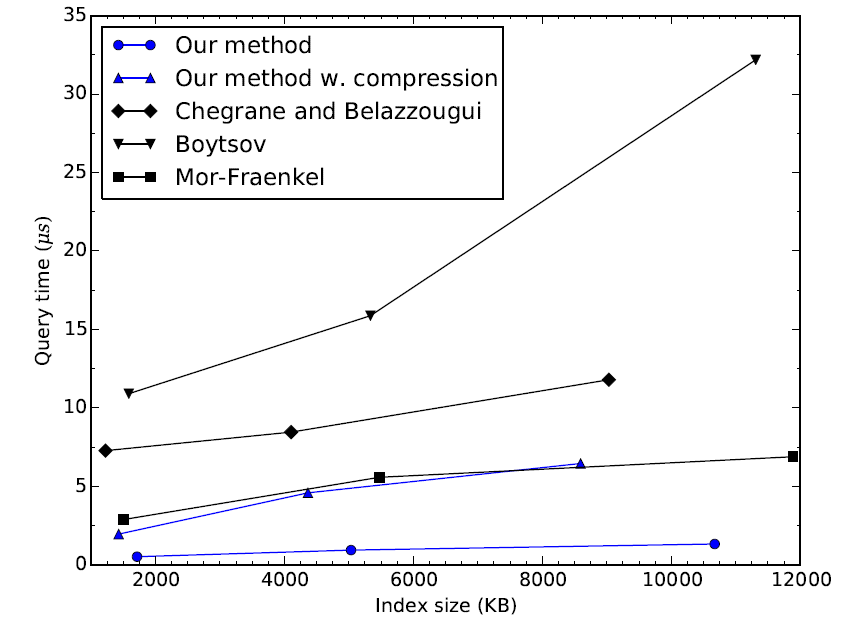}
\caption{Les résultats expérimentaux de~\cite{cislak2015practical} sur le dictionnaire anglais pour k=1 avec la distance de Hamming.}
\label{fig:ASM_Cislak}
\end{figure}

On peut remarquer depuis la figure \ref{fig:ASM_Cislak} (les résultats de \cite{cislak2015practical}) les points suivants :

\begin{enumerate}

\item Le graphe montre que les méthode de~\cite{cislak2015practical} donnent de meilleurs résultats comparés à ceux de notre méthode de recherche approchée avec hachage (chapitre \ref{chap:recherche_approchee_avec_hachage} précédent) et ceux de Boytsov. Mais cette comparaison n'est pas totalement correcte, car la méthode de~\cite{cislak2015practical} ne traite qu'un seul type d'erreur (la substitution). Le temps de calcul augmente lorsqu'on traite les trois types d'erreurs (intuitivement, le temps de traitement pour trois erreurs est de 3 fois plus).
Donc, par conséquence, on ne peut pas dire si les résultats de ~\cite{cislak2015practical} sont meilleurs ou pas.

\medskip
\item Notre méthode de recherche approchée avec hachage (chapitre \ref{chap:recherche_approchee_avec_hachage} précédent) donne de meilleurs résultats par rapport à ceux de Boytsov. Les deux méthodes traitent les 3 types d'erreurs de distance d'édition, donc ils sont directement comparables.

\medskip
\item Le temps de calcul de la méthode sans compression de~\cite{cislak2015practical}, qui donne le meilleur résultat dans le graphe est proche de $1 \mu s$. Le résultat de notre méthode \emph{TRT\_CI} est d'environ 0,77 $\mu s$ donc proche de 1 $\mu s$ aussi. Sauf que la méthode de~\cite{cislak2015practical} traite une seule erreur (la substitution) seulement, par contre la nôtre traite les trois types d'erreurs en même temps. Cela implique que notre méthode donne de meilleurs résultats par rapport à ceux de~\cite{cislak2015practical} (on peut dire qu'elle est 3 fois plus rapide).

\medskip
\item Le temps de calcul de la méthode sans compression de~\cite{cislak2015practical} (qui donne le meilleur résultat) augmente avec le changement de la taille du dictionnaire anglais (la taille du dictionnaire anglais est petite et ne permet pas de donner un bon aperçu), d'après le graphe, on peut déduire que si la taille du dictionnaire augmente, le temps de calcul augmente lui aussi. Quant à notre méthode \emph{TRT\_CI}, le temps de calcul est invariant avec l'augmentation de la taille du dictionnaire. Cela signifie que notre méthode est meilleure même avec ce rapport de différence dans le nombre d'erreurs traitées.

\end{enumerate}

\subsection{L'efficacité de la méthode \emph{TRT\_CI}, par rapport à la recherche exacte}
La recherche exacte est l'opération la plus basique qui peut être effectuée sur un dictionnaire de mots. Lorsque la recherche est effectuée dans un Trie, la complexité est proportionnelle à la longueur du mot requête. Pour cette raison, on peut la considérer comme une opération de référence qui donne le temps d'exécution le plus optimal par rapport à toutes les autres opérations (la recherche approchée avec $k>0$).

Puisque le temps d'exécution de notre méthode \emph{TRT\_CI} est indépendant de la taille du dictionnaire $D$, comme dans la recherche exacte en utilisant un Trie, nous avons fait des expérimentations sur les mêmes ensembles de données (les dictionnaires Anglais et WikiTitle), afin de trouver une relation entre la recherche exacte en utilisant un Trie et notre méthode \emph{TRT\_CI}, voir la figure \ref{fig:exact_with_TRT_CI}.

\begin{figure}[h]
\centering
\includegraphics[width=0.9\linewidth]{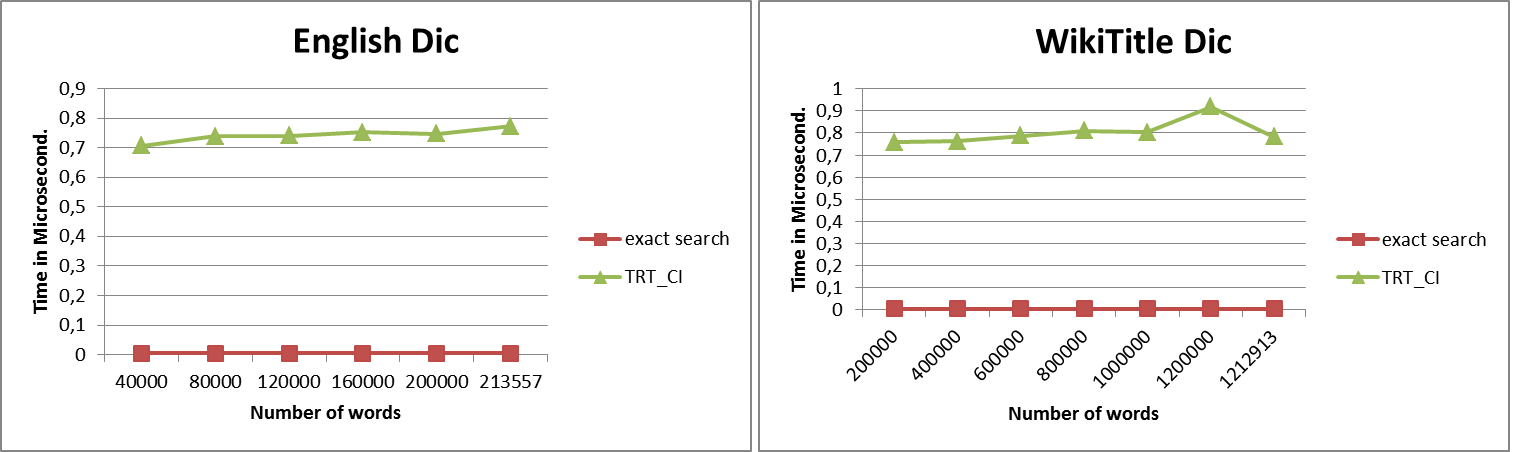}
\caption{La relation entre la méthode \emph{TRT\_CI} et la recherche exacte.}
\label{fig:exact_with_TRT_CI}
\end{figure}

Les résultats expérimentaux de la comparaison entre la méthode \emph{TRT\_CI} et la recherche exacte montrent qu'effectivement, il existe une relation presque constante entre le temps d'exécution de ces deux méthodes, cette relation est décrite comme suit :

Le temps d'exécution de \emph{TRT\_CI} est d'environ $167$ fois le temps d'exécution de la recherche exacte.
Le facteur $167$ s'explique par la différence entre la recherche approchée et la recherche exacte. La recherche exacte (en utilisant un Trie) retourne juste une seule solution, mais la recherche approchée renvoie toutes les solutions valides pour le mot requête, ce qui signifie qu'elle peut impliquer de nombreuses recherches exactes.\\

Enfin, on remarque que les expérimentations sont réalisées sur des dictionnaires de langues (Anglais, WikiTitle), où les ensembles des mots comme :\\
$\{P_1 a_1 P_2, P_1 a_2 P_2,\dots, P_1 a_{|\Sigma|} P_2,\} ,\: (a_i \in \Sigma)$ sont rares. L'étape de vérification est donc réduite à quelques mots seulement.

\bigskip
Dans ce qui suit, nous interprétons seulement les tests et les analyses de nos deux méthodes \emph{TRT\_WNI} et \emph{TRT\_CWNI} afin de souligner l'amélioration pratique apportée à l'algorithme de Amir et al.

Pour notre première méthode \emph{TRT\_CI}, nous allons juste mettre les résultats pour voir la différence avec les autres méthodes. Car cette méthode est beaucoup plus performante que toutes les méthodes testées dans ce chapitre.
Entre les quatre méthodes qui se basent sur le Trie et le Trie inversé, le fonctionnement de notre première méthode \emph{TRT\_CI} est lié à la taille de l'alphabet du dictionnaire qui est très petit $\sigma$. Contrairement aux autres méthodes (TRT\_WNI, TRT\_CWNI, et celle de Amir et al.) elles sont tous liées au nombre total de mots du dictionnaire $n$ qui est un grand nombre et cela rend le temps de l'intersection très grand.

Donc on a deux groupes de méthodes, celles qui se basent juste sur la taille du mot et la taille de l'alphabet, le temps de recherche est proportionnel à $m$ et $\sigma$ qui sont de petits nombres, et l'autre qui est liée à la taille du mot requête et le nombre total des mots du dictionnaire, donc le temps de recherche est proportionnel à $n$.

\subsection{Test des deux méthodes TRT\_WNI et TRT\_CWNI}

Amir et al. ont fait un travail théorique et ils n'ont pas une implémentation de leur méthode, afin de pouvoir comparer nos méthodes avec la leur, nous avons implémenté leur méthode. Pour l'intersection, nous avons utilisé le tableau de bits comme dans nos deux méthodes. Cela rend l'opération de l'intersection la même dans les trois méthodes.

Nous avons fait deux types d'expérimentations, le premier avec les mots exacts où toutes les positions sont considérées comme des positions d'erreur, le deuxième avec des mots qui contiennent exactement une seule erreur.
Les résultats sont résumés dans le tableau \ref{tab:English_Wiki_all_methode_exact_1err}.

% \usepackage{multirow}
\begin{table}[h]
\centering
\begin{tabular}{clll}
\cline{3-4}
                                                             & \multicolumn{1}{l|}{}                 & \multicolumn{1}{l|}{\textbf{Anglais}} & \multicolumn{1}{l|}{\textbf{WikiTitle}} \\ \hline
\multicolumn{1}{|c|}{\multirow{4}{*}{\textbf{Mot 1 erreur}}} & \multicolumn{1}{l|}{Amir et al}       & \multicolumn{1}{l|}{297,86}           & \multicolumn{1}{l|}{2840,50}            \\ \cline{2-4} 
\multicolumn{1}{|c|}{}                                       & \multicolumn{1}{l|}{\textit{TRT\_CI}} & \multicolumn{1}{l|}{\textit{0,77}}    & \multicolumn{1}{l|}{\textit{0,78}}      \\ \cline{2-4} 
\multicolumn{1}{|c|}{}                                       & \multicolumn{1}{l|}{TRT\_WNI}         & \multicolumn{1}{l|}{54,31}            & \multicolumn{1}{l|}{306,67}             \\ \cline{2-4} 
\multicolumn{1}{|c|}{}                                       & \multicolumn{1}{l|}{TRT\_CWNI}        & \multicolumn{1}{l|}{28,94}            & \multicolumn{1}{l|}{231,86}             \\ \hline
\multicolumn{1}{l}{}                                         &                                       &                                       &                                         \\ \hline
\multicolumn{1}{|c|}{\multirow{4}{*}{\textbf{Mot exact}}}    & \multicolumn{1}{l|}{Amir et al}       & \multicolumn{1}{l|}{1516,27}          & \multicolumn{1}{l|}{8406,68}            \\ \cline{2-4} 
\multicolumn{1}{|c|}{}                                       & \multicolumn{1}{l|}{\textit{TRT\_CI}} & \multicolumn{1}{l|}{\textit{2,52}}    & \multicolumn{1}{l|}{\textit{1,65}}      \\ \cline{2-4} 
\multicolumn{1}{|c|}{}                                       & \multicolumn{1}{l|}{TRT\_WNI}         & \multicolumn{1}{l|}{185,23}           & \multicolumn{1}{l|}{701,80}             \\ \cline{2-4} 
\multicolumn{1}{|c|}{}                                       & \multicolumn{1}{l|}{TRT\_CWNI}        & \multicolumn{1}{l|}{131,51}           & \multicolumn{1}{l|}{548,02}             \\ \hline
\end{tabular}
\caption{Temps d'exécution en $\mu s$ pour l'ensemble de données Anglais et WikiTitle.}
\label{tab:English_Wiki_all_methode_exact_1err}
\end{table}

Dans le tableau \ref{tab:English_Wiki_all_methode_exact_1err}, on peut voir clairement que nos deux approches (TRT\_WNI et TRT\_CWNI) sont meilleures que la méthode de Amir et al. et cela car elles utilisent des heuristiques pour localiser les positions possibles de l'erreur, et  évitent ainsi de tester toutes les positions de façon naïve. De plus, dans nos méthodes on évite de faire l'intersection avec la racine car cela signifie l'intersection avec tout le dictionnaire.
Dans les résultats, on voit aussi que notre méthode $TRT\_CWNI$ améliore la méthode $TRT\_WNI$, car on diminue le nombre de feuilles considérées dans l'intersection, lorsqu'on choisit les n\oe{}uds du premier niveau qui sortent du n\oe{}ud de l'erreur, et ces n\oe{}uds mènent à des solutions possibles.

En utilisant des mots avec juste une seule erreur, notre méthode $TRT\_WNI$ dans le dictionnaire Anglais est 5 fois plus rapide que la méthode de Amir et al. et 9 fois plus rapide dans le dictionnaire WikiTitle.
Avec les mots exacts où toutes les positions sont considérées comme une position d'erreur, $TRT\_WNI$ est 8 fois plus rapide dans le dictionnaire Anglais, et environ 12 fois plus rapide dans le dictionnaire WikiTitle.
Il y a une différence dans le temps d'exécution lorsque nous testons avec les mots exacts, et les mots avec une seule erreur, parce que dans le mot exact nous devons vérifier tous les n\oe{}uds dans le chemin du mot requête de la racine à la feuille, et ceci augmente le temps d'exécution.\\

Nous expérimentons avec des tailles différentes d'un dictionnaire, nous changeons le nombre de mots dans le dictionnaire pour voir l'influence de ce paramètre sur le temps d'exécution. Voir la figure \ref{fig:englisg_wiki_test_Amir_vs_ibra}.

\begin{figure}[h]
\centering
\includegraphics[width=0.49\linewidth]{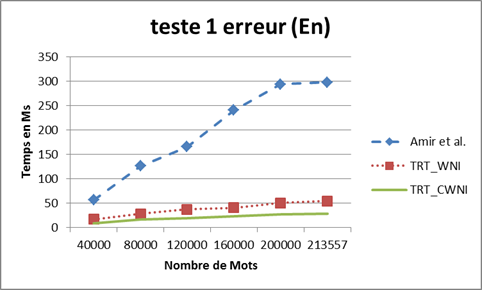}
\includegraphics[width=0.49\linewidth]{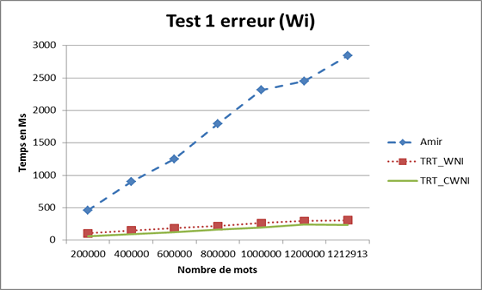}
\caption{Temps d'exécution en fonction du nombre de mots dans les dictionnaires Anglais et WikiTitle.}
\label{fig:englisg_wiki_test_Amir_vs_ibra}
\end{figure}

Dans la figure \ref{fig:englisg_wiki_test_Amir_vs_ibra}, nous présentons l'évaluation des trois méthodes en utilisant des mots avec juste une seule erreur.
Il est clair que le temps d'exécution de la méthode de Amir et al. est très élevé comparé à nos méthodes. En plus, le temps d'exécution de leur méthode augmente rapidement lorsque la taille de dictionnaire s'agrandit.
Cela pourrait s'expliquer par le fait que la méthode de Amir et al. vérifie toutes les positions dans le mot requête et donc cela revient à vérifier tous les n\oe{}uds dans le Trie depuis la racine jusqu'à feuille.
De plus, l'intersection utilisée dépend du nombre de mots dans le dictionnaire et cela signifie qu'à chaque fois que l'ensemble de mots devient grand, le temps augmente lui aussi.
Si nous avons beaucoup de mots qui ont le même préfixe, et beaucoup de mots qui ont le même suffixe, donc lorsque on vérifie le préfixe dans le Trie on aura un grand ensemble de mots du dictionnaire, et ce sera le même cas pour le suffixe, et donc nous aurons deux grands ensembles à vérifier. Le cardinal de chaque ensemble peut être proportionnel à la taille du dictionnaire.

Le temps d'exécution de notre méthode $TRT\_WNI$ augmente lentement. C'est parce qu'elle ne vérifie que les n\oe{}uds depuis la racine jusqu'à la position de l'erreur trouvée par la recherche exacte, comme nous l'avons expliqué dans la sous-section \ref{Possible_error_positions}. En plus de cela, nous traitons les caractères qui sont à l'extrémité du mot requête directement avec la distance d'édition, et cela afin d'éviter l'intersection avec le n\oe{}ud racine (voir la sous-section \ref{Special_cases}).

Dans ce qui suit nous donnons un tableau récapitulatif de toutes les méthodes testées dans ce travail (voir le tableau \ref{tab:results_all_method}).

\begin{table}[h]
\centering
\begin{tabular}{|l|c|c|}
\hline 
Méthode & Anglais & Wikititle \tabularnewline
\hline 
\hline 
Amir et al & 297,86  & 2840,50 \tabularnewline
\hline 
Karch et al & 8 & 34\tabularnewline
\hline 
Chegrane et Belazzougui & 4,55  & 8,28\tabularnewline
\hline 
Méthode 1 : TRT\_CI & 0,77 & 0,78\tabularnewline
\hline 
Méthode 2 : TRT\_WNI & 54,31  & 306,67 \tabularnewline
\hline 
Méthode 3 : TRT\_CWNI & 28,94  & 231,86 \tabularnewline
\hline 
\end{tabular}

\caption{Comparaison des méthodes existantes avec nos trois méthodes sur les ensemble de données Anglais et WikiTitle (temps en $\mu s$).}
\label{tab:results_all_method}
\end{table}

\section{L'application de notre méthode dans l'indexation de texte}
\label{sec:TRT:application_notre_methode_dans_texte}

On peut facilement adapter notre solution à la recherche des occurrences des mots ou la recherche des facteurs dans un texte.

\noindent
\paragraph{La recherche des occurrences des mots :}~\\
Dans le chapitre précédent on a dit qu'un texte T peut se formuler avec $T'={l_1,l_2,...,l_d}$ où chaque mot peut avoir plusieurs occurrences $l_i={pos_1,pos_2,...,pos_{nb}}$.

Dans la structure de données Trie, dans chaque feuille, on ajoute un champ qui pointe vers une liste qui contient toutes les positions des mots dans le texte.

\noindent
\paragraph{La recherche des facteurs dans un texte :}~\\
Il suffit de remplacer la structure de données Trie par un arbre des suffixes.

Il est clair qu'avec l'arbre des suffixes, on peut faire aussi la recherche de toutes les occurrences des mots, car un mot est aussi un facteur, mais on utilise un Trie avec les positions des mots afin d'optimiser l'espace mémoire, car l'arbre des suffixes stocke tous les suffixes de texte, alors que le Trie stocke les mots une seule fois.

\section{Conclusion}
\label{sec:Conclusion}

Dans ce chapitre, nous avons proposé trois méthodes pour résoudre le problème de la recherche approchée dans un dictionnaire avec $k=1$.
Toutes les trois méthodes opèrent sur une structure de données bidirectionnelle (le Trie et le Trie inversé), pour exécuter une recherche exacte sur deux morceaux du motif. 

\medskip
Notre première méthode \emph{TRT\_CI} effectue une intersection  des caractères communs de transitions sortantes entre les deux n\oe{}uds où l'erreur est rencontrée, dans le Trie (préfixe) et dans le Trie inversé (suffixe) pour déterminer les chemins qui peuvent conduire à des solutions.

\medskip

Dans la deuxième méthode \emph{TRT\_WNI}, nous effectuons une intersection entre deux ensembles de numéros de mots. Cette méthode améliore la méthode de Amir et al \cite{Amir2000}, grâce à l'utilisation d'heuristiques pour trouver les positions possibles de l'erreur, et pour éviter de faire une intersection avec la racine.

\medskip

Dans la troisième méthode \emph{TRT\_CWNI}, nous combinons les deux méthodes \emph{TRT\_CI} et \emph{TRT\_WNI}, pour réduire le nombre de feuilles et réduire les ensembles sur lesquels l'intersection de numéros de mots sera effectuée.

\medskip

Les résultats numériques montrent que notre première méthode \emph{TRT\_CI} surpasse toutes les autres implémentations testées jusqu'à ce jour en termes de temps d'exécution. De plus, cette performance est en proportion constante par rapport à la recherche exacte, indépendamment de la taille du dictionnaire.

\medskip
Nos deux autres méthodes TRT\_WNI et TRT\_CWNI donnent des résultats meilleurs que la méthode de Amir et al. car elles utilisent des heuristiques pour localiser les positions possibles de l'erreur, et évitent de tester toutes les positions.

\medskip
Notre méthode $TRT\_CWNI$ améliore la méthode $TRT\_WNI$ en diminuant le nombre de feuilles qui entrent  dans le calcul de l'intersection lorsqu'on choisit les n\oe{}uds du premier niveau qui sortent depuis le n\oe{}ud de l'erreur.

\medskip

La méthode $TRT\_CWNI$ améliore la seconde méthode $TRT\_WNI$, mais la première approche $TRT\_CI$ reste meilleure, parce qu'elle ne dépend que de la longueur du mot requête et de la taille de l'alphabet, mais la seconde dépend de tous les mots du dictionnaire.

\medskip
Nos trois méthodes et la méthode de Amir et al, se basent sur le Trie et le Trie inversé. 
Notre méthode TRT\_CI donne les meilleurs résultats et elle surpasse toutes les autres méthodes pratiques testées. Son fonctionnement dépend de la taille du mot requête et l'alphabet du dictionnaire, alors que les trois autres méthodes dépendent toutes du nombre total $n$ qui est la taille de tout le dictionnaire.

\bigskip
Dans ce chapitre et le chapitre précédent, nous avons décrit deux solutions du problème général de la recherche approchée dans un dictionnaire/texte.

Les performances théoriques de la 1\iere{} solution (la recherche approchée avec hachage présentée dans le chapitre précédent) et son application pour un nombre d'erreurs $k \geq 2$ semble suggérer qu'elle est applicable avec d'aussi bonnes performances dans d'autres contextes. Nous avons vu qu'en pratique et pour $k=1$ erreur la seconde solution (qui utilise le Trie et le Trie inversé proposée dans ce chapitre) est meilleure.

Par ailleurs, dans un problème de recherche approchée particulier, \emph{l'auto-complétion}, où il s'agit de trouver tous les suffixes d'un préfixe contenant des erreurs, nous avons apporté une nouvelle solution meilleure, en pratique, que la solution 1. Cette solution fait l'objet du chapitre qui suit.

\chapter{L'auto-complétion approchée dans une architecture Client-Serveur}

\ifpdf
    \graphicspath{{Autocompletion/}{Autocompletion/PDF}}
\else
    \graphicspath{{Autocompletion/}}
\fi

% %\section*{Les mots clés}
% % Auto-complétion, auto-suggestion, auto-complétion approchée, top-K, la recherche approchée, auto complétion client/serveur.

\section{Introduction}

%\subsection{Motivation} le contexte :
Les champs de saisie de texte sont utilisés pour entrer les données dans l'ordinateur, et ils ont été améliorés au cours des années. Aujourd'hui, ils sont équipés avec de nombreuses fonctionnalités afin d'aider l'utilisateur. Parmi l'une des plus intéressantes on trouve l'auto-complétion (on l'appelle aussi l'auto-suggestion).\\

%\subsection{Problem definition}
L'auto-complétion est une technique qui facilite et accélère l'écriture, en proposant une liste de mots ou de phrases (les suggestions) qui complètent les quelques caractères tapés dans le champ de texte, dans un temps très court (généralement, quelques millisecondes). Généralement, la liste de suggestions s'affiche en dessous du champ de saisie.

La manière habituelle (mais pas obligatoire) dont ces mots sont choisis est telle que les quelques caractères tapés sont préfixe des mots~\footnote{On trouve aussi la complétion où les quelques caractères tapés ne sont pas juste préfixe mais un facteur de la phrase, une sous-chaine apparaissant dans n'importe quelle position.}. Ensuite, prenant en compte cette liste, l'utilisateur peut choisir l'un des mots ou continuer de taper plus de caractères pour faire plus de filtrage, ce qui fait apparaître de nouvelles listes de mots.

À chaque fois que l'utilisateur tape un autre caractère, le nombre de résultats diminue pour devenir juste quelques mots dans la liste de suggestions. L'utilisateur trouve ce qu'il recherche, ou il continue à taper toute sa requête jusqu'à la fin.\\

%\subsection{The aim of auto completing, et domaine of use :}
L'auto-complétion est très utile et populaire dans plusieurs domaines et systèmes.

%WEB
Dans le web, cette fonctionnalité est assurée par les navigateurs pour compléter les URLs. Elle est utilisée dans les pages HTML (côté client), avec des champs de saisie de texte spécifiques équipés par un système d'auto-complétion basé sur l'historique de champ récent \cite{WWW2013InputElement,WWW2012InputAutocomplete}. Dans le web côté serveur pour fournir une liste de suggestions et l'envoyer au client comme par exemple les moteurs de recherche.

%Desktop
Sur les ordinateurs de bureau, cette fonctionnalité est intégrée dans de nombreuses applications, en utilisant par exemple la touche \emph{Tab} dans un interpréteur de ligne de commande (Shell bash sous UNIX). Dans les éditeurs des codes sources pour les programmeurs, cette fonctionnalité est appelée \emph{intellisense} \cite{ICC14}.

%Mobile
Dans les appareils mobiles, taper avec précision est une tâche fastidieuse et l'écriture des utilisateurs a tendance à contenir des erreurs typographiques.
L'auto-complétion approchée est une fonctionnalité très importante dans ce type d'environnements (les appareils mobiles simples et les Smartphones et les tablettes). 

La figure \ref{fig:autocompletion_image}, représente les différentes catégories où l'auto-complétion peut être utilisée, dans le web, le PC de bureau, le Smartphone.

\begin{figure}[h]
\centering
\includegraphics[width=0.8\linewidth]{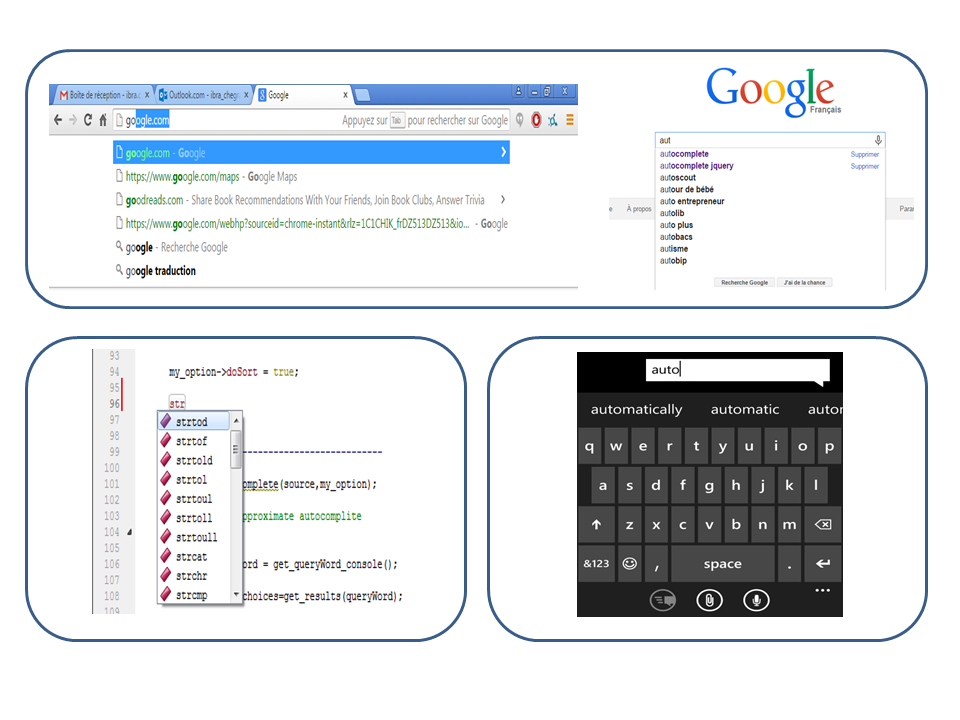}
\caption{L'auto-complétion dans différentes applications et dispositifs (web, application de PC, Smartphone).}
\label{fig:autocompletion_image}
\end{figure}

%\subsection{Our Motivation}

Considérant les bibliothèques de programmation de l'auto-complétion dans le Web, l'une des plus utilisée et qui permet aux programmeurs d'ajouter cette fonctionnalité à leurs champs de saisie est JQuery UI auto-complete (voir \url{http://api.jqueryui.com/autocomplete/}~\footnote{Visité le: 02-07-2016.}). 

Cependant, et cela est la motivation principale de ce travail, la chaîne de caractères d'entrée (les quelques caractères) peut contenir des erreurs. L'erreur peut être soit une erreur de frappe, en particulier lors de la saisie rapide, soit une méconnaissance par l'utilisateur de l'orthographe correcte (nom de personne, nom du produit ... etc.). 
La bibliothèque JQuery UI auto-complete et les autres bibliothèques de l'auto-complétion ne permettent pas de tolérer les erreurs dans la chaine de caractères entrée par l'utilisateur. Dans certains cas ce comportement peut induire l'utilisateur en erreur.

Pour résoudre ce problème, nous devons tolérer un certain nombre d'erreurs dans le préfixe tapé dans les champs de saisie de texte et dans la liste des suggestions afin d'obtenir une liste de complétion exacte et approchée (voire la figure \ref{fig:approximate_autocompletion}).

\begin{figure}[h]
\centering
\includegraphics[width=0.7\linewidth]{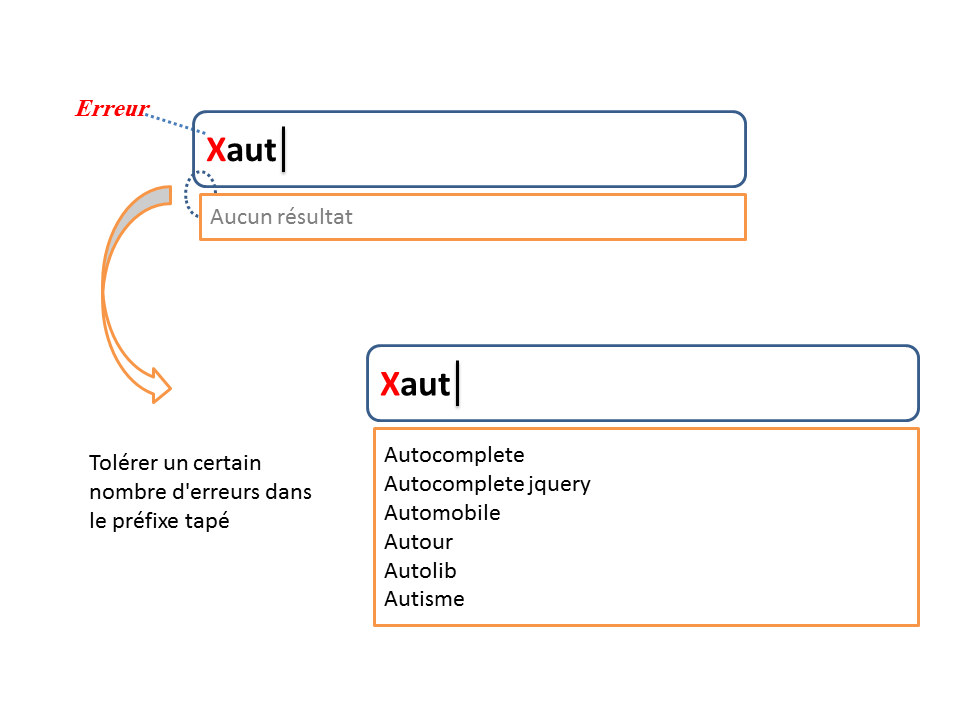}
\caption{L'auto-complétion approchée.}
\label{fig:approximate_autocompletion}
\end{figure}

Lorsque l'utilisateur tape les quelque caractères (préfixe) de sa requête, le système lui affiche une liste de suggestions qui contient tous les mots qui sont lexicalement similaires au préfixe tapé.
Dans cette liste, les mots qui ont exactement le même préfixe sont proposés en premier (l'auto-complète exacte), ensuite, les mots qui diffèrent d'un certain nombre d'erreurs $k$ (l'auto-complète approchée).
La mesure la plus utilisée pour déterminer la différence entre deux chaînes de caractères $x$ et $y$ est la \emph{distance d'édition} \cite{Le66}.\\

%\subsection{Our work}

Une conséquence évidente de la tolérance aux erreurs est que la liste de suggestions pourrait être très longue.
Pour réduire le nombre de résultats, nous proposons une nouvelle bibliothèque d'auto-complétion nommée {\tt appacolib} qui (a) limite  le nombre d'erreurs possibles à au plus à une erreur, et (b) rapporte les $k$ suggestions ayant les scores les plus élevés (dans un ordre décroissant).
On appelle cette méthode top-k complétion (top-k suggestion). Le paramètre $k$ est donné par l'utilisateur.\\

Répondre à des requêtes de préfixe avec erreur sur un dictionnaire est un sujet de recherche étudié depuis un certain temps, mais qui reste néanmoins très actif.
Il existe plusieurs solutions algorithmiques qui résolvent ce problème de façon efficace. Cependant, très peu de ces solutions font face au problème dans une architecture client-serveur et proposent une solution pratique utilisable dans les systèmes 
actuels. 
Ceci est le but de ce travail. Nous avons étudié le problème de l'auto-complétion approchée sous la mesure de la distance d'édition, et en particulier la possibilité d'améliorer la qualité des résultats en utilisant un système de classement selon l'importance (score) des mots dans le dictionnaire. 
Nous avons aussi étudié le problème de la redondance des résultats et nous avons proposé une solution efficace pour les éliminer.\\

En fait, la bibliothèque {\tt appacolib} est un ensemble de bibliothèques pour différents langages pouvant être utilisés soit sur le serveur (C/C++) ou sur le client (JavaScript) ou sur les deux en même temps afin de répondre rapidement à des requêtes avec erreurs et faire l'auto-complétion en se basant sur un dictionnaire $\mathtt{UTF}\mbox{-}8$.

Typiquement, si on a un dictionnaire très volumineux (plus de 3 Mo d'entrées), on peut utiliser la bibliothèque écrite en C/C++ dans le coté serveur avec un module FASTCGI pour fournir une liste de suggestions dans un format JSON, et renvoyer les résultats à travers des appels AJAX depuis notre bibliothèque version client ou n'importe quelle autre interface utilisée dans le coté client. Si le dictionnaire n'est pas très volumineux (par exemple moins de 3 Mo d'entrées), on peut effectuer toutes les opérations au niveau local (sur le navigateur).
La bibliothèque version client est utilisée aussi pour envoyer des requêtes d'auto-complétion à des serveurs qui se basent sur une base de données ou un autre système permettant d'offrir une liste de suggestions au format JSON (par exemple un système écrit en PHP).\\
 
Lorsque l'on traite des données du monde réel, des problèmes pratiques apparaissent. Généralement, ils ne sont pas pris en compte dans le travail théorique. En particulier, la gestion efficace des fichiers $\mathtt{UTF}\mbox{-}8$  en termes de temps et d'espace mémoire qui n'est, techniquement pas triviale.
  
Notre bibliothèque permet le traitement des dictionnaires $\mathtt{UTF}\mbox{-}8$, quel que soit la langue utilisée (Arabe, Chinois, Latin...etc).

\paragraph{Le chapitre est organisé comme suit : }
Dans la section \ref{sec:auto:data_structure}, nous détaillons les différentes structures de données utilisées pour faire une auto-complétion approchée efficace. Dans la section \ref{sec:auto:methode_recherche}, nous expliquons les méthodes de recherche et comment les listes de suggestions sont générées. Dans la section \ref{sec:auto:autocompletion_user_typing} nous proposons  différentes méthodes permettant de s'adapter au comportement de l'utilisateur dans le but de réduire les opérations faites dans la recherche, et donc de gagner en temps de calcul.
Dans la section \ref{sec:auto:reduire_nombre_branche_sortante_tester} nous proposons une méthode permettant de réduire le nombre des branches sortantes a tester dans chaque n\oe{}ud valide, à l'aide d'une table de hachage qui permet de générer des caractères candidats à tester dans chaque branche.
Dans la section \ref{sec:auto:resultat_top-k} nous expliquons comment réduire le nombre de résultats, et mettre en avant les résultats qui satisfont le plus le besoin de l'utilisateur, en utilisant le système dit \emph{Top-k}.

Dans la section \ref{sec:auto:Eliminer_la_redondance}, nous expliquons comment éliminer les résultats en double.
Nous présenterons ensuite la stratégie client/serveur dans la section \ref{sec:auto:strategies_client_serveur}.
Dans la section \ref{sec:auto:quelle_vitesse_devrait_etre_l'auto-completion} nous présentons et nous discutons de la vitesse nécessaire pour que l'auto-complétion soit efficace. La section \ref{sec:auto:test} est dédiée aux tests et aux expérimentations. Enfin la dernière section \ref{sec:auto:conclusion} conclut ce chapitre.

\section{La structure de données}
\label{sec:auto:data_structure}

La plupart des algorithmes efficaces de recherche de préfixe approchée construisent d'abord un index sur le dictionnaire, puis utilisent cet index pour effectuer une recherche efficace pour chaque requête. Dans notre travail, nous construisons un index aussi, semblable à un Trie compact avec un système de classement de tous les mots du dictionnaire.

Dans notre travail, nous utilisons un Trie compact dans lequel nous optimisons l'espace mémoire; chaque n\oe{}ud qui est fils unique de son n\oe{}ud parent est fusionné avec ce dernier. Un détail important de cette représentation est que chaque transition est en fait marqué par un facteur de l'un des mots dans le dictionnaire (et non pas par un caractère unique comme dans le cas du trie).
Ainsi, pour réduire la mémoire nécessaire, en restant toujours efficace, nous codons une transition par son premier caractère, un pointeur dans l'index vers le début d'une occurrence de l'élément correspondant (facteur), et la longueur de la transition qui est donc la longueur du facteur.
Cela implique que, dans cette approche, nous gardons à la fois, le dictionnaire des mots et la structure de données (le Trie compact) dans la mémoire centrale.

\subsection{L'ordre lexicographique de dictionnaire}
\label{sub:Lordre_lexicographique_de_dictionnaire}

Dans notre travail, nous ordonnons le dictionnaire selon l'ordre lexicographique dans le but de pouvoir construire le tableau des longueurs des préfixes communs (LCP) entre toutes paires de mots successifs, permettant ainsi d'accélérer la construction du Trie.

L'opération de tri du dictionnaire prend un temps relativement grand, spécialement si la taille de dictionnaire est grande. Comme en généralement, les dictionnaires se présentent sous forme de fichiers (par exemple le dictionnaire Anglais), il est préférable d'ordonner le fichier avant de l'utiliser dans la construction de système d'auto-complétion, pour éviter de l'ordonner à chaque fois qu'on veut l'utiliser.

\subsection{La construction de Trie à l'aide de tableau LCP}
\label{sub:construction_Trie_avec_LCP}

Pour construire notre Trie compact d'une façon efficace, nous utilisons comme indiqué plus haut une autre structure de données appelée \emph{tableau des plus longs préfixes communs}. Le tableau LCP stocke les longueurs des plus longs préfixes communs entre chaque paire de mots consécutifs du dictionnaire, tel que la première case du tableau contient la longueur du plus long préfixe commun entre le premier et le deuxième mot du dictionnaire.\\

\noindent
Exemple : soit le dictionnaire suivant $D=\{AbCdeAa , AbCdeBbbOo , AbCdeBbbSss , AbEfg\}$. On prend les mots deux par deux, et on calcule leurs plus long préfixes communs. 

\begin{enumerate}

\item $\{AbCdeAa , AbCdeBbbOo\}$, le plus long préfixe commun est $AbCde$ , il est de longueur $5$.
\item $\{AbCdeBbbOo , AbCdeBbbSss\}$, le plus long préfixe commun est $AbCdeBbb$ , il est de longueur $8$.
\item $\{AbCdeBbbbSss , AbEfg\}$, le plus long préfixe commun est $Ab$ , il est de longueur $2$.
\end{enumerate}
À la fin, notre tableau LCP est comme suit : $\mathtt{Tab\_LCP}= [5,8,2]$.

\bigskip
La profondeur d'un n\oe{}ud $nd$ est le nombre de caractères lus sur le chemin depuis la racine jusqu'au ce n\oe{}ud $nd$. La profondeur de chaque n\oe{}ud interne $nd$ est en fait la longueur du plus long préfixe commun entre tous les mots pointé par les feuilles se trouvant sous ce n\oe{}ud. La profondeur d'un mot pour un n\oe{}ud feuille est la longueur du mot pointé par le n\oe{}ud.
Dans la construction du Trie, pour chaque mot, on crée un n\oe{}ud feuille, et pour chaque préfixe commun, on crée un n\oe{}ud interne.

\bigskip
En supposant que l'on a déjà calculé et construit le tableau LCP ($\mathtt{Tab\_LCP}$), les étapes de construction du Trie compact à l'aide de tableau LCP sont comme suit :

% -----------------------------------------------------------------
\noindent
\rule{8cm}{0.1pt}\\
\textbf {Algorithme de construction du Trie compact}\\
\textbf{Entrée :} Un fichier trié contenant les mots du dictionnaire. Le tableau LCP.\\
\textbf{Sortie :} La structure de données Trie compact.\\

\begin{enumerate}

\item Créer le n\oe{}ud racine.
\item Ajouter le premier mot au Trie, donc créer un n\oe{}ud feuille pour le premier mot.

\item Pour tous les mots restants $i\in [2..n]$ :

	\begin{enumerate}

	\item Si le mot $x$ précédent $i-1$ ($x_{i-1}$) est un préfixe propre dans le mot $x$ actuel $i$ ($x_i$). Alors on teste si le plus long préfixe commun entre le mot actuel et le mot précédent ($\mathtt{Tab\_LCP[i-1]}$) est supérieur ou égal à la longueur du mot précédent (la profondeur $\mathtt{Prof}$ du n\oe{}ud feuille $\mathtt{nd_f}$ pointant vers le mot précédent $x_{i-1}$, ($\mathtt{Prof(nd_{f(x_{i-1})}}$)).

		\begin{itemize} 
	
		\item Si les deux mots ont la même longueur, alors on ne fait rien (le dictionnaire ne contient pas deux copies d'un même mot). 
	
		\item Sinon, ajouter un nouveau n\oe{}ud feuille $\mathtt{nd}_{f (x_i)}$ qui sort du n\oe{}ud père de mot précédent $\mathtt{nd}_{p (x_{i-1})}$. Ajouter les caractères restants du mot actuel à la transition entre le n\oe{}ud père et le nouveau n\oe{}ud $\mathtt{Tr}(\mathtt{nd}_{p (x_{i-1})}, \mathtt{nd}_{f (x_i)})$. On ajoute les caractères qui ne sont pas dans le préfixe commun représenté par le n\oe{}ud père, le nombre de caractères à ajouter est égal à la profondeur du n\oe{}ud feuille (qui est la taille de mot) moins la profondeur de n\oe{}ud père ($\mathtt{Nb\_char\_add} = \mathtt{Prof}(\mathtt{nd}_{f}) - \mathtt{Prof}(\mathtt{nd}_{p})$).
	
		\end{itemize}

	\item Sinon, le mot précédent ($x_{i-1}$) n'est pas un préfixe propre du mot actuel (la taille du préfixe commun est différente de la longueur du mot précédent) donc :\\

		\begin{itemize} 
	
		\item Depuis le n\oe{}ud feuille pointant vers le mot précédent $\mathtt{nd}_{f(x_{i-1})}$, remonter dans ses n\oe{}uds ancêtres jusqu'à arriver au n\oe{}ud $\mathtt{nd}_{p(\mathtt{LCP}(x_{i-1},x_i))}$ qui représente le préfixe commun entre le mot précédent et le mot actuel. Pour cela, il suffit juste de comparer la taille de préfixe commun ($\mathtt{Tab\_LCP}[i-1]$) avec la profondeur des n\oe{}uds internes ($\mathtt{Prof}(\mathtt{nd}_{p_i})$), tant que $\mathtt{Tab\_LCP}[i-1] < \mathtt{Prof}(\mathtt{nd}_p)$, alors remonter au n\oe{}ud parent et ainsi de suite, jusqu'à arriver au n\oe{}ud $\mathtt{nd}_{p(\mathtt{LCP}(x_{i-1},x_i))}$. Donc on a $\mathtt{Prof}(\mathtt{nd}_{p_1}) < \mathtt{Prof}(\mathtt{nd}_{p_2}) < ... < \mathtt{Prof}(\mathtt{nd}_{f})$, et $ \mathtt{nd}_{p_1} = \mathtt{nd}_{p(\mathtt{LCP}(x_{i-1},x_i))}$.  Le n\oe{}ud $\mathtt{nd}_{p_i}$ est le père du n\oe{}ud $\mathtt{nd}_{p_{i+1}}$, le n\oe{}ud $nd_f$ est le dernier n\oe{}ud dans le chemin depuis $nd_{p_1}$.\\
	 
		\item Entre le n\oe{}ud $\mathtt{nd}_{p_1}$ et son n\oe{}ud fils $\mathtt{nd}_{p_2}$, créer un n\oe{}ud intermédiaire $\mathtt{nd}_{\mathtt{int}}$. L'ordre des n\oe{}uds devient alors comme suit : {$\mathtt{nd}_{p_1},\, \mathtt{nd}_{\mathtt{int}},\, \mathtt{nd}_{p_2}$}. La chaîne de caractères qui a été sur la transition $\mathtt{Tr}(\mathtt{nd}_{p_1},\mathtt{nd}_{p_2})$ va être découpée entre les deux nouvelles transitions $\mathtt{Tr}(\mathtt{nd}_{p_1},\mathtt{nd}_{\mathtt{int}}),\, \mathtt{Tr}(\mathtt{nd}_{\mathtt{int}},\mathtt{nd}_{p_2})$, le nombre de caractères à mettre dans la transition $\mathtt{Tr}(\mathtt{nd}_{p_1},\mathtt{nd}_{\mathtt{int}})$ est la profondeur de $\mathtt{nd}_{\mathtt{int}}$ moins la profondeur de n\oe{}ud $\mathtt{nd}_{p_1}$, ($\mathtt{Nb\_char\_add} = \mathtt{Prof}(\mathtt{nd}_{\mathtt{int}}) - \mathtt{Prof}(\mathtt{nd}_{p_1})$).\\
	
	 	\item Insérer un nouveau  n\oe{}ud feuille $\mathtt{nd}_{f}$ qui sort de ce n\oe{}ud intermédiaire $\mathtt{nd}_{\mathtt{int}}$ et le faire pointer le mot actuel à insérer. Ensuite, ajouter les caractères qui ne sont pas dans le préfixe commun à la nouvelle transition $\mathtt{Tr}(\mathtt{nd}_{\mathtt{int}},\mathtt{nd}_{f})$. Le nombre de caractères à ajouter est égal à la profondeur du n\oe{}ud feuille (qui est la taille du mot) moins la profondeur du n\oe{}ud père qui est le n\oe{}ud intermédiaire : $\mathtt{Nb\_char\_add} = \mathtt{Prof}(\mathtt{nd}_{f}) - \mathtt{Prof}(\mathtt{nd}_{\mathtt{int}})$.
	
		\end{itemize}
	
	\end{enumerate}

\end{enumerate}
\rule{8cm}{0.1pt}\\

\noindent
Exemple :\\
Soit le dictionnaire suivant $D=\{\mathtt{AbCdeAa} , \mathtt{AbCdeBbbOo} , \mathtt{AbCdeBbbSss} , \mathtt{AbEfg}\}$.
Le tableau LCP est comme suit : $\mathtt{Tab\_LCP} = [5,8,2]$.
On construit notre Trie étape par étape avec les explications :

\begin{enumerate}

\item La 1\iere{} étape est de créer le n\oe{}ud racine $\mathtt{nd}_0$.

\item Le 1\ier{} mot $\mathtt{AbCdeAa}$ : insérer le 1\ier{} n\oe{}ud feuille $\mathtt{nd}_1$ qui est relier avec la transition qui sort de la racine. Et on met le mot $\mathtt{AbCdeAa}$ sur la transition ($\mathtt{Tr}(\mathtt{nd}_0,\mathtt{nd}_1)$).

\item Le 2\ieme{} mot $\mathtt{AbCdeBbbOo}$ : d'après le tableau LCP, il y a un préfixe commun d'une longueur $5$ avec le mot précédent, et le mot précédent n'est pas un préfixe propre de mot actuel. La transition précédente $\mathtt{Tr}(\mathtt{nd}_0,\mathtt{nd}_1)$ va être découpée, pour cela, on crée un n\oe{}ud intermédiaire $\mathtt{nd}_2$ entre $\mathtt{nd}_0$ et $\mathtt{nd}_1$, de telle sorte que la transition $\mathtt{Tr}(\mathtt{nd}_0,\mathtt{nd}_2)$ contienne le préfixe commun $\mathtt{AbCde}$, le reste du mot précédent est dans la nouvelle transition $\mathtt{Tr}(\mathtt{nd}_2,\mathtt{nd}_1)$ donc $\mathtt{Aa}$. Ensuite, depuis $\mathtt{nd}_2$ on ajoute un n\oe{}ud fils $\mathtt{nd}_3$, de telle sorte que la nouvelle transition $\mathtt{Tr}(\mathtt{nd}_2,\mathtt{nd}_3)$ va contenir le reste du mot actuel donc $\mathtt{BbbOo}$.

\item Le 3\ieme{} mot $\mathtt{AbCdeBbbSss}$ : récupérer la longueur du préfixe commun depuis le tableau LCP, donc $8$. Depuis le n\oe{}ud feuille du dernier mot insérer ($\mathtt{nd}_3$), on retrouve le n\oe{}ud ancêtre le plus profond qui à une profondeur inférieure ou égale à la longueur de préfixe commun $8$, on trouve le n\oe{}ud $\mathtt{nd}_2$. On crée alors un n\oe{}ud intermédiaire $\mathtt{nd}_4$ entre la transition $\mathtt{Tr}(\mathtt{nd}_2,\mathtt{nd}_3)$ et on met $\mathtt{Bbb}$ dans la nouvelle Transition $\mathtt{Tr}(\mathtt{nd}_2,\mathtt{nd}_4)$ car le n\oe{}ud $\mathtt{nd}_2$ a une profondeur $5$, et le nouveau n\oe{}ud a une profondeur $8$, donc on ajoute juste $3$ caractères à la nouvelle transition. La transition $\mathtt{Tr}(\mathtt{nd}_4,\mathtt{nd}_3)$ va contenir $\mathtt{Oo}$.
Ensuite, on crée un nouveau n\oe{}ud feuille $\mathtt{nd}_5$ qui sort du $\mathtt{nd}_4$, et on met $\mathtt{Sss}$ sur la nouvelle transition $\mathtt{Tr}(\mathtt{nd}_4,\mathtt{nd}_5)$.

\item Le 4\ieme{} mot $\mathtt{AbEfg}$ : d'après le tableau LCP, on a un préfixe commun d'une longueur $2$.
Alors depuis le n\oe{}ud feuille de mot précédent $\mathtt{nd}_5$, on remonte pour trouver le bon père, on remonte jusqu'au n\oe{}ud racine $\mathtt{nd}_0$. On crée un nouveau n\oe{}ud intermédiaire $\mathtt{nd}_6$ entre la transition $\mathtt{Tr}(\mathtt{nd}_0,\mathtt{nd}_2)$, la nouvelle transition $\mathtt{Tr}(\mathtt{nd}_0,\mathtt{nd}_6)$ va représenter le préfixe commun $\mathtt{aB}$.
Ensuite, on crée un nouveau n\oe{}ud feuille $\mathtt{nd}_7$ qui sort de $\mathtt{nd}_6$, et la transition $\mathtt{Tr}(\mathtt{nd}_6,\mathtt{nd_7})$ va contenir la chaîne de caractères $\mathtt{Efg}$.

\end{enumerate}

\medskip
Le résultat de la construction de la structure de données est illustré dans la figure \ref{fig:Créer_trie_LCP_exemple}.

\begin{figure}[H]
\centering
\includegraphics[width=0.6\linewidth]{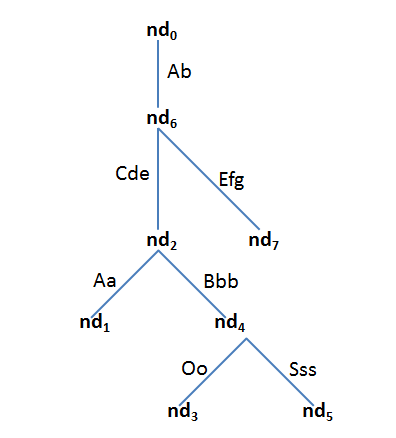}
\caption{Le Trie compact construit avec le tableau LCP.}
\label{fig:Créer_trie_LCP_exemple}
\end{figure}

\subsection{Réduire la mémoire nécessaire au Trie tout en restant efficace en temps de calcul}
\label{sub:Reduire_espace_memoire_de_Trie}

Dans cette approche, nous gardons à la fois, le dictionnaire des mots et la structure de données (le Trie compact) dans la mémoire centrale.

Afin de réduire la taille du Trie, toutes les transitions vont être encodées seulement avec leur 1\ier{} caractère, et on leur ajoute 1) la longueur de la transition,  donc la longueur du facteur, 2) un pointeur vers le début d'une occurrence du facteur dans le dictionnaire des mots.
Donc chaque transition, au lieu de contenir un facteur (d'un mot ou plusieurs mots partageant cette transition), elle ne contiendra que $3$ éléments, le 1\ier{} caractère du facteur, un pointeur vers le dictionnaire et la longueur de la transition.

La réduction de la mémoire du Trie compact avec l'exemple précédent (de la figure \ref{fig:Créer_trie_LCP_exemple}) est illustrée dans la figure \ref{fig:Trie_compression}.

\begin{figure}[H]
\centering
\includegraphics[width=0.6\linewidth]{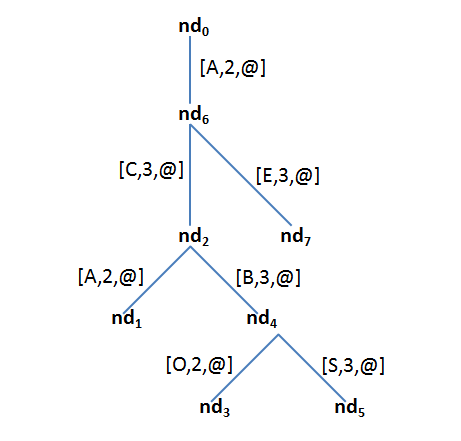}
\caption{Réduire l'espace mémoire occupé par le Trie. Chaque transition est codée par son 1\ier{} caractère, la longueur de la transition (le facteur), et un pointeur vers le début d'une occurrence du facteur dans le dictionnaire.}
\label{fig:Trie_compression}
\end{figure}

\medskip
Pour maintenir notre structure de données aussi petite que possible, les caractères UFT-8 et les nombres entiers sont codés par bits~\cite{Elias1975}, les pointeurs sont également codés comme des entiers.

\subsection{Ajouter les scores des mots au Trie compact}
\label{sub:Ajouter_scores_mots_Trie}

Chaque mot est associé un score statique, dans la construction du Trie compact, lorsqu'on atteint une feuille, on stocke le score du mot associé dans le n\oe{}ud feuille pointant vers le mot. Après l'insertion de tous les mots dans le Trie, et afin de supporter une complétion Top-k d'une façon efficace et rapide, pour chaque n\oe{}ud interne, on garde de manière récursive le score maximal parmi ses enfants, jusqu'à ce qu'on arrive à la racine.

Généralement, les scores sont stockés à la fin de chaque mot, donc on doit les extraire pour les insérer dans le Trie.

\medskip
\noindent
Exemple : soit le dictionnaire précédent avec des scores\\
 $D=\{\mathtt{AbCdeAa}\#55 , \mathtt{AbCdeBbbOo}\#9 , \mathtt{AbCdeBbbSss}\#11 , \mathtt{AbEfg}\#33\}$. Le résultat d'ajouter les scores au Trie est illustré dans la figure \ref{fig:Trie_add_scores} (dans cette figure, nous avons utilisé la version du Trie compact sans l'étape de réduction de l'espace mémoire) :

\begin{figure}[h]
\centering
\includegraphics[width=0.7\linewidth]{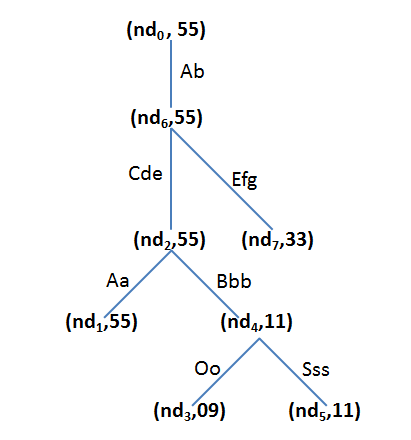}
\caption{Ajouter les scores des mots dans le Trie compact.}
\label{fig:Trie_add_scores}
\end{figure}

\subsection{Préparer la file de priorité et le tableau de hachage}
\label{sub:Preparer_file_priorite_tableau_hachage}

Nous proposons d'utiliser deux structures de données supplémentaires pour supporter les requêtes Top-k d'une façon efficace et rapide, et aussi pour supprimer les résultats en double de la liste finale.

La file de priorité est une structure de données qui permet de garder une liste d'éléments dans l'ordre, et dont l'élément avec la plus grande priorité est stocké en première position. Cette structure de données est utilisée pour permettre a trouver les résultats Top-k afin de les retourner à l'utilisateur. Pour implémenter notre file de priorité, nous avons utilisé un tas classique  implémenté avec tableau.\\

Nous utilisons un tableau de hachage avec sondage linéaire, afin de vérifier si le résultat existe dans la liste des solutions dans le but d'éliminer les résultats en double.

\subsection{Résumé de toutes les étapes de construction de la structure de données}

\noindent
\rule{8cm}{0.1pt}\\
\textbf {Algorithme de construction de la structure de données (les grandes lignes)}\\
\textbf{Entrée :} fichier contenant les mots du dictionnaire.\\
\textbf{Sortie :} structure de données pour l'auto-complétion.\\

\begin{enumerate}[1.]

\item Ordonner le dictionnaire dans l'ordre lexicographique. Pour plus de détails voir la sous-section \ref{sub:Lordre_lexicographique_de_dictionnaire}.

\item  Construire le tableau des plus longs préfixes communs (LCP array) de tous les mots du dictionnaire. Pour plus de détails voir la sous-section \ref{sub:construction_Trie_avec_LCP}.
 
\item  Construire le Trie compact en se basant sur le tableau LCP pour accélérer la construction. Pour plus de détails voir la sous-section \ref{sub:construction_Trie_avec_LCP}.

\item  Ajouter les scores des mots au Trie compact. Pour plus de détails voir la sous-section \ref{sub:Ajouter_scores_mots_Trie}.

\item Préparer la file de priorité. Pour plus de détails voir la sous-section \ref{sub:Preparer_file_priorite_tableau_hachage}.

\item Préparer un tableau de hachage avec sondage linéaire. Pour plus de détails voir la sous-section \ref{sub:Preparer_file_priorite_tableau_hachage}.

\end{enumerate}
\rule{8cm}{0.1pt}\\

\section{La méthode de recherche}
\label{sec:auto:methode_recherche}

On commence d'abord par expliquer l'algorithme qui permet de rechercher les résultats et de  proposer la liste de suggestions. Dans notre méthode, le Trie est utilisé afin de répondre aux requêtes de manière approchée, en utilisant un classement transversal de certains n\oe{}uds du Trie. Soit $q$ le mot requête avec $|q|=m$ et $k$ le nombre des résultats demandés dans la liste de complétion.

\subsection{Trouver les n\oe{}uds valides}
\label{sub:Trouver_les_noeuds_valides}

Étant donné un mot $q$, on appelle \emph{locus} la position dans le Trie compact où la recherche exacte de $q$ s'est arrêtée. Cette position peut être sur un n\oe{}ud, ou au milieu d'une arête (transition).
La recherche exacte s'arrête dans la position nommée \emph{locus} car soit on est arrivé à la profondeur $|q|$, donc on a trouvé le mot requête complet (la profondeur $|q|$ est le nombre total des caractères du mot requête $q$), ou on a obtenu une erreur avant de terminer la recherche et de trouver le mot entier $q$.\\

On appelle \emph{n\oe{}ud locus} le n\oe{}ud qui représente la position \emph{locus}. Si la position \emph{locus} se termine dans un n\oe{}ud, ce dernier est alors, le \emph{n\oe{}ud locus}; sinon, la position \emph{locus} se termine au milieu d'une arête et donc on considère le n\oe{}ud à destination de cette arête comme le \emph{n\oe{}ud locus}.\\

On appelle un n\oe{}ud \emph{n\oe{}ud valide}, le n\oe{}ud qui mène à des solutions exactes ou des solutions approchées.
Pour les solutions exactes, nous n'avons qu'un seul n\oe{}ud, par contre pour les solutions approchées, on peut avoir plusieurs n\oe{}uds possibles.

Dans ce qui suit, on donne l'algorithme pour trouver tous les n\oe{}uds valides. On utilise l'algorithme naïf de la recherche approchée dans un Trie.

\noindent
\rule{8cm}{0.1pt}\\
\textbf {Algorithme\_tous\_les\_n\oe{}uds\_valides}\\
\textbf{Entrée :} la requête  (préfixe) $q$.\\
\textbf{Sortie :} une liste de n\oe{}uds valides.\\
\begin{enumerate}[a.]

\item \label{item:nd_v_1} Trouver le \emph{n\oe{}ud locus} $\mathtt{nd}$ de $q$.

\item Si le \emph{n\oe{}ud locus} $\mathtt{nd}$ est à une profondeur d'au moins $|q|$, cela signifie que l'on a une solution exacte, alors ajouter ce n\oe{}ud à la liste des n\oe{}uds valides.

\item Pour chaque n\oe{}ud sur le chemin menant au \emph{n\oe{}ud locus} trouvé à l'étape \ref{item:nd_v_1} :
	\begin{itemize}
	\item Prendre un n\oe{}ud comme \emph{n\oe{}ud locus}.

	\item Faire une opération d'édition sur toutes les transitons sortantes (sauf pour celle qui est sur le chemin qui mène au \emph{n\oe{}ud locus} trouvé à l'étape \ref{item:nd_v_1}).
	
	\item Continuer une recherche exacte pour le suffixe restant de la requête dans le sous-arbre descendant de ce nouveau \emph{n\oe{}ud locus}.

	\item Si une occurrence approchée de $q$ est trouvée, alors ajouter le \emph{n\oe{}ud locus} de cette solution comme un n\oe{}ud valide à la liste des solutions.
	\end{itemize}

\end{enumerate}
\rule{8cm}{0.1pt}\\

\begin{figure}[h]
\centering
\includegraphics[width=0.7\linewidth]{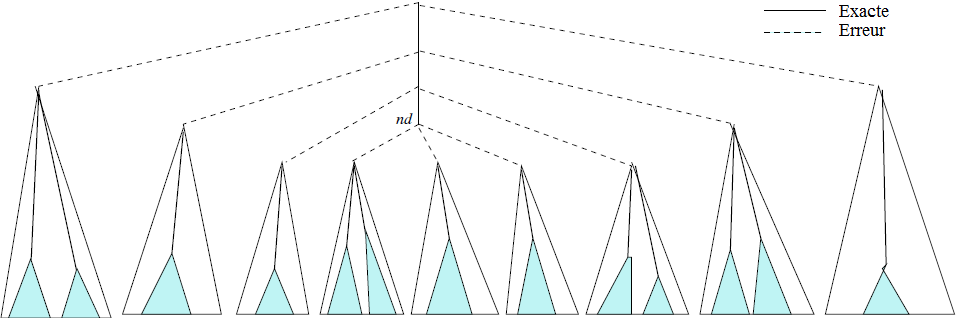}
\caption{La recherche avec une distance d'édition à 1-erreur dans un Trie compact.}
\label{fig:recherche_des noeud_trie}
\end{figure}

La figure \ref{fig:recherche_des noeud_trie} représente la recherche des n\oe{}uds valides qui mènent à des solutions exacte et approchées dans un Trie compact. Le chemin gras représente une correspondance exacte, alors que les chemins avec des points correspondent à celles avec 1-erreur d'édition (insertion, suppression, substitution). 

Dans la section \ref{sec:auto:reduire_nombre_branche_sortante_tester}, on détaille une méthode qui permet de faire une recherche approchée d'une manière efficace (sans explorer tout l'espace de Trie).

\subsection{Obtenir une liste de résultats à suggérer}
\label{sub:Obtenir_liste_top_K}

On se basant sur la liste des n\oe{}uds valides trouvés dans l'étape précédente, on calcule l'ensemble des $k$ résultats. Pour cela, on utilise les scores enregistrés dans chaque n\oe{}ud interne du Trie pour trouver quel chemin on doit prendre afin d'arriver aux résultats ayant les plus hauts scores. On utilise une file de priorité afin de trouver les $k$ meilleurs résultats et non juste le meilleur résultat.
L'algorithme permettant de calculer l'ensemble de la liste de suggestions Top-k, une fois que l'ensemble des n\oe{}uds valides a été calculé est le suivant :\\

\noindent
\rule{8cm}{0.1pt}\\
\textbf {Algorithme\_de\_liste\_de\_suggestions}\\
\textbf{Entrée :} liste des n\oe{}uds valides.\\
\textbf{Sortie :} liste des Top-k complétions.\\
\begin{enumerate}[a.]
		
\item Il y a une solution exacte : on prend tous les n\oe{}uds destinations des transitions sortantes depuis le n\oe{}ud valide représentant une solution exacte, et on les ajoute à la file de priorité.

\item \label{item:al_sol_pre_que_1} Obtenir le n\oe{}ud avec le plus grand score (celui en haut de la file de priorité).
Si ce n\oe{}ud représente une feuille, alors on ajoute le mot correspondant à la liste des résultats. Sinon on insère toutes les n\oe{}uds destinations de ses transitions (ces fils) dans la file de priorité.

\item \label{item:al_sol_pre_que_2} Faire le traitement de l'étape \ref{item:al_sol_pre_que_1} jusqu'à obtenir $k$ mots dans la liste de suggestions, ou jusqu'à ce que la file de priorité soit vide.
	
\item Si la file de priorité est vide avant d'avoir $k$ mots dans la liste des suggestions, ou si l'on n'a aucune solution exacte, alors : ajouter tous les n\oe{}uds valides restants qui représentent des solutions approchées avec leurs scores dans la file de priorité

\item Faire le traitement de l'étape \ref{item:al_sol_pre_que_2}

\end{enumerate}
\rule{8cm}{0.1pt}\\	

Dans cet algorithme qui trouve les $k$ résultats de complétions, on commence par le n\oe{}ud de la solution exacte afin de s'assurer que les premiers résultats dans la liste de suggestions sont des complétions exactes, ensuite, on ajoute les solutions approchées, car on doit d'abord présenter les solutions exactes (si elles existent) et ensuite seulement les solutions approchées.

À cette fin, dans la phase qui trouve les n\oe{}uds valides (voir la sous-section \ref{sub:Trouver_les_noeuds_valides}), lorsqu'on trouve une solution exacte, on l'ajoute dans la liste des n\oe{}uds valides dans le champ réservé à la solution exacte (où on marque juste qu'il s'agit d'une solution exacte) afin de le distinguer des autres n\oe{}uds. Ensuite, on ajoute les autres n\oe{}uds qui représentent des solutions approchées.

\section{L'auto-complétion et la saisie de l'utilisateur}
\label{sec:auto:autocompletion_user_typing}

Afin de ne pas recalculer tous les \emph{n\oe{}uds locus}\footnote{Voir la sous-section \ref{sub:Trouver_les_noeuds_valides}.} à partir du début, à chaque fois qu'un nouveau caractère est tapé, on stocke l'ensemble des positions des \emph{n\oe{}uds locus} après chaque nouveau caractère, et on continue la recherche dans les sous-arbres dont ils sont racines.

Une analyse des comportements d'un utilisateur en train de taper une requête, permet de les diviser en trois catégories : 

\begin{enumerate}

\item L'utilisateur tape sa requête pour la première fois, ou alors il tape une requête totalement différente de celle qui la précède.

\item Lorsque l'utilisateur tape sa requête, le système d'auto-complétion en temps réel lui propose des listes de suggestions. Et l'utilisateur continue à taper sa requête donc il ajoute d'autres caractères à la fin si aucune des suggestions ne lui convient.

Cette deuxième catégorie est la manière habituelle de l'interaction de l'utilisateur avec un système d'auto-complétion.

\item Lorsque l'utilisateur écrit sa requête, ensuite, il fait une modification par exemple une suppression à la fin de la chaîne pour qu'il retape d'autres caractères. En général, la modification peut-être : une suppression ou un ajout à n'importe quelle position dans la requête, à la fin, au milieu ou au début.
Généralement, cette troisième catégorie n'est pas la manière habituelle que l'utilisateur utilise avec le système d'auto-complétion.
\end{enumerate}
 
Dans ce qui suit, nous détaillons chaque catégorie et nous proposons une méthode permettant de s'adapter à l'utilisateur afin d'essayer de réduire les opérations faites dans la recherche et donc de gagner en temps de calcul.

\subsection{Chercher depuis la racine}
\label{sub:methode1_search_from_root}

La première méthode est la plus simple. C'est la méthode habituellement utilisée pour faire la recherche dans un Trie : 
elle consiste à faire la recherche à chaque fois en commençant le parcours du Trie depuis la racine.
On applique les algorithmes {\bf Tous\_n\oe{}uds\_valides} et {\bf Liste\_des\_suggestions}.
On appelle cette méthode {\em search\_from\_root}. Quelle que soit la requête tapée, on refait à chaque fois la recherche depuis la racine.

\subsection{Ajout à la fin}
\label{sub:methode2_search_end_node}

Le système d'auto-complétion est interactif, lorsqu'on tape une requête dans un champ de texte, les suggestions sortent au fur et à mesure qu'on tape les caractères, et cela dépend de \textbf{la durée d'attente entre chaque paire de requêtes consécutives}.\\

Exemple : on tape \textit{AB} donc on a une première liste de résultats pour le préfixe \textit{AB}. Ensuite, on continue à taper et on ajoute \textit{CD}, le mot requête devient \textit{ABCD}, donc on a une autre liste de suggestion pour le préfixe \textit{ABCD}. De même, à chaque fois qu'on ajoute de nouveaux caractères, on aura une nouvelle liste.\\

On remarque que chaque requête est un préfixe de la requête suivante. Généralement entre chaque paire de requêtes consécutives traitées par le système, il y a un temps d'attente prédéfini $t$. 
Pour ne pas refaire la recherche à chaque fois depuis le début (la racine) on procède comme suit :

\medskip
Sur les quelques premières frappes de la requête : on applique les algorithmes {\bf Tous n\oe{}uds valides} et {\bf Liste des suggestions}.

Quand des caractères sont ajoutés à la fin (en fonction du temps d'attente $t$, on peut avoir un ou plusieurs caractères), on ne commence pas la recherche à partir de la racine du Trie, mais plutôt, de la position du {\em n\oe{}ud locus} dans la dernière recherche exacte (i.e. dans la requête précédente). En d'autres termes, le  {\em n\oe{}ud locus} va être considéré comme une nouvelle racine et on continuera la recherche à partir de cette position, et on refait exactement les mêmes étapes dans {\bf Tous\_n\oe{}uds\_valides} et {\bf Liste\_des\_suggestions}.

\medskip
En général, on considère qu'on a deux mots requêtes différents. On commence la recherche à partir de la position du {\em n\oe{}ud locus} du dernier mot requête. Si c'est un préfixe du mot requête courant. Autrement, on commence la recherche à partir de la racine.\\

Cette manière de faire est la manière habituelle de taper une requête par l'utilisateur (caractère après caractère). Donc cette méthode est adaptée pour le besoin de proposer des listes de complétions d'une façon rapide au fur et à mesure que l'utilisateur tape les caractères de sa requête : à chaque fois qu'il ajoute des caractères, on affiche une liste d'une façon rapide et on continue la recherche depuis le n\oe{}ud auquel a abouti la recherche précédente.\\

On nomme cette méthode {\em search\_end\_node}.

\subsection{Une modification dans la requête} 
\label{sub:methode3_search_tab_nodes}

En tapant sa requête, l'utilisateur peut de temps en temps faire certaines modifications telles que la suppression, l'insertion au milieu, un ajout à la fin... Etc. 
On note ici qu'il y a des parties de la requête qui restent inchangées.

Pour ne pas recommencer à chaque fois la recherche depuis la racine, on enregistre tous les n\oe{}uds sur le chemin menant de la racine au {\em n\oe{}ud locus} dans un tableau $T$ de la 1\iere{} requête. La taille du tableau $T$ est la longueur du plus long chemin dans le Trie. On stocke chaque n\oe{}ud dans la case qui correspond à sa profondeur. Le n\oe{}ud $\mathtt{nd}$ est stocké dans la case numéro $i$ où la profondeur de $\mathtt{nd}$ est $i$, cela signifie qu'on à $i$ caractères depuis la racine jusqu'à ce n\oe{}ud $\mathtt{nd}$.

S'il y a un préfixe commun entre la requête modifiée (la 2\ieme{} requête) et celle qui la précède, alors on trouve le n\oe{}ud correspondant à ce préfixe commun à partir du tableau précédemment sauvegardé. Pour cela, il suffit juste d'aller directement à la case numéro (profondeur du préfixe commun). Le n\oe{}ud récupéré est considéré comme étant la racine, ensuite, on continue la recherche en faisant les mêmes étapes des algorithmes {\bf Tous\_n\oe{}uds\_valides} et {\bf Liste\_des\_suggestions}.

\medskip
En général, s'il y a un préfixe commun entre deux mots requêtes, on commence la recherche à partir du n\oe{}ud dont la profondeur est égale à la longueur du plus long préfixe commun.

\paragraph{Cette 2\ieme{} méthode a une mauvaise complexité temporelle :}

D'un premier point de vue, on dit que cette méthode permet de gagner un peu de temps de calcul puisqu'on ne recommence pas la recherche depuis le début, il suffit juste de trouver le bon n\oe{}ud et de le considérer comme une nouvelle racine pour continuer la recherche.

Mais en réalité, pour que cette méthode puisse s'exécuter, on doit faire des modifications ou une ré-initialisation dans le tableau ($Tab\_node$) qui sauvegarde les n\oe{}uds sur le chemin de la recherche après chaque nouveau mot requête.

Au début, toutes les cases du tableau sont initialisées à -1. Avec le 1\ier{} mot requête, on sauvegarde les n\oe{}uds qui correspondent au chemin de la recherche.
Dans le 2\ieme{} mot requête : si on a un préfixe commun,  on réinitialise la partie droite avec -1 à partir de la position du préfixe commun. Ensuite, on continue la recherche de suffixe et on sauvegarde les n\oe{}uds traversés pour ce dernier. Si on a un mot complètement différent on doit ré-initialiser tout le tableau pour l'utiliser de nouveau avec ce nouveau mot.

\medskip
Les étapes de ré-initialisation et modification du tableau ajoutent un temps supplémentaire au temps de recherche.

\medskip
Pour voir si cette méthode est utile ou non, on doit calculer la complexité de ces opérations.
Pour le 1\ier{} mot requête, on a $2m$. La recherche de mot requête dans le Trie est en $O(m)$, et l'ajout de tous les n\oe{}uds au tableau $Tab\_node$ est aussi en $O(m)$.

Pour le 2\ieme{} mot requête, on a $2m$. Calculer le plus grand préfixe commun et ré-initialiser la partie droite de tableau $Tab\_node$ donnent $O(m)$. Continuer à chercher le suffixe qui reste et stocker les n\oe{}uds trouvés de ce suffixe dans le tableau s'exécutent en $O(m)$. 

Cela donne entre le 1\ier{} et le 2\ieme{} mot requête une complexité de $4m$.

Au final, d'après le calcul de la complexité, cette méthode n'est pas efficace, car elle prend beaucoup plus d'opérations que de faire une simple recherche depuis la racine à chaque fois. Faire une simple recherche depuis la racine pour 2 requêtes est en $2m$ par contre cette méthode est en $4m$, donc elle prend un temps double.

On appelle cette méthode {\em search\_tab\_node}.

\section{Réduire le nombre des branches sortantes testées dans le Trie}
\label{sec:auto:reduire_nombre_branche_sortante_tester}

Afin de réduire le nombre des branches sortantes à tester dans chaque n\oe{}ud valide, nous avons effectué les expérimentations sur  la version côté serveur à l'aide d'une table de hachage, qui permet de générer des caractères candidats à tester dans chaque branche.

On stocke un dictionnaire de listes de substitutions tel que défini dans ~\cite{ibra.Chegrane.simple,Be09}, voir la sous-section \textit{dictionnaire des listes de substitutions}~\ref{sub:dictionnaire_des_liste_de_substitution} dans le chapitre~\ref{chap:recherche_approchee_avec_hachage}. Ce dictionnaire est construit seulement sur les préfixes ayant une longueur limitée. Pour chaque longueur, on stocke les caractères candidats pour toutes les positions.

Un dictionnaire de listes de substitutions, pour une longueur $d$ stocke une liste de caractères $c$ associée à des motifs de substitutions $p \phi q$, de telle sorte que $p c q$ est une chaîne de caractères de longueur $d$ qui est elle-même le préfixe de certaines chaînes de caractères dans le dictionnaire.

Le dictionnaire de listes de substitutions permet de faire la recherche approchée comme suit : étant donné un mot requête de la forme d'un motif de substitution $p \phi q$, le dictionnaire de listes de substitutions retourne un ensemble des caractères $c$ (une liste) tels que s'ils sont substitués à la place de $\phi$, ils génèrent des préfixes pour les chaînes de caractères qui sont dans le dictionnaire.

Dans notre cas, nous avons un paramètre $D$ (par exemple $D=6$), de telle sorte qu'on stocke des listes de substitutions pour tous les préfixes de la longueur $d\leq D$.

Pour chaque préfixe commun du dictionnaire de longueur $D$, on prend tous ces préfixes de longueur de $2$ (longueur minimum d'un mot) à $D$, pour générer des listes de substitutions.
Ainsi, quelle que soit la longueur du préfixe requête tapé (de $2$ à $D$) on peut appliquer cette méthode de recherche approchée.

Nous avons implémenté trois stratégies. étant donné un mot requête $x[1..m]$ avec $m\leq D$.

\subsection{La 1\iere{} stratégie : construire les mots candidats et les vérifier dans le Trie.}
\label{sub:methode_1er_SL}

Dans la première stratégie, on applique la méthode expliquée dans le chapitre~\ref{chap:recherche_approchee_avec_hachage} (la recherche approchée avec hachage), qui trouve les mots candidats qui peuvent être une solution approchée, ensuite, on les vérifie dans le dictionnaire exact.

Pour chaque motif de substitution $p \phi q$, on cherche les caractères candidats depuis la liste de substitutions, et on les remplace à la place de $\phi$ pour obtenir un mot $p c q$ qui a une très forte chance d'être une solution approchée. Pour vérifier si le mot $p c q$ représente une solution approchée, on le vérifier dans le Trie, en faisant une simple recherche exacte, donc on commence depuis la racine jusqu'à trouver le dernier caractère du motif $p c q$.

L'algorithme détaillé de cette méthode est comme suit :

\noindent
\rule{8cm}{0.1pt}\\
\textbf {Algorithme\_tous\_les\_n\oe{}uds\_valides\_1err\_SL}\\
\textbf{Entrée :} le préfixe requête $x$.\\
\textbf{Sortie :} une liste de n\oe{}uds valides.\\
\begin{enumerate}[a.]

\item \label{item:1err_SL_a} Trouver le \emph{n\oe{}ud locus} $nd$ de $x$.

\item Si le \emph{n\oe{}ud locus} $nd$ est à une profondeur d'au moins $|x|$, cela signifie qu'on a une solution exacte, donc ajouter ce n\oe{}ud à la liste des n\oe{}uds valides.

\item Pour chaque n\oe{}ud sur le chemin menant au \emph{n\oe{}ud locus} trouvé en \ref{item:1err_SL_a}

	\begin{itemize}
	\item Prendre un n\oe{}ud comme un \emph{n\oe{}ud locus}.

	\item Placer le caractère spécial $\phi$ à la position correspondant au n\oe{}ud choisi, donc obtenir le motif $p \phi q$.
	
	\item Calculer la valeur de hachage de mot $x' = p \phi q$ donc $h(x')$.

	\item Interroger le dictionnaire des listes de substitutions (avec $h(x')$)  afin d'obtenir une liste de caractères.
	
	\item Remplacer chaque caractère $c$ candidat à la place de $\phi$ pour obtenir le mot candidat $pcq$.
	
	\item Faire une recherche exacte dans le Trie pour le mot $pcq$.

	\item Si on le trouve, ce mot $pcq$ est une solution approchée, on ajoute le \emph{n\oe{}ud locus} de cette solution comme un n\oe{}ud valide à la liste des solutions.
	\end{itemize}

\end{enumerate}
\rule{8cm}{0.1pt}\\

On appelle cette 1\iere{} méthode {\em 1\_err\_SL}.

\subsection{La 2\ieme{} stratégie : choisir les chemins candidats avec les caractères de substitutions}
\label{sub:methode_1er_SL_node}

Dans la recherche approchée dans un Trie, on doit tester tous les fils d'un n\oe{}ud pour arriver à la solution approchée. Notre but est de savoir quel est le chemin qui mène à la solution sans tester tous les chemins sortants d'un n\oe{}ud donné.
Pour cela, dans cette deuxième stratégie, dans les n\oe{}uds où on doit appliquer l'opération de distance d'édition pour trouver les solutions approchées, on interroge le dictionnaire des listes de substitutions en même temps que l'on traverse le Trie du haut vers le bas pour trouver les caractères qui marquent les transitions à vérifier.
Plus précisément, supposons pour un n\oe{}ud traversé en profondeur $d$, on interroge le dictionnaire des listes de substitutions pour le motif de substitution $p[1..d-1]\phi[d + 1..m]$ afin d'obtenir une liste de caractères.
Ensuite, on continue à traverser les transitions marquées par les caractères de la liste obtenue seulement.

Les étapes de cette méthode sont détaillées dans l'algorithme suivant :

\noindent
\rule{8cm}{0.1pt}\\
\textbf {Algorithme\_tous\_les\_n\oe{}uds\_valides\_1err\_SL\_node}\\
\textbf{Entrée :} le préfixe requête $x$.\\
\textbf{Sortie :} une liste de n\oe{}uds valides.\\
\begin{enumerate}[a.]

\item Trouver le \emph{n\oe{}ud locus} $nd$ de $x$.

\item Si le \emph{n\oe{}ud locus} $nd$ est à une profondeur d'au moins $|x|$, cela signifie que l'on a une solution exacte, donc ajouter ce n\oe{}ud à la liste des n\oe{}uds valides.

\item Pour chaque n\oe{}ud sur le chemin menant au \emph{n\oe{}ud locus} trouvé en a.
	
	\begin{itemize}
	\item Prendre un n\oe{}ud comme un \emph{n\oe{}ud locus}.

	\item Placer le caractère spécial $\phi$ à la position correspondant au n\oe{}ud choisi, donc obtenir le motif $p \phi q$.

	\item Interroger le dictionnaire des listes de substitutions afin d'obtenir une liste de caractères.
	
	\item Continuer à traverser les transitions enfants marquées par les caractères dans la liste de substitutions, afin de vérifier le suffixe $q$ qui reste.

	\item Si une correspondance approchée de $x$ est trouvée, on ajoute le n\oe{}ud de cette solution comme un n\oe{}ud valide à la liste des solutions.
	\end{itemize}

\end{enumerate}
\rule{8cm}{0.1pt}\\

On appelle cette deuxième méthode \emph{1-err\_SL\_Node}. 

\medskip
Cette 2\ieme{} méthode est très proche de la première méthode \emph{1-err\_SL}. En effet, au lieu de refaire la recherche depuis le début, il suffit juste de choisir le bon chemin dans le Trie avec le caractère récupéré de la liste de substitutions, ensuite, continuer la vérification du suffixe qui reste.

\subsection{La 3\ieme{} stratégie : utiliser la liste de substitutions dans seulement les 1\iers{} niveaux du Trie}
\label{sub:methode_1err_SL_3_level}

La recherche approchée dans un arbre (Trie ou autre) peut se diviser en deux (2), les premiers niveaux qui sont proches de la racine, et les niveaux inférieurs qui sont proches des feuilles.

Dans les premiers niveaux, chaque n\oe{}ud à un grand sous-arbre et un nombre de fils important, et cela influe sur la recherche approchée et la rend lente, car il y a beaucoup de possibilités et de chemins à tester avant de trouver le bon qui représente la solution. Donc la méthode simple de traiter ce problème est inefficace (celle qui applique directement l'opération de distance d'édition dans chaque n\oe{}ud).

Par contre si on considère que l'on recherche l'erreur seulement dans les derniers niveaux, la méthode simple (naïve) est suffisante et efficace, car chaque n\oe{}ud à un petit sous-arbre et un nombre de fils très petit, et cela rend la recherche rapide.

\medskip
En se basant sur cette observation, nous avons proposé une méthode hybride qui utilise les caractères de substitutions juste dans les 3 premiers niveaux, et pour le reste des niveaux, on applique la recherche simple dans le Trie.

À chaque fois qu'on supprime un niveau, un très grand nombre de chemins s'éliminent. Avec seulement trois niveaux, on élimine un nombre très important de possibilités, et on réduit la taille des sous-arbres qui restent.
Le but, c'est de minimiser l'espace mémoire utilisé par le dictionnaire des listes de substitutions (pour tous les préfixes de 2 à $D$, avec $D=6$, comme nous l'avons expliqué au début de cette section), et en même temps gagner dans le temps d'exécution. Avec  trois niveaux, les tests donnent de bons résultats.
						
Si l'erreur est dans ces 3 premiers niveaux, on applique la recherche avec la méthode qui utilise les caractères de substitutions afin de choisir les bons chemins. Sinon on applique la méthode de recherche naïve dans le Trie.

Dans la partie test (voir la sous-section \ref{sub:test_list_substitution}) la méthode \emph{1err\_SL\_node} donne de meilleurs résultats par rapport à la première méthode \emph{1err\_SL}. Pour cela, dans cette 3\ieme{} méthode hybride, on utilise la 2\ieme{} méthode \emph{1err\_SL\_node} dans les 3 premiers niveaux, et pour le reste, on utilise la méthode classique expliqué dans la sous-section \ref{sub:Trouver_les_noeuds_valides}.\\

Les étapes de cette méthode sont résumées dans l'algorithme suivant :\\

\noindent
\rule{8cm}{0.1pt}\\
\textbf {Algorithme\_tous\_les\_n\oe{}uds\_valides\_1err\_SL\_3\_level}\\
\textbf{Entrée :} le préfixe requête $x$.\\
\textbf{Sortie :} une liste de n\oe{}uds valides.\\
\begin{enumerate}[a.]

\item Trouver le \emph{n\oe{}ud locus} $nd$ de $x$.

\item Si le \emph{n\oe{}ud locus} $nd$ est à une profondeur d'au moins $|x|$, cela signifie qu'on a une solution exacte, donc on ajoute ce n\oe{}ud à la liste des n\oe{}uds valides.

\item Pour chaque n\oe{}ud sur le chemin menant au \emph{n\oe{}ud locus} trouvé en (a).

	\begin{itemize}
	\item Prendre un n\oe{}ud $nd$ comme un \emph{n\oe{}ud locus}.
	
	\item Si la profondeur de \emph{n\oe{}ud locus} < 3 ($Prof(nd) < 3$) alors :
			
		\begin{itemize}
		\item Placer le caractère spécial $\phi$ à la position correspondant au n\oe{}ud choisi, pour obtenir le motif $p \phi q$.
	
		\item Interroger le dictionnaire des listes de substitutions afin d'obtenir une liste de caractères.
		
		\item Continuer à traverser les transitions marquées par les caractères qui sont dans la liste seulement, afin de vérifier le suffixe $q$ qui reste.
	
		\item Si une correspondance approchée de $x$ est trouvée, alors ajouter le n\oe{}ud de cette solution comme un n\oe{}ud valide à la liste des solutions.
		\end{itemize}
			
	\item Sinon, $Prof(nd) \geq 3$ alors:
		\begin{itemize}
		\item Faire une opération de distance d'édition sur toutes les transitions sortantes (sauf pour celle qui est sur le chemin qui mène au \emph{n\oe{}ud locus} trouvé en a).
			
		\item Continuer une recherche exacte pour le suffixe restant de la requête dans le sous-arbre descendant de ce nouveau \emph{n\oe{}ud locus}.
		
		\item Si une correspondance approchée de $x$ est trouvée alors on ajoute le n\oe{}ud de cette solution comme un n\oe{}ud valide à la liste des solutions.
		\end{itemize}
			
	\end{itemize}

\end{enumerate}
\rule{8cm}{0.1pt}\\

On appelle cette méthode : \emph{1-err\_SL\_3\_level}.

\subsection{Évaluation de la complexité}

Nous donnons une évaluation de la complexité moyenne pour les deux méthodes {\em 1\_err\_SL} et \emph{1-err\_SL\_Node}.

\begin{theorem}
Étant donné un dictionnaire $D$ de $d$ mots, et un mot requête $q$ d'une longueur $m$.
La complexité temporelle moyenne pour une requête préfixe de la recherche approchée pour l'auto-complétion est de $O(m^2)$.
\end{theorem}

\begin{proof}

L'algorithme fait une recherche exacte jusqu'à la position de l'erreur. Tous les n\oe{}uds depuis la racine à la position de l'erreur vont être vérifiés pour trouver les solutions approchées. Le nombre maximal de positions est égal à $m$.

Pour une seule position $i \in [1..m]$, on applique la méthode de listes de substitutions pour obtenir les lettres pouvant mener à une solution. En moyenne on a juste une seule position (voir la preuve dans le chapitre \ref{chap:recherche_approchee_avec_hachage} (la recherche avec hachage), dans la preuve \ref{prof:chap:hachage:liste_sub}, et la section expérimentation section \ref{sec:hash:experimentation}). Donc on aura juste un seul chemin à vérifier, ce qui donne un temps $O(m)$ pour vérifier le préfixe avant la position de l'erreur et le suffixe après la position de l'erreur. Le calcul de la liste de substitutions se fait quant à lui en temps $O(1)$ (voir le détail dans le chapitre sur la recherche avec hachage). Donc le temps total pour une seule position est $O(m)$ et le temps total pour toutes les positions est $O(m \times m) = O(m^2)$.
\end{proof}

\section{Les résultats Top-K}
\label{sec:auto:resultat_top-k}

La liste des résultats dans l'auto-complétion pourrait être très grande, surtout dans les deux cas suivants :

1) Lorsque l'utilisateur tape moins de 3 caractères, on aura un très grand sous-arbre pour rendre toutes ses feuilles comme résultats, et donc beaucoup de mots. Par exemple, si on a qu'un seul caractère dans certains cas, on aura presque la totalité de l'arbre. Beaucoup de bibliothèques (par exemple \textit{JQuery UI Autocomplete}) proposent comme paramètre de réglage un nombre minimum de caractères à taper pour traiter la requête afin d'éviter de rendre une grande partie de dictionnaire comme résultat.

2) Lorsqu'on accepte les erreurs dans la requête, une conséquence évidente est que la liste de suggestions devient très grande aussi.\\

Afin de réduire le nombre de résultats, et mettre en avant les résultats les plus importants qui satisfont le plus le besoin de l'utilisateur, on utilise le système dit \textbf{Top-k} qui rapporte les $k$ suggestions les plus hautement classées dans un ordre décroissant par rapport à leur score.

Le paramètre $k$ est donné par l'utilisateur, le nombre de résultats dépend de la taille  de l'interface du dispositif utilisé. Généralement, on utilise une valeur entre 10 et 15. Cette méthode est appelée \textbf{Top-k complétion} (Top-k suggestion).

Comme on a expliqué dans la sous-section \ref{sub:Obtenir_liste_top_K}, on utilise un Trie qui stocke tous les scores des mots dans les n\oe{}uds feuilles, et chaque n\oe{}ud interne garde le score maximal de ces enfants. Le Trie avec les scores et une autre structure de donnée nommée la file de priorité sont combinées dans le but de trouver les $k$ résultats qui ont les scores les plus élevés. L'algorithme permettant de trouver la liste Top-k  comme expliqué précédemment est: {\bf Algorithme de liste de suggestions}.\\

Pour faire la recherche Top-k, les mots du dictionnaire doivent avoir un score qui représente leur importance. Dans le cas où certains mots n'ont pas de score, on leur attribue le score NULL, donc la valeur 0.\\

Exemple $D=\{"abcd\#10" ,  "abcdefg\#5" , "xyz\#02", "ibra", "nasa" \}$.\\

Dans cet exemple, les mots $\{"abcd" , "abcdefg" , "xyz"\}$ ont des scores (ils apparaissent à la fin avec un séparateur spécial $\#$). Par contre, les deux mots $\{ "ibra", "nasa"\}$ n'ont pas de score, dans ce cas, on leur donne la valeur $0$.\\

Dans une recherche, les résultats ayant un même score sont classés par ordre lexicographique.

Dans les deux sous-sections qui suivent, nous expliquons le concept des résultats Top-k statique et dynamique, et nous donnons une méthode simple qui donne la possibilité à l'utilisateur de visionner tous les résultats possibles d'une requête groupe par groupe sans encombrer l'interface graphique et sans refaire la recherche.

\subsection{Classement Top-k statique et dynamique.}
\label{sub:Classement_top_k_statique_dynamique.}

Les résultats dans un classement statique restent toujours les mêmes dans le temps. Si on tape la même requête, plusieurs fois, l'algorithme retournera comme résultat  toujours la même liste, même si par exemple, on choisit toujours le dernier mot de la liste.

Comme exemple des classements statiques : les noms des villes du monde selon leur superficie ou leur nombre d'habitants (le nombre d'habitants change, mais généralement les statistiques sont faites dans des périodes un peu éloignées), l'ordre lexicographique du dictionnaire,... etc.\\

Dans un classement Top-k dynamique, les résultats dans la liste peuvent changer avec la même requête, car les scores changent. Par exemple si dans une application donnée le score du mot choisi dans la liste augmente, donc à la prochaine requête sa position dans la liste change, alors  si l'utilisateur choisit toujours le même mot, ce dernier va apparaître au top de la liste (comme par exemple dans l'IDE de programmation ''Visuel Studio'').
 
Le score statique est le nombre donné pour chaque mot, ces scores peuvent être mis à jour. Le score pour un classement dynamique peut être une valeur positive/négative, et le score initial sera changé avec une unité spécifique. Si on a une valeur positive, le score augmente, par exemple : le système de Google +1.

Dans notre travail, les scores sont des nombres entiers donnés pour chaque mot dans le dictionnaire. Lorsque l'utilisateur choisit un mot depuis la liste des résultats Top-k affichée, on récupère le mot choisi par l'utilisateur, et on fait une mise à jour de son score (par exemple avec +1).

\paragraph{Comment faire une mise à jour de score : }
\label{subsub:comment_faire_update_top_k}

Pour faire la mise à jour de score du mot choisi, on doit faire la mise à jour à tous les n\oe{}uds qui sont sur le chemin dans le Trie depuis la racine jusqu'au n\oe{}ud feuille. Pour cela,  nous avons choisi dans notre travail, de refaire la recherche du mot depuis le début et de sauvegarder dans une liste tous les n\oe{}uds trouvés sur le chemin jusqu'au n\oe{}ud feuille. Lorsqu'on arrive au dernier n\oe{}ud, on récupère le score et on lui ajoute l'unité de mise à jour (par exemple +1). Ensuite, on prend le nouveau score et on le compare avec les scores des n\oe{}uds qui sont sur le chemin de la recherche et qui ont été sauvegardés pour leur appliquer une mise à jour. Donc si le nouveau score est plus grand que les scores de ces n\oe{}uds, alors on les modifie. On applique la mise à jour à tous les n\oe{}uds qui sont sur le chemin de la recherche pour rééquilibrer le Trie par le nouveau score, puisque chaque n\oe{}ud interne contient le score le plus élevé de ces fils.

\subsection{Afficher tous les résultats ordonnés en Top-k groupe par groupe}
\label{sun:top_k_groupe_by_groupe}

Dans certains cas, l'utilisateur veut voir une grande partie ou la totalité, des résultats de sa requête. L'affichage restreint aux résultats Top-k ne satisfait pas le besoin de l'utilisateur, comme par exemple une requête de recherche d'un produit donné, un livre ou un film, etc.

Dans ce cas il est plus adéquat de proposer un système qui affiche tous les résultats à l'utilisateur, à sa demande.

Dans notre travail, nous nous inspirons des moteurs de recherche comme Google qui affiche les résultats Top-k et propose des boutons/liens en bas de la page (next) pour afficher les pages suivantes.

En affichant les résultats Top-k, on fournit une option pour obtenir les autres résultats, et les afficher groupe par groupe. Pour une démonstration, voir dans \url{http://5.135.166.57/APPACOLIB/}~\footnote{Visité le: 02-07-2016.}, le bouton \fbox{{\tt next}} , ou la combinaison (CTRL + flèche droite) pour obtenir la même fonction.\\

Lorsque l'on utilise cette technique, on ne refait pas la recherche du mot requête à chaque fois. Il suffit juste d'aller à l'index (la file de priorité et le Trie) et extraire les résultats par groupe de $k$ éléments.

Lorsque on a une requête, on recherche d'abord les n\oe{}uds valides par l'algorithme {\bf Algorithme tous les n\oe{}uds valides}, ensuite, on les range dans une file de priorité pour trouver les $k$ résultats par l'algorithme ({\bf Algorithme de liste de suggestions}) expliquer dans la section \ref{sec:auto:methode_recherche}. S'il reste des n\oe{}uds dans la file de priorité, cela veut dire que nous avons un nombre de résultats plus que $k$. 
Ainsi lorsque l'utilisateur veut afficher les $k$ éléments suivants, il suffit juste d'exécuter les mêmes opérations de l'algorithme ({\bf Algorithme de liste de suggestions}) pour extraire les $k$ résultats depuis les n\oe{}uds qui restent dans la file de priorité.
  
Les $k$ résultats ne sont pas pré-calculés pour les afficher par groupe de $k$ éléments, car cela prend un temps considérable pour trouver tous les résultats. Mais à chaque fois que l'utilisateur demande de voir le groupe suivant des résultats, on utilise la file de priorité et on applique les étapes de l'algorithme ({\bf Algorithme de liste de suggestions}).
Généralement, l'utilisateur est satisfait juste par le premier groupe des résultats; dans certains cas, il visualise le premier groupe suivant.

\paragraph{Remarque :}
Si l'utilisateur ne trouve pas sa requête dans le premier groupe des résultats et ne veut pas consulter les autres groupes, il lui suffit juste de taper quelque caractères supplémentaire de sa requête pour faire plus de filtrage sur les résultats.

\section{Éliminer la redondance (les résultats en double)}
\label{sec:auto:Eliminer_la_redondance}

Dans la recherche approchée il y a toujours des résultats en double. Ce phénomène est encore plus large dans l'auto-complétion approchée. Lorsqu'on applique les trois opérations de distance d'édition pour rechercher des solutions avec une erreur, on peut atteindre les mêmes n\oe{}uds valides, et donc tous leurs fils sont des résultats en double.\\

\noindent
Il existe deux cas où les n\oe{}uds valides causent une redondance des résultats : 
\begin{enumerate}

\item On trouve le même n\oe{}ud valide plus qu'une seule fois, par exemple avec
l'opération de suppression on trouve $nd_1$ comme un n\oe{}ud valide, et avec
la substitution on retrouve le même n\oe{}ud valide $nd_1$.

\item On trouve $nd_1$ comme un n\oe{}ud valide (par exemple avec une opération 
de suppression), et ensuite, on trouve son fils $nd_2$ comme un n\oe{}ud valide (par exemple avec une opération d'insertion), et par conséquent, tous les fils du deuxième n\oe{}ud sont des résultats en double.
\end{enumerate}

\bigskip
Comment résoudre ce problème des n\oe{}uds valides en double?\\

Avant d'ajouter les n\oe{}uds valides à la liste des solutions, on vérifie d'abord s'ils existent dans cette liste ou non. L'auto complétion ne donne que les premiers éléments Top-k, cela signifie une petite liste d'éléments. Pour cela, on peut utiliser une petite table de hachage.
Cette table de hachage est utilisée dans l'algorithme {\bf Algorithme tous les n\oe{}uds valides}, avant l'ajout des n\oe{}uds valides à la liste des solutions.\\

Pour la relation entre les fils et les n\oe{}uds valides père, on utilise une table de hachage dans l'algorithme {\bf Algorithme de liste de suggestions}, pour vérifier si un mot solution existe dans la liste de suggestions ou non.

\section{Les stratégies client/serveur de l'auto-complétion}
\label{sec:auto:strategies_client_serveur}

Il y a deux stratégies extrêmes pour l'auto-complétion : soit le programme s'exécute sur le coté serveur, et le client envoie les requêtes et gère les résultats seulement, ou bien le programme s'exécute du coté client (en utilisant JavaScript), et le client fait tous les traitements en local et affiche les choix (résultats) finaux.

Entre ces deux extrêmes, plusieurs scénarios peuvent être considérés. Ils dépendent de plusieurs paramètres, parmi les plus importants :

\begin{enumerate}[a)]

\item La taille du dictionnaire. %Dictionary size;

\item La charge du serveur. %Server load;

\item La vitesse de connexion.  %Connection speed;

\item Le nombre de connexions courantes au niveau du serveur (nombre de connexions simultanées).

\item La puissance du calcul du coté client (la puissance de calcul du coté serveur est supposée  beaucoup plus puissante ($\gg$) que celle du client, mais partagée).

\item Le classement statique ou dynamique. %Static vs dynamic rank

\item Le délai entre deux requêtes. %Update delay.

\end{enumerate}

\noindent
Notre bibliothèques est disponible sur
\url{https://github.com/AppacoLib/api.appacoLib}, elle permet d'exécuter nos algorithmes soit sur le serveur ou du côté client, selon le meilleur scénario recherché.
Le c\oe{}ur de nos bibliothèques est écrit en (C/C++) du côté serveur, et en JavaScript pour le côté client.

Nous expliquons ci-dessous ses deux utilisations principales.

\paragraph{Appacolib du coté serveur, } Typiquement si on a un dictionnaire très volumineux (plus de 3 Mo d'entrées), on peut utiliser la bibliothèque écrite en C/C++ avec un module FASTCGI pour fournir une liste de suggestions dans un format JSON et renvoyer les résultats à travers des appels AJAX, puis afficher la liste en utilisant la bibliothèque écrite en JavaSscript (les résultats sont affichés généralement en dessous du champ de texte qui s'auto-complète). Pour une démonstration voir \url{http://5.135.166.57/APPACOLIB/result_from_server_CGI.html}.

\paragraph{Appacolib du coté client, } Du coté client, on a plusieurs stratégies dépendant du scénario voulu :

\begin{enumerate}[1)]

\item Toute l'opération de l'auto-complétion approchée peut être effectuée au niveau du navigateur du coté client. Si le dictionnaire n'est pas très volumineux (par exemple moins de 3 Mo d'entrées), on peut effectuer toutes les opérations au niveau local : la construction de l'index, la recherche de requête préfixe, et l'affichage de la liste des suggestions. Dans ce cas, on doit fournir une liste de mots pour construire l'index. La source du dictionnaire peut être locale ou au niveau du serveur, et peut être en différents formats (fichier, une chaîne de caractères avec des séparateurs, un tableau, ...etc.). Pour plus de détails sur le format de source pour construire l'index voir la documentation de \emph{AppacoLib} dans \url{https://github.com/AppacoLib/api.appacoLib/tree/master/doc_appaco_lib}.

\item Afficher les résultats au niveau du client uniquement : en utilisant un appel AJAX pour envoyer la requête et recevoir une liste de réponses venues du serveur.
Dans ce cas, on peut combiner les deux parties de notre bibliothèque: utiliser la bibliothèque en C/C++ du coté serveur pour gérer la création de l'index et répondre aux requêtes de l'auto-complétion, et utiliser la bibliothèque du coté client pour recevoir les résultats et les afficher.

Lorsqu'on utilise la bibliothèque du coté client pour juste envoyer la requête au serveur, et gérer  l'affichage des résultats finaux, le serveur n'est pas forcément notre bibliothèque écrite en C/C++, mais il peut être n'importe quel système qui accepte les requêtes AJAX et qui fournit les résultats en format JSON, comme par exemple un simple système écrit en PHP qui lit les résultats depuis un simple tableau, ou peut être un système plus complexe comme la gestion des complétions depuis les bases de données.

\item Télécharger le Trie construit dans le serveur, et exécuter les requêtes localement.

\end{enumerate}

\section{À quelle vitesse devrait être l'auto-complétion}
\label{sec:auto:quelle_vitesse_devrait_etre_l'auto-completion}

Pour être utile, l'auto-complétion doit être sensible, réactive et instantanée.
Quand l'auto-complétion est gérée du coté serveur, le calcul du temps de traitement prend en considération le temps pour taper un caractère et le temps du transfert des données entre le client et le serveur, et le temps du traitement du côté client pour gérer l'interface.

Le nombre de frappes par seconde est de $7,5$ \cite{Starner1996}, cela donne pour un caractère, l'utilisateur met ($ 133 ms $). Miller \cite{Miller1968}, a montré que pour être instantané pour l'utilisateur, le temps de réponse doit être inférieur à $100 ms$ ($ms$ : MilliSeconde).
Cela signifie que l'on a seulement environ $100 ms$ pour traiter la requête et donner une liste de suggestions qui contient les résultats Top-k. Cette période de $100 ms$ inclut le temps de la communication entre le serveur et le client et le temps de gestion de l'interface.
Ainsi, le côté serveur doit traiter et renvoyer les résultats très rapidement.

Si on considère seulement le côté client pour faire de l'auto-complétion indépendamment du serveur, on aurait exactement $100 ms$ pour effectuer tout le traitement (sans compter le temps de transfert, car tous les calculs sont faits en local).

\section{Tests et expérimentations}
\label{sec:auto:test}
% % ms : Millisecond

Les tests ont été effectués sur deux fichiers: le dictionnaire Anglais ($213,557$ mots,
$2,4$ Mégaoctets) et un extrait des titres des articles de Wikipédia ($1,2$ millions titres,
$11,5$ Mégaoctets).

Notre bibliothèque est disponible sur \url{https://github.com/AppacoLib/api.appacoLib}. Le c\oe{}ur de nos bibliothèques principales est écrit en (C/C ++) pour le côté serveur, et en JavaScript pour le côté client.
L'implémentation de notre bibliothèque est modulaire, elle est divisée en plusieurs fichiers sources indépendants selon les fonctionnalités voulues. \textbf{La bibliothèque en C/C++ contient 7 593 lignes de code}, et la bibliothèque écrite en \textbf{JavaScript  contient 3 884  lignes} de code source.
Pour une démonstration de notre bibliothèque voir : \url{http://5.135.166.57/APPACOLIB/}.\\

Dans nos tests, nous avons considéré les deux scénarios extrêmes cités avant, l'auto-complétion est faite seulement par le client, et l'auto-complétion faite par le serveur et gérée par le client dans le navigateur.

Nous avons testé les différentes méthodes proposées, la complétion exacte ou approchée, les méthodes pour s'adapter aux modifications de l'utilisateur sur la requête, l'utilisation du hachage pour générer une liste de substitutions pour améliorer la recherche approchée.

Les paramètres pris en considération lors des tests sont :
\begin{itemize}

\item Le temps de construction de la structure de données.

\item L'espace occupé par l'index.

\item Le temps de réponse aux  requêtes en fixant le $k$ et on changeant la longueur des préfixes de 2 à 6, pour l'auto-complétion exacte et approchée.

\item L'affectation du nombre de résultats $k$ sur le temps du calcul.

\end{itemize}

\medskip
Les mots requêtes ont été obtenus en choisissant au hasard parmi les préfixes des mots dans le dictionnaire et introduisant des erreurs aléatoires dans des positions aléatoires.
Le temps final a été obtenu en faisant la moyenne de 1000 mots requêtes distinctes.

\subsection{Le coté serveur, C/C++}
\label{sub:test_serveur}

Les tests ont été effectués sur un ordinateur Intel core-2 duo e8400, windows 7, 3.0 GHz, 2GB de RAM, avec un compilateur GNU GCC version 4.4.1.

\subsubsection{Test du temps de réponse pour l'auto-complétion exacte et approchée}
\label{subsub:test_exact_err}

Dans ce test, nous avons focalisé sur le temps de réponse à une seule requête. Nous avons fixé $k$ à 10 résultats dans la liste, et nous avons changé les tailles des préfixes de 2 à 6.
Les tests sont faits avec la totalité du dictionnaire, l'espace mémoire occupé par l'index est d'environ 5 Mo pour le dictionnaire Anglais (En) et 29 Mo pour le dictionnaire WikiTitle (Wi). 
Le temps de construction du Trie pour le dictionnaire Anglais (En) est d'environ $177 ms$ et $896 ms$ pour le dictionnaire WikiTitle (Wi).

\begin{table}[h]
\begin{tabular}{|l|ll|ll|}
\hline
Longueur   & Exacte & 1-erreur  & Exacte  & 1-erreur \\
de requête & (En)   & (En)      & (Wi)    & (Wi) \\
\hline
2          & 0,02   & 0,11      & 0,02    & 0,34     \\
3          & 0,017  & 0,16      & 0,017   & 0,44   \\
4          & 0,013  & 0,19      & 0,014   & 0,53    \\
5          & 0,010  & 0,21      & 0,012   & 0,58   \\
6          & 0,009  & 0,25      & 0,01    & 0,61 \\
\hline
\end{tabular}
\centering
\caption{Temps de requête $+$ Top-$k$ (en millisecondes) pour les dictionnaires Anglais (En) et WikiTitles (Wi).}
\label{tab:temps_reponse_exact_err}
\end{table}

Dans le tableau \ref{tab:temps_reponse_exact_err},  pour l'auto-complétion exacte, il est intéressant de noter que le temps de réponse pour les requêtes courtes (long = 2) est plus long que pour celui des requêtes longues (long = 6).
Cela est dû principalement à l'existence de beaucoup plus de résultats, au delà du seuil fixé, lorsque la requête est courte. En effet, avec une requête courte,  la recherche s'arrête dans les premiers niveaux de l'arbre, donc on aura un sous-arbre très grand où on doit trouver les résultats Top-k, ce qui augmente le temps de calcul.

Cependant, pour l'auto-complétion avec erreur, le temps est beaucoup plus élevé comparé à l'auto-complétion exacte. Cela est dû au nombre d'opérations de la complétion exacte et approchée. Dans l'auto-complétion exacte, il suffit juste d'exécuter une seule recherche exacte puis de récupérer les résultats depuis un seul n\oe{}ud.
Par contre, dans l'auto-complétion approchée, on est obligé de faire une recherche exacte plus une opération de distance d'édition sur tous les n\oe{}uds sur le chemin depuis la racine jusqu'au dernier n\oe{}ud où la recherche exacte s'est arrêtée (faire une recherche approchée avec 1 seule erreur), ensuite, trouver les résultats depuis plusieurs n\oe{}uds, ceux de la solution exacte (1 seul n\oe{}ud), et ceux de la solution approchée (plusieurs n\oe{}uds).

La différence de temps dans les résultats de recherche exacte et approchée existe, mais elle devient négligeable par rapport à l'ensemble du temps total de réponse qui ne doit pas dépasser les $100 ms$.

\subsubsection{Test avec différentes tailles du dictionnaire}

Dans ce test, nous avons fait varier la taille du dictionnaire (le nombre des mots) à chaque fois pour voir l'influence de ce paramètre sur les 3 aspects suivants : le temps de création de l'index, le temps de réponse de requête et l'espace mémoire occupé par l'index.

\paragraph{Le temps de création de l'index :}

nous commençons par le temps de création de notre index, voir la figure \ref{fig:En_Wi_creation} et le tableau \ref{tab:espace_memoire_EN_WI}, qui illustrent le changement du temps de création des 2 dictionnaires selon leurs tailles.

\begin{figure}[h]
\centering
\includegraphics[width=0.49\linewidth]{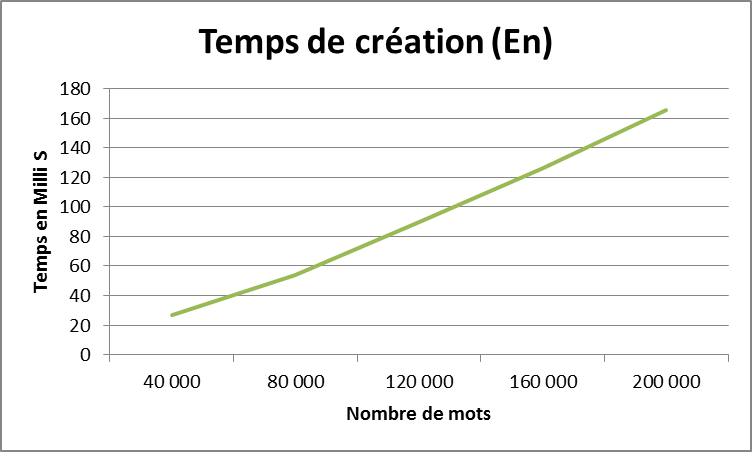}\hfill
\includegraphics[width=0.49\linewidth]{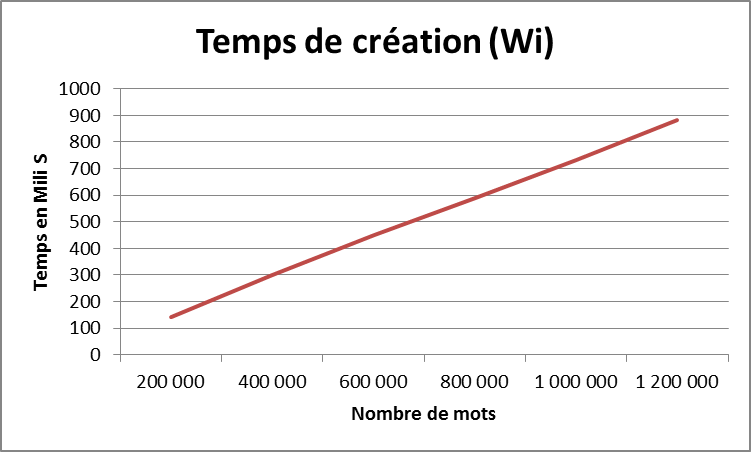}
\caption{Le temps de création pour différentes tailles des dictionnaires Anglais (En) et WikiTitle (Wi).}
\label{fig:En_Wi_creation}
\end{figure}

Dans la figure \ref{fig:En_Wi_creation}, on remarque bien que dans les 2 dictionnaires, le temps de création augmente en corrélation avec la taille du dictionnaire d'une façon linéaire. Pour le dictionnaire Anglais (En), le temps commence de $20ms$ jusqu'à $170ms$,  et pour WikiTitle (Wi) le temps est entre $150ms$ et $900ms$.\\

\begin{table}[h]
\begin{tabular}{|c|c|c|}
\hline
\begin{tabular}[c]{@{}l@{}}Nombre\\ Mots\end{tabular} & Création (EN)  & Création (Wi)  \\ \hline
\multicolumn{1}{|c|}{40 000}                          & 27  & 24  \\ \hline
\multicolumn{1}{|c|}{80 000}                          & 54  & 48  \\ \hline
\multicolumn{1}{|c|}{120 000}                         & 90  & 78  \\ \hline
160 000                                               & 126 & 106 \\ \hline
200 000                                               & 165 & 141 \\ \hline
\end{tabular}
\centering
\caption{Le temps de création en millisecondes pour les deux dictionnaires Anglais (En) et WikiTitle (Wi) avec les mêmes tailles.}
\label{tab:espace_memoire_EN_WI}
\end{table}

Il est intéressant de remarquer dans le tableau \ref{tab:espace_memoire_EN_WI}, que les temps de création pour les deux dictionnaires avec la même taille sont proches. La différence qui se trouve entre ces deux dictionnaires est due à la nature de chaque dictionnaire (la distribution des : mots, caractères, préfixes, suffixes, la longueur moyenne des mots,  etc.)

\paragraph{Le temps de recherche d'une requête : }
Le temps de l'auto-complétion des requêtes change en fonction de la taille de l'index et du type de l'index. Nous utilisons un Trie compact où le temps de recherche d'une requête exacte ne dépend pas de la taille de l'index, mais il dépend du nombre de caractères de la requête elle-même. Dans notre cas, nous avons la recherche de Top-k et la recherche approchée, les deux opérations sont influencées par la taille de l'index (le nombre des n\oe{}uds dans le Trie, et si la distribution des scores est équilibrée ou non). Les résultats sont illustrés dans les tableaux \ref{tab:Temps_rechercher_En_Wi_dif_size} et \ref{tab:Temps_rechercher_Wi_dif_grand_size}.

\begin{table}[h]
\begin{tabular}{|l|l|l|l|l|}
\hline
Nb des Mots & exacte (En) & exacte (Wi) & 1-err (En) & 1-err (Wi) \\
\hline
40 000      & 0,0156      &  0,0179 & 0,092 & 0,177 \\
80 000      & 0,0164      &  0,0183 & 0,106 & 0,224 \\
120 000     & 0,0159      &  0,0170 & 0,132 & 0,251 \\
160 000     & 0,0166      &  0,0177 & 0,139 & 0,252 \\
200 000     & 0,0160      &  0,0176 & 0,146 & 0,254 \\
\hline
\end{tabular}
\centering
\caption{Le temps de recherche + Top-k en millisecondes pour différentes tailles des dictionnaires Anglais (En) et WikiTitle (Wi).}
\label{tab:Temps_rechercher_En_Wi_dif_size}
\end{table}

\begin{table}[h]
\begin{tabular}{|l|l|l|}
\hline
Nb des Mots & exacte (Wi) & 1-err (Wi) \\
\hline
200 000     & 0,0176  & 0,254   \\
400 000     & 0,0180  & 0,282   \\
600 000     & 0,0175  & 0,320   \\
800 000     & 0,0181  & 0,350   \\
1 000 000   & 0,0178  & 0,373   \\
1 200 000   & 0,0179  & 0,485  \\
\hline
\end{tabular}
\centering
\caption{Le temps de recherche + Top-k en millisecondes pour WikiTitle (Wi) avec des grandes tailles.}
\label{tab:Temps_rechercher_Wi_dif_grand_size}
\end{table}

Dans les tableaux \ref{tab:Temps_rechercher_En_Wi_dif_size} et \ref{tab:Temps_rechercher_Wi_dif_grand_size}, on remarque  que dans les deux dictionnaires, le changement de taille n'influe pas sur le temps de l'auto-complétion exacte (recherche+Top-k), le temps est presque stable. Dans le dictionnaire Anglais, le temps est autour de $0,016ms$. Pour le dictionnaire Wikititle, le temps est autour de $0,018ms$. On remarque qu'entre la taille de 40 mille mots jusqu'à la totalité du dictionnaire Wikititle (Wi) ($1,2$ millions mots), le temps reste presque stable.
Entre les deux dictionnaires, même avec la grande différence du nombre de mots 200 mille mots (En), et $1,2$ million pour (Wi), on remarque que le temps ne change pas beaucoup, il y a une différence de $0,02ms$.

Dans ce test, nous avons fait la recherche dans le Trie plus la liste Top-k avec $k=10$. Le temps de la liste Top-k dépend des scores des mots dans le dictionnaire. Dans notre Trie chaque n\oe{}ud interne sauvegarde le score maximal de ses fils pour qu'on puisse trouver rapidement le mot qui a le score le plus élevé (cela prend le temps de traverser la hauteur du Trie), et avec l'utilisation de la file de priorité le temps reste stable pour un nombre de résultats $k$ fixe.

Le temps de l'auto-complétion avec erreur, change avec le changement de la taille du dictionnaire. Lorsque la taille du dictionnaire augmente cela implique que la taille de l'index augmente aussi (plus spécialement le Trie). Lorsque le Trie est grand cela signifie que le nombre des chemins à vérifier pour trouver des solutions approchées est grand, et cela influe sur le temps d'exécution final. Pour le Top-k,  (l'auto-complétion approchée) on a plusieurs n\oe{}uds à traiter, et donc plusieurs sous-arbres, contrairement au cas de l'auto-complétion exacte où on a juste un seul n\oe{}ud et donc un seul sous-arbre. Lorsque la taille du dictionnaire est grande, cela signifie que les tailles des sous-arbres aussi sont grandes, et cela, influe sur l'augmentation du temps d'exécution de Top-k.

\paragraph{L'espace mémoire : }

Le changement de la taille du dictionnaire initial influe directement sur la taille de l'index (le Trie). Lorsqu'on augmente le nombre de mots dans le dictionnaire, cela implique aussi une augmentation de la taille de l'index, car des nouveaux mots sont ajoutés au Trie. Voir la figure \ref{fig:changement_taille_En_Wi}.

\begin{figure}[h]
\centering
\includegraphics[width=0.49\linewidth]{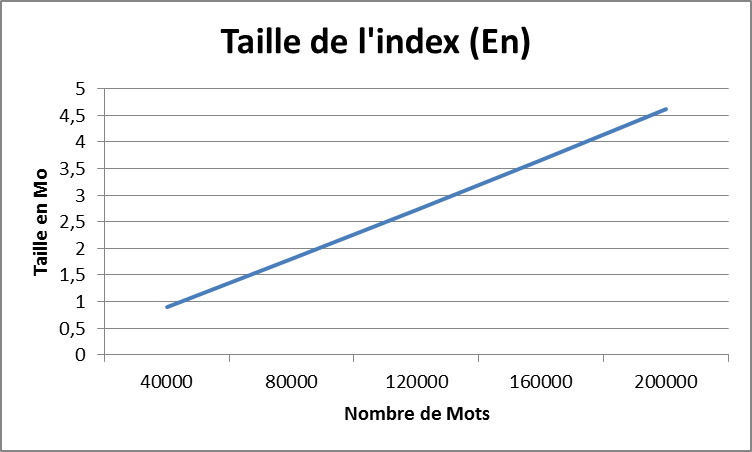}\hfill
\includegraphics[width=0.49\linewidth]{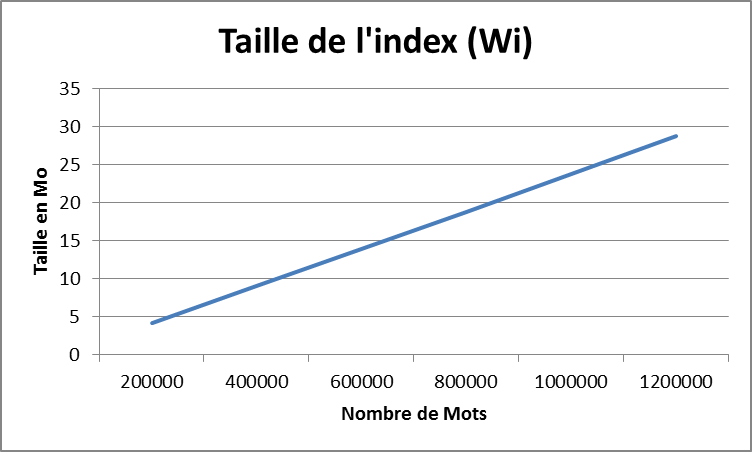}
\caption{Changement de la taille de l'index en fonction de la taille des dictionnaires Anglais (En) et WikiTitle (Wi).}
\label{fig:changement_taille_En_Wi}
\end{figure}

La figure \ref{fig:changement_taille_En_Wi}, montre l'augmentation de la taille de l'index en corrélation avec l'augmentation du nombre de mots.

L'augmentation de la taille de l'index est linéaire, pour le dictionnaire Anglais (En). Elle commence à 1 Mo pour 40 mille mots jusqu'à 5 Mo pour 200 mille mots. La taille initiale du dictionnaire Anglais est $2,4$ Mo. Et il en est de même pour le dictionnaire WikiTitle (Wi): elle commence à 5 Mo pour 200 mille mots jusqu'à 29 Mo pour $1,2$ million mots (la taille initiale de dictionnaire WikiTitle est de $11,5$ Mo).

\subsubsection{Test des trois méthodes qui s'adaptent au comportement de l'utilisateur}

Dans la section \ref{sec:auto:autocompletion_user_typing}, nous avons proposé et expliqué trois méthodes afin de s'adapter à l'utilisateur lorsque il tape sa requête.
Les trois méthodes sont :

\begin{enumerate}
\item \emph{search\_from\_root}, la méthode simple, faire la recherche depuis la racine à chaque fois.

\item \emph{search\_end\_node}, faire la recherche depuis le dernier n\oe{}ud, si cela est possible.

\item \emph{search\_tab\_node}, faire la recherche depuis le n\oe{}ud du plus grand préfixe commun.
\end{enumerate}

\medskip
Nous n'avons fait que les tests pour les deux premières méthodes. La troisième méthode (\emph{search tab node}), a une complexité élevée par rapport à la recherche simple. Pour plus de détails voir la sous-section \ref{sub:methode3_search_tab_nodes}. Pour cette troisième méthode une fois que nous avons implémenté et commencé à la tester, les premiers tests n'ont pas donné de bons résultats et la méthode cause beaucoup de Bugs. C'est la raison pour laquelle  nous avons arrêté les tests.\\

Pour réaliser les tests, nous avons pris un tableau qui contient 100 mots, et nous avons généré un autre tableau qui contient les préfixes du premier dans l'ordre, avec une distance entre le mot et son préfixe de \emph{nb\_char\_def}.
Pour comparer les méthodes \emph{search\_from\_root} et \emph{search\_end\_node} : nous avons fait le test avec le préfixe ensuite le mot, l'un après l'autre, exemple : \emph{abc} ensuite \emph{abcd}, et on mesure le temps.
			
Nous avons 2 tableaux (les mots avec leurs préfixes), donc 200 mots requêtes. Nous avons refait le test 10 fois, et à la fin, nous avons pris la moyenne. Les résultats sont illustrés dans la figure \ref{fig:En_Wi_nb_char_3}.

\begin{figure}[h]
\centering
\includegraphics[width=0.49\linewidth]{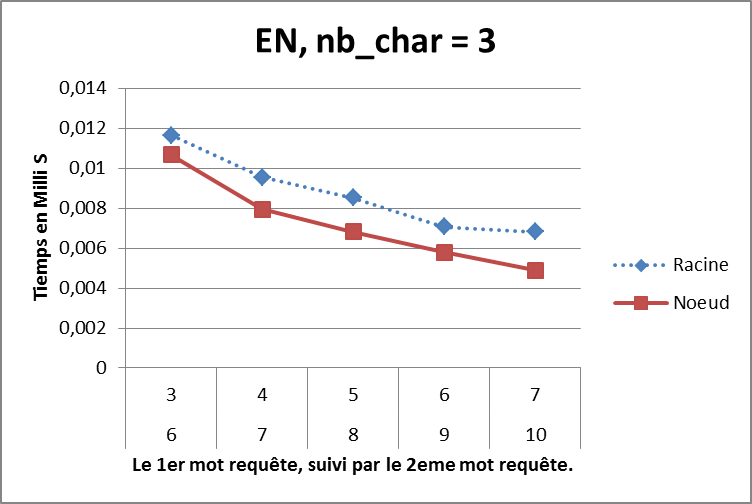}
\includegraphics[width=0.49\linewidth]{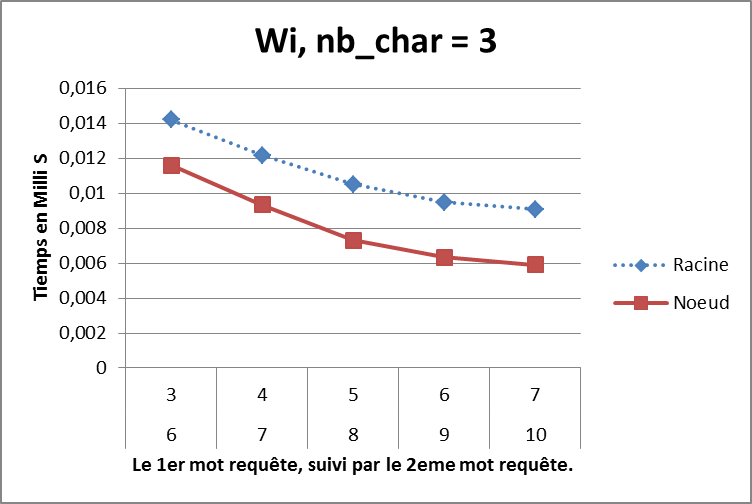}
\caption{Temps de requête en millisecondes pour les 2 méthodes \emph{search\_from\_root} et \emph{search\_end\_node}, pour les dictionnaires Anglais (En) et WikiTitle (Wi), avec $nb\_char\_def=3$.}
\label{fig:En_Wi_nb_char_3}
\end{figure}

Dans la figure \ref{fig:En_Wi_nb_char_3},  l'axe des x, représente la longueur de la première requête suivie de la longueur de la seconde requête en dessous, et l'axe des y représente le temps en millisecondes.
Dans ces deux graphes, on observe clairement que la méthode \emph{search\_end\_node} (continuer depuis le dernier n\oe{}ud ) donne de meilleurs résultats.\\

\noindent
Remarque :\\
Nous observons que la méthode \emph{search\_end\_node} améliore légèrement les résultats, mais en réalité, cela n'a pas une grande influence sur le comportement du système dans l'interaction avec l'utilisateur si tout le traitement s'effectue du coté client, car les temps des deux méthodes sont très petits comparés au temps nécessaire de l'interaction avec l'utilisateur qui est $100ms$.
Mais si le traitement s'effectue du coté serveur, cette amélioration est nécessaire, car le temps de l'interaction avec l'utilisateur est la somme de tous les temps de toutes les étapes : le traitement, le temps du transfert réseau (l'envoi de la requête et la réception des résultats), etc.

\subsubsection{Influence du paramètre k (Top-k) sur le temps de traitement}
\label{subsub:Influence_parametre_k}

Le paramètre $k$ est le nombre de résultats qui doivent être affichés à l'utilisateur selon l'importance des résultats. Dans cette partie, nous changeons à chaque fois ce paramètre $k$, pour observer son influence sur le temps d'exécution.

Dans ce test, on ne présente que le test de dictionnaire WikiTitle, car ses résultats sont plus clairs (pour le dictionnaire Anglais, on obtient le même résultat pour le changement du temps du calcul en fonction du paramètre $k$).
Nous utilisons des mots requêtes d'une longueur de trois caractères, afin d'avoir un grand sous-arbre et donc beaucoup de résultats, et cela permet de choisir le nombre de résultats à afficher $k$ librement. Car dans le cas contraire, un mot requête long va nous donner juste une petite liste de résultats. Le résultat du test est illustré dans la figure \ref{fig:top_k_test}.

\begin{figure}[h]
\centering
\includegraphics[width=0.7\linewidth]{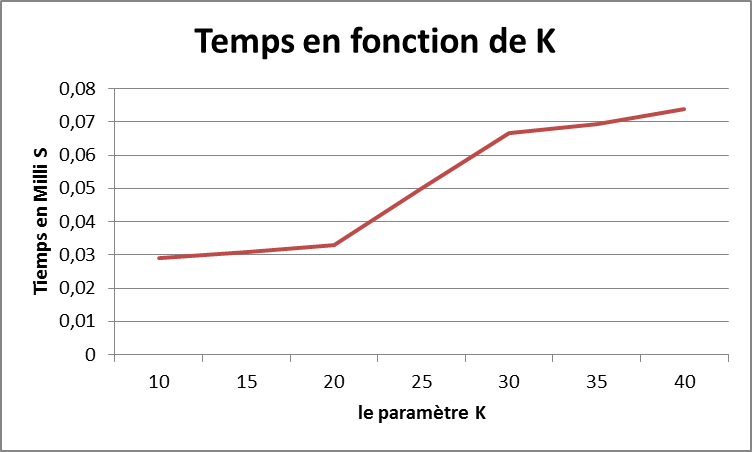}
\caption{Temps de requête en millisecondes en fonction du paramètre $k$, pour le dictionnaire Wikititle (Wi).}
\label{fig:top_k_test}
\end{figure}

Dans la figure \ref{fig:top_k_test}, on remarque que lorsque $k$ varie entre 10 et 20, le temps ne change pas beaucoup, ensuite il augmente de façon linéaire.

Comme nous l'avons expliqué dans la section \ref{sec:auto:methode_recherche}, la méthode de recherche est divisée en deux étapes: la première étape est celle de la recherche dans le Trie pour trouver les n\oe{}uds valides, la seconde étape  insère ces derniers dans une file de priorité pour trouver la liste Top-k.
Dans notre test, la première étape n'influe pas sur le temps d'exécution, car quelque soit $k$, la méthode de recherche restera la même.
La seconde étape consiste à insérer les n\oe{}uds dans la file de priorité, puis itérativement, vérifier si le n\oe{}ud en tête de cette file représente une feuille et d'ajouter le mot correspondant, à la liste des résultats, sinon ajouter ces fils à la file de priorité. On exécute cette opération jusqu'à l'obtention de $k$ éléments dans la liste des résultats.
Cela implique qu'à chaque fois que $k$ augmente le nombre d'opérations augmente aussi, et cela implique un temps de calcul plus élevé.

\subsection{La liste de substitutions, coté serveur}
\label{sub:test_list_substitution}

Nous avons expérimenté avec l'utilisation du dictionnaire des listes de substitutions afin d'accélérer le temps de la recherche approchée.

Le dictionnaire des listes de substitutions augmente la taille de notre index avec 0,7 Mégaoctet pour le dictionnaire Anglais et 5 Mo pour WikiTitle.

Pour comparer l'efficacité du dictionnaire des listes de substitutions, nous l'avons comparé avec la recherche approchée simple dans un Trie.

Le temps des requêtes par les trois méthodes, la recherche simple dans un Trie (\emph{1-err}) et les 2 méthodes avec l'utilisation du dictionnaire des listes de substitutions (\emph{1-err\_SL} et \emph{1-err\_SL\_Node}) pour des préfixes d'une longueur $2$ jusqu'à $6$ sont illustrés dans le tableau \ref{tab:SL_test} et la figure \ref{fig:En_Wi_1_erreur_SL_Methode}.

\begin{table}[h]
\begin{tabular}{|c|l|l|l|l|l|l|l|l|}
\cline{1-1} \cline{3-5} \cline{7-9}
\multicolumn{1}{|l|}{\begin{tabular}[c]{@{}l@{}}Longueur \\ préfixe\end{tabular}} &  & \begin{tabular}[c]{@{}l@{}}1-err \\ (En)\end{tabular} & \begin{tabular}[c]{@{}l@{}}1-err\_SL \\ (En)\end{tabular} & \begin{tabular}[c]{@{}l@{}}1-err\_SL\\ \_Node (En)\end{tabular} &  & \begin{tabular}[c]{@{}l@{}}1-err\\ (Wi)\end{tabular} & \begin{tabular}[c]{@{}l@{}}1-err\_SL\\ (Wi)\end{tabular} & \begin{tabular}[c]{@{}l@{}}1-err\_SL\\ \_Node (Wi)\end{tabular} \\ \cline{1-1} \cline{3-5} \cline{7-9} 
2                                                                                 &  & 0.11                                                   & 0,18                                                    & 0,11                                                          &  & 0,35                                                  & 0,50                                                   & 0,37                                                         \\ \cline{1-1} \cline{3-5} \cline{7-9} 
3                                                                                 &  & 0.14                                                   & 0,17                                                    & 0,13                                                          &  & 0,45                                                  & 0,56                                                   & 0,40                                                         \\ \cline{1-1} \cline{3-5} \cline{7-9} 
4                                                                                 &  & 0.16                                                   & 0,15                                                    & 0,11                                                          &  & 0,51                                                  & 0,50                                                   & 0,32                                                         \\ \cline{1-1} \cline{3-5} \cline{7-9} 
5                                                                                 &  & 0.17                                                   & 0,13                                                    & 0,09                                                          &  & 0,53                                                  & 0,36                                                   & 0,20                                                         \\ \cline{1-1} \cline{3-5} \cline{7-9} 
6                                                                                 &  & 0.17                                                   & 0,10                                                    & 0,06                                                          &  & 0,55                                                  & 0,25                                                   & 0,13                                                         \\ \cline{1-1} \cline{3-5} \cline{7-9} 
\end{tabular}
\centering
\caption{Le temps des requêtes avec les 3 méthodes de la recherche approchée.}
\label{tab:SL_test}
\end{table}

Dans le tableau \ref{tab:SL_test}, on remarque que la méthode \emph{1-err\_SL\_Node} améliore les résultats, elle est plus rapide par rapport à la méthode de recherche simple dans le Trie (\emph{1-err}).
Par contre, la méthode \emph{1-err\_SL} n'est pas aussi bonne que la deuxième méthode. D'après le tableau, elle donne un temps plus long pour les préfixes de petite longueur, et de bons temps avec les préfixes de longueur 5 et 6.

\begin{figure}[h]
\centering
\includegraphics[width=0.49\linewidth]{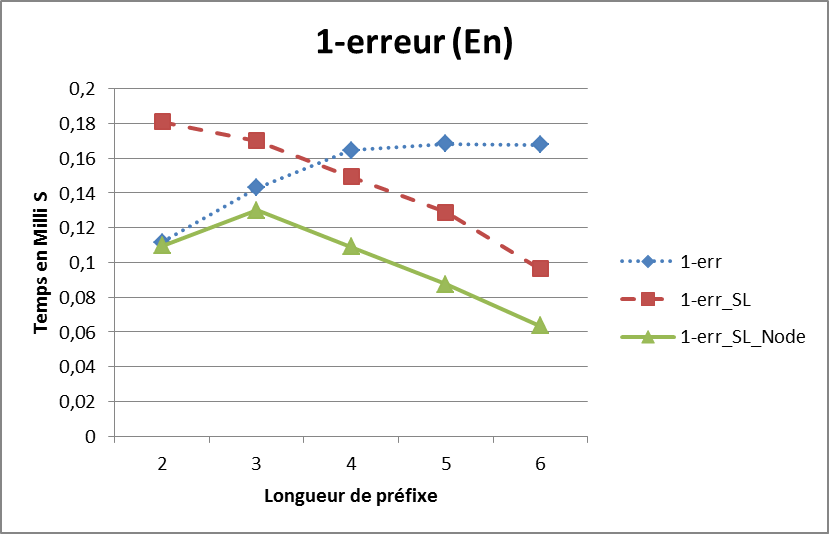}\hfill
\includegraphics[width=0.49\linewidth]{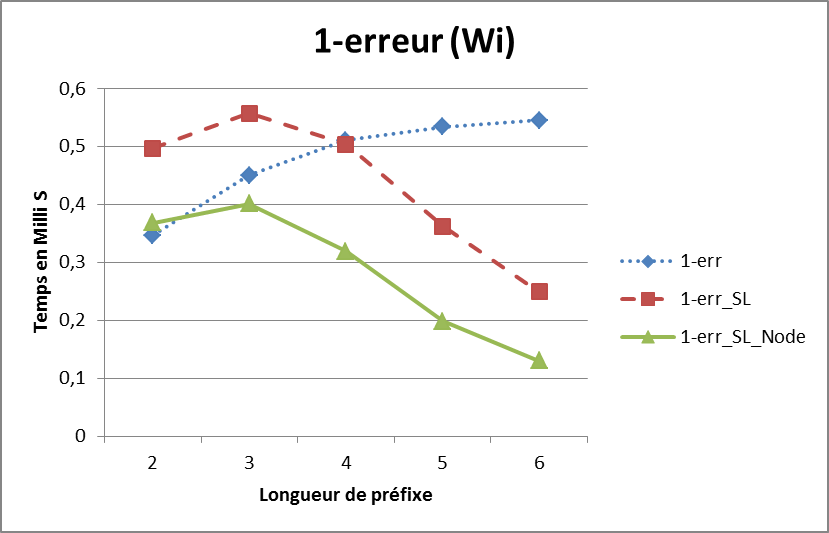}
\caption{Temps de l'auto-complétion approchée avec les 3 méthodes 1-err (naïve), 1-err\_SL et 1-err\_SL\_Node, pour les dictionnaires Anglais (En) et WikiTitle (Wi).}
\label{fig:En_Wi_1_erreur_SL_Methode}
\end{figure}

La figure \ref{fig:En_Wi_1_erreur_SL_Methode}, montre que la méthode \emph{1-err\_SL\_Node} donne de meilleurs résultats.
La première méthode \emph{1-err\_SL} donne des résultats meilleurs que la méthode naïve \emph{1-err} (la recherche simple dans le Trie) avec des préfixes de longueur $> 5$, mais cela n'est pas intéressant dans l'auto-complétion, car ce qui nous intéresse, ce sont les premiers caractères tapés par l'utilisateur (généralement 2, 3 ou 4).

Un fait intéressant: on peut observer que, lorsqu'on utilise le dictionnaire de substitutions, le temps de recherche diminue quand on augmente la longueur de la requête.
Cela est dû au fait que la taille des listes de substitutions est plus courte pour les préfixes longs que les préfixes courts.
Intuitivement, un motif de substitution plus long $p \phi q$, va correspondre à un nombre de préfixe moindre, par rapport au motif de substitution court.

\subsubsection{Utiliser les caractères de substitutions juste dans les premiers niveaux du Trie}

Dans cette partie, nous testons la méthode qui divise le Trie en deux parties, dans les premiers niveaux du Trie, nous utilisons le dictionnaire de substitutions, et dans le reste de Trie, nous utilisons la méthode de recherche naïve. Pour plus de détails, voir la sous-section \ref{sub:methode_1err_SL_3_level}.

Pour les tests, nous procédons de la même manière:  l'erreur est aléatoire dans le préfixe, elle peut être dans tous les niveaux.

Nous faisons les tests pour des préfixes d'une longueur 4, 5 et 6 seulement, car avec les préfixes de longueur 2 et 3, nous n'utilisons que la méthode du dictionnaire de substitutions, qui donne les mêmes résultats que la méthode \emph{1-err\_SL\_node}.

Nous avons vu dans les tests précédents que la méthode \emph{1-err\_SL\_node} a donné les meilleurs résultats comparée à la méthode naïve, et à la méthode \emph{1-err\_SL}. Pour cela, dans ce test, nous faisons la comparaison juste avec la méthode \emph{1-err\_SL\_node}. Les résultats sont illustrés dans le tableau \ref{tab:test_1-err_SL_3_level}.

\begin{table}[h]
\begin{tabular}{|c|l|l|l|l|l|l|}
\cline{1-1} \cline{3-4} \cline{6-7}
\multicolumn{1}{|l|}{\begin{tabular}[c]{@{}l@{}}Longueur \\ préfixe\end{tabular}} &  & \begin{tabular}[c]{@{}l@{}}1-err\_SL\_Node \\ En\end{tabular} & \begin{tabular}[c]{@{}l@{}}1-err\_SL\_3L\\ En\end{tabular} &  & \begin{tabular}[c]{@{}l@{}}1-err\_SL\_Node\\ Wi\end{tabular} & \begin{tabular}[c]{@{}l@{}}1-err\_SL\_3L\\ Wi\end{tabular} \\ \cline{1-1} \cline{3-4} \cline{6-7} 
4                                                                                 &  & 0,11                                                          & 0,10                                                       &  & 0,32                                                         & 0,26                                                       \\ \cline{1-1} \cline{3-4} \cline{6-7} 
5                                                                                 &  & 0,09                                                          & 0,07                                                       &  & 0,20                                                         & 0,17                                                       \\ \cline{1-1} \cline{3-4} \cline{6-7} 
6                                                                                 &  & 0,06                                                          & 0,06                                                       &  & 0,13                                                         & 0,10                                                       \\ \cline{1-1} \cline{3-4} \cline{6-7} 
\end{tabular}
\centering
\caption{Temps de l'auto-complétion en utilisant les 2 méthodes, \emph{1-err\_SL\_node} et \emph{1-err\_SL\_3\_level} pour les 2 dictionnaires Anglais (En) et WikiTitle (Wi).}
\label{tab:test_1-err_SL_3_level}
\end{table}

Dans le tableau \ref{tab:test_1-err_SL_3_level}, on observe que la méthode hybride \emph{1-err\_SL\_3\_level}, améliore les résultats (les résultats   avec le dictionnaire WikiTitle le montrent plus).

L'utilisation de cette méthode permet de combiner les caractéristiques de la méthode simple qui a un temps de construction rapide et un espace mémoire petit par rapport à \emph{1-err\_SL\_Node} avec la rapidité de cette dernière du temps de recherche.

La construction du dictionnaire de substitutions juste pour les trois premiers caractères (préfixe jusqu'à 3) au lieu de la totalité de la longueur du mot (ou comme nous avons fait dans les tests jusqu'à la taille 6) divise le temps de construction et l'espace du dictionnaire de substitutions par deux.

\subsection{Le coté Client, JavaScript}

Les tests ont été effectués sur un ordinateur Netbook qui a les caractéristiques suivantes :  un processeur Intel Atom CPU N455 1,67 Ghz, 2Go RAM et window 7 starter 32 bits, (Aujourd'hui, cet ordinateur Netbook est moins puissant qu'un Smartphone).

Les différents tests pour l'auto-complétion exacte et approchée sur les dictionnaires Anglais et WikiTitle sont illustrés dans les tableaux (\ref{tab:test_nav_Firefoxe}, \ref{tab:test_nav_chrome}, \ref{tab:test_nav_inter_explorer}, \ref{tab:test_nav_opera}).

% % ------------------------- Le navigateur Google Chrome.

\begin{table}[H]

\begin{tabular}{|l|ll|ll|}
\hline
Longueur & Exacte & 1-erreur  & Exacte  & 1-erreur \\
requête   & (En) & (En) & (Wi) & (Wi) \\
\hline
2                         & 0,006       & 0,47        & 0,0123      & 1,88  \\
3                         & 0,009       & 0,59        & 0,0242      & 2,41  \\
4                         & 0,0114      & 0,65        & 0,0186      & 2,42  \\
5                         & 0,0125      & 0,62        & 0,0206      & 2,28  \\
6                         & 0,0143      & 0,58        & 0,0233      & 2,08  \\
\hline
\end{tabular}
\centering
\caption{ Le navigateur Google Chrome 39. Le temps des requêtes $+$ Top-$k$ (en millisecondes) sur les dictionnaires Anglais (En) et WikiTitle (Wi). La construction du Trie pour le dictionnaire Anglais est $2,7$ secondes et 13 secondes pour le dictionnaire Wikititle.}
\label{tab:test_nav_chrome}
\end{table}

% % ------------------------- Le navigateur Fire fox 29.

\begin{table}[H]

\begin{tabular}{|l|ll|ll|}
\hline
Longueur & Exacte & 1-erreur  & Exacte  & 1-erreur \\
requête   & (En) & (En) & (Wi) & (Wi) \\
\hline
2	& 0,0021	& 0,4541	& 0,0026	& 1,1547\\
3	& 0,0035	& 0,5117	& 0,0059	& 1,453\\
4	& 0,0042	& 0,4474	& 0,0061	& 1,2086\\
5	& 0,0061	& 0,3554	& 0,0076	& 0,9461\\
6	& 0,005	    & 0,2808	& 0,0081	& 0,7408\\
\hline

\end{tabular}
\centering
\caption{ Le navigateur Firefox 29. Le temps des requêtes $+$ Top-$k$ (en millisecondes) sur les dictionnaires Anglais (En) et les Titres de Wikipédia (Wi). La construction du Trie pour le dictionnaire Anglais est $7,4$ secondes et 47 secondes pour le dictionnaire Wikititle.}
\label{tab:test_nav_Firefoxe}
\end{table}

% % ------------------------- Le navigateur Opera 26.

\begin{table}[H]
\begin{tabular}{|l|ll|ll|}
\hline
Longueur & Exacte & 1-erreur  & Exacte  & 1-erreur \\
requête   & (En) & (En) & (Wi) & (Wi) \\
\hline
2	& 0,003	    & 0,5975	& 0,005	    & 0,9819\\
3	& 0,0054	& 0,3052	& 0,0087	& 0,9652\\
4	& 0,005	    & 0,355	    & 0,0072	& 1,0334\\
5	& 0,0064	& 0,2933	& 0,007	    & 0,6881\\
6	& 0,0059	& 0,2412	& 0,0078	& 0,6455\\
\hline
\end{tabular}
\centering
\caption{ Le navigateur Opera 26. Le temps des requêtes $+$ Top-$k$ (en millisecondes) sur les dictionnaires Anglais (En) et les Titres de Wikipédia (Wi). La construction du Trie pour le dictionnaire Anglais est $1,5$ secondes et $6,2$ secondes pour le dictionnaire Wikititle.}
\label{tab:test_nav_opera}
\end{table}

% % ------------------------- Le navigateur Internet Explorer.

\begin{table}[H]

\begin{tabular}{|l|ll|ll|}
\hline
Longueur & Exacte & 1-erreur  & Exacte  & 1-erreur \\
requête   & (En) & (En) & (Wi) & (Wi) \\
\hline
2                             & 0,0471      & 1,64    & 0,0868      & 7,07   \\
3                             & 0,0492      & 1,93    & 0,0883      & 8,72    \\
4                             & 0,0608      & 2,37    & 0,1063      & 10,56   \\
5                             & 0,0896      & 3,29    & 0,1172      & 11,07 \\
6                             & 0,1013      & 3,59    & 0,1308      & 10,59   \\
\hline
\end{tabular}
\centering
\caption{ Le navigateur Internet Explorer 11. Le temps des requêtes $+$ Top-$k$ (en millisecondes) sur les dictionnaires Anglais (En) et les Titres de Wikipédia (Wi). La construction du Trie pour le dictionnaire Anglais est 30 secondes et 150 secondes pour le dictionnaire Wikititle.}
\label{tab:test_nav_inter_explorer}
\end{table}

Les résultats des quatre navigateurs sont beaucoup plus petits que $100ms$, cela veut dire que nos méthodes d'auto-complétion fonctionnent très bien dans tous les navigateurs.\\

Afin de mieux visualiser les résultats et afin d'avoir une idée de la performance de nos méthodes de l'auto-complétion exacte et approchée sur les différents navigateurs, nous avons fait une comparaison entre eux illustrée par les deux figures \ref{fig:client_exacte_En_Wi}, \ref{fig:client_1_erreur_En_Wi} suivantes :

\begin{figure}[H]
\centering
\includegraphics[width=0.49\linewidth]{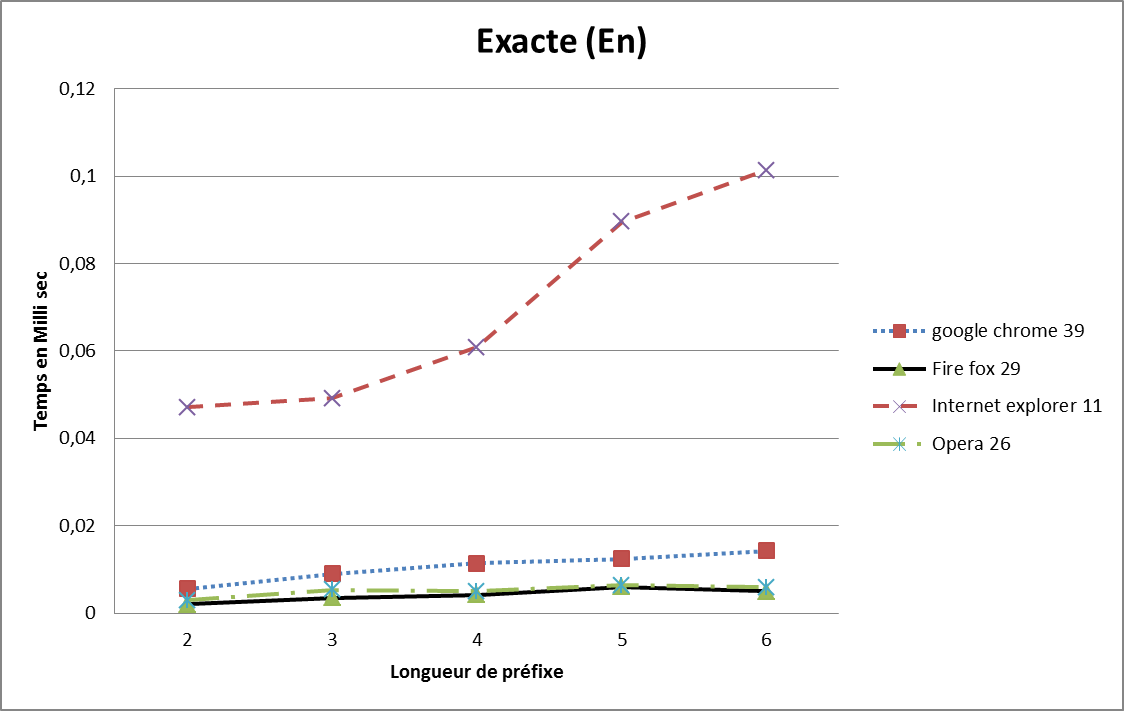}\hfill
\includegraphics[width=0.49\linewidth]{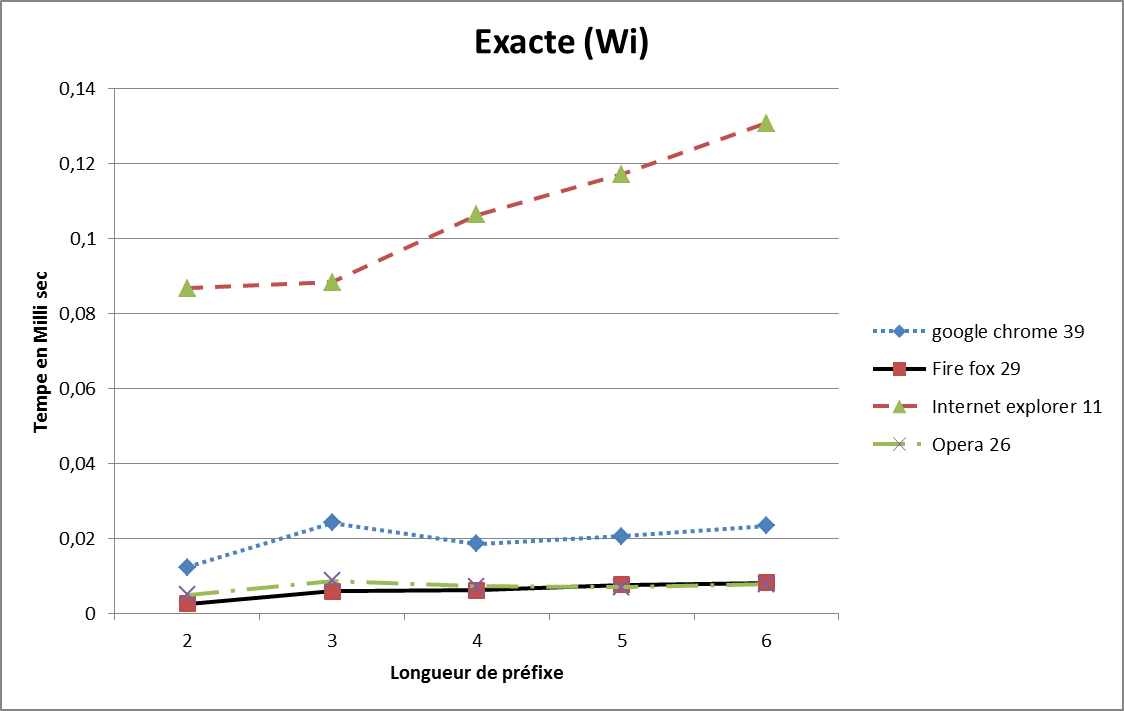}
\caption{Temps des 4 navigateurs pour une auto-complétion exacte avec les dictionnaires Anglais (En) et WikiTitle (Wi).}
\label{fig:client_exacte_En_Wi}
\end{figure}

\begin{figure}[H]
\centering
\includegraphics[width=0.49\linewidth]{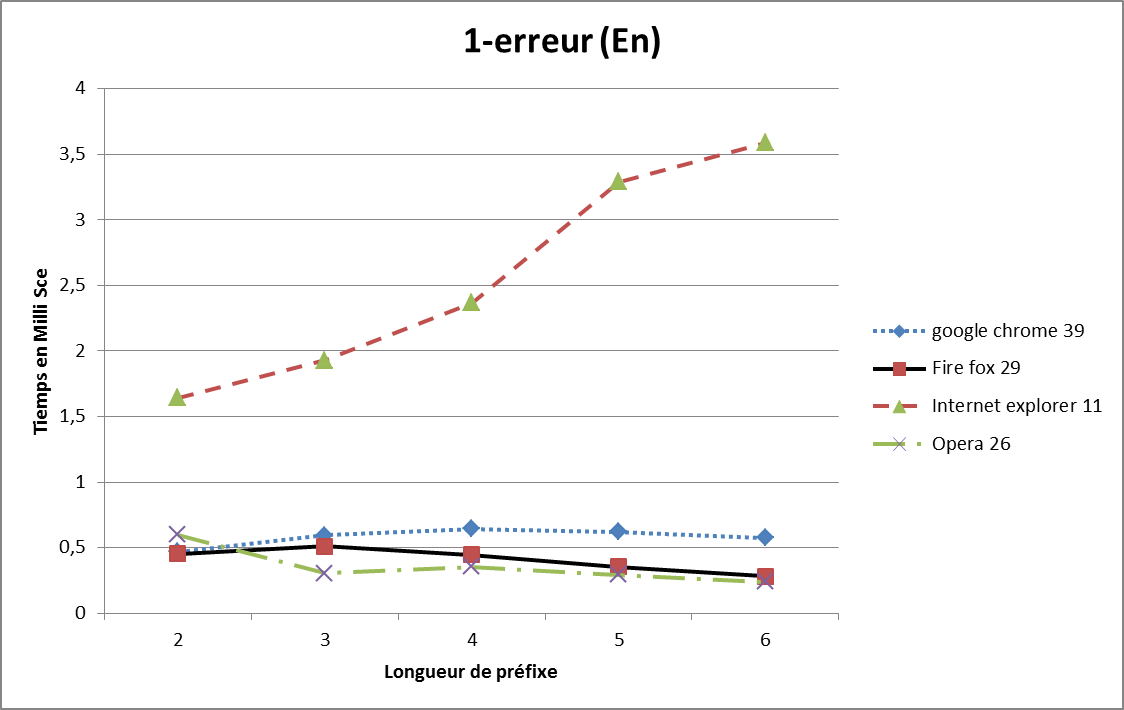}\hfill
\includegraphics[width=0.49\linewidth]{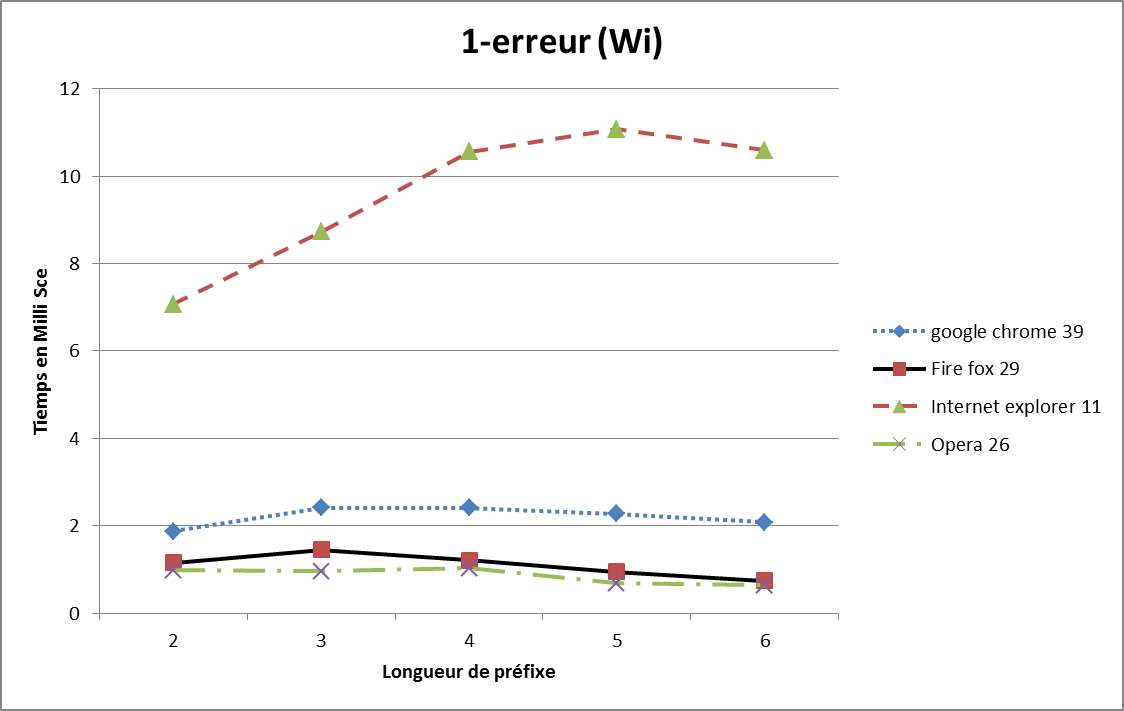}
\caption{Temps des 4 navigateurs pour une auto-complétion approchée (1-erreur) pour les dictionnaires Anglais (En) et WikiTitle (Wi).}
\label{fig:client_1_erreur_En_Wi}
\end{figure}

Les quatre graphes présentés dans les deux figures (\ref{fig:client_exacte_En_Wi}, \ref{fig:client_1_erreur_En_Wi}) illustrent bien la différence entre les quatre navigateurs et montrent que la meilleure performance a été obtenue avec \emph{Opera 26} et \emph{Firefox 29} qui sont presque au même niveau dans l'auto-complétion exacte, \emph{Firefox} est légèrement plus lent dans l'auto-complétion approchée. Ces deux navigateurs sont suivis par \emph{Google Chrome 39} (légèrement derrière), et enfin \emph{Microsoft Internet Explorer 11}.

On remarque qu'avec les trois navigateurs (Chrome, Firefox et Opera) le temps ne change pas beaucoup avec le changement de la taille des requêtes.
Par contre avec \emph{Internet explorer 11} le temps augmente considérablement. Chrome, qui est un navigateur populaire, est seulement légèrement plus lent que les 2 plus rapides. Internet Explorer est le plus lent.

\section{Conclusion}
\label{sec:auto:conclusion}

Dans ce chapitre, nous avons présenté une approche efficace pour résoudre le problème de l'auto-complétion approchée sous la contrainte de la distance d'édition dans une architecture client-serveur, et nous avons présenté la méthode Top-k pour améliorer la qualité des résultats en utilisant un système de classement dynamique et statique, selon l'importance (score) des mots.

Nous avons également présenté une méthode efficace d'élimination des résultats en double, ce qui est un réel problème dans les résultats de l'auto-complétion approchée.

Nous avons présenté une bibliothèque (nommée Appaco\_lib), pour être utilisée soit du coté serveur ou du coté client ou les deux en même temps pour répondre rapidement à des requêtes d'auto-complétion approchée pour des dictionnaires $\mathtt{UTF}\mbox{-}8$ (donc toutes les langues).

Pour les détails de l'utilisation de notre bibliothèque voir le document qui explique toutes les fonctionnalités et les options proposées dans : \url{https://github.com/AppacoLib/api.appacoLib/tree/master/doc_appaco_lib}. \\

L'intérêt de cette étude, c'est d'avoir démontré que ce problème, bien qu'il soit un type particulier de la recherche approchée, ne pouvait pas être résolu de manière optimale par les solutions présentées dans les deux chapitres précédents.
Nous avons ainsi étendu l'étude de la recherche approchée à un contexte différent qui nécessite une solution différente.
Sachant que $|prefixe(m)| \leq |m|$, et que l'auto-complétion nécessite une deuxième étape pour trouver tous les suffixes, les solutions présentées dans les chapitres précédents ne pouvaient pas rester optimales.\\

Remarquons enfin que la recherche approchée pour laquelle nous avons présenté plusieurs solutions, se base sur une comparaison entre mots. Il s'agit fondamentalement, quelle que soit la solution, de comparer la proximité entre mots. Il existe plusieurs fonctions mathématiques qui mesurent cette proximité, elles sont appelées fonction de distance.

Une fonction de distance $dist(x,y)$ détermine l'écart entre les mots $x$ et $y$, c'est-à-dire le nombre d'opérations\footnote{Chaque fonction à ses propres opérations. Voir la définition dans le chapitre \ref{chap:Preliminaires_definitions}.} minimum nécessaires pour transformer le mot $x$ en $y$.

Si $x$ et $y$ sont des séquences biologiques (ex: ADN, ARN), le problème de l'alignement revient à la recherche de cette proximité entre $x$ et $y$, et c'est donc un problème de recherche approchée dans lequel le nombre d'erreurs $k$ est inconnu.

Trouver le meilleur alignement entre $x$ et $y$ revient à trouver le plus petit $k$, autrement dit, la meilleure façon de mettre $x$ face à $y$ avec un nombre minimum de différences.\\

Nous nous intéressons à ce nouveau problème de recherche approchée pour les séquences génomiques (des longues chaînes de texte) \emph{l'alignement}, dans le chapitre qui suit.

\chapter{Alignement multiple de séquences d'ADN avec un nouvel algorithme \emph{DiaWay}}

\ifpdf
    \graphicspath{{Alignement/}}
\else
    \graphicspath{{Alignement/}}
\fi

% % \section*{Les mots clés}

% % Bio-informatique, alignement des séquences biologiques, alignement multiple, DiAlign, diagonales, matrice (dot plot), consistance.

\section{Introduction}
\label{sec:ali:introduction}

%\subsection{Motivation}
L'alignement des séquences biologiques est parmi les domaines d'application les plus importants de la recherche approchée de motifs. 
Les séquences biologiques (ADN, ARN, protéines) sont considérées comme de longues chaînes de caractères sur un alphabet spécifique, ( $\Sigma_{ADN}=\{A,C,T,G\}$ pour l'ADN, $\Sigma_{ARN}=\{A,C,U,G\}$ pour l'ARN, et $ \Sigma_{prot\acute{e}ine}=\{20 \; lettres  \; d'acides \;  amin\acute{e}s\}$ pour les protéines). L'alignement est un problème appartenant au domaine de la bio-informatique permettant de visualiser les ressemblances entre deux séquences ou plus et de déterminer leurs éventuelles homologies. 
Pour plus de détails sur les méthodes et les techniques d'alignement des séquences biologiques, référer à \cite{bioinformatics_waterman1995,Gusfield1997, bioinformatics2001,bioinformatics_lesk2013}.

%\subsection{Définition}
L'alignement des séquences biologiques est la représentation de deux ou plusieurs séquences les unes sous les autres, afin de faire ressortir les régions similaires, plus précisément déterminer les différences et les similitudes.

Généralement pour mesurer le taux de similarité entre les séquences on utilise la distance d'édition où on a les trois opérations suivantes: la substitution (le $A$ est substitué à $T$ et le $G$ est substitué à $C$), l'insertion et la suppression qui nécessitent en général l'introduction de \emph{trous} (appelés \emph{gaps}) à certaines positions dans les séquences, de manière à aligner les caractères communs sur des colonnes successives, voir la figure \ref{fig:alignement_par_paire}.

\begin{figure}[h]
\centering
\includegraphics[width=0.7\linewidth]{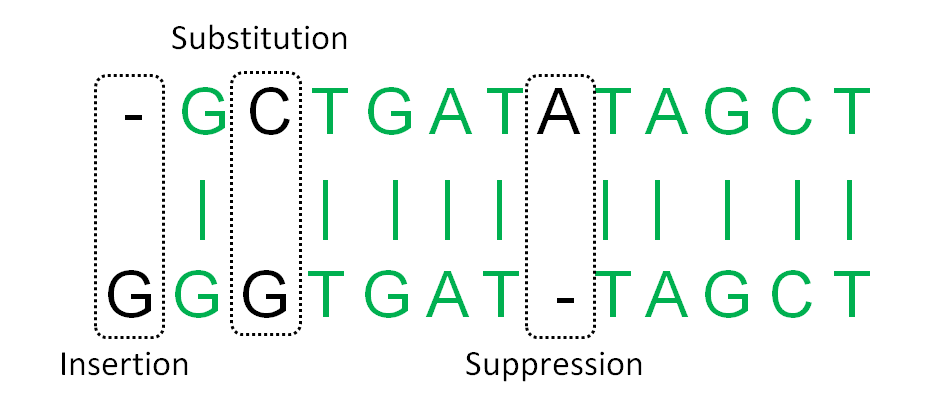}
\caption{Alignement de deux séquences d'ADN.}
\label{fig:alignement_par_paire}
\end{figure}

Le défi important pour la bio-informatique est de concevoir des algorithmes capables de trouver automatiquement des alignements biologiquement corrects.\\

Le problème d'alignement par paire a été largement considéré comme étant résolu \cite{Needleman_Wunsch,Smith_waterman}.
La plupart des efforts sont concentrés sur l'amélioration des algorithmes d'alignements multiple, pour trouver des solutions optimales ou raisonnablement bonnes. L'alignement multiple correspond à une généralisation de l'alignement par paire, la recherche de similitude ne se fait pas qu'au niveau de deux séquences, mais au niveau de toutes les séquences.
Il existe de nombreuses méthodes qui fournissent des solutions à ce problème \cite{Thompson1994,Morgenstren1996,Notredame2000,Edgar2004,plasma,FASTA_1988,BLAST_1990basic}.\\

% % Notre solution 
Dans ce chapitre, nous apportons une solution au problème de l'alignement multiple.
Notre solution est basée sur une méthode d'alignement connue \emph{DIALIGN} \cite{Morgenstren1996,Morgenstern1998b}, qui est un algorithme conçu pour l'alignement par paire, ainsi que l'alignement multiple des séquences d'acides nucléiques et de protéines.

Notre nouvelle approche (appelée \emph{DiaWay}) s'exécute selon les étapes suivantes: après extraction des diagonales de la matrice dot plot, au lieu de supprimer automatiquement les diagonales inconsistantes, on essaye d'aligner les sous-diagonales (fragments) consistantes de celles-ci afin d'améliorer la qualité de l'alignement.
L'étape d'indexation des diagonales, où nous créons une structure d'index pour accélérer l'accès aux différents fragments des diagonales sera suivie par la vérification de l'inconsistance.
L'étape de tri des diagonales consiste à trier toutes les diagonales consistantes pour l'alignement final, en fonction de leurs positions dans les séquences biologiques.
La dernière étape consiste simplement à rassembler les fragments des diagonales et ensuite insérer des gaps (le symbole $-$ ) dans les vides.\\

Les expérimentations, comparées à \emph{DIALIGN}, montrent que notre approche \emph{DiaWay} donne de meilleurs résultats dans les deux cas d'alignement d'ADN par paire et multiple.

Le travail présenté dans ce chapitre a été publié en 2015 dans deux articles \cite{DIAWAY,DIAWAY_b}, où nous avons présenté deux implémentations de notre approche \emph{DiaWay 1.0}, \emph{DiaWay 2.0}.\\

%\subsection{Suite d'un travail de master}
Ce travail de doctorat est l'aboutissement d'un premier travail réalisé en master \cite{masters_ibra}, qui a été poursuivi pendant la durée du doctorat en cherchant à améliorer notre  approche et faire les différents tests nécessaires afin d'arriver à une version performante qui fait un alignement efficace et rapide. Ce travail a abouti à la publication de deux articles scientifiques, ainsi qu'à l'implémentation de deux versions logiciels.\\

% % Organisation du chapitre :
Le chapitre est organisé comme suit : La section \ref{sec:DIALIGN} couvre le concept de l'approche originale de \emph{DIALIGN}, et une brève explication de ses étapes. Les étapes de notre nouvelle approche \emph{DiaWay} de l'extraction des diagonales au processus d'alignement final sont présentées dans les sections \ref{sec:extraction_des_diagonals}, \ref{sec:Diagonal_indexation}, \ref{sec:diagonals_consistency}, \ref{sec:sorting_diagonals} et \ref{sec:ali:Mettre_les_fragments_ensemble}. La section \ref{sec:ali:analyse_comparaison} portes sur les tests, les comparaisons, et les analyses. Et enfin, une conclusion est donnée dans la section \ref{sec:ali:Conclusion}.

\section{DIALIGN}
\label{sec:DIALIGN}

\emph{DIALIGN} \cite{Morgenstren1996,Morgenstern1998b}, est un algorithme connu qui fait un alignement par paire et multiple des séquences biologiques. Il combine les caractéristiques d'alignement locales et globales. La méthode consiste à aligner les similitudes de paires locales. L'alignement final est composé d'une collection de diagonales qui répondent à un certain critère de consistance puis sélectionne un ensemble consistant de diagonales avec une somme maximale de poids.

Les grandes lignes de l'algorithme d'alignement multiple \emph{DIALIGN} sont :

\begin{enumerate}

\item Construire une matrice de similarité pour chaque paire des séquences afin de récupérer toutes les diagonales possibles. Une diagonale (également appelé fragment ou bloc) est définie par des régions similaires appartenant aux deux séquences biologiques.

\item Donner pour chaque diagonale un score puis les trier selon leurs scores et selon leurs degré de chevauchement avec les autres diagonales. La liste de diagonales est ensuite utilisée pour assembler un alignement multiple d'une manière gloutonne (gourmande).

\item  Les diagonales sont intégrées une par une dans l'alignement multiple en commençant par la diagonale de poids (score) maximum. Avant d'ajouter une diagonale à l'alignement, on doit vérifier si elle est consistante avec les diagonales déjà intégrées.

\item Dans une étape finale, le programme introduit des gaps (-) dans les séquences jusqu'à ce que tous les résidus liés par les diagonales sélectionnées soient correctement disposées.

\end{enumerate}

\paragraph{Le poids d'une diagonale, et le score d'alignement :}

\emph{DIALIGN} essaye de sélectionner un ensemble consistant de diagonales avec une somme maximale de poids.
Les diagonales d'un poids élevé sont plus susceptibles d'être biologiquement pertinentes, et, par conséquent, sont ajoutées d'abord à l'alignement.
La qualité des alignements produits par la méthode \emph{DIALIGN} dépend de la manière dont le score ou le poids des diagonales est définis. \emph{DIALIGN} a défini un nouveau système de score qui est différent des autres approches traditionnelles.

Le premier système de score est défini dans \cite{Morgenstren1996,Morgenstern1998a}, dans cette version \emph{DIALIGN 1}, la dépendance sur le seuil $T$ ($T$ représente une valeur de seuil positive ou nulle qui rentre dans le processus du calcul des poids des diagonales), était un sérieux inconvénient, car il n'y avait pas de règle générale pour choisir la valeur $T$ qui permet de produire des alignements de qualité. En plus, la longueur minimum requise d'une diagonale était arbitraire.
Dans la 2\ieme{} version de Dialign (\emph{DIALIGN 2}) Morgenstern et al. \cite{Morgenstern1998b} introduisent une nouvelle fonction de poids pour les diagonales afin de remédier au problème de seuil $T$ et la taille minimale de diagonale.

\paragraph{La consistance des diagonales avec la méthode DIALIGN :}

Un concept crucial et délicat décrit dans \emph{DIALIGN}. Comment décider si une diagonale est consistante avec les diagonales déjà intégrées dans l'alignement ou non.
Une collection de diagonales est dite consistante  s'il n'y a pas un double conflit ou croisement des résidus des fragments des diagonales \cite{Morgenstren1996}.
De même, pour les diagonales impliquant la même paire de séquences ne sont pas autorisées à avoir un chevauchement. Voir la figure \ref{fig:D1_croise_et_chovoche_D0}.

\begin{figure}[h]
\centering
\includegraphics[width=0.8\linewidth]{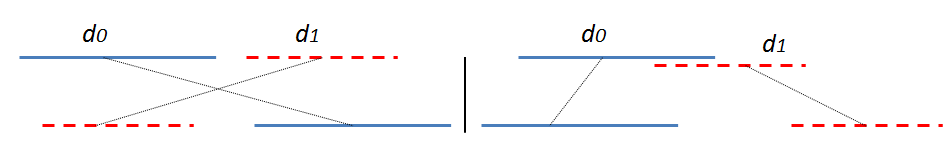}
\caption{L'inconstance des diagonales, $d_1$ se croise/chevauche avec $d_0$}
\label{fig:D1_croise_et_chovoche_D0}
\end{figure}

Une diagonale qui traverse (se croise) ou se chevauche avec une autre diagonale qui a un poids plus élevé, est une diagonale inconsistante, parce qu'elle ne permet pas un alignement correct avec cette diagonale qui a un poids plus élevé.
À chaque fois qu'on ajoute une diagonale à l'alignement, on vérifie si elle est consistante avec les diagonales déjà intégrées. Les diagonales qui ne sont pas consistantes avec l'ensemble croissant des diagonales consistantes seront rejetées et donc supprimées. 
Trouver un ensemble de diagonales consistantes est un problème NP-complet \cite{Subramanian2008}.

\section{L'extraction des diagonales}
\label{sec:extraction_des_diagonals}

Soit $S$ l'ensemble de toutes les séquences biologiques, et $S=\{s_1,s_2,...,s_l\}$ où $l$ est le nombre total de toutes les séquences, $s_k \in S$ où $k \in [1..l]$.

Soit $D$ l'ensemble de toutes les diagonales, $ D=\{d_1,d_2,...,d_n\} $ où $n$ est le nombre total de toutes les diagonales initiales. Soit $D^{'}$ l'ensemble des diagonales consistantes seulement, $ d_j \in D^{'} $ où $j=[1..m]$ et $m$ est le nombre total de toutes les diagonales consistantes seulement. Chaque diagonale comporte une paire de fragment (sous-séquence), le premier segment est désigné par $x$, et l'autre par $y$. Pour désigner une diagonale et ses deux fragments, on écrit $d_i =(x_i,y_i)$ où $i \in [1..n]$ .\\

Les diagonales sont extraites à partir d'une matrice dot plot construite pour chaque paire de séquences, tels que les éléments $[i, j]$ de cette matrice sont remplis par $1$ dans le cas de correspondance entre les caractères de la position $i$ et $j$ dans les deux séquences respectivement, et $0$ dans le cas contraire.\\

Une diagonale est définie par des régions similaires appartenant à deux séquences. Elle est représentée par une paire de fragments (sous-séquences) d'une taille égale et qui contiennent des similitudes tout en acceptant la mutation c'est-à-dire : concordance entre $A$ et $T$ et entre $C$ et $G$.

Soient $S_1$ et $S_2$ deux séquences d'ADN, et soit $M$ leur matrice dot plot. La diagonale $d$ est donc comme suit $d = \{M[i_1][j_1], M[i_2][j_2],\dots,M[i_m][j_m]\}$, où $m$ est la longueur de cette diagonale, et $i \in [1..|S_1|]$ et $j \in [1..|S_2|]$.
Comme nous l'avons dit précédemment, une diagonale est désignée  par ses deux fragments et on écrit $d_i =(x_i,y_i)$ où $i \in [1..n]$.\\

Exemple : Soient deux séquences d'ADN (de petites longueurs) $S=\{s_1,s_2\}$, avec $s_1={ATTCCGACT}$ et $s_2={AATTCGCGT}$. Le résultat de la construction de la matrice dot plot est illustré dans la figure \ref{fig:Matric_dot_plot}.

\begin{figure}[h]
\centering
\includegraphics[width=0.7\linewidth]{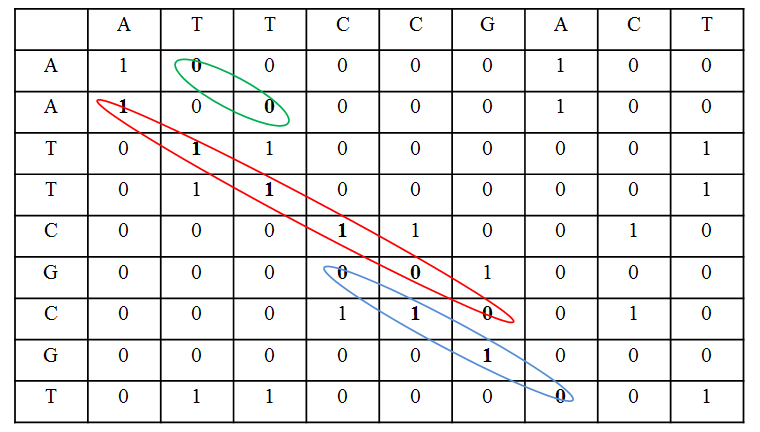}
\caption{La matrice de similarité dot plot de deux séquences d'ADN.}
\label{fig:Matric_dot_plot}
\end{figure}

Dans la figure \ref{fig:Matric_dot_plot}, nous donnons comme exemple les trois diagonales qui sont entourées par des ellipses :

\begin{itemize}
\item[$d_1 =$]
\begin{tabular}{cc}
  		  & TT \\
          & AA 
\end{tabular}

\smallskip

\item[$d_2 =$]
\begin{tabular}{cc}
          & ATTCCG \\
          & ATTCGC 
\end{tabular}

\smallskip

\item[$d_3 =$]
\begin{tabular}{cc}
          & CCGA \\
          & GCGT 
\end{tabular}
			 
\end{itemize}

% % % % % % % % % % % % % %
% % notre travail commence ici :)

\bigskip

Nous avons remarqué que lorsqu'on supprime les diagonales inconsistantes, même si c'est seulement une partie de leurs résidus qui cause le problème de l'inconsistance, on perd dans la qualité de l'alignement. Et nous avons remarqué qu'on peut aligner les parties des diagonales (les sous-diagonales) qui ne posent pas de problème d'inconsistance pour améliorer le résultat de l'alignement final. Pour cela, nous procédons comme suit :

\subsection{Méthode 1}

Lorsqu'on extrait les diagonales, on extrait aussi leurs sous-diagonales, qui sont elles-mêmes considérées comme des diagonales. Car une diagonale peut être inconsistante, tandis que l'une de ses sous-diagonales est consistante, et donc ses sous-diagonales vont avoir une chance d'être alignées.

Pour extraire les sous-diagonales, à chaque fois, on enlève les premières paires de caractères des deux fragments qui constituent la diagonale, donc on enlève le premier caractère gauche du premier fragment et on fait la même chose pour le deuxième fragment, et dans les deux fragments, il reste leurs suffixes. Les deux nouveaux sous-fragments vont constituer une nouvelle sous-diagonale.
Tous les suffixes d'une diagonale $d_i$ sont considérés comme des sous-diagonales. On prend seulement les suffixes qui ont une longueur supérieure ou égale à la longueur minimale de la diagonale permise.\\

Exemple:
Soit la longueur minimale des diagonales égales à $2$, et soit 
la diagonale $d_2$ (de l'exemple précédent) suivante :\\

\begin{itemize}

\item[$d_2 =$]
\begin{tabular}{cc}
          & ATTCCG \\
          & ATTCGC 
\end{tabular}
			 
\end{itemize}

Donc ses sous-diagonales sont extraites comme suit :\\

\begin{itemize}

\item[$\{d_{2.1}, \, d_{2.2}, \, d_{2.3}, \, d_{2.4} \} =$]
\begin{tabular}{cccc}
          TTCCG , & TCCG , & CCG , & CG \\
          TTCGC , & TCGC , & CGC , & GC 
\end{tabular}
			 
\end{itemize}

\bigskip

Dans le processus de vérification de l'inconsistance d'une diagonale $d_i$, si on trouve qu'elle est consistante, on supprime directement toutes ses sous-diagonales de l'ensemble $D$.
Car la diagonale principale est consistante, donc on n'a pas besoin de ses sous-diagonales. En plus, toutes ses diagonales sont consistantes à leur tour, et leur fragments sont inclus dans la diagonale principale.

Les diagonales sont indexées en fonction de leurs numéros de séquences afin de permettre un accès rapide pour les trouver. Voir la section \ref{sec:Diagonal_indexation}.
Cette méthode est utilisée dans notre première implémentation appelée \emph{DiaWay $1.0$}.

\subsection{Méthode 2}

Il s'est avéré qu'avec la méthode précédente, le nombre de diagonales devient très importants, et cela influx sur le temps de calcul, en particulier pendant le processus de vérification de l'inconsistance des diagonales.

Pour cette raison, dans la deuxième version de notre approche, nous utilisons la méthode expliquée dans \cite{Subramanian2005}, au lieu d'extraire toutes les sous-diagonales, on supprime seulement les parties (les résidus) inconsistantes à partir des deux fragments de la diagonale principale, et on donne aux sous-fragments restants une autre chance dans le processus de sélection des diagonales consistantes.
On calcule le poids de la partie restante (la sous-diagonale restante), et on la place dans la liste des diagonales qui restent à vérifier.

Cette seconde méthode est utilisée dans la version appelée \emph{DiaWay $2.0$}.

\section{L'indexation des diagonales}
\label{sec:Diagonal_indexation}

Certaines étapes de notre approche nécessitent la recherche des diagonales et donc un accès rapide à ces derniers.
Pour remédier à ce problème, on construit une structure d'index (une matrice) qui contient les numéros des diagonales des fragments en fonction de leurs appartenances aux séquences biologiques.\\

On commence par la construction de la structure de données, on utilise un tableau de pointeur qui pointe vers des tableaux avec des tailles différentes. La taille de chaque tableau est le nombre maximal des fragments des diagonales possibles dans la séquence qui porte le numéro de ce tableau. On peut aussi utiliser une simple matrice, si le nombre de fragments dans chaque ligne est approximativement le même.

La deuxième étape consiste à remplir notre index.
Au moment de l'extraction des diagonales, pour chaque paire de séquences $(s_i,s_j)$ avec $i \neq j$, lorsqu'on extrait une diagonale $d_k$, on ajoute ses deux fragments à l'index, le premier fragment de la séquence $s_i$ dans le tableau(ou la ligne numéro) $i$, et le deuxième fragment dans le tableau (ou la ligne numéro) $j$.

Lorsqu'on extrait les diagonales, si la méthode d'extraction ne permet pas d'obtenir les fragments dans l'ordre d'appartenance à la séquence (par exemple l'extraction des diagonales de la partie supérieure de la diagonale principale de la matrice (M[0,0] M[$|t|$][$|t|$] avec $t=min(|s_2|,|s_1|)$ ), ensuite, la partie inférieure). Dans ce cas, on doit trier les fragments qui sont dans la même ligne (tableau) de notre index, selon l'indice de leur début dans la séquence biologique.\\

Exemple:\\
Soit les diagonales suivantes $D=\{d_0, d_1, d_2, d_3, d_4, d_5, d_6, d_7\}$, ordonnées selon leur poids. Les fragments des diagonales sont répartis sur les quatre séquences biologiques comme il est montré dans la figure \ref{fig:index_diagonales}.

\begin{figure}[h]
\centering
\includegraphics[width=0.9\linewidth]{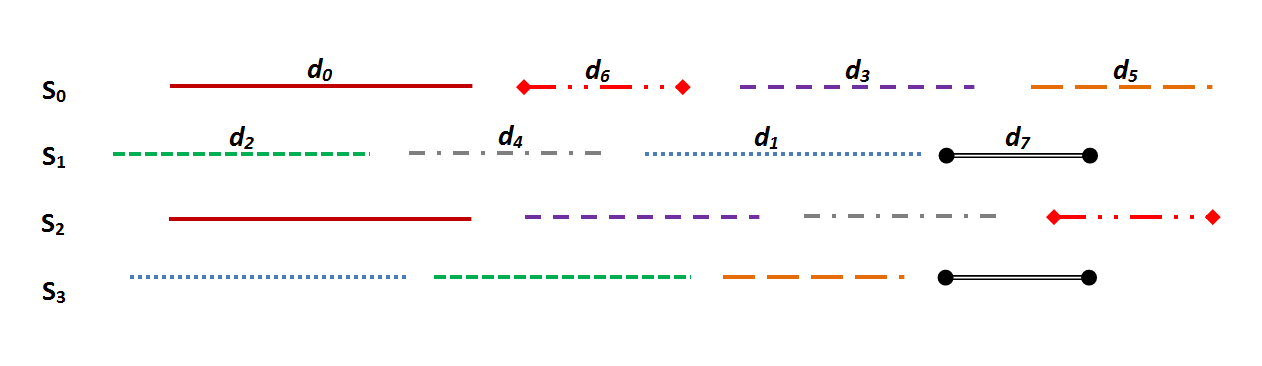}
\caption{Les fragments des 7 diagonales répartis sur les 4 séquences biologiques.}
\label{fig:index_diagonales}
\end{figure}

 Le résultat de la construction de l'index est illustré dans le tableau \ref{tab:tab_ordanance_des_diagonales} suivant (dans ce simple exemple, on utilise une matrice).

\begin{table}[h] 
\centering
\begin{tabular}{|c|c|c|c|c|}
\hline $S_0$ & $d_{0.1}$ & $d_{6.1}$ & $d_{3.1}$ & $d_{5.1}$ \\ 
\hline $S_1$ & $d_{2.1}$ & $d_{4.1}$ & $d_{1.1}$ & $d_{7.1}$ \\ 
\hline $S_2$ & $d_{0.2}$ & $d_{3.2}$ & $d_{4.2}$ & $d_{6.2}$ \\ 
\hline $S_3$ & $d_{1.2}$ & $d_{2.2}$ & $d_{5.2}$ & $d_{7.2}$ \\ 
\hline 
\end{tabular}
\caption{Les diagonales $d_0, d_1, d_2, d_3, d_4, d_5, d_6, d_7$ indexées selon les positions de leurs fragments dans les séquences.}
\label{tab:tab_ordanance_des_diagonales}
\end{table}

Si on veut obtenir les diagonales situées sur les mêmes séquences que $d_0$, il suffit d'accéder directement à l'index de $d_0 (S_0, S_2)$ qui correspond aux séquences $0$ et $2$, et dans ce cas on trouve les diagonales $d_3$ et $d_6$.\\

La méthode de recherche consiste à aller à la première ligne portant le premier fragment et enregistrer les numéros de toutes les diagonales trouvées dans cette ligne. Ensuite, aller à la deuxième ligne qui contient le deuxième fragment de la diagonale, et on cherche s'il y a des numéros de diagonales déjà enregistrées dans la première ligne.\\

Pour une raison d'optimisation du processus précédent, on utilise un simple tableau de bits, pour vérifier l'intersection entre les numéros des diagonales qui sont dans les deux lignes de la matrice comme suit :
On utilise un tableau de bits initialisé à $0$, la longueur de ce tableau de bits $L$ est le nombre total des diagonales. On prend la première ligne de la matrice et on marque toutes les cases, correspondantes aux numéros des diagonales, dans le tableau de bits par un $1$.
On prend la deuxième ligne de la matrice, et pour chaque numéro de diagonale, dans la position de tableau de bits représenté par ce numéro, on vérifie simplement si on a $1$. Si oui, alors la diagonale correspondante est prise comme une solution.\\

L'étape de l'indexation des diagonales est utilisée :
\begin{itemize}
\item Dans le processus de l'inconsistance simple. Voir la sous-section \ref{subsec:simple_inconsistency}.
\item Pour trier les diagonales pendant le processus d'alignement. Voir la section \ref{sec:sorting_diagonals}.
\end{itemize}

\section{Une nouvelle méthode pour la consistance des diagonales}
\label{sec:diagonals_consistency}

La consistance des diagonales est un problème difficile, elle est un concept crucial dans l'algorithme \emph{DIALIGN}.
La diagonale qui se croise ou chevauche avec une diagonale consistante qui a un poids plus élevé est considérée comme une diagonale inconsistante  \cite{Morgenstren1996}, voir la figure \ref{fig:D1_croise_et_chovoche_D0}, et cela, parce qu'elle ne permet pas aux diagonales avec un poids plus élevé d'être alignées correctement.

L'approche originale \cite{Morgenstren1996} vérifie la consistance d'une diagonale avec les diagonales déjà incorporées afin de l'ajouter à l'alignement. Nous avons proposé un concept totalement différent du concept original. Contrairement à l'algorithme \emph{DIALIGN}, nous essayons de prouver l'inconsistance des diagonales en utilisant une modélisation basée sur les graphes. Notre approche prouve qu'une diagonale est inconsistante avec les diagonales déjà testées et acceptées comme consistantes.

Dans notre travail, on utilise un graphe orienté, on considère les débuts des fragments des séquences $x_i$ ou $y_i$ comme des n\oe{}uds (sommets) du graphe, et les arcs correspondent à la relation entre le début d'un fragment d'une diagonale $d_i=(x_i,y_i)$ avec le début d'un autre fragment $(x_j \; ou \; y_j)$ d'une diagonale $d_j$ située sur la même séquence, où la position de début du fragment de $d_i$ est juste avant celui de $d_j$. Ainsi que la relation de début du fragment $x_i$ de la diagonale $d_i$ vers l'autre fragment $y_i$ de la même diagonale (ou vice versa).\\

Le graphe orienté $G=(V, E)$ est défini comme suit :\\

$V=\{x_{1},x_{2},...,x_{n} \}\cup\{y_{1},y_{2},...,y_{n}\}$ où $x_{i}$ et $y_{i}$ : le début du $1\ier{}$ et $2\ieme{}$ fragment de la diagonale $d_i$ , et $i \in [0..n]$ \\

$ E =  E_1 \cup E_2 \cup E_3 \cup E_4 \cup E_5 $

$ E_1=\{(x_i,y_i ),(y_i,x_i ) \;| i \in [0..n]\}$

$ E_2=\{(x_i, x_j) \;|i,j \in [0..n]$  et  $\exists k$  où  $x_i,x_j \in s_k$  et  $x_j$  le premier fragment qui apparait après $x_i$, et $s_k$ le numéro de la séquence biologique $k$ $\}$

$ E_3=\{(x_i, y_j ) \;|i,j \in [0..n]$  et $\exists k \; | x_i, y_j \in s_k$  et $y_j$ le premier fragment qui apparait après $x_i$  $\}$

$ E_4=\{(y_i, x_j ) \;|i,j \in [0..n]$  et $\exists k \; | y_i, x_j \in s_k$  et $x_j$ le premier fragment qui apparait après $y_i$  $\}$

$E_5=\{(y_i, y_j ) \;|i,j \in [0..n]$  et $\exists k \; | y_i, y_j \in s_k$ et $y_j$  le premier fragment qui apparait après $y_i$  $\}$\\

Exemple: on a 4 diagonales (8 fragments), situés sur 3 séquences biologiques différentes ($S_1$, $S_2$ et $S_3$).

\begin{figure}[h]
\centering
\includegraphics[width=0.7\linewidth]{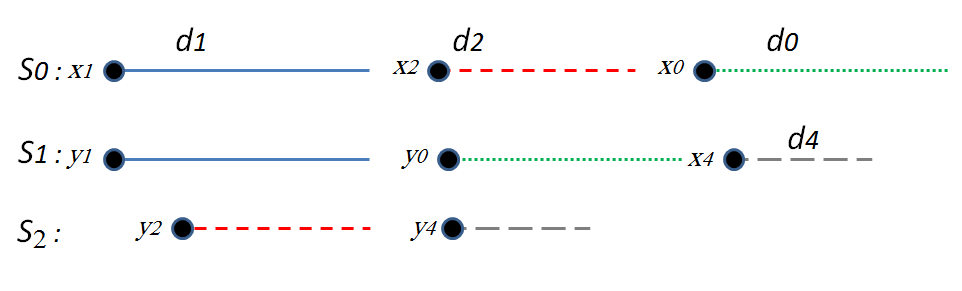}
\caption{Quatre diagonales et leurs fragments.}
\label{fig:4_diagonals}
\end{figure}

Donc le graphe sera $G(8, 13)$ avec 8 sommets et 13 arcs, voir la figure \ref{fig:graphe}:

\begin{figure}[h]
\centering
\includegraphics[width=0.7\linewidth]{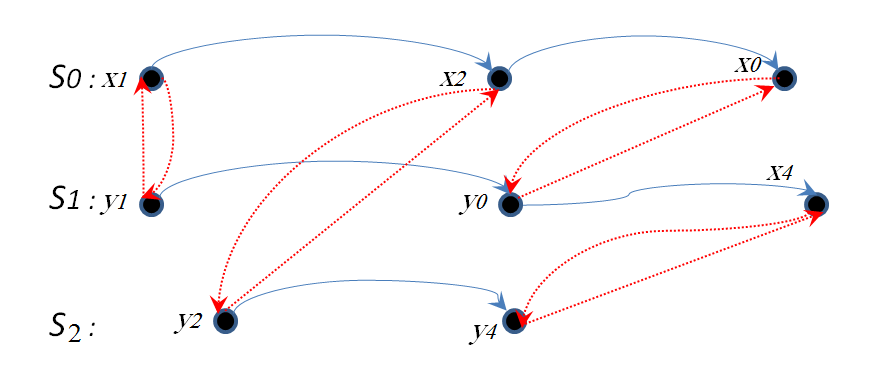}
\caption{Graphe de consistance $G(8,13)$. Les arcs en pointillés (les arcs verticaux avec la couleur rouge) représentent l'ensemble $E_1$, les autres arcs horizontaux représentent les ensembles $E_2$, $E_3$, $E_4$ et $E_5$.}
\label{fig:graphe}
\end{figure}

\subsection{L'inconsistance d'une diagonale avec un chemin simple}
\label{sec:path_consistency}

La diagonale $d_i$ est appelée inconsistante avec toutes les diagonales déjà acceptées dans $D^{'}$, si et seulement s'il y a un chemin simple \footnote{ Le chemin est un graphe dirigé $G = (V, E)$. C'est une séquence alternative de sommets et d'arcs :$\lambda=x_0 u_1 x_1...x_{k-1} u_k x_k$ où $i \in [1..k]$, et le sommet $x_i$ est le bout initial de l'arc $u_i$ et le sommet $x_{i+1}$ est son bout final. $\lambda$ est un chemin de $x_0$  à $x_k$ de longueur $k$. Dans un chemin simple tous ses arcs sont distincts l'un de l'autre \cite{graph1976}.}
qui connecte $x_{i}$ et $ y_{i}$ passant à travers les fragments des diagonales consistantes dans $D^{'}$. $(x_{i} \longrightarrow y_{i}$ ou $x_{i} \longleftarrow y_{i})$. En d'autres termes, pour prouver que $d_i$ est une diagonale inconsistante avec $D^{'}$ on doit trouver un chemin simple $\lambda$ entre $x_{i}$ et $y_{i}$ avec $\lambda=x_0 u_1 x_1...x_{k-1} u_k x_k$  où ($x_0=x_{i}$  et  $x_k=y_{i}$)   ou  ($x_0=y_{i}$   et  $x_k=  x_{i}$ ). \\

L'exemple suivant illustre ce concept. Soit la diagonale $d_{5}(x_5,y_5)$ dans la figure \ref{fig:chemin_simple}, pour vérifier l'inconsistance de cette diagonale, on essaye de trouver un chemin simple qui passe à travers les fragments des diagonales consistantes dans $D^{'}$ ($D^{'}=\{d_0,d_1,d_2,d_3,d_4\}$).

\begin{figure}[h]
\centering
\includegraphics[width=1\linewidth]{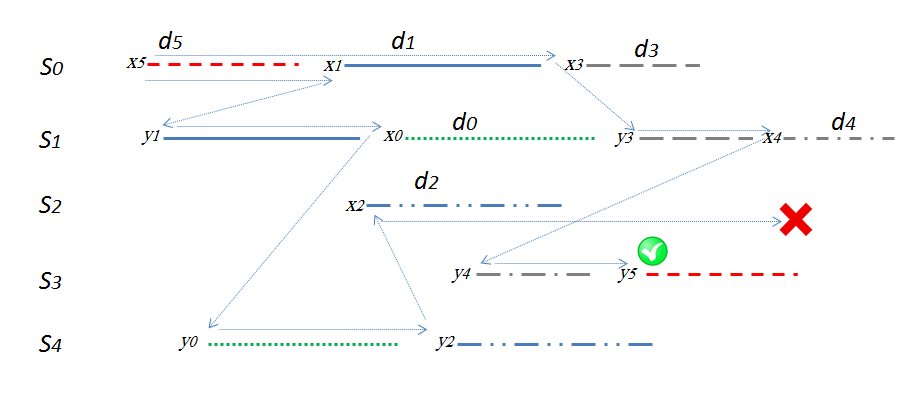}
\caption{Un chemin simple entre les deux fragments de la diagonale $d_5$ qui passe à travers les fragments des diagonales consistantes ($\lambda=x_5 \, x_1 \, x_3 \, y_3 \, x_4 \, y_4 \, y_5$).}
\label{fig:chemin_simple}
\end{figure}

Dans la figure \ref{fig:chemin_simple}, on trouve un chemin simple qui passe à travers les fragments des diagonales consistantes  ($\lambda=x_5 \, x_1 \, x_3 \, y_3 \, x_4 \, y_4 \, y_5$),  donc la diagonale $d_5$ est considérée comme inconsistante, et par conséquent elle est supprimée de la liste des diagonales $D$.
Lorsqu'on essaye de trouver un chemin simple, on peut avoir plusieurs possibilités. Dans cette figure (figure \ref{fig:chemin_simple}), on peut voir qu'on a deux chemins, le 1\ier{} mène à la solution et l'autre pas. On appel cette méthode \emph{inconsistance de chemin}.\\

Pour faire la recherche de chemin simple dans notre graphe orienté, on considère que la racine est le sommet représenté par le premier fragment de la diagonale $d_i$, ça peut être soit $x_i$ ou $y_i$. Le but est de trouver l'autre sommet dans le graphe, c'est-à-dire si la racine est $x_i$, on essaye de trouver $y_i$ et vice versa. La recherche de chemin simple utilise l'algorithme de \textit{recherche en profondeur d'abord}, pour plus de détails sur cet algorithme voir \cite{cormen2001}.

Dans notre graphe orienté, deux branches différentes peuvent converger vers un même sommet, voir la figure \ref{fig:arbreAumbigue}. Cela signifie qu'un sous-graphe peut être vérifié plus d'une seule fois inutilement, pour cela on doit s'assurer que dans la phase de recherche, le chemin passe qu'une seule fois par un sommet donné. Donc si dans la recherche on arrive à un sommet qui a été déjà vérifié, on arrête, pour ne pas vérifier le sous-graphe plusieurs fois.

\begin{figure}[h]
\centering
\includegraphics[width=0.7\linewidth]{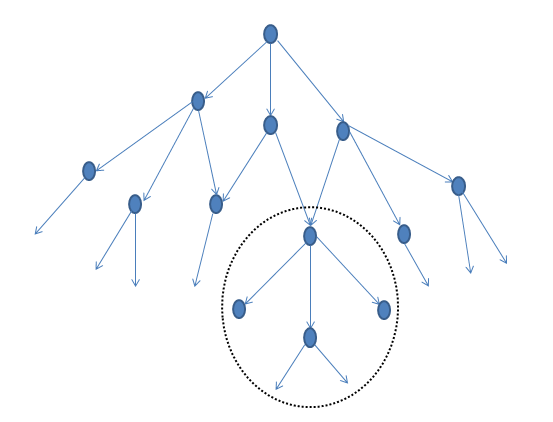}
\caption{Le graphe de recherche d'un chemin simple. Le sous-graphe entouré, n'est vérifié qu'une seule fois.}
\label{fig:arbreAumbigue}
\end{figure}

\subsection{L'inconsistance avec le chemin simple trouvé}
\label{sub:Inconsistency_with_found_simple_path}

Lors de l'étape de recherche de chemin simple, pour vérifier l'inconsistance d'une diagonale $d_i$, si ce chemin existe, toutes les diagonales dans $D$ qui appartiennent à ce chemin dans la direction spécifiée sont supprimées. voir la figure \ref{fig:appartientChemin}.

\begin{figure}[h]
\centering
\includegraphics[width=0.8\linewidth]{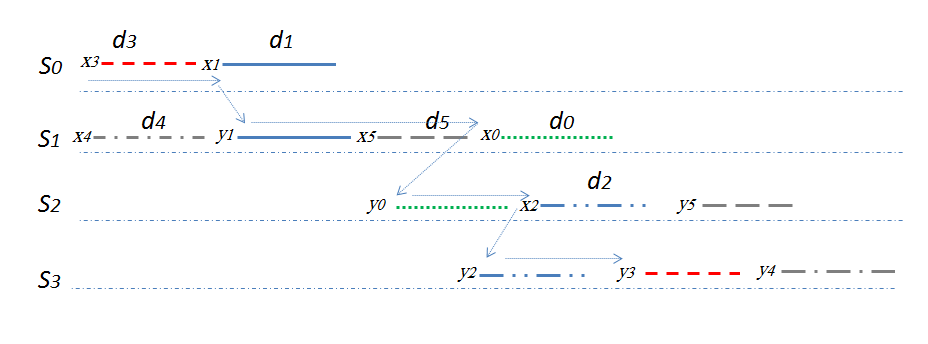}
\caption{Les diagonales $d_4$ et $d_5$ appartiennent au chemin simple.}
\label{fig:appartientChemin}
\end{figure}

Le chemin simple trouvé, pour prouver l'inconsistance de la diagonale $d_3$, est utilisé pour éliminer les diagonales $d_4$ et $d_5$ qui appartiennent à ce chemin et suivent le bon sens. Avec cette technique on évite de vérifier ces diagonales avec la recherche de chemin simple qui peut prendre beaucoup de temps.

\subsection{L'inconsistance simple}
\label{subsec:simple_inconsistency}

La première méthode \emph{ inconsistance de chemin } (voir la sous-section \ref{sec:path_consistency}) suffit à elle seule pour éliminer toutes les diagonales inconsistantes, mais avec cette approche la recherche peut emprunter plusieurs chemins avant de trouver le bon, et cela implique un temps d'exécution très élevé.
Pour cela, on a choisi d'employer une seconde méthode appelée \emph{inconsistance simple} qui permet d'éliminer les inconsistances entre deux diagonales situées dans les mêmes séquences, et cela permet de réduire considérablement le temps d'exécution.

L'algorithme prend une diagonale consistante $d_i$ ($d_i \in D^{'}$) et la compare avec la diagonale $d_j$ où $d_j \in D$), si on trouve que leurs sous-séquences se trouvent dans les mêmes séquences et $d_j$ se croise ou se chevauche avec la diagonale  $d_i$ alors  $d_j$ est inconsistante, donc on la supprime directement de $D$. On vérifie toutes les diagonales qui se situent sur les mêmes séquences de $d_i$. \\

Exemple :\\
On considère deux séquences biologiques $S=\{s_0,s_1\}$, et soit $d_0=(x_0,y_0)$ une diagonale consistante, et soit les 3 diagonales $D=\{d_1(x_1,y_1) , d_2(x_2,y_2) , d_3(x_3,y_3)\}$ qui ne sont pas encore vérifiées. Tous leurs fragments sont sur les mêmes séquences. voir la figure \ref{fig:inconsistance_simple}.

\begin{figure}[h]
\centering
\includegraphics[width=0.8\linewidth]{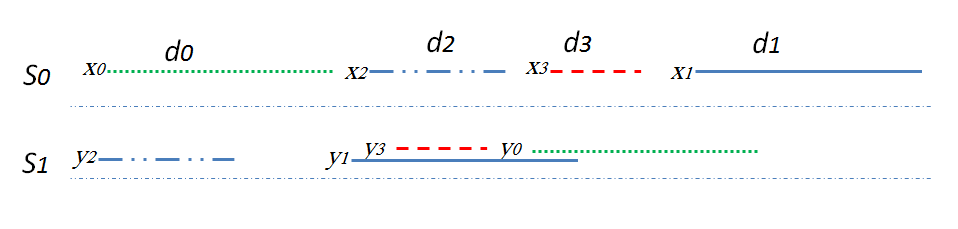}
\caption{Inconsistance simple : les diagonales  $d_1, d_2$ se croisent avec $d_0$ et $d_3$ se chevauche avec $d_0$.}
\label{fig:inconsistance_simple}
\end{figure}

Dans la figure \ref{fig:inconsistance_simple}, on voit clairement que les diagonales  $d_1, d_2$ se croisent avec $d_0$ et $d_3$ se chevauche avec $d_0$, donc toutes ces diagonales $d_1,d_2,d_3$ sont inconsistantes et seront donc supprimées.

\subsection{Chemin vertical}
\label{sub:ali:Chemin_vertical}

Un chemin est dit vertical si et seulement s'il ne passe que par les sommets (les fragments) d'une même diagonale $d_i$, et si on a plus d'une diagonale, dans ce cas leurs fragments doivent partager le même sommet (c'est-à-dire leurs fragments sont dans les mêmes séquences et ont le même début).

Les diagonales qui sont reliées par un chemin simple vertical ne sont pas considérées comme des diagonales inconsistantes, car elles ne causent pas de problème pour l'alignement. Pour cela, on ne les supprime pas, on les garde pour les utiliser dans la suppression des autres diagonales avec la méthode \emph{inconsistance simple}.

Soit l'exemple dans la figure \ref{fig:cheminVertical} suivante.

\begin{figure}[h]
\centering
\includegraphics[width=0.8\linewidth]{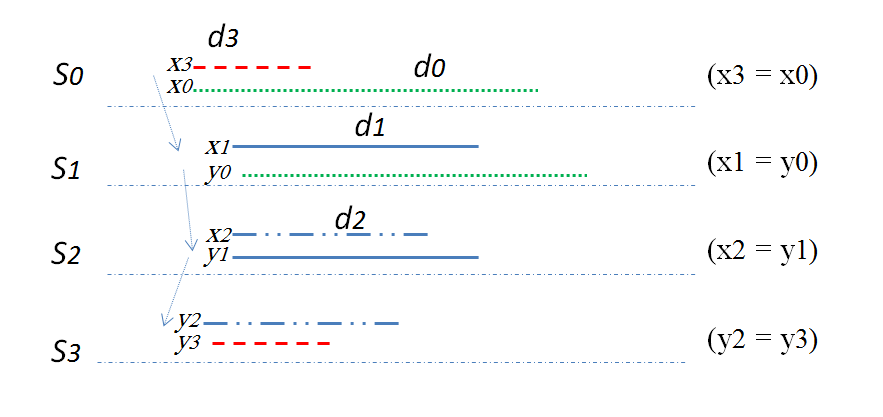}
\caption{Un chemin simple vertical qui relie les deux fragments ($x_3, \, y_3$) de la diagonale $d_3$.}
\label{fig:cheminVertical}
\end{figure}

Selon la méthode \emph{inconsistance de chemin} (voir la sous-section \ref{sec:path_consistency}), la diagonale $d_3$ est une diagonale inconsistante car il y a un chemin simple entre ses deux fragments, donc on la supprime, même si elle ne cause pas de problème pour l'alignement, $d_3$ sera alignée par transitivité lorsque les diagonales $d_0$, $d_1$ et $d_2$ seront alignées.

Mais, afin de supprimer les diagonales inconsistantes avec $d_3$ en utilisant la méthode \emph{inconsistance simple} ( voir la sous-section \ref{subsec:simple_inconsistency}), on doit la garder, et l'ajouter à l'ensemble des diagonales consistantes $D^{'}$.

\subsection{Le processus du traitement de l'approche de l'inconsistance}

Étape I: initialisation de l'ensemble des diagonales consistantes $D^{'}$ :

\begin{itemize}

\item On ajoute la première diagonale $d_0$ à $D^{'}$. La première diagonale $d_0$ a le plus grand poids dans l'ensemble $D$, donc elle est automatiquement consistante.

\item On utilise la méthode \emph{inconsistance simple} pour supprimer les diagonales de $D$ qui sont inconsistantes simple avec $d_0$.

\item On ajoute la première diagonale dans $D$ (appelée $d_1$) à l'ensemble consistant $D^{'}$. On a prouvé que la diagonale $d_1$ est consistante à l'aide de la méthode \emph{inconsistance simple}, par conséquent $d_1$ est consistante avec $d_0$, maintenant $D^{'} = \{d_0, d_1\}$.

\item On fait appel à la méthode \emph{inconsistance simple} pour supprimer toutes les diagonales de $D$ qui sont inconsistantes simple avec $d_1$.

\end{itemize}

\noindent
Étape II : Traiter toutes les diagonales inconsistantes restantes dans $D$ :

\begin{itemize}

\item Trouver une diagonale consistante :\\
Appliquer la méthode \emph{inconsistance de chemin} (voir la sous section \ref{sec:path_consistency}), pour éliminer toutes les diagonales inconsistantes dans $D$ avec toutes les diagonales de $D^{'}$, une par une, jusqu'à ce qu'on trouve une diagonale consistante avec les diagonales déjà acceptées (qui sont dans $D^{'}$), ou on arrive à la fin de l'ensemble $D$, et cela veut dire qu'on a supprimé toutes les diagonales de $D$.

Dans cette méthode, lorsqu'on trouve un chemin simple entre les deux fragments d'une diagonale $d_i$ on la supprime. Toutes les diagonales qui sont dans ce chemin simple sont supprimées en utilisant la méthode \emph{inconsistance avec chemin trouvé}, voir la sous section \ref{sub:Inconsistency_with_found_simple_path}.

Lorsqu'on fait la vérification d'une diagonale $d_i$ avec la méthode \emph{inconsistance de chemin} et on ne trouve pas de chemin simple entre ces deux fragments $(x_i,y_i)$, cela signifie que cette diagonale est consistante, donc on l'ajoute à l'ensemble $D^{'}$.\\

\item Appliquer la méthode \emph{inconsistance simple} :\\
On applique la méthode \emph{inconsistance simple} pour supprimer toutes les diagonales de $D$ qui sont inconsistantes simple avec $d_i$ trouvé dans l'étape précédente.\\

\item  Répéter ces deux étapes précédentes, jusqu'à traiter toutes les diagonales qui sont dans $D$. À la fin, $D^{'}$ ne contiendra que des diagonales consistantes, et l'ensemble $D$ sera vide.

\end{itemize}

\section{Le tri des diagonales}
\label{sec:sorting_diagonals}

Après la phase de la vérification de la consistance des diagonales, voir la section \ref{sec:diagonals_consistency}, où on supprime toutes les diagonales inconsistantes qui posent un problème pour l'alignement, on obtient un ensemble contenant uniquement des diagonales consistantes. L'alignement est effectué sur les diagonales restantes, de façon à ce que chaque fragment est placé en face de l'autre fragment de la même diagonale.

La méthode triviale pour aligner les  diagonales restantes (placer les 2 fragments de chaque diagonales l'un en face de l'autre) est la suivante :

\subsection{Méthode 1}

Les diagonales restantes sont triées en fonction de leur poids dans l'ensemble final. Pour les aligner, on commence par la première diagonale dans la liste (celle qui a le plus grand poids), et on met ses deux fragments $(d_i=(x_i,y_i))$, l'un devant l'autre.
Si le début du fragment $x_i$ est, avant le début du fragment $y_i$, alors $x_i$ est déplacé vers $y_i$, et vice versa. Pour déplacer le fragment ($x_i/y_i$), toute la sous-séquence qui commence à partir de ce fragment doit être déplacée.

Lorsqu'on passe à la prochaine diagonale ($d_j$ avec $j>i$) pour rassembler ses fragments, le déplacement des fragments de cette diagonale $d_j$ peut détériorer l'alignement des diagonales déjà alignées. Pour cela, on commence d'abord par vérifier si les fragments des diagonales déjà alignées sont déplacées (leurs positions ont été changées et leur alignement est détérioré) afin de ré-effectuer à nouveau leur alignement. Ce processus est répété pour toutes les diagonales déjà alignées une par une.
Par ailleurs, pour l'ensemble de toutes les diagonales, chaque fois qu'une nouvelle diagonale $d_j$ est alignée, on doit vérifier de manière récurrente pour chaque diagonale $d_k$ ($k \in [1..j]$), toutes les diagonales précédentes dans l'alignement ($[1..k]$), jusqu'à arriver à la première diagonale $d_1$.\\

Cette méthode est très couteuse en matière de temps d'exécution, car pour chaque diagonale alignée, on doit vérifier toutes les diagonales précédentes ajoutées dans le processus d'alignement (l'étape finale qui met ensemble (regroupe) les fragments des diagonales).\\

Exemple: Si on considère qu'on est dans la 5\ieme{} diagonale, la diagonale alignée $d_j$ ($j$ = 5), donc le nombre de vérifications est comme suit :\\

\noindent
$5+4+3+2+1$\\
$ +4+3+2+1$\\
$   +3+2+1$\\
$     +2+1$\\
$       +1$\\
     
Le nombre de vérifications à faire pour une seule position est donné par la formule suivante : $\sum_{k=1}^{k=j} \frac{k(k+1)}{2} $ \footnote{Il y a une petite erreur dans notre article \cite{DIAWAY_b}. Nous avons utilisé la formule $\sum_{k=1}^{j-1} j+(j-k)^{k+1}$ au lieu de la formule $j+ \sum_{k=1}^{j-1} (j-k)(k+1)$, et bien sur cette erreur a fait que la complexité temporelle obtenue qui est de l'ordre de $O(n^n)$ est fausse. Dans cette thèse, nous obtenons donc la formule $\sum_{k=1}^{k=j} \frac{k(k+1)}{2} $ qui est égale à $j+ \sum_{k=1}^{j-1} (j-k)(k+1)$.}.

Le nombre total d'opérations pour vérifier $n$ diagonales est :  
$\sum_{j=1}^{j=n} \sum_{k=1}^{k=j} \frac{k(k+1)}{2} $ dont la complexité temporelle est de l'ordre  
$\Omega(n^4)$.

\subsection{Méthode 2}

La méthode précédente ayant une complexité temporelle élevée, nous avons proposé une nouvelle méthode qui est plus efficace et qui peut assurer l'alignement final seulement en $O(k \times n)$, où $n$ est le nombre de diagonales consistantes restantes, et $k$ est le nombre de séquences biologiques à aligner.\\

Cette méthode consiste à trier toutes les diagonales consistantes pour l'alignement final, en fonction de leurs positions dans les séquences biologiques.
La condition de vérification pour le tri n'est pas triviale, car les fragments des diagonales ne sont pas situés dans les mêmes séquences biologiques, pour cette raison, on applique la méthode suivante pour les trier :

\begin{enumerate}

\item L'indexation des fragments des diagonales en fonction des numéros des séquences de leurs fragments, comme nous l'avons déjà expliqué, dans la section \ref{sec:Diagonal_indexation}.
\medskip
\item Construire progressivement l'ensemble final des diagonales pour l'alignement de telle sorte que : la diagonale du premier tour est celle dont ses deux fragments n'ont pas de fragments qui les précèdent dans la même séquence. C'est-à-dire les fragments qui sont situés les premiers dans les séquences biologiques. 
Après cela, on supprime les fragments de la structure d'index, et on ajoute leurs diagonales à la liste triée.
Ce processus est répété jusqu'à ce que tous les fragments soient supprimés et leurs diagonales sont ajoutées à la liste triée.

Pour obtenir les fragments satisfaisant ces conditions on procède comme suit : 
on parcourt la première colonne de la matrice ($i\in [1..k]$), $k$ est le nombre des lignes de la matrice (en d'autre terme, on parcourt la première case des lignes de la matrice ou la structure d'index ligne après ligne). On marque le premier fragment d'une diagonale à la ligne $i$. Si le second fragment situé à la ligne $j$ (où $j > i$) est situé en premier sur sa ligne, la condition est satisfaite. Dans le cas contraire, si la condition n'est pas satisfaite, on saute à la ligne suivante $i+1$. Ce processus est répété jusqu'à ce qu'on arrive à la dernière ligne de la matrice.

Un fragment est considéré le premier dans sa ligne, s'il n'y a aucun fragment avant lui. Dans le cas où on trouve une diagonale qui satisfait la condition, on recommence la recherche depuis la 1\iere{} ligne. Si on arrive à la dernière ligne sans qu'on trouve une nouvelle diagonale qui satisfait la condition de tri, on passe à la colonne suivante dans la matrice.

\end{enumerate}

\paragraph{Le calcul de la complexité :} Pour trouver la bonne diagonale qui satisfait la condition (d'être la première), il se peut qu'on doive parcourir toute la colonne (tous les niveaux). On est sure que dans chaque étape (le parcourir d'une seule colonne, autrement un seul niveau), on trouve au moins une seule diagonale. Dans le pire cas, pour une seule diagonale, il faut parcourir toute la colonne donc les (k) lignes de la matrice (le nombre des séquences biologiques). Au total, on a $n$ diagonales, donc pour traiter toutes les diagonales, il nous faut $O(k \times n)$.

Ensuite, dans la dernière étape \ref{sec:ali:Mettre_les_fragments_ensemble}, on regroupe les fragments des diagonales  avec un temps de l'ordre de $O(n)$. La complexité globale de ces deux étapes devient $O(k \times n)$.

L'étape finale (regrouper les fragments des diagonales) peut être faite directement en même temps que l'opération de tri. À chaque fois qu'on trouve la diagonale qui satisfait la condition de tri, on la supprime de l'index et on regroupe ses fragments.\\

Exemple :\\
On considère l'ensemble suivant après la phase de suppression des diagonales inconsistantes :
$D'= \{d_0, d_1, d_2, d_3, d_4, d_5, d_6\}$ (voir la figure \ref{fig:tri_diag_alignement1}).

\begin{figure}[h]
\centering
\includegraphics[width=0.8\linewidth]{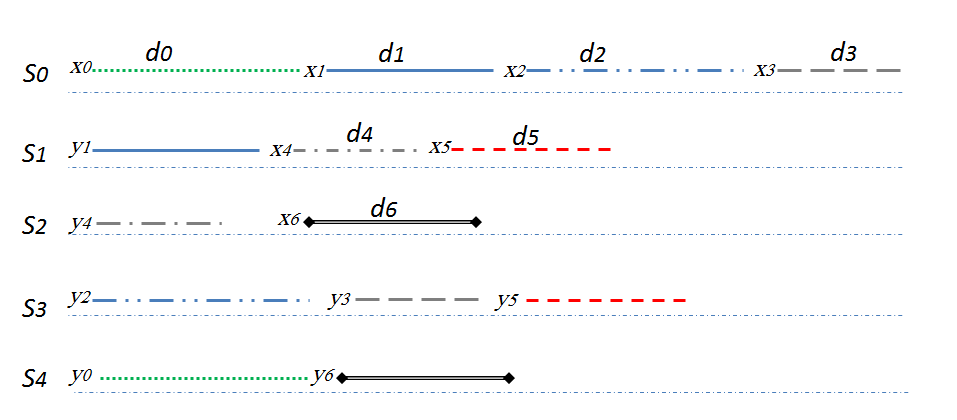}
\caption{L'ensemble des diagonales indexées dans une matrice.}
\label{fig:tri_diag_alignement1}
\end{figure}

On constate dans la figure \ref{fig:tri_diag_alignement1} que la diagonale $d_0$ satisfait la condition expliquée précédemment, donc on l'ajoute dans l'ensemble trié des diagonales  $D^{''}$, et on supprime ses deux fragments.

\begin{figure}
    \centering
    \begin{subfigure}[b]{0.49\textwidth}
        \centering
        \includegraphics[width=\textwidth]{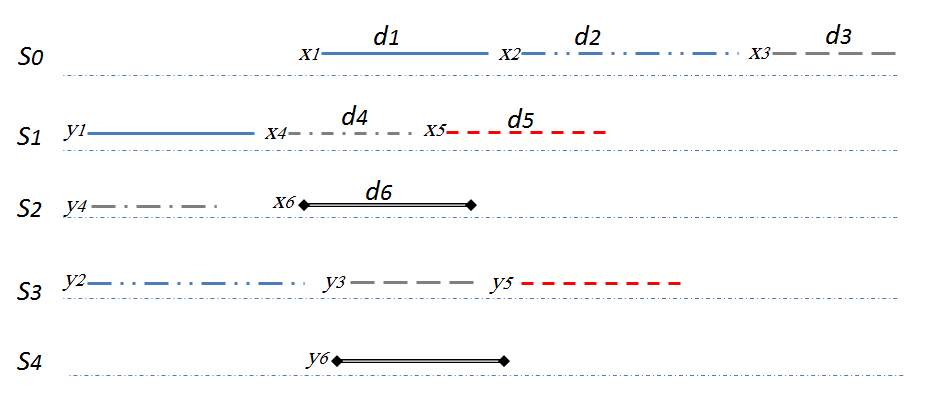}
        \caption{} \label{fig:operation_2}
    \end{subfigure}
    \hfill
    \begin{subfigure}[b]{0.49\textwidth}
        \centering
        \includegraphics[width=\textwidth]{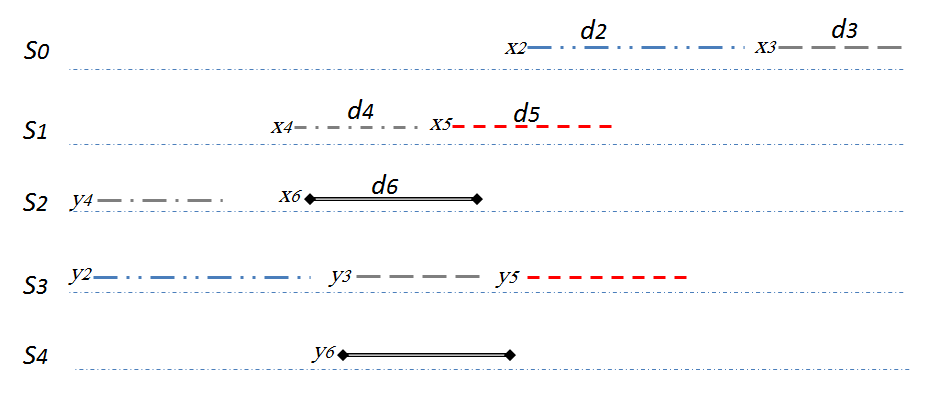}
        \caption{} \label{fig:operation_3}
    \end{subfigure}

    \begin{subfigure}[b]{0.49\textwidth}
        \centering
        \includegraphics[width=\textwidth]{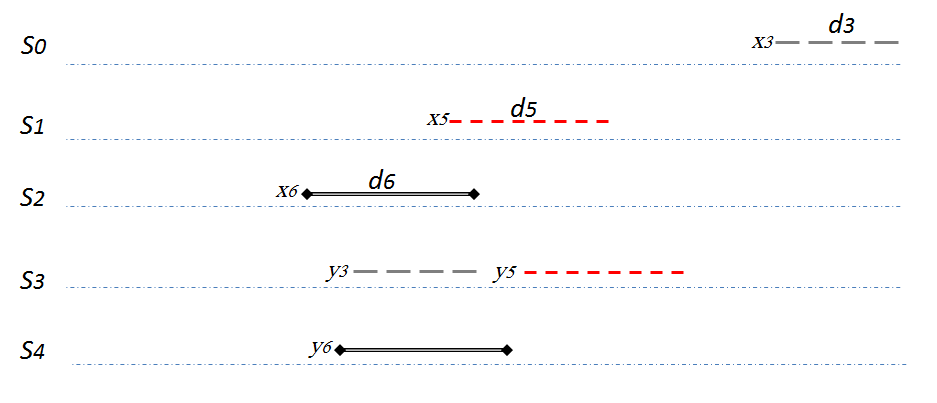}
        \caption{} \label{fig:operation_4}
    \end{subfigure}
    \hfill
    \begin{subfigure}[b]{0.49\textwidth}
        \centering
        \includegraphics[width=\textwidth]{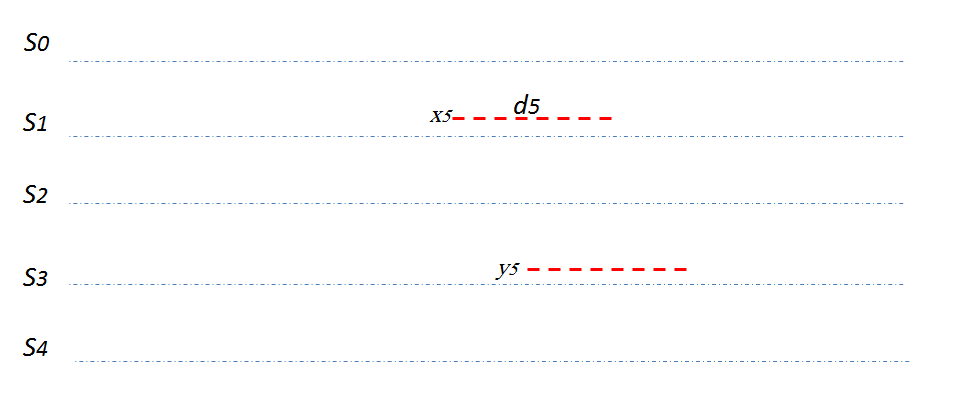}
        \caption{} \label{fig:operation_5}
    \end{subfigure}
    
    \caption{Les itérations du tri des diagonales.}
    \label{fig:tri_diag_alignement2...}
\end{figure}

Dans la figure \ref{fig:operation_2}, la diagonale $d_1$, est la diagonale qui satisfait la condition expliquée précédemment.

Dans la figure \ref{fig:operation_3}, on a deux diagonales $d_2$, et $d_4$ qui satisfont la condition expliquée précédemment en même temps. Donc on suit l'ordre d'apparition de leurs fragments dans la matrice, on commence par la diagonale $d_2$, on l'ajoute à $D^{''}$ et on supprime ses fragments de la matrice. Ensuite on fait la même chose pour la diagonale $d_4$. L'étape montrée dans cette figure se fait en deux itérations, celle de $d_2$, ensuite celle de $d_4$.

Dans la figure \ref{fig:operation_4}, on a aussi deux $d_3$, et $d_6$ qui satisfont la condition en même temps.

Dans la figure \ref{fig:operation_5}, la diagonale $d_5$, est la diagonale qui satisfait la condition expliquée précédemment.\\

Pour notre exemple, l'ensemble des diagonales initiales qui est trié selon leurs poids était comme suit : 
$D'= \{d_0, d_1, d_2, d_3, d_4, d_5, d_6\}$.

Après l'opération de tri selon les positions des fragments des diagonales dans les séquences biologiques, l'ensemble devient: $D^{''}= \{d_0, d_1, d_2, d_4, d_3, d_6, d_5\}$.

\section{Mettre les fragments des diagonales ensemble, et insérer les gaps (-)}
\label{sec:ali:Mettre_les_fragments_ensemble}

Dans cette dernière étape, on prend l'ensemble des diagonales triées avec la deuxième méthode expliquée dans la section précédente, et pour chaque diagonale ($d_i=(x_i,y_i)$), on met ses deux fragments l'un en face de l'autre (l'un devant l'autre). Si le début du fragment $x_i$ est, avant le début du fragment $y_i$, alors $x_i$ est déplacé vers $y_i$, et vice versa. Pour déplacer le fragment, toute la sous-séquence qui commence à partir du fragment ($x_i/y_i$) doit être déplacée.

Après cette étape, on doit insérer les gaps ($-$) dans les espaces vides (qui sont le résultat du déplacement des fragments).

\section{L'analyse et la comparaison}
\label{sec:ali:analyse_comparaison}

Nous avons réalisé deux implémentations de notre approche. La première version de notre approche est appelée \emph{DiaWay 1.0} (Diagonal Way). Dans cette implémentation, pour chaque diagonale, on extrait toutes ses sous-diagonales comme nous l'avons expliqué dans la section \ref{sec:extraction_des_diagonals}, méthode $1$. Notre deuxième implémentation est appelée \emph{DiaWay 2.0}. Dans cette version, nous n'extrayons pas les sous-diagonales, on utilise une autre méthode pour améliorer la qualité d'alignement, voir la section \ref{sec:extraction_des_diagonals} méthode $2$ (dans la version actuelle de \emph{DiaWay 2.0}, la méthode $2$ expliquée dans la section \ref{sec:extraction_des_diagonals} n'est encore implémentée, on utilise simplement les diagonales seules), et aussi, nous trions les diagonales avant leurs alignement final.
La deuxième version de notre approche est une amélioration de la première version. Comme nous l'avons expliqué, nous avons changé les deux méthodes (l'extraction des diagonales, et leurs tri avant la phase finale) afin de rendre le programme beaucoup plus rapide.

Notre implémentation est modulaire, elle a été programmée en langage C++ avec l'utilisation du Framework Qt. Les tests ont été effectués sur Windows 7, 32 bits, avec un processeur Intel Core 2 duo E8400 3,00 GHz et 2 Gigabyte de RAM. Nous avons utilisé un seul c\oe{}ur.
Le code source et les deux versions de notre approche sont disponibles sur : \url{https://github.com/chegrane/DiaWay\_1.0} et \url{https://github.com/chegrane/DiaWay\_2.0}. 
Notre projet est distribué sous la licence publique générale limitée GNU (GNU LGPL) (Anglais: GNU Lesser General Public License).

Nous avons fait une analyse sur les résultats des tests, et une comparaison avec l'algorithme original \emph{DIALIGN}. Les ensembles de données (benchmarks) des séquences d'ADN sont récupérés depuis les bases de données \textbf{BAliBASE} et \textbf{SMART}~\cite{BAliBASE_carroll2007}.

Dans notre travail, nous utilisons la même formule de score définie dans \cite{Morgenstren1996,Morgenstern1998a,Morgenstern1998b}, pour cela, l'analyse est basée sur le temps de calcul seulement.\\

% il faut dire pour quoi on a pas fait des comparaision avec les autres algorithm d'alignmenet?? où je doit faire les test...

Le temps d'exécution varie en fonction du nombre de séquences à aligner, de leurs longueurs et la longueur minimale de la diagonale autorisée.
Soit un ensemble de séquences biologiques $S$ avec une longueur $|S|$, la longueur des séquences moyennes est notée par $\bar{S}$.
$|D|$ représente la longueur de l'ensemble $D$ de toutes les diagonales (donc le nombre total des diagonales initiales). Et $|D^{'}|$ la taille de l'ensemble $D^{'}$ qui ne contient que les diagonales consistantes.

Le tableau \ref{tab:tab_nb_diagonal} suivant, illustre la relation entre le nombre de diagonales consistantes et le nombre total de diagonales initiales.

\begin{table}[h] 
\centering
\begin{tabular}{|c|c|c|c|}
\hline $|S|$ & $\bar{S}$  & $|D|$ & $|D^{'}|$ \\ 
\hline 3 & 139 & 43757 & 156 \\ 
\hline 5 & 202 & 101261 & 499 \\ 
\hline 8 & 170 & 243845 & 661 \\ 
\hline 10 & 101 & 131515 & 1474 \\ 
\hline 
\end{tabular}
\caption{Le nombre total de diagonales initiales et les diagonales consistantes restantes.}
\label{tab:tab_nb_diagonal}
\end{table}

Dans le tableau \ref{tab:tab_nb_diagonal}, les résultats montrent qu'avant d'arriver à l'ensemble des diagonales consistantes, nous sommes obligés de supprimer presque toutes les diagonales initiales. Le nombre des diagonales consistantes restantes est inférieur à $1\%$ de toutes les diagonales initiales.\\

%\subsection{influence of diagonals number in the execution time}

Le temps d'exécution varie selon le nombre de toutes les diagonales initiales. Dans la figure \ref{fig:nb_diagonal_temps_7}, nous donnons le temps d'exécution en millisecondes pour une longueur minimale autorisée de diagonale $=7$.

\begin{figure}[h]
\centering
\includegraphics[width=0.7\linewidth]{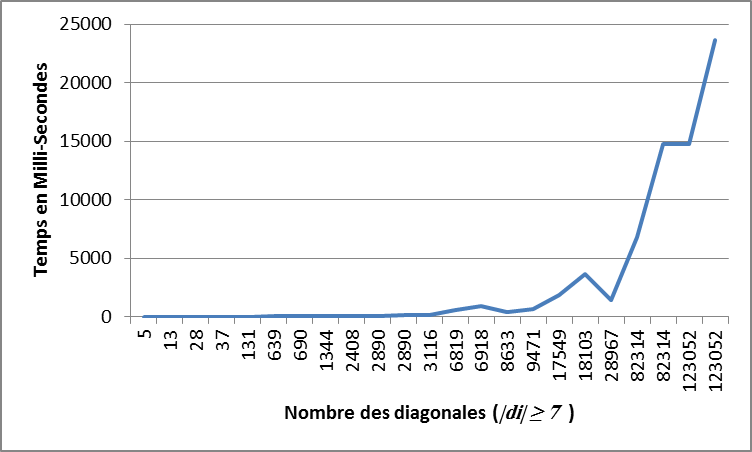}
\caption{Le temps d'exécution en millisecondes pour une longueur minimale de diagonale = $7$.}
\label{fig:nb_diagonal_temps_7}
\end{figure}

% chercher commant ecrir millisecondes  avec des symbole latex

Le graphe de la figure \ref{fig:nb_diagonal_temps_7}, montre que le temps d'exécution augmente en corrélation avec le nombre de diagonales.
Cela signifie que si le nombre initial des diagonales est petit (juste après la première étape qui permet leur extraction voir \ref{sec:extraction_des_diagonals}), cela implique un temps d'exécution petit. Et dans le cas contraire, où on a un très grand nombre de diagonales, le temps de calcul est très élevé.\\

Si nous avons des séquences biologiques avec un degré élevé de similitude, le nombre de diagonales consistantes sera énorme. Cela veut dire que le graphe de vérification de l'inconsistance des diagonales (voir la section \ref{sec:path_consistency}) devient aussi très grand. Cela rend le traitement de la méthode inconsistance de chemin plus lent, vu le nombre de chemins qui doivent être testés avant de trouver le bon, et qui implique un temps d'exécution plus élevé. Voir le tableau \ref{tab:similarity_degree}.

\begin{table}
\centering
\begin{tabular}{|c|c|c|c|c|}
\cline{2-5} 
\multicolumn{1}{c|}{} & $|S|$  &  $\bar{S}$ & $|D|$ & T en min\tabularnewline
\hline 
 & 9 & 40 & 15318  & 0,59\tabularnewline
\cline{2-5} 
Haut  & 9 & 76 & 57983  & 3,61\tabularnewline
\cline{2-5} 
similarité  & 10 & 73 & 63137  & 3,44\tabularnewline
\cline{2-5} 
 & 10 & 75 & 67058  & 7,42\tabularnewline
\hline 
\hline
 & 8 & 180 & 226217  & 0,64\tabularnewline
\cline{2-5} 
Similarité & 8 & 295 & 583156  & 1,36\tabularnewline
\cline{2-5} 
régulière & 10 & 166 & 303685  & 1,52\tabularnewline
\cline{2-5} 
 & 11 & 190 & 551049  & 2,27\tabularnewline
\hline 
\end{tabular}
\caption{Le temps (en minute) d'exécution en fonction du degré de similitude des séquences.}
\label{tab:similarity_degree}
\end{table}

Dans le tableau \ref{tab:similarity_degree}, on a deux groupes de séquences. Le premier contient les séquences qui ont un degré élevé de similarité, et le second celles qui ont un degré de similarité régulier. On voit que la taille des séquences du deuxième ensemble (avec des similarités ordinaires) est plus grande que le premier ensemble avec une similarité élevée, et également, le nombre de diagonales du deuxième ensemble est excessivement plus grand que celui du premier ensemble. Malgré cela, le temps d'exécution de l'ensemble de similitude régulier est plus petit que l'ensemble de similarité élevée. Car, comme nous l'avons expliqué ci-dessus, dans le cas du degré de similitude élevé, le nombre de diagonales consistantes est très grand, et cela rend le traitement d'inconsistance plus lent (spécialement, la méthode de l'inconsistance de chemin, voir la section \ref{sec:path_consistency}). Et cela implique un temps d'exécution plus élevé.\\

%\subsection{Comparison with DIALIGN 2.2}

Nous avons comparé nos deux implémentations avec \emph{DIALIGN 2.2} \cite{Morgenstern1999}, son code source est disponible sur \url{http://bibiserv2.cebitec.uni-bielefeld.de/dialign/}~\footnote{Visité le: 02-07-2016.}.
Le temps d'exécution dépend du nombre de séquences biologiques, et leur longueur moyenne, voir le tableau \ref{tab:T_nbSequance_averageL}.

\begin{table}[h]
\centering
\begin{tabular}{|c|c|c|c|c|}
\hline $|S|$ & $\bar{S}$ & DW.2.0 & DW.1.0 & DIALIGN 2.2 \\ 
\hline 3 & 530 & 63 & 421  & 1312  \\ 
\hline 4 & 250 & 31 & 125 & 766 \\ 
\hline 5 & 200 & 47 & 125 & * \\ 
\hline 6 & 119 & 32 & 1203  & 1750  \\ 
\hline 7 & 195 & 125 & 422 & 1640  \\ 
\hline 8 & 294 & 437  & 1453  & 4062  \\ 
\hline 9 & 325 & 1047  & 6828  & 6297  \\ 
\hline 10&	476 & 3204 & 14782 & 16094  \\
\hline 11&	190 & 875  & 3656  & 3281  \\
\hline 
\end{tabular}
\caption{Le temps d'exécution de DIALIGN 2.2, DiaWay $1.0$ (DW.1.0), et DiaWay $2.0$ (DW.2.0) en millisecondes selon le nombre des séquences et leurs longueurs.}
\label{tab:T_nbSequance_averageL}
\end{table}

Le symbole étoile (*) dans le tableau \ref{tab:T_nbSequance_averageL}, signifie que l'application plante et ne fonctionne pas pour cet exemple. D'après le tableau \ref{tab:T_nbSequance_averageL}, on peut constater les points suivants :

\begin{itemize}
\item Le temps pour aligner 4 ou 5 séquences est petit par rapport au temps pour aligner 3 séquences, parce que leurs longueurs sont plus petites que celles des 3 séquences.

\item La méthode \emph{DiaWay 1.0} (DW.1.0) a un meilleur temps d'exécution par rapport à \emph{DIALIGN 2.2} pour un petit ensemble de séquences avec une longueur moyenne relativement petite. Mais pour un grand nombre de séquences, son temps d'exécution est très mauvais parce que : 1) \emph{DiaWay 1.0} (DW.1.0) extrait pour chaque diagonale toutes ses sous-diagonales et cela rend l'ensemble des diagonales énorme, 2) dans l'étape finale de l'alignement (mettre les fragments ensembles), on ne fait pas le tri des diagonales. Ces deux raisons rendent le temps de traitement très élevé.

\item Le temps d'exécution pour notre version \emph{DiaWay 2,0} (DW.2.0) est beaucoup plus petit que les deux autres méthodes, \emph{DIALIGN} et \emph{DiaWay 1.0} (DW.1.0). Voir la figure \ref{fig:comparison_multiple_alignement}.

\end{itemize}

\medskip
Dans le prochain test, nous allons prendre 10 séquences biologiques, et nous changeons à chaque fois la longueur des séquences afin de voir l'influence de ce paramètre sur le temps d'exécution, voir la figure \ref{fig:comparison_multiple_alignement}.

\begin{figure}[h]
\centering
\includegraphics[width=0.7\linewidth]{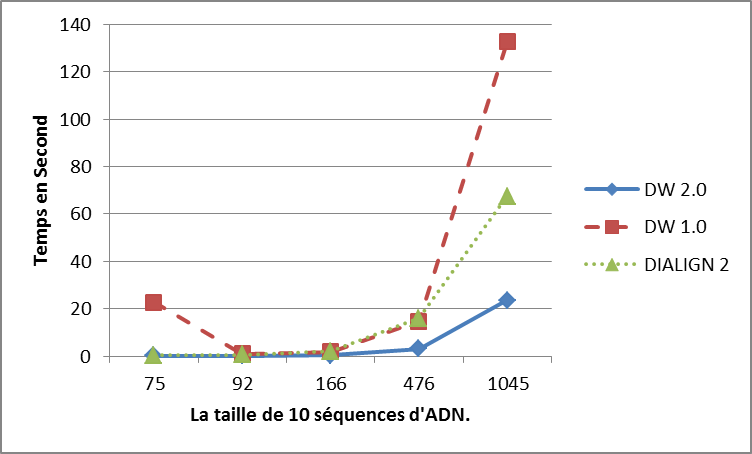}
\caption{Le temps d'exécution pour aligner 10 séquences de longueurs différentes.}
\label{fig:comparison_multiple_alignement}
\end{figure}

Le temps d'exécution de notre version \emph{DiaWay 2.0} (DW.2.0) est meilleur que celui des deux autres méthodes. Notre version \emph{DiaWay 1.0} (DW.1.0) donne un temps d'exécution très mauvais pour une grande taille de séquences. Parce que, comme nous l'avons expliqué avant, elle extrait pour chaque diagonale toutes ses sous-diagonales, et aussi on ne fait pas le tri des diagonales dans l'étape finale de l'alignement, ce qui implique un temps d'exécution très important.\\

Pour montrer l'efficacité de notre méthode \emph{DiaWay 2.0} (DW.2.0) par rapport à celle de \emph{DIALIGN 2.2}, nous avons fait un test pour l'alignement par paire, nous prenons seulement deux séquences biologiques avec de très grande tailles. De plus, nous changeons à chaque fois la longueur. Voir le tableau \ref{tab:T_2_Sequance_lenght}, et la figure \ref{fig:comparaison_dialign_Newdiaway}.

\begin{table}[h]
\centering
\begin{tabular}{|c|c|c|c|c|}
\hline $|S|$ & $\bar{S}$ &	 DW.2.0 & DW.1.0 &	DIALIGN 2.2 \\
\hline 2 &	2000 &	203 &	57640 &	6796  \\
\hline 2 &	3000 &	437 & 	172359 &	14843  \\
\hline 2 &	4032 &	812 &	392703 &	25718  \\
\hline 2 &	5000 &	1235 &	710984 &	38828  \\
\hline 2 & 	7499 &	2828 &	2223297 &	81281  \\
\hline
\end{tabular} 
\caption{Le temps d'exécution en millisecondes selon la longueur de 2 séquences seulement.}
\label{tab:T_2_Sequance_lenght}
\end{table}

Dans le tableau \ref{tab:T_2_Sequance_lenght}, ont peut voir que le temps d'exécution de notre méthode DW.2.0 est beaucoup plus petit que \emph{DIALIGN 2.2}.

\begin{figure}[h]
\centering
\includegraphics[width=0.7\linewidth]{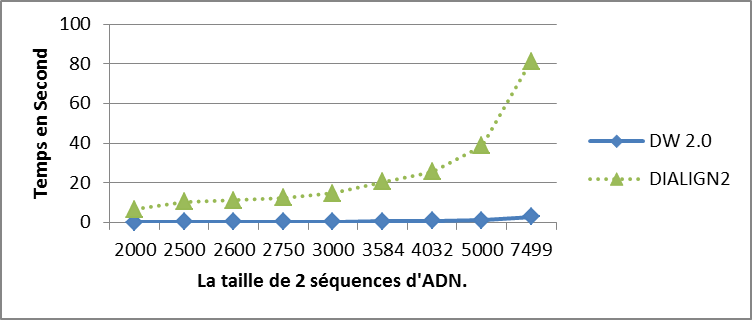}
\caption{Comparaison de l'alignement par paire, avec des séquences de très grandes tailles.}
\label{fig:comparaison_dialign_Newdiaway}
\end{figure}

Dans ce graphe (la figure \ref{fig:comparaison_dialign_Newdiaway}) nous avons omis la méthode \emph{DiaWay 1.0} (DW.1.0) car elle donne un temps d'exécution très mauvais.
Les résultats montrent que \emph{DiaWay 2.0} (DW 2.0) est très efficace en comparant avec l'algorithme \emph{DIALIGN 2.2}. Le temps d'exécution de notre méthode est presque stable.

\section{Conclusion}
\label{sec:ali:Conclusion}
Dans ce chapitre, nous avons présenté un nouvel algorithme d'alignement des séquences  biologiques par paire et multiple. Un algorithme robuste et efficace en terme de temps d'exécution. Notre algorithme est basé sur l'algorithme \emph{DIALIGN} qui est très connue pour l'alignement des séquences d'ADN et de protéine.

Nous avons présenté les différentes étapes de notre approche,  de l'étape d'extraction des diagonales à l'étape finale qui consiste à regrouper les fragments des diagonales et l'insertion des gaps, en passant par l'étape de l'indexation des diagonales et l'étape de la vérification de l'inconsistance des diagonales, où nous avons proposé une nouvelle méthode qui utilise les graphes afin de résoudre ce problème de l'inconsistance.

Dans l'étape de tri des diagonales pour l'alignement final, nous avons proposé une nouvelle technique de tri, car le besoin est différent et on ne peut pas appliquer les algorithmes de tri classique.

Nous avons réalisé deux implémentations de notre approche, \emph{DiaWay 1.0} et \emph{DiaWay 2.0}. Nous avons fait différents tests sur des séquences par paires et multiple. Les résultats montrent que notre approche \emph{DiaWay 2.0} est très efficace dans les deux cas d'alignement d'ADN par paire et multiple, et donne un très bon temps d'exécution comparé à \emph{DIALIGN 2.2}.

\def\baselinestretch{1}
\chapter{Conclusion générale}
\label{chap:Conclusion_generale}

\def\baselinestretch{1.66}

L'objet de cette thèse porte sur l'étude du problème de la recherche approchée en général, et son étude dans différents contextes importants de l'informatique.

\begin{enumerate}[a]
\item La recherche approchée dans un dictionnaire ou un texte.
\item L'auto-complétion approchée.
\item L'alignement de séquences biologiques.
\end{enumerate}

Chaque contexte à de nombreuses applications. Nous avons apporté des solutions qui donnent aussi bien en théorie qu'en pratique des résultats satisfaisants en complexité.\\

\paragraph{La recherche approchée dans un dictionnaire/texte.} Nous avons apporté deux solutions.

La première solution utilise les tables de hachage, nous avons présenté une solution pour résoudre le problème de la recherche approchée pour $k \geq 2$. L'idée de notre algorithme est basée sur l'utilisation des tableaux de hachage avec sondage linéaire et des signatures de hachage afin d'accélérer le temps de recherche.
Nous avons proposé deux variantes de notre solution, une version non-compacte qui est incrémentale (possibilité d'insertion des nouveaux éléments), et une version compacte qui optimise l'espace mémoire.

\medskip

La deuxième solution résout le problème de la recherche approchée dans un dictionnaire/texte pour $k=1$ erreur. La solution utilise le Trie et le Trie inversé afin d'obtenir les caractères possibles pour une position d'erreur donnée et de vérifier les parties exactes du mot requête rapidement. Nous avons présenté trois variantes de notre solution, la première \emph{TRT\_CI} effectue une intersection sur des caractères communs des transitions sortantes entre les deux n\oe{}uds d'erreur du Trie et le Trie inversé pour déterminer les chemins qui peuvent conduire à des solutions. Dans la seconde méthode \emph{TRT\_WNI}, nous faisons une intersection entre deux ensembles de numéros de mots pour déterminer les solutions. Dans la troisième variante \emph{TRT\_CWNI}, nous combinons les deux méthodes \emph{TRT\_CI} et \emph{TRT\_WNI}, pour réduire le nombre de feuilles avant de faire une intersection entre leurs numéros.

% % ajouter les perfermance théorique et pratique 

Dans la première partie, dans les deux chapitres \ref{chap:recherche_approchee_avec_hachage} et \ref{chap:Trie_R_Tire} où on a proposé deux solutions pour le problème de la recherche approchée (une solution qui se base sur le hachage et l'autre qui se base sur le Trie et Trie inversé), les mots du dictionnaire/texte recherchés sont de longueur moyenne, en particulier lorsque en considère les langues naturelles, et le nombre d'erreurs $k$ est donnée d'hypothèse. La complexité dépend de la taille des mots, la taille de dictionnaire, la taille de l'alphabet et $k$.\\

\paragraph{L'auto-complétion approchée}

Dans ce travail, nous avons présenté une solution pratique pour résoudre le problème de l'auto-complétion approchée dans une architecture client-serveur, et nous avons présenté la méthode Top-k pour améliorer la qualité des résultats en utilisant un système de classement dynamique et statique.

\medskip

Dans le problème de l'auto-complétion, la requête recherchée par un algorithme de recherche approchée est un préfixe d'un mot, donc sa taille est relativement plus petite que celle des mots du cas précédent, ou de même ordre de grandeur. Après l'étape de la recherche approchée de préfixe taper, on a besoin d'une 2\ieme{} étape pour obtenir tous les suffixes de ce préfixe. Pour cela, on est obligé de proposer des solutions spécifiques à ce problème (l'auto-complétion approchée) indépendamment de problème général de la recherche approchée. C'est la raison pour laquelle nous avons étudié ce problème et proposé quatre variantes pour la recherche approchée du préfixe.
La première recherche dans un Trie le préfixe avec retour arrière (la méthode naïve). La seconde recherche le préfixe par l'algorithme 1 (avec hachage) du problème général. Et la troisième solution combine une structure de Trie avec hachage pour utiliser une liste de substitutions qui permet de choisir les chemins des solutions. La quatrième variante, prend en compte la profondeur du Trie dans lequel le préfixe est recherché. Dans les niveaux proches de la racine, la recherche se fait par la variante 3, et dans les niveaux proches des feuilles, elle se fait par la variante 1. Les expérimentations ont montré que cette dernière variante donne les meilleures performances de recherche.\\

L'auto-complétion nécessite une étape qui trouve tous les suffixes d'un préfixe. Pour cela nous avons proposé trois variantes d'algorithmes pour s'adapter au comportement de l'utilisateur. 1) Commencer la recherche à partir de la racine. 2) Commencer la recherche depuis le dernier n\oe{}ud de la requête précédente lorsque le mot requête précédent est un préfixe complet du mot actuel. 3) Commencer la recherche depuis le plus grand préfixe commun.

Pour réduire le nombre des résultats obtenus, nous appliquons une fonction qui retourne les résultats les plus pertinents (Top-k) vis-vis d'un score (statique ou dynamique) associé.\\

Nous avons réalisé une bibliothèque (nommée Appaco\_lib), pour être utilisée soit du côté serveur ou du côté client ou les deux en même temps pour répondre rapidement à des requêtes d'auto-complétion approchée pour des dictionnaires UTF-8 (donc toutes les langues).

\paragraph{Alignement multiple des séquences d'ADN :}

Nous avons aussi examiné la recherche approchée dans le contexte de la bio-informatique. 

Aligner deux séquences $x$ et $y$ est de mettre les régions similaires de $x$ en face de $y$. Trouver le bon alignement revient à minimiser le nombre d'opérations qui permet de transformer la séquence $x$ en la séquence $y$, ou soit à maximiser la similarité entre $x$ et $y$.

Ce problème ainsi exprimé, est donc un problème de recherche approchée où les mots (les séquences biologiques) sont de taille relativement grande, le seuil d'erreur $k$ est inconnu, et la taille de l'alphabet est réduite.

Dans ce travail, nous avons présenté un nouvel algorithme d'alignement des séquences  biologiques par paire et multiple robuste et efficace en terme de temps d'exécution, notre algorithme est basé sur l'algorithme \emph{DIALIGN} qui est très connu pour l'alignement des séquences d'ADN et de protéine.

Nous avons réalisé deux implémentations de notre approche, \emph{DiaWay 1.0} et \emph{DiaWay 2.0}. Nous avons fait différents tests sur des séquences multiples et par paires. Les résultats montrent que notre approche (\emph{DiaWay 2.0}) est très efficace dans les deux cas, l'alignement d'ADN multiple et par paire, et donne un très bon temps d'exécution comparé à \emph{DIALIGN 2.2}.\\

Toutes nos implémentations sont distribuées sous la licence publique générale limitée GNU (GNU LGPL) (Anglais: GNU Lesser General Public License). Les résultats de notre travail sont disponibles sous forme de bibliothèques sur :

\begin{itemize}

\item La recherche approchée avec les tableaux de hachage : \url{https://code.google.com/p/compact-approximate-string-dictionary/}

\item La recherche approchée avec le Trie et le Trie inversé : \url{https://github.com/chegrane/TrieRTrie}

\item L'auto-complétion approchée avec top-k : \url{https://github.com/AppacoLib/api.appacoLib}

\item L'alignement multiple des séquences biologiques : \url{https://github.com/chegrane/diaway} pour la première version, et sur \url{https://github.com/chegrane/DiaWay_2.0} pour la deuxième version.

\end{itemize}

%% ajouter les lien des publications.

\paragraph{}

Au terme de ce travail et avant d'envisager quelques perspectives rappelons quelques problèmes qui se sont posés dans toute cette problématique de la recherche approchée dans un mot autrement dit la présence de(s) erreur(s) dans les mots. Nous en mentionnons trois:

\begin{enumerate}
\item La position de l'erreur dans le mot requête : 
	tout simplement lorsqu'on a une requête $q$ d'une longueur $m$, la difficulté provient du fait que l'on ne sait pas à quelle position l'erreur se trouve, elle peut être à n'importe quelle position de l'intervalle $[1..m]$.

\medskip
	
	Lorsque le nombre d'erreurs est $\geq 2$, le problème devient plus complexe étant donné la combinaison des positions des erreurs que l'on doit considérer:  une combinaison de $\binom{m}{k}$ positions.

Inversement, si on connait par avance la position de l'erreur (ou les positions possibles), la solution devient facile et rapide à obtenir.
	Dans le cas où $k \geq 2$, si seulement on connaît juste une partie des positions de l'erreur par exemple celle d'une partie de gauche ou de droite du mot requête, la complexité des combinaisons des positions peut être réduite.

On peut alors envisager quelques questions à explorer qui pourraient améliorer ce travail, par exemple:

\medskip
\item Dans une position donnée, quels sont les caractères qui représentent une solution où au moins qui mènent à des solutions possibles. Ce problème peut aussi dépendre de la taille de l'alphabet, petite ou grande.

\medskip
\item Quels types d'erreurs  considérer? (en fonction des différentes mesures de distance entre deux mots). Dans le cas de $k \geq 2$ erreurs, on a une combinaison à faire entre les types d'erreurs et les positions où elles pourraient se trouver.
Quel est l'impact du choix de la fonction de distance dans ce problème?

\medskip

Si on sait d'avance quels sont les types d'erreurs concernées dans le travail (par exemple dans une recherche approchée où les erreurs sont celles d'un système d'OCR, il est possible de caractériser, expérimentalement, les erreurs possible d'un OCR), il serait alors possible de diminuer le nombre de combinaisons et donc un temps de calcul plus réduit. Cela permet aussi de proposer des algorithmes moins complexes adaptés au besoin.
\end{enumerate}

\paragraph{}

Il y a d'autres paramètres que l'on pourrait aussi prendre en considération afin de résoudre ce problème comme l'espace mémoire occupé par la structure de donnée et le temps du pré-traitement. Une réflexion dans cette direction peut apporter de nouvelles solutions meilleures.

\backmatter % book mode only

\bibliographystyle{splncs}
\renewcommand{\bibname}{Références} % changes default name Bibliography to References
\bibliography{References/references} % References file

\end{document}